\documentclass[12pt,a4paper]{book}
\usepackage[latin1]{inputenc}   % Pour pouvoir utiliser directement les accents
\usepackage{amsmath}   % Pour les maths
\usepackage{amsfonts}  % Symboles
\usepackage{amssymb}  % Symboles
\usepackage{latexsym}  % Symboles
\usepackage{graphicx} % Pour inclure des figures eps
\usepackage{epstopdf}
\usepackage[a4paper]{geometry}
\usepackage{color} 
\usepackage{subfigure}
\usepackage{url}
\usepackage{minitoc}
\usepackage{hyperref} %% pour PDF
\usepackage{meta-donnees}

\newtheorem{theorem}{Theorem}
\newtheorem{proposition}{Proposition}
\newtheorem{conjecture}{Conjecture}
\newtheorem{lemma}{Lemma}

\newtheorem{claim}{Claim}
 \newtheorem{definition}{Definition}

\newcommand{\mc}[1]{\mathcal{#1}}
\newcommand{\mb}[1]{\mathbb{#1}}
\newcommand{\pdc}[1]{\hat{#1}}
\newcommand{\Out}{\normalfont{Out}}

\newcommand{\ccw}{counterclockwise }
\newcommand{\ccww}{counterclockwise}
\newcommand{\cw}{clockwise }
\newcommand{\cww}{clockwise}
\newcommand{\pss}{PS }
\newcommand{\ps}{Poulalhon and Schaeffer }
\newcommand{\npss}{Poulalhon and Schaeffer's }
\newcommand{\aps}{{\sc Algorithm PS} }
\newcommand{\apss}{{\sc Algorithm PS}}

\newcounter{sclaim}
\newcounter{ssclaim}

\newenvironment{proof}{\noindent \emph{Proof.}\ }{\hfill
    $\Box$\vspace{1em}}

  \newenvironment{proofclaim}{\noindent \emph{Proof.}\ }{\hfill
    $\Diamond$\vspace{1em}}

\newenvironment{sclaim}[1][]%
{\refstepcounter{sclaim}\vspace{1ex}\noindent{\it  (\arabic{sclaim})  {#1}{}}\it}{\vspace{1ex}}

\newenvironment{proofsclaim}[1][]%
	{\noindent {}{#1}{}}{ This proves (\arabic{sclaim}).\vspace{1ex}}

  \newcommand{\md}{\mathcal D} \newcommand{\mv}{\mathcal V}
  \newcommand{\ms}{\mathcal S_\mv} \newcommand{\mr}{\mathbb R}

\newcommand{\mz}{\mathbb Z}

\begin{document}
\frontmatter
\thispagestyle{empty}
% la ligne ci-dessous est à insérer obligatoirement dans le préambule du document avant \begin{document}

% les lignes en bas sont à insérer obligatoirement après \begin{document}

%%%%%%%%%%%%%%%%%%%%%%%%%%%%%%%%%%%%%%%%%%%%%%%%%%%%%%
%%             Commandes Meta-données               %%
%%   à renseigner par les auteurs pour générer      %%
%%     la couverture modèle Univ. Grenoble          %%
%%%%%%%%%%%%%%%%%%%%%%%%%%%%%%%%%%%%%%%%%%%%%%%%%%%%%%
%%      Fichier encodé au format ISO-8859-16        %%

%\Sethpageshift{???mm}   %%optionnel : à décommenter si besoin pour ajout d'espace afin de center la couvérture horizontalement (valeur par défaut est -5.5mm)
%\Setvpageshift{???mm}   %%optionnel : à décommenter si besoin pour ajout d'espace afin de center la couvérture verticalement (valeur par défaut est -15.5mm)

%\Universite{}    %%optionnel : à décommenter et à renseigenr si vous voulez changer le non d'université
%\Grade{}         %%optionnel : à décommenter et à renseigenr si vous voulez changer le grade
\Specialite{X}
\Arrete{X}
\Auteur{Benjamin Lévêque}
%\Directeur{X}
%\CoDirecteur{}    %%optionnel : à décommenter et à renseigenr si présence d'un Co-directeur de thèse
%\Laboratoire{X}
%\EcoleDoctorale{X}         
\Titre{Generalization of Schnyder woods to orientable surfaces and
  applications}
%\Soustitre{}      %%optionnel : à décommenter et à renseigenr si
%présence  d'un sous-titre de thèse
\Depot{19 octobre 2016}

% Commande pour création de nouvelles catégories dans le jury:
 
%\UGTNewJuryCategory{...NomDeLaCategorie...}{...Definition...}

% Exemple \UGTNewJuryCategory{UGTFamille}{Membre de la famille} que nous ajoutons dans la commande \Jury ci-dessous sous la forme \UGTFamille{Jean Rousseau}{(...titre_et_affiliation...s'il_y_en_a...)}

 \Jury{
% \UGTPresident{...Civilité, Prénom_et_Nom...}{...titre_et_affiliation...}
% \UGTPresidente{...Civilité, Prénom_et_Nom...}{...titre_et_affiliation...}

 \UGTExaminateur{Vincent Beffara}{Directeur de recherche CNRS}% Institut Fourier, Grenoble}
 \UGTExaminatrice{Nadia Brauner}{Professeur, Université Grenoble
   Alpes}%, Laboratoire G-SCOP, Grenoble}
 \UGTExaminateur{Victor Chepoi}{Professeur, Université Aix-Marseille}%,   Laboratoire LIF, Marseille}
 \UGTExaminateur{\'Eric Colin de Verdière}{Directeur de recherche  CNRS}%, Ecole Normale Supérieure, Paris}
\UGTRapporteur{Stefan Felsner}{Professeur, TU Berlin}
\UGTRapporteur{Marc Noy}{Professeur, UPC Barcelone}
 \UGTRapporteur{Gilles Schaeffer}{Directeur de recherche CNRS}%,   Laboratoire LIX, Palaiseau}
 \UGTExaminateur{Andr\'as Seb\H{o}}{Directeur de recherche CNRS}%, Laboratoire G-SCOP}

% \UGTDirecteur{...Civilité, Prénom_et_Nom...}{...titre_et_affiliation...}       %% Directeur de thèse
% \UGTCoDirecteur{...Civilité, Prénom_et_Nom...}{...titre_et_affiliation...}     %% Co-Directeur de thèse s'il y en a

% \UGTInvite{...Civilité, Prénom_et_Nom...}{...titre_et_affiliation...}
% \UGTInvitee{...Civilité, Prénom_et_Nom...}{...titre_et_affiliation...}
 }

\MakeUGthesePDG    %% très important pour générer la couvérture de thèse

 \newpage
 \thispagestyle{empty}
 \
 \newpage

 \thispagestyle{empty}

%\title{Generalization of Schnyder woods to orientable surfaces and
%  applications}

%\maketitle

 \noindent {\bf Abstract :} 
 Schnyder woods are particularly elegant combinatorial structures
 with numerous applications concerning planar triangulations and more
 generally 3-connected planar maps.  We propose a simple
 generalization of Schnyder woods from the plane to maps on orientable
 surfaces of any genus with a special emphasis on the toroidal case.
 We provide a natural partition of the set of Schnyder woods of a
 given map into distributive lattices depending on the surface
 homology.  In the toroidal case we show the existence of particular
 Schnyder woods with some global properties that are useful for
 optimal encoding or graph drawing purpose.

\


\tableofcontents

\mainmatter

\chapter*{Introduction}
\addcontentsline{toc}{part}{Introduction}
\markboth{}{}

Schnyder woods (see Part~\ref{part:1}) are today one of the main tools in the area of planar
graph representations. Among their most prominent applications are the
following: They provide a machinery to construct space-efficient
straight-line drawings~\cite{Sch90, Kan96,Fel01}, yield a
characterization of planar graphs via the dimension of their
vertex-edge incidence poset~\cite{Sch89,Fel01}, and are used to encode
triangulations~\cite{PS06,Bar12}. Further applications lie in
enumeration~\cite{Bon05}, representation by geometric
objects~\cite{FOR94, GLP11}, graph spanners~\cite{BGHI10}, etc.

We propose a simple generalization of Schnyder woods from the plane to
maps on orientable surfaces of any genus (see Part~\ref{part:2}). This
is done in the language of angle labellings.  Generalizing results of
De Fraysseix and Ossona de Mendez~\cite{FO01}, and
Felsner~\cite{Fel04}, we establish a correspondence between these
labellings and orientations and characterize the set of orientations of
a map that corresponds to such a Schnyder wood. Furthermore, we study
the set of orientations of a given map and provide a natural partition
into distributive lattices depending on the surface homology. This
generalizes earlier results of Felsner~\cite{Fel04} and Ossona de
Mendez~\cite{Oss94}.  Whereas many questions remain open for the
double torus and higher genus (like already the question of existence
of the studied objects), in the toroidal case we are able to push our
study quite far.

The torus can serve as a model for planar periodic surfaces and is for
instance often used in statistical physics context since it enables to
avoid dealing with particular boundary conditions.  It is in a sense
the most homogeneous oriented surface since Euler's formula sums
exactly to zero. For our purpose this means that one can asks for
orientations satisfying the same local condition everywhere.

We study structural properties of toroidal Schnyder woods extensively
(see Part~\ref{part:torus}). We analyze the behavior of the
monochromatic cycles of a toroidal Schnyder wood to define the notion
of ``crossing'' that is useful for graph drawing purpose. We also
exhibit a kind of ``balanced'' property that enables to define a
canonical lattice and thus a unique minimal element, used in
bijections.

We are able to provide several proofs of existence of Schnyder woods
in the toroidal case (see Part~\ref{chap:existence}) and this is
particularly interesting since this problem is open in higher genus.
Some consequences are a linear time algorithm to compute either a
crossing Schnyder wood or a minimal balanced Schnyder wood for
toroidal triangulations, and the fact that a toroidal map admits a
Schnyder wood if and only if it is ``essentially 3-connected''.

Concerning the applications (see Part~\ref{part:5}), we generalize a
method devised by Poulalhon and Schaeffer~\cite{PS06} to linearly
encode a planar triangulation optimally.  In the plane, this leads to
a bijection between planar triangulations and some particular
trees. For the torus we obtain a similar bijection but with particular
unicellular maps (maps with only one face).  We also show that
toroidal Schnyder woods can be used to embed the universal cover of an
essentially 3-connected toroidal map on an infinite and periodic
orthogonal surface. We use this embedding to obtain a straight-line
flat torus representation of any toroidal map in a polynomial size
grid.

Most of the results presented in this manuscript appear in the papers
\cite{DGL15,GL13,GKL15}. The goal of this manuscript is to merge these
papers and to present them in a unified way. After the first
paper~\cite{GL13} our view on the objects have evolved, the
definitions have been generalized, enabling to reveal more structural
properties. Thus we feel that there was a need to restructure this
research in a unique document that also contains additional details,
results, corrections and simplifications.

As briefly explained in the conclusion, many of the structural
properties that we exhibit here for Schnyder woods can be generalized
for other kind of orientations/maps, like transversal structure for
4-connected triangulations, 2-orientations for quadrangulations or
more generally $\frac{d}{d-2}$-orientations for d-angulations, etc.
There are also probable other applications of Schnyder woods in the
plane that can be extended to higher genus.  For example we now have
all the ingredients to obtain a random generator for toroidal
triangulations. For these reasons we believe that this manuscript is
just the beginning of  more research to come and we hope that it
can serve as a starting point for one interested in studying
orientations of maps in higher genus and their applications.

I was happy to work on this topic for the last few years and I would
like to thank my colleagues, Nicolas Bonichon, Luca Castelli Aleardi
and Eric Fusy, whose fruitful discussions have allowed this work to be
achieved and also my co-authors, Vincent Despr\'e, Daniel
Gon\c{c}alves and Kolja Knauer, who have been following me on this
journey. We were mostly guided by the beauty of the structural
properties of the studied objects, and the nice applications in the
end were then just consequences of this quest.

\part{Planar case}
\label{part:1}
\chapter{From planar triangulations to Schnyder woods}
\label{chap:introSW}

Before giving the formal definition of Schnyder woods in the plane
(see Section~\ref{sec:SWplane}), we explain how they appear quite
naturally while studying planar triangulations.

Given a general graph $G$, let $n$ be the number of vertices and $m$
the number of edges. If the graph is embedded on the plane (or a surface), let $f$ be
the number of faces (including the outer face).  Euler's formula says
that a connected graph embedded on the plane satisfies $n-m+f=2$.

A general graph (i.e. not embedded on a surface) is \emph{simple} if
it contains no loop and no multiple edges.  We start this manuscript
with the study of \emph{planar triangulations}, that are simple graphs
embedded in the plane such that every face, including the outer face,
has size three (see example of Figure~\ref{fig:noorientation}).

\begin{figure}[!h]
\center
\includegraphics[scale=0.6]{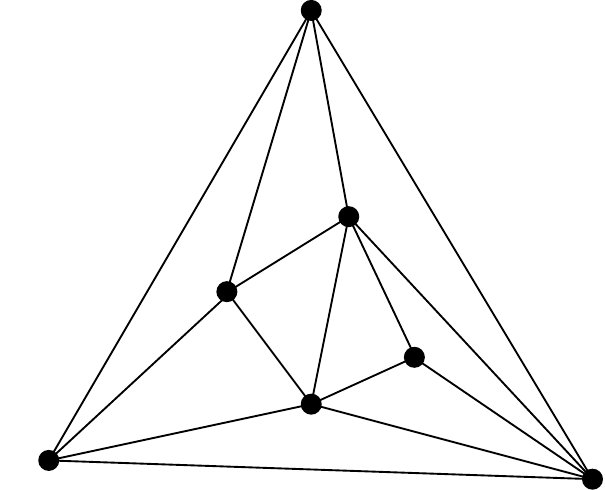}
\caption{Example of a planar triangulation.}
\label{fig:noorientation}
\end{figure}

Consider a planar triangulation $G$. Since every face has size $3$, we
have $3f=2m$. Then by Euler's formula we obtain $m=3n-6$. Thus there
is ``almost'' $3$ times more edges that vertices in a planar
triangulation. In fact, the relation can be re-written $(m-3)=3(n-3)$
so there is exactly three times more internal edges than internal
vertices. Thus there is hope to be able to assign to each internal
vertex of $G$, three incident edges such that each internal edge is
assigned exactly once. By orienting the edges from the vertices to
which they are assigned, one obtain an orientation of the graph $G$
where every inner vertex has outdegree exactly $3$ (see
Figure~\ref{fig:3orientation}). Such an orientation is called a
\emph{$3$-orientation}.

\begin{figure}[!h]
\center
\includegraphics[scale=0.6]{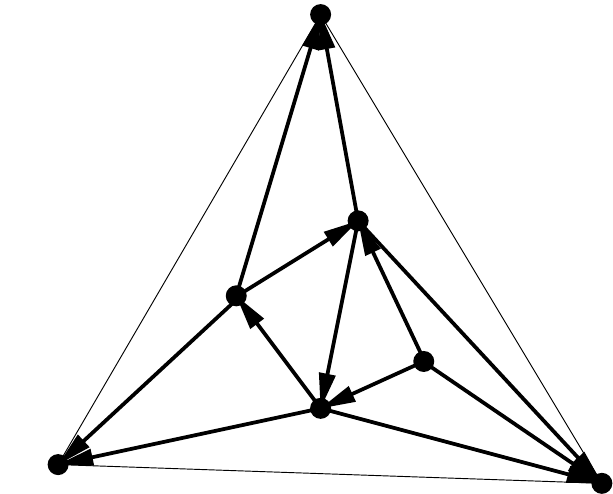}
\caption{Every inner vertex has outdegree exactly $3$.}
\label{fig:3orientation}
\end{figure}

Consider a $3$-orientation of $G$. First note that since
$(m-3)=3(n-3)$, all the inner edges incident to outer vertices are
entering the outer vertices.

Some natural objects to consider in a $3$-orientation are ``middle
walks''.  For an internal edge $e$ of $G$, we define the \emph{middle
  walk from $e$} as the sequence of edges $(e_i)_{i\geq 0}$ obtained
by the following method.  Let $e_0=e$. If the edge $e_i$ is entering
an internal vertex $v$, then the edge $e_{i+1}$ is chosen in the three
edges leaving $v$ as the edge in the ``middle'' coming from $e_i$
(i.e. $v$ should have exactly one edge leaving on the left of the path
consisting of the two edges $e_i,e_{i+1}$ and thus exactly one edge
leaving on the right). If the edge $e_i$ is entering an outer vertex,
then the walk ends there.

These middles walks have interesting structural properties. First, one
can show that a middle walk cannot intersect itself, otherwise it
would form a cycle whose interior region would contradicts Euler's
formula by a counting argument (triangulation and 3-orientation inside
+ middle walk on the border). Thus a middle walk is in fact a middle
path and has to end on one of the three outer vertices since it is not
infinite.

Let $x_0,x_1,x_2$ be the three vertices appearing on the outer face in
\ccw order. One can assign to each inner edge $e$ the color
$i\in\{0,1,2\}$ if the middle path starting from $e$ ends at vertex
$x_i$ (see Figure~\ref{fig:3orientation-}). Note that a subwalk of a
middle path is also a middle path. So if $e'$ is an edge of a middle
path then $e$ and $e'$ receive the same color. Thus all the edges of a
middle path have the same color.

\begin{figure}[!h]
\center
\scalebox{0.6}{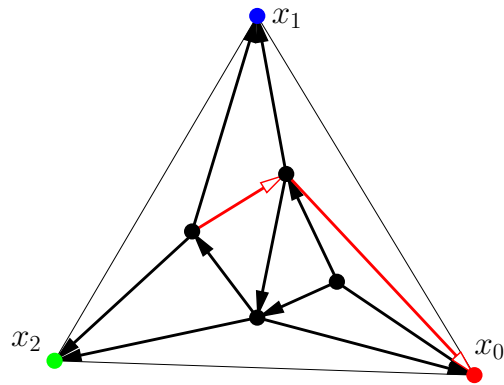}
\caption{All the edges of a middle path are ending at the same outer
  vertex by following middle edges.}
\label{fig:3orientation-}
\end{figure}

Moreover, consider two distinct outgoing edges $e,e'$ of a inner
vertex and the two middle paths $W$ and $W'$ starting from $e$ and
$e'$. With similar counting arguments as before, one can show that $W$
and $W'$ do not intersect. Thus each of the three middle paths
starting from the three outgoing edges of an inner vertex, ends at a
different outer vertex. So every vertex has exactly one edge leaving
in color $0$, $1$, $2$, respectively, and these edges appear in
counterclockwise order (see Figure~\ref{fig:planar-triangulation}).

\begin{figure}[!h]
\center
\scalebox{0.6}{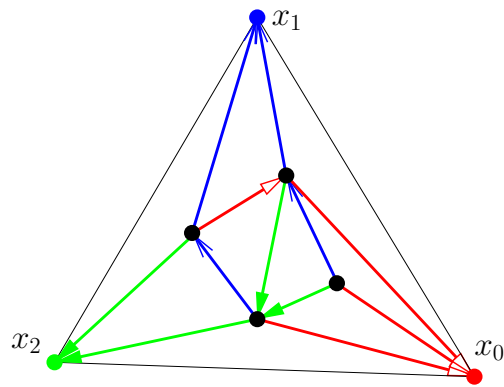}
\caption{Every inner vertex has exactly one edge leaving in each color.}
\label{fig:planar-triangulation}
\end{figure}

By the middle path property every edge entering a vertex in color $i$
has to enter in the sector between the outgoing edges of color $i+1$
and $i-1$ (throughout the manuscript colors are given modulo 3). So
the inner vertices are satisfying the local property of
Figure~\ref{fig:LSP} where the depicted correspondence between red, blue,
green, 0, 1, 2, and the arrow shapes is used through the entire
manuscript.

\begin{figure}[!h]
\center
\includegraphics[scale=0.5]{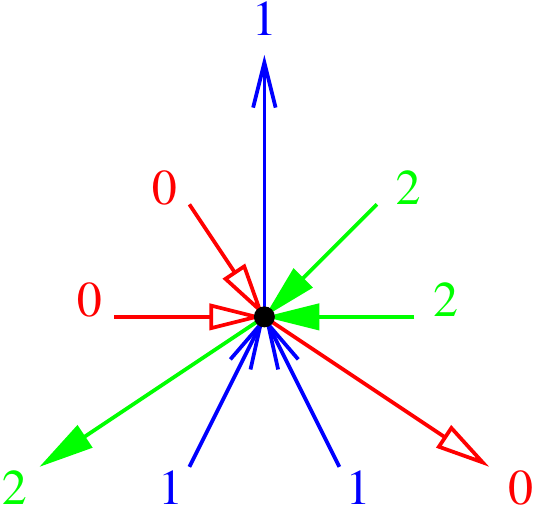}
\caption{Local property satisfied at every inner vertex.}
\label{fig:LSP}
\end{figure}

For each color $i$, every inner vertex is the starting point of a
middle path of color $i$. This path is ending at vertex $x_i$ by
definition. Thus the set of edges colored $i$ forms an oriented tree
rooted at $x_i$ (edges are oriented toward $x_i$) that is spanning
all the inner vertices of $G$ (see Figure~\ref{fig:ex-tree}).

\begin{figure}[!h]
\center
\includegraphics[scale=0.4]{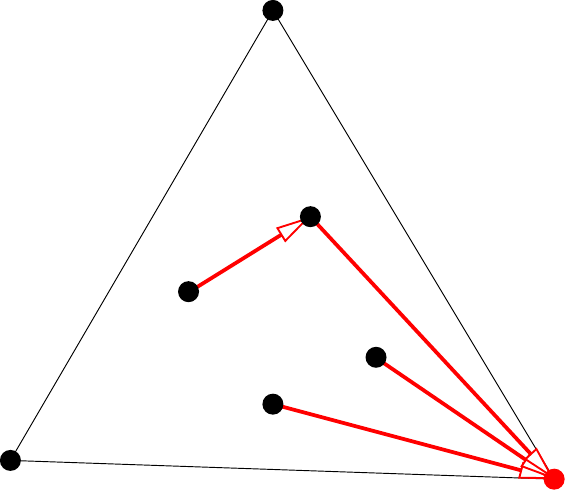} \ \  \  \ 
\includegraphics[scale=0.4]{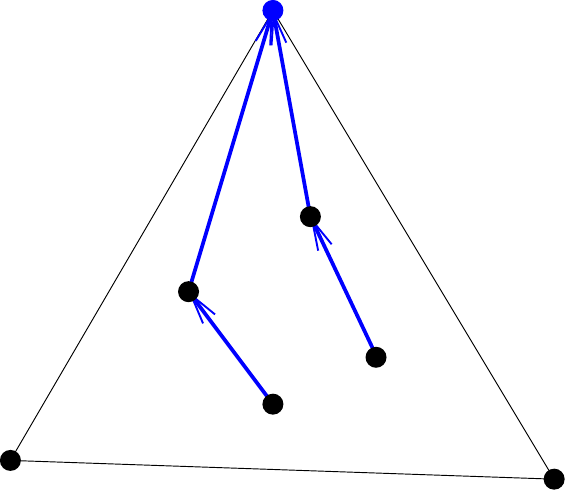} \ \   \ \ 
\includegraphics[scale=0.4]{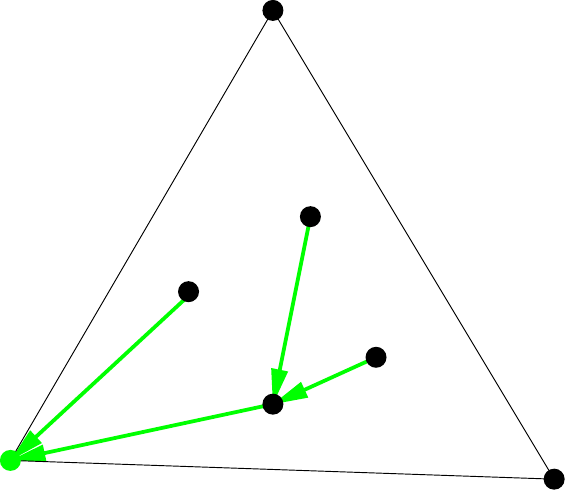}
\caption{Partition of the inner edges into three spanning trees.}
\label{fig:ex-tree}
\end{figure}

These orientations and colorings of the internal edges of a planar
triangulation where first introduced by Schnyder~\cite{Sch89} who
proved their existence for any planar triangulation. The partition
into three trees is the reason of their usual name: Schnyder woods.

\chapter{Generalization to 3-connected planar maps}
\label{sec:SWplane}

To define Schnyder woods formally we use the following local property
introduced by Schnyder~\cite{Sch89} (see Figure~\ref{fig:LSP}):

\begin{definition}[Schnyder property]
\label{def:schnyderproperty}
Given a map $G$, a vertex $v$ and an orientation and coloring of the
edges incident to $v$ with the colors $0$, $1$, $2$, we say that $v$
satisfies the \emph{Schnyder property}, if $v$ satisfies the following
local property:

\begin{itemize}
\item Vertex $v$ has out-degree one in each color.
\item The edges $e_0(v)$, $e_1(v)$, $e_2(v)$ leaving $v$ in colors
  $0$, $1$, $2$, respectively, occur in counterclockwise order.
\item Each edge entering $v$ in color $i$ enters $v$ in the
  counterclockwise sector from $e_{i+1}(v)$ to
  $e_{i-1}(v)$.
\end{itemize}
\end{definition}

Then the formal definition of Schnyder woods is the following:

\begin{definition}[Schnyder wood]
\label{def:schnyder}
Given a planar triangulation $G$, a \emph{Schnyder wood} is an
orientation and coloring of the inner edges of $G$ with the colors
$0$, $1$, $2$ (edges are oriented in one direction only), where each
inner vertex $v$ satisfies the \emph{Schnyder property}.
\end{definition}

See Figure~\ref{fig:planar-triangulation} for an example of a Schnyder
wood.

By a result of De Fraysseix and Ossona de Mendez~\cite{FO01}, there is
a bijection between orientations of the internal edges of a planar
triangulation where every inner vertex has outdegree $3$ and Schnyder
woods. Indeed, Chapter~\ref{chap:introSW} gives the ideas behind this
bijection and show how to recover the colors from the orientation.

Originally, Schnyder woods were defined only for planar
triangulations~\cite{Sch89}. Felsner~\cite{Fel01, Fel03} extended this
definition to planar maps. To do so he allowed edges to be oriented in
one direction or in two opposite directions and when an edge is
oriented in two direction, then each direction have a distinct color
and is outgoing (see Figure~\ref{fig:felsnerdoubleedge}).

\begin{figure}[!h]
\center
\begin{tabular}{cc}
\includegraphics[scale=0.5]{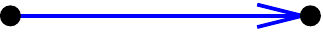} \ \ \  &  \ \ \
\includegraphics[scale=0.5]{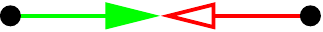} \\
\end{tabular}
\caption{The two types of edges in planar Schnyder woods}
\label{fig:felsnerdoubleedge}
\end{figure}

Then the formal definition of the generalization, called
\emph{planar Schnyder wood} in this manuscript, is the following:

\begin{definition}[Planar Schnyder wood]
\label{def:felsner}
Given a planar map $G$. Let $x_0$, $x_1$, $x_2$ be three vertices
occurring in counterclockwise order on the outer face of $G$. The
\emph{suspension} $G^\sigma$ is obtained by attaching a half-edge that
reaches into the outer face to each of these special vertices.  A
\emph{planar Schnyder wood} rooted at $x_0$, $x_1$, $x_2$ is an
orientation and coloring of the edges of $G^\sigma$ with the colors
$0$, $1$, $2$, where every edge is oriented in one direction or in
two opposite directions (each direction having a distinct color and
being outgoing), satisfying the following conditions:

\begin{itemize}
\item Every vertex satisfies the Schnyder property and the
  half-edge at $x_i$ is directed outward and colored $i$.
\item There is no  face whose boundary is a
monochromatic cycle.
\end{itemize}
\end{definition}

See Figure~\ref{fig:planar-felsner} for two examples of planar
Schnyder woods.

\begin{figure}[!h]
\center
\includegraphics[scale=0.5]{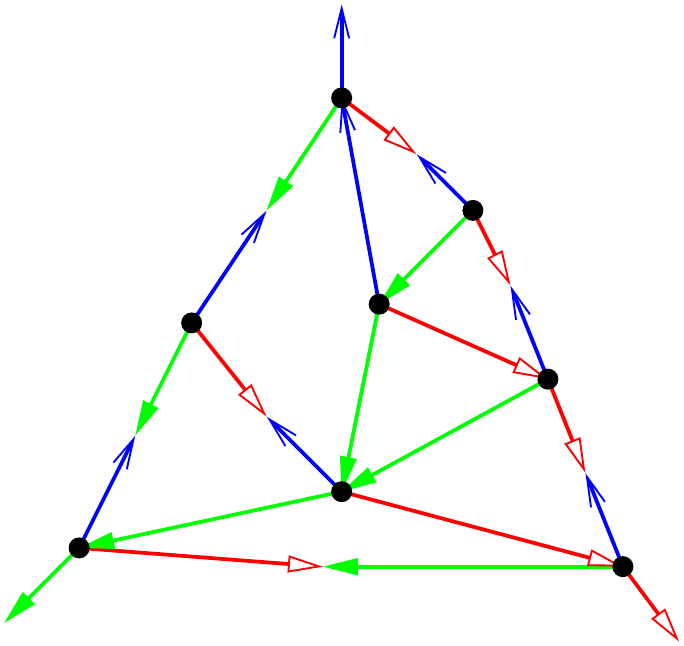}
\hspace{1cm}
\includegraphics[scale=0.5]{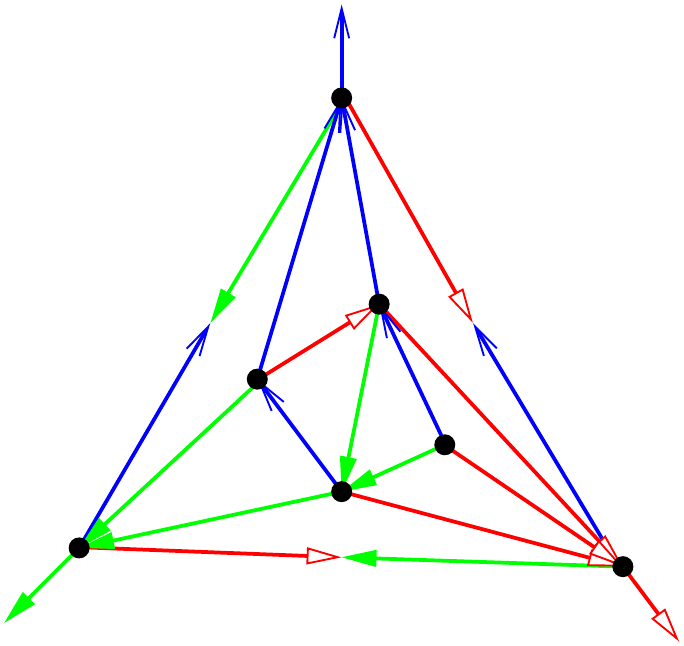}
\caption{A planar Schnyder wood of a planar map and of a planar
  triangulation.}
\label{fig:planar-felsner}
\end{figure}

When there is no ambiguity we may omit the word ``planar'' in
``planar Schnyder wood''.

Note that a planar triangulation has exactly $3n-6$ edges and this
explains why in Definition~\ref{def:schnyder} just inner vertices are
required to satisfy the Schnyder property. There are $3$ vertices in
the outer face that together should have $9$ outgoing edges in order
to satisfy the Schnyder property but there is just $3$ non-colored
edges on the outer face (see for instance
Figure~\ref{fig:planar-triangulation}). In
Definition~\ref{def:felsner} restricted to planar triangulations, the
$6$ missing outgoing edges are obtained by adding $3$ half-edges
reaching into the outer face and by orienting the $3$ outer edges in
two directions. On the right of Figure~\ref{fig:planar-felsner} the
triangulation of Figure~\ref{fig:planar-triangulation} is represented
with a planar Schnyder wood.

A planar map $G$ is \emph{internally 3-connected} if there exists
three vertices on the outer face such that the graph obtain from $G$
by adding a vertex adjacent to the three vertices is
3-connected. Miller~\cite{Mil02} proved the following (see
also~\cite{Fel01} for existence of Schnyder woods for 3-connected
planar maps and~\cite{BFM07} where the following result is stated in
this form):

\begin{theorem}[\cite{Mil02}]
\label{th:schnyder} 
A planar map admits a planar Schnyder wood if and only if it is
internally 3-connected.
\end{theorem}

Consider a planar Schnyder wood. Let $G_i$ be the directed graph
induced by the edges of color $i$ of $G$. This definition includes
edges that are half-colored $i$, and in this case, the edges get only
the direction corresponding to color $i$.  The graph $(G_i)^{-1}$ is
the graph obtained from $G_i$ by reversing all its edges. The graph $G_i\cup
G_{i-1}^{-1}\cup G_{i+1}^{-1}$ is obtained from the graph $G$ by
orienting edges in one or two direction depending on whether this
orientation is present in $G_i$, $G_{i-1}^{-1}$ or $G_{i+1}^{-1}$.
Felsner~\cite{Fel01} proved the following essential property of planar
Schnyder wood:

\begin{lemma}[\cite{Fel01}]
  \label{lem:nodirectedcycleplan}
  The graph $G_i\cup (G_{i-1})^{-1}\cup (G_{i+1})^{-1}$ contains no
  directed cycle.
\end{lemma}

Every vertex, is the starting point of an outgoing edge of color
$i$. Since $G_i$ is acyclic by Lemma~\ref{lem:nodirectedcycleplan}, we
have the following:

\begin{lemma}[\cite{Fel01}]
  \label{lem:schnydertree}
  For $i\in\{0,1,2\}$, the digraph $G_i$ is a tree rooted at $x_i$.
\end{lemma}

Thus in planar Schnyder woods, we still have, like for the
triangulation case, three spanning trees. But now some edges
appear in two trees, so it is no more a partition of the edges into
three spanning trees. Figure~\ref{fig:3c-tree} illustrate this
property on the planar Schnyder wood of the left of Figure~\ref{fig:planar-felsner}.

\begin{figure}[!h]
\center
\begin{tabular}{ccc}
\includegraphics[scale=0.4]{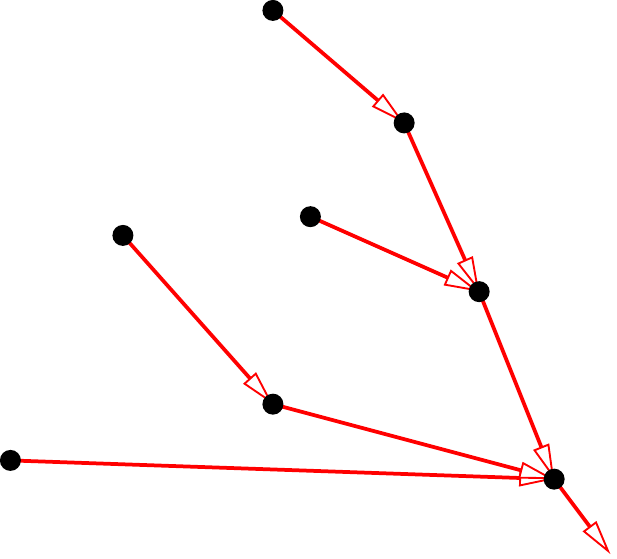} \ \ & \  \ 
\includegraphics[scale=0.4]{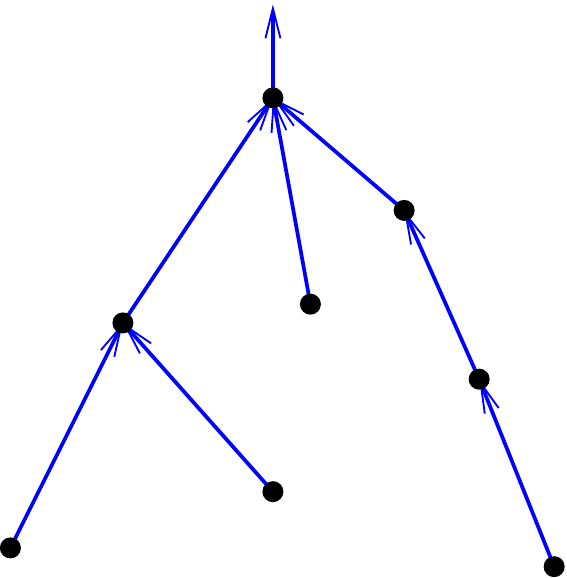} \ \ & \  \ 
\includegraphics[scale=0.4]{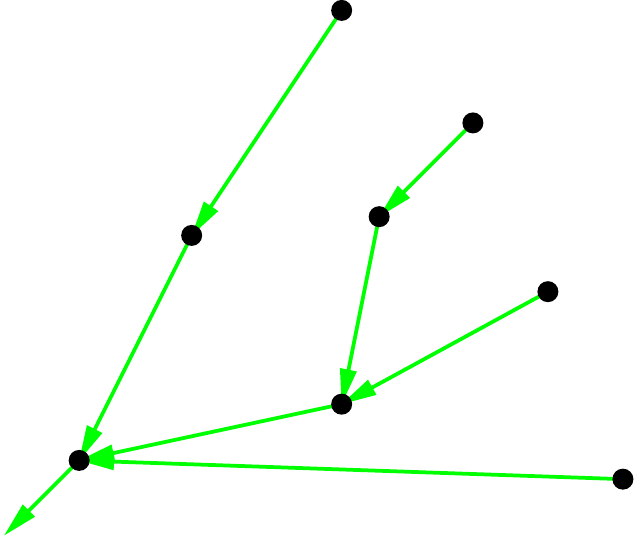}\\
 $G_0$ \ \ & \  \ $G_1$ \ \ & \  \ $G_2$ \\
\end{tabular}
\caption{Three spanning trees.}
\label{fig:3c-tree}
\end{figure}

For every vertex $v$ and color $i$, let $P_i(v)$ denote the directed
path from $v$ to $x_i$, composed of edges colored $i$. Then another
consequence of Lemma~\ref{lem:nodirectedcycleplan} is the following:

\begin{lemma}[\cite{Fel01}]
  \label{lem:nocommon-plan}
  For every vertex $v$ and $i,j\in\{0,1,2\}$, $i\neq j$, the two paths
  $P_i(v)$ and $P_j(v)$ have $v$ as only common vertex.
\end{lemma}

Thus, like for triangulations, for every vertex $v$, the three
monochromatic paths $P_0(v)$, $P_1(v)$, $P_2(v)$ do not intersect each
other except on $v$. These paths divide $G$ into regions (see
Figure~\ref{fig:regions-planar}) which form the basis of several graph
drawing methods using Schnyder woods, see for
instance~\cite{Sch90,Fel01}.

\begin{figure}[!h]
\center
\includegraphics[scale=0.5]{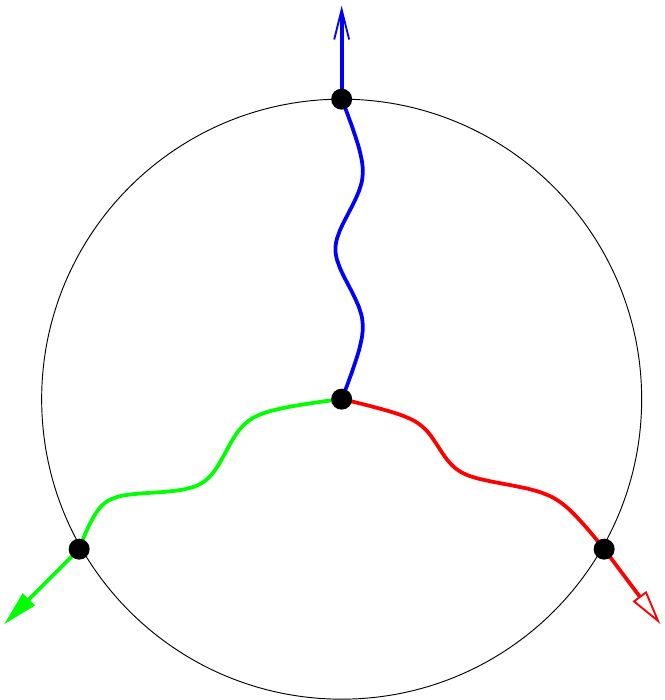}
\caption{Regions delimited by the three monochromatic paths from a vertex.}
\label{fig:regions-planar}
\end{figure}

Allowing edges to be oriented in one or two directions as in
Definition~\ref{def:felsner} also enables to define the dual of a
planar Schnyder wood. Indeed, by applying the correspondence of
Figure~\ref{fig:edgedualrule}, a planar Schnyder wood of $G^\sigma$
automatically defines a planar Schnyder wood of the dual map (with a
special rule for the outer face that gets three vertices plus three
half-edges).  Figure~\ref{fig:planar-felsner-duality} illustrate this
property on the planar Schnyder wood of the left of
Figure~\ref{fig:planar-felsner} where vertices of the primal are black
and vertices of the dual are white (this serves as a convention for
the rest of the manuscript).

\begin{figure}[!h]
\center
\includegraphics[scale=0.5]{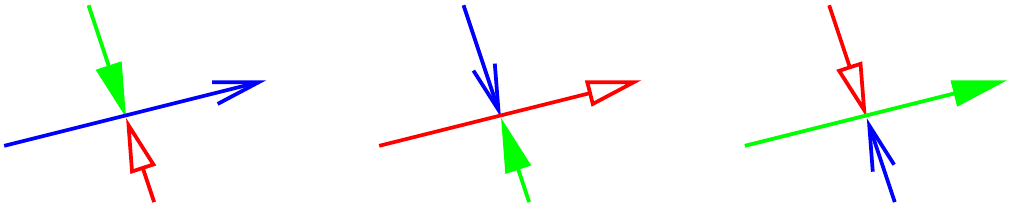}
\caption{Rules for the dual Schnyder wood.}
\label{fig:edgedualrule}
\end{figure}

\begin{figure}[!h]
\center
\includegraphics[scale=0.5]{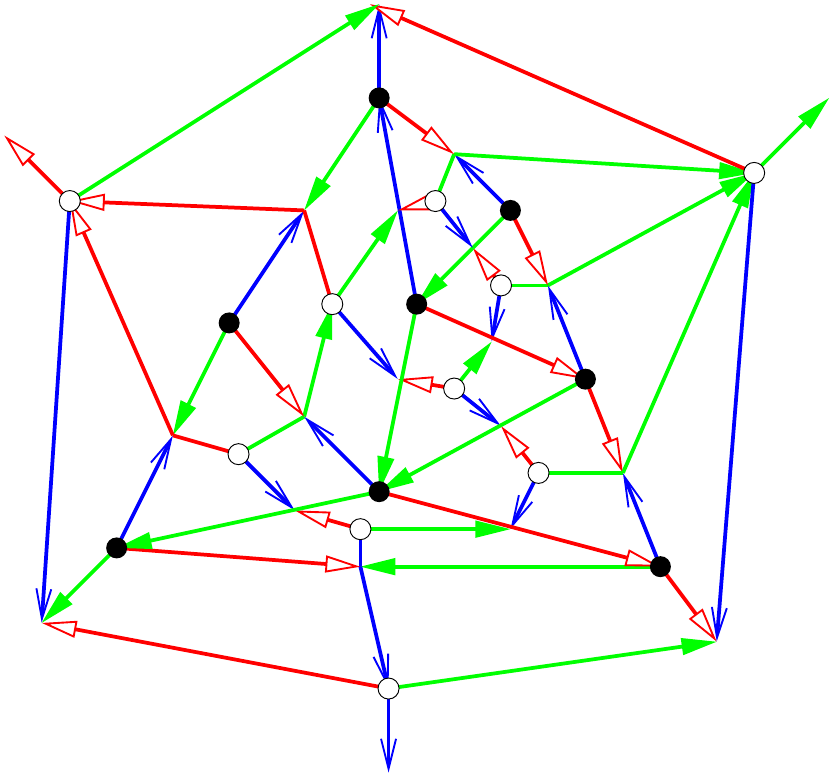}
\caption{Superposition of a planar Schnyder wood and its dual.}
\label{fig:planar-felsner-duality}
\end{figure}

\part{Generalization to orientable surfaces}
\label{part:2}
\chapter{Generalization of Schnyder woods}
\label{sec:generalization}

\section{Euler's formula and consequences}

For higher genus triangulated surfaces, a generalization of Schnyder
woods has been proposed by Castelli Aleardi, Fusy and
Lewiner~\cite{CFL09}, with applications to encoding.  In this
definition, the simplicity and the symmetry of the original definition
of Schnyder woods are lost. Here we propose an alternative
generalization of Schnyder woods for higher genus.

A closed curve on a surface is \emph{contractible} if it can be
continuously transformed into a single point.  Given a graph embedded
on a surface, a \emph{contractible loop} is an edge forming a
contractible curve. Two edges of an embedded graph are called
\emph{homotopic multiple edges} if they have the same extremities and
their union encloses a region homeomorphic to an open disk.  Except if
stated otherwise, we consider graphs embedded on surfaces that do not
have contractible loops nor homotopic multiple edges. Note that this
is a weaker assumption, than the graph being \emph{simple}, i.e. not
having loops nor multiple edges. A graph embedded on a surface is
called a \emph{map} on this surface if all its faces are homeomorphic
to open disks.  A map is a triangulation if all its faces have size
three.  A \emph{triangle} of a map is a closed walk of size $3$
enclosing a region that is homeomorphic to an open disk.  This region
is called the \emph{interior} of the triangle. Note that a triangle is
not necessarily a face of the map as its interior may be not empty. A
\emph{separating triangle} is a triangle whose interior is non empty.
Note also that a triangle is not necessarily a cycle since
non-contractible loops are allowed.

Euler's formula says that any map on an orientable surface of genus
$g$ satisfies $n-m+f=2-2g$. In particular, the plane is the surface of
genus $0$, the torus the surface of genus $1$, the double torus the
surface of genus $2$, etc.  By Euler's formula, a triangulation of
genus $g$ has exactly $3n+6(g-1)$ edges. For a toroidal triangulation,
Euler's formula gives exactly $m=3n$ so there is hope for a nice
object satisfying the Schnyder property for every vertex.  But having
a generalization of Schnyder woods in mind, for all $g\ge 2$ there are
too many edges to force all vertices to have outdegree exactly
$3$. This problem can be overcome by allowing vertices to fulfill the
Schnyder property ``several times'', i.e. such vertices have outdegree
3, 6, 9, etc. with the color property of Figure~\ref{fig:LSP} repeated
several times (see Figure~\ref{fig:369}).

\begin{figure}[!h]
\center
\begin{tabular}{ccc}
\includegraphics[scale=0.5]{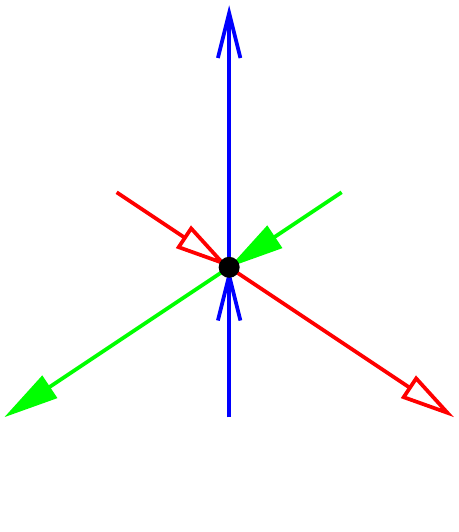} \  &  \
\includegraphics[scale=0.5]{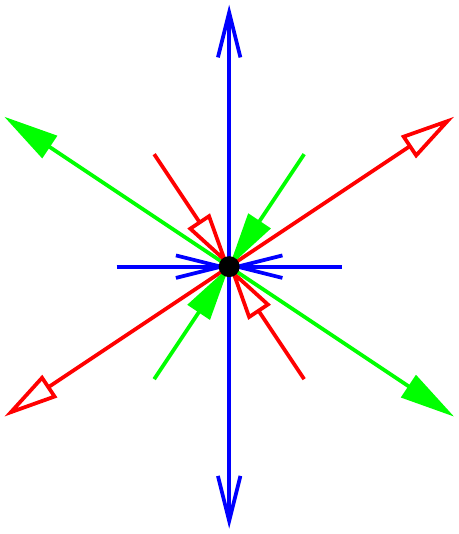} \  &  \
\includegraphics[scale=0.5]{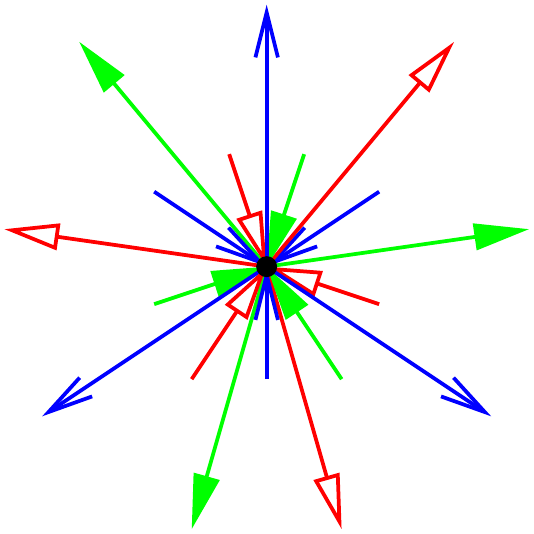} \\
Outdegree 3 \  & \ Outdegree 6 \  &  \ Outdegree 9 \\
\end{tabular}
\caption{The Schnyder property repeated several times around a
  vertex.}
\label{fig:369}
\end{figure}

In this manuscript we formalize this idea to obtain a concept of Schnyder
woods applicable to general maps (not only triangulations) on
arbitrary orientable surfaces.

\section{Schnyder angle labeling}
\label{sec:anglelabeling}

Consider a map $G$ on an orientable surface. An \emph{angle labeling}
of $G$ is a labeling of the angles of $G$ (i.e. face corners of $G$)
in colors $0$, $1$, $2$.  More formally, we denote an angle labeling
by a function $\ell:\mc{A}\to \mb{Z}_3$, where $\mc{A}$ is the set of
angles of $G$.  Given an angle labeling, we define several properties
of vertices, faces and edges that generalize the notion of Schnyder
angle labeling in the planar case~\cite{Fel-book}.

Consider an angle labeling $\ell$ of $G$.
A vertex or a face $v$ is of \emph{type $k$}, for $k\geq 1$, if the
labels of the angles around $v$ form, in counterclockwise order, $3k$
nonempty intervals such that in the $j$-th interval all the angles
have color $(j \mod 3)$. A vertex or a face $v$ is of \emph{type
  $0$}, if the labels of the angles around $v$ are all of color $i$
for some $i$ in $\{0,1,2\}$. 

An edge $e$ is of \emph{type $1$ or $2$} if the labels of the four
angles incident to edge $e$ are, in clockwise order, $i-1$, $i$,
$i$, $i+1$ for some $i$ in $\{0,1,2\}$. The edge $e$ is of \emph{type
  $1$} if the two angles with the same color are incident to the same
extremity of $e$ and of \emph{type $2$} if the two angles are incident
to the same side of $e$. An edge $e$ is of
\emph{type $0$} if the labels of the four angles incident to edge $e$
are all $i$ for some $i$ in $\{0,1,2\}$ (see
Figure~\ref{fig:edgelabeling}).

If there exists a function $f:V\to \mathbb N$ such that every vertex
$v$ of $G$ is of type $f(v)$, we say that $\ell$ is $f$-{\sc vertex}.
If we do not want to specify the function $f$, we simply say that
$\ell$ is {\sc vertex}.  We sometimes use the notation $K$-{\sc
  vertex} if the labeling is $f$-{\sc vertex} for a function $f$ with
$f(V)\subseteq K$. When $K=\{k\}$, i.e. $f$ is a constant function,
then we use the notation $k$-{\sc vertex} instead of $f$-{\sc
  vertex}. Similarly we define {\sc face}, $K$-{\sc face}, $k$-{\sc
  face}, {\sc edge}, $K$-{\sc edge}, $k$-{\sc edge}.

The following lemma shows that property {\sc edge} is the central
notion here.  Properties $K$-{\sc vertex} and $K$-{\sc face} are used
 to express additional requirements on the angle labellings
that are considered.

\begin{lemma}
\label{lem:EDGElabeling}
An {\sc edge} angle labeling is {\sc vertex} and {\sc face}.
\end{lemma}

\begin{proof}
  Consider $\ell$ an {\sc edge} angle labeling. Consider two
  counterclockwise consecutive angles $a,a'$ around a vertex (or a
  face). Property {\sc edge} implies that $\ell(a')=\ell(a)$ or
  $\ell(a')=\ell(a)+1$ (see Figure~\ref{fig:edgelabeling}). Thus by
  considering all the angles around a vertex or a face, it is clear
  that $\ell$ is also {\sc vertex} and {\sc face}.
\end{proof}

Thus we define a Schnyder labeling as follows:

\begin{definition}[Schnyder labeling]
\label{def:schnyderlabeling}
Given a map $G$ on an orientable surface, a \emph{Schnyder labeling} of $G$ is
an {\sc edge} angle labeling of $G$.
\end{definition}

Figure~\ref{fig:edgelabeling} shows how a Schnyder labeling defines an
orientation and coloring of the edges of the graph with edges oriented
in one direction or in two opposite directions.  Compared to
Definition~\ref{def:felsner}, we allow an edge to be oriented in two
opposite directions that are incoming, and in this case the two
direction have the same color.

\begin{figure}[!h]
\center
\begin{tabular}{ccc}
\includegraphics[scale=0.5]{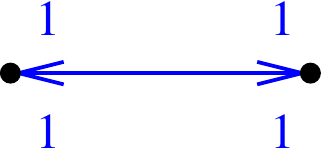} \ \ \  &  \ \ \
\includegraphics[scale=0.5]{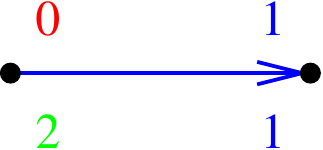} \ \ \  &  \ \ \
\includegraphics[scale=0.5]{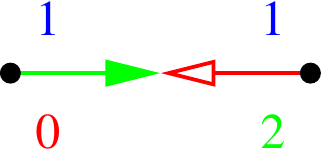} \\
Type 0  \ \ \  &  \ \ \ Type 1  \ \ \  &  \ \ \ Type 2  \\
\end{tabular}
\caption{Correspondence between Schnyder labellings and some
  bi-orientations and colorings of the edges.}
\label{fig:edgelabeling}
\end{figure}

The correspondence of Figure~\ref{fig:edgelabeling} gives the
following bijection, as proved by Felsner~\cite{Fel03} (see
Figure~\ref{fig:planar-felsner-angle}):

\begin{proposition}[\cite{Fel03}]
\label{prop:bijfelsner}
  If $G$ is a planar map and $x_0$, $x_1$, $x_2$ are three vertices
  occurring in counterclockwise order on the outer face of $G$, then
  the planar Schnyder woods of $G^\sigma$ are in bijection with the
  \{1,2\}-{\sc edge}, 1-{\sc vertex}, 1-{\sc face} angle labellings of
  $G^\sigma$ (with the outer face being 1-{\sc face} but in clockwise
  order w.r.t.~itself).
\end{proposition}

\begin{figure}[!h]
\center
\includegraphics[scale=0.5]{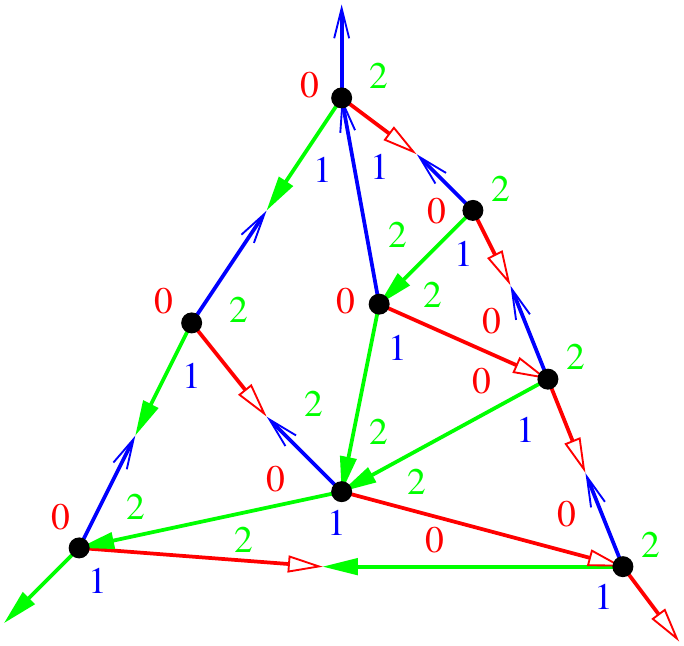}
\caption{Angle labeling corresponding to a planar Schnyder wood.}
\label{fig:planar-felsner-angle}
\end{figure}

In Sections~\ref{sec:sphericalSW} and~\ref{sec:generalSW}, we show how
Schnyder labellings can be used to generalize Schnyder woods on
surfaces and then exhibit some  properties in the dual
(Section~\ref{sec:duality}) and in the universal cover
(Section~\ref{sec:universalcover}).  In
Section~\ref{sec:conjectureexistence}, we raise some conjectures about
the existence of Schnyder labeling in higher genus that leads us to
consider {\sc edge}, $\mathbb{N}^*$-{\sc vertex}, $\mathbb{N}^*$-{\sc
  face} angle labellings. These particular angle labellings seem to be
the most interesting case of Schnyder labellings when $g\geq 1$ but the
core of the structural properties (Chapters~\ref{sec:characterization}
and~\ref{sec:structure}) is written for the general situation of {\sc
  edge} angle labellings with no additional requirement on properties
{\sc vertex} and {\sc face}.

\section{Spherical Schnyder woods}
\label{sec:sphericalSW}

We generalize angle labellings letting vertices have outdegree equal to
$0\bmod 3$ and not necessarily exactly $3$ like in the previous
definitions of Schnyder woods. Such a relaxation allows us to define a
new kind of Schnyder wood in the plane with no ``special rule'' on the
outer face (i.e. not just considering inner vertices like in
Definition~\ref{def:schnyder} and without adding half-edges reaching
the outer face like in Definition~\ref{def:felsner}). We call them
\emph{spherical Schnyder woods} since there is no face
playing the particular role of the outer face.

\begin{definition}[spherical Schnyder wood]
\label{def:bipolar}
Given a planar map $G$, a \emph{spherical Schnyder wood} is an
orientation and coloring of the edges of $G$ with the colors $0$, $1$,
$2$, where every edge is oriented in one direction or in two
opposite directions (each direction having a distinct color and being
outgoing), satisfying the following conditions:
 
\begin{itemize}
\item Every vertex, except exactly two vertices called \emph{poles},
  satisfies the Schnyder property and each pole has only incoming
  edges, all of the same color.
\item There is no face whose boundary  is a
monochromatic cycle.
\end{itemize}
\end{definition}

See Figure~\ref{fig:example-planar-bipolar} for two examples of
spherical Schnyder woods.

\begin{figure}[!h]
\center
  \includegraphics[scale=0.5]{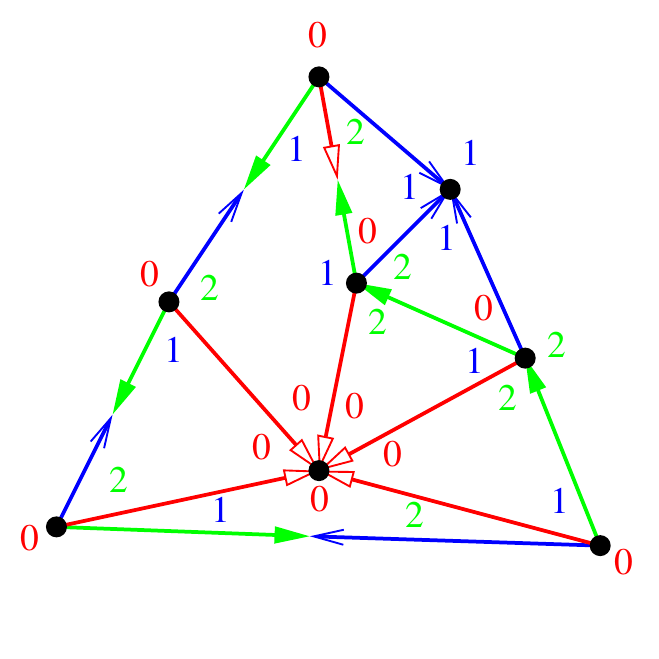}
\hspace{1cm}
\includegraphics[scale=0.5]{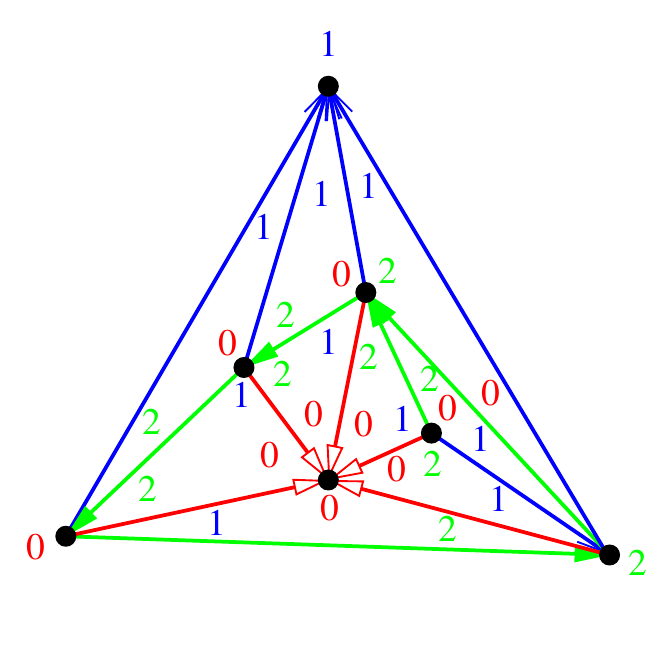}
\caption{A spherical Schnyder wood of a planar map and of a planar
  triangulation.}
\label{fig:example-planar-bipolar}
\end{figure}

Like for planar Schnyder woods, there is a bijection between spherical
Schnyder woods and particular angle labellings (see
Figure~\ref{fig:example-planar-bipolar}):

\begin{proposition}
\label{prop:bijbipolar}
If $G$ is a planar map, then the spherical Schnyder woods of $G$ are in
bijection with the \{1,2\}-{\sc edge}, \{0,1\}-{\sc vertex}, 1-{\sc
  face} angle labellings of $G$.
\end{proposition}

\begin{proof}
  ($\Longrightarrow$) Consider a spherical Schnyder wood of $G$.  We
  label the angles around a pole $v$ with the color of its incident
  edges. We label the angles of a non-pole vertex $v$ such that all
  the angles in the counterclockwise sector from $e_{i+1}(v)$ to
  $e_{i-1}(v)$ are labeled $i$. Then one can easily check that the
  two poles are of type 0, that all the non-poles are of type 1 and
  that all the edges are of type 1 or 2. Then by
  Lemma~\ref{lem:EDGElabeling}, the labeling is also {\sc face}.

  It remains to prove that the labeling is 1-{\sc face}.  For this
  purpose, we count the color changes of the angles at vertices, faces
  and edges, and denote it by the function $c$.  Since the labeling is
  \{1,2\}-{\sc edge}, for an edge $e$ there are exactly three changes
  around $e$, so $c(e)=3$ and $\sum_e c(e)=3m$.  These changes around
  an edge can happen either around one of the two incident vertices or
  one of the two incident faces, so
  $\sum_e c(e)=\sum_v c(v)+\sum_f c(f)$. The vertices $v$ of type 1
  have $c(v)=3$. The two vertices $v$ of type 0 have $c(v)=0$. So
  $\sum_v c(v)=3(n-2)$. Thus finally, $\sum_f c(f)=3m-3(n-2)=3f$ (the
  last equality is by Euler's formula).  Since the labeling is {\sc
    face}, we have $c(f)=0\bmod 3$ for every face $f$. As there is no
  face whose boundary is a monochromatic cycle, all the faces have
  $c(f)\geq 3$ and thus all the faces have exactly $c(f)=3$ and the
  labeling is 1-{\sc face}.

  ($\Longleftarrow$) Consider a \{1,2\}-{\sc edge}, \{0,1\}-{\sc
    vertex}, 1-{\sc face} angle labeling of $G$.  Again, we count the
  color changes of the angles at vertices, faces and edges.  We have
  $3m=\sum_e c(e)=\sum_v c(v)+\sum_f c(f)$ and $\sum_f c(f)=3f$.  Then
  Euler's formula gives $\sum_v c(v)=3m-3f=3(n-2)$. The vertices $v$
  of type 1 have $c(v)=3$ and the vertices $v$ of type 0 have
  $c(v)=0$. So there are exactly two vertices of type 0.  Consider the
  coloring and orientation of the edges of $G$ obtained by the
  correspondence shown in Figure~\ref{fig:edgelabeling}. It is clear
  that every vertex, except exactly the two vertices of type 0,
  satisfies the Schnyder property and that each vertex of type 0 has
  only incoming edges, all of the same color.  Since the angle
  labeling is 1-{\sc face}, there is no face whose boundary is
  a monochromatic cycle as such a face is of type 0. So the considered
  coloring and orientation of the edges of $G$ is a spherical Schnyder
  wood.
\end{proof}

A spherical Schnyder wood of a planar triangulation has all its faces
that are triangular and 1-{\sc face}. Thus the three angles of a face
are labels $0,1,2$ and there is no bi-directed edge. So the angle
labeling is 1-{\sc edge}.

Spherical Schnyder woods seem somewhat less ``regular'' than usual
Schnyder woods. For example, the colors of the edges at the two poles
can be distinct or identical and the edges having the same color can
induce a connected graph or not and may contains cycles (see
Figures~\ref{fig:example-planar-bipolar}
and~\ref{fig:planar-bipolar-strange}).

\begin{figure}[!h]
\center
\includegraphics[scale=0.3]{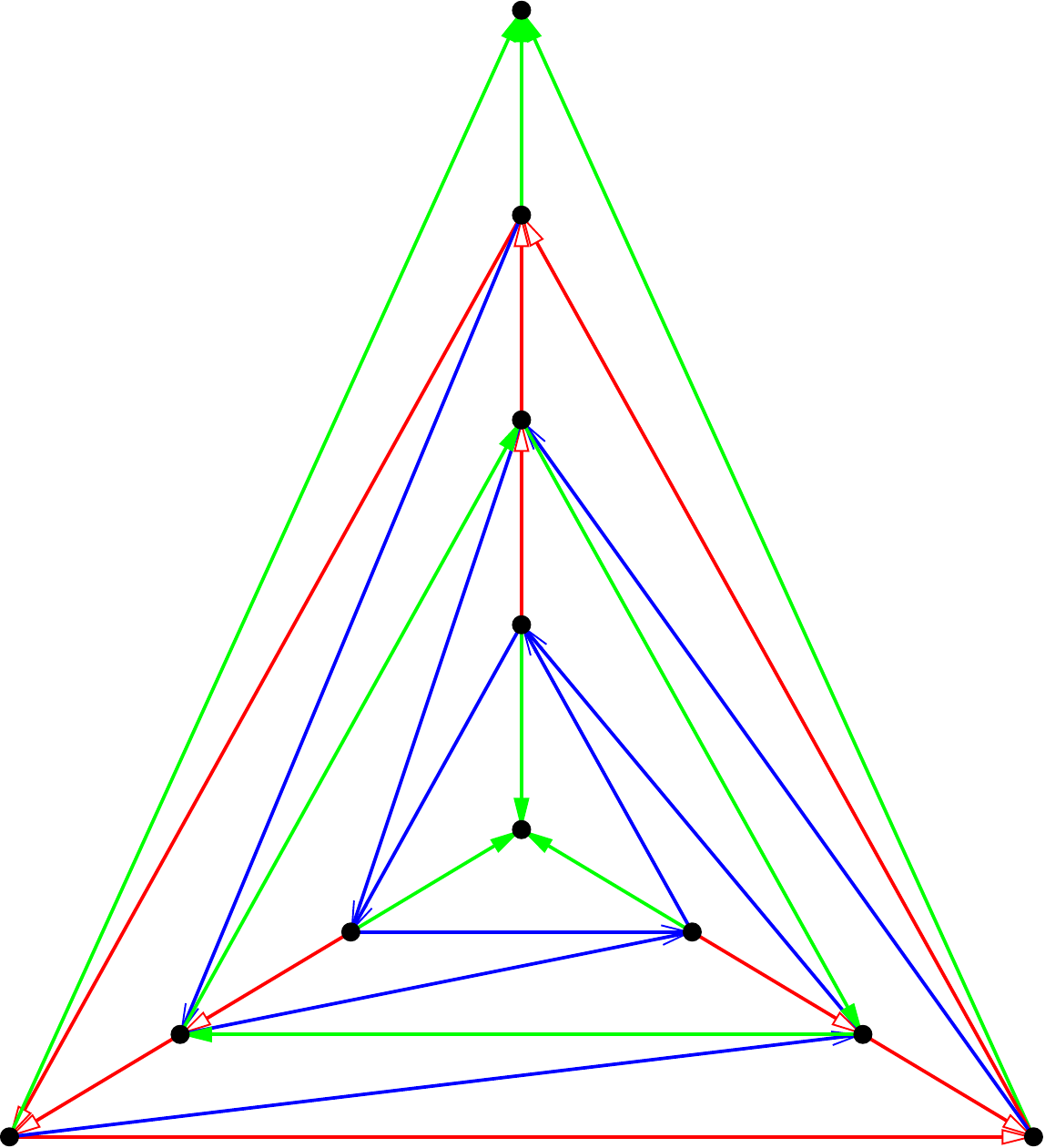}
\caption{A spherical Schnyder wood of a planar triangulation
  where the green color has three components including the two poles.}
\label{fig:planar-bipolar-strange}
\end{figure}

Note that the two poles
must be non-adjacent vertices since they have only incoming
edges. Then it is not difficult to see that $K_4$ is the only planar
triangulation that does not admit a spherical Schnyder wood.

\begin{proposition}
\label{prop:notk4}
  A planar triangulation that is not $K_4$  admits a  spherical Schnyder wood.
\end{proposition}

\begin{proof}
  Let $G$ be a planar triangulation that is not $K_4$. Since $G$ is
  not $K_4$, we can consider $x_0, x_1, x_2, u$ such that $x_0$,
  $x_1$, $x_2$ are the three vertices occurring in counterclockwise
  order on the outer face of $G$, vertex $u$ is not incident to $x_1$,
  and $x_0ux_2$ is a face of $G$.  Consider a planar Schnyder wood of
  $G$ rooted at $x_0$, $x_1$, $x_2$ (see Definition~\ref{def:felsner}
  and example of Figure~\ref{fig:planar-felsner}). 
  Let $v$ be the entering vertex of  $e_1(u)$.  We now transform
  the planar Schnyder wood into a spherical Schnyder wood by the following:

  \begin{itemize}
  \item Remove the half-edges at $x_0,x_1,x_2$. 
  \item Orient  the three outer edges in one direction only such that
$x_0x_1$ and $x_2x_1$ are entering $x_1$ in color $1$, and $x_2x_0$
is leaving $x_2$ in color $2$.
  \item Reverse the three edges $ux_0, uv, ux_2$ such that they are
    entering $u$ in color $0$.
  \item Reverse all the edges of the path $P_0(v)$ from $x_0$ to $v$
    and color them in color $2$.
  \item In the interior of the region delimited by $vu$, $ux_0$ and
    $P_0(v)$, add $+1$ to all the colors (without modifying the
    orientation).
  \end{itemize}

  Then it is not difficult to see that the obtained orientation and
  coloring is a spherical Schnyder wood of $G$ with poles $x_1$ and
  $u$.
\end{proof}

The proof of Proposition~\ref{prop:notk4} gives a general method to
transform a planar Schnyder wood into a spherical Schnyder wood.  The
spherical Schnyder wood on the right of
Figure~\ref{fig:example-planar-bipolar} is obtained from the one on the
right of Figure~\ref{fig:planar-felsner} by such a transformation.
This transformation can be done in a more general framework where
$u,v$ are any inner vertices satisfying: $v$ is on $P_1(u)$ and
$P_0(u),P_0(v)$ are edge disjoint. This is illustrated on
Figure~\ref{fig:transfo-bipolar} where the ``+1'' and ``-1'' indicates
the change of colors in the interior of the considered region.

\begin{figure}[!h]
\center
\includegraphics[scale=0.3]{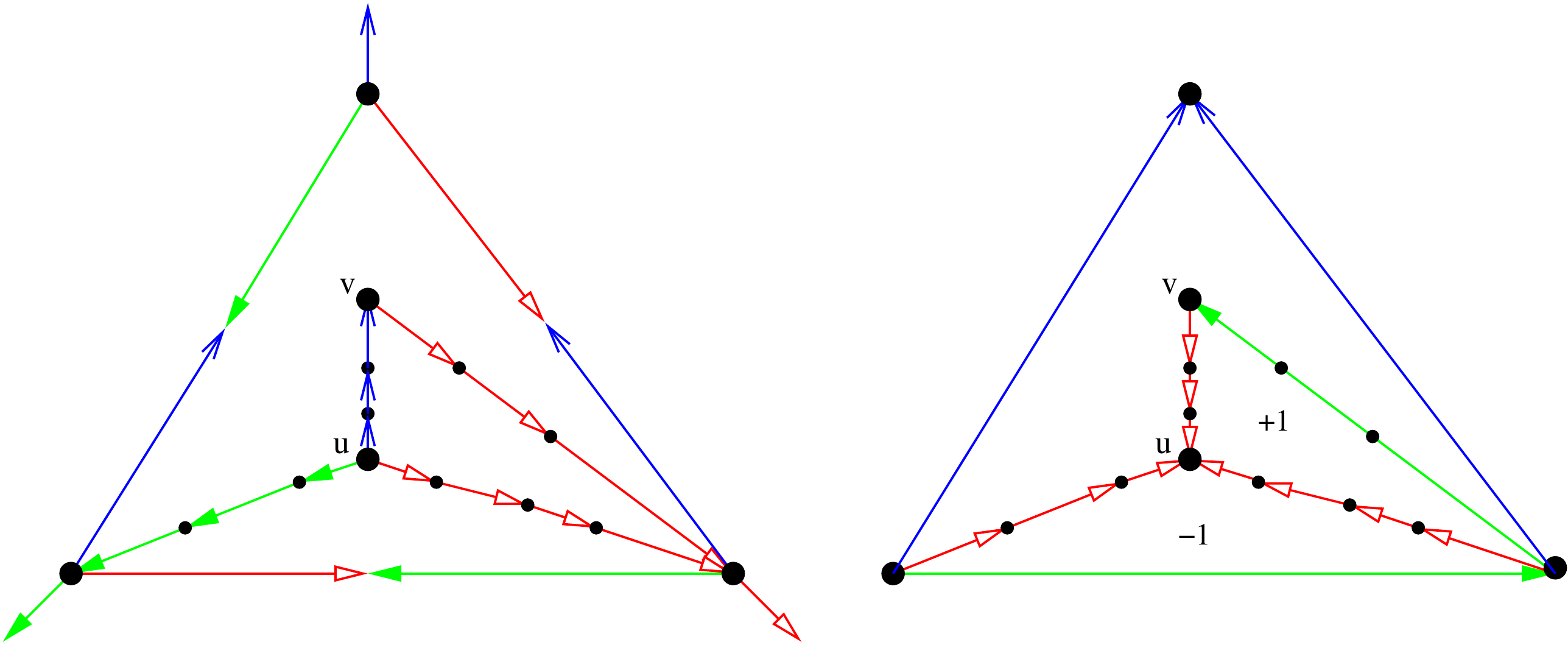}
\caption{Transforming a planar Schnyder wood into a spherical
  Schnyder wood.}
\label{fig:transfo-bipolar}
\end{figure}

Note that it is also possible to define variants of
Definition~\ref{def:bipolar}. For example one can allow one pole or
the two poles to be in the dual map. For triangulation, this implies
the use of edges of type 2 by a counting argument.
Figure~\ref{fig:bipolarpoles} gives two such examples. The example on
the left of Figure~\ref{fig:bipolarpoles} is obtained from the
triangulation of Figure~\ref{fig:planar-felsner} by choosing an inner
vertex $v$, reversing the three monochromatic paths from $v$ and
removing the half edges.  The example on the right of
Figure~\ref{fig:bipolarpoles} is obtained from the triangulation of
Figure~\ref{fig:planar-felsner} by choosing an inner face $f$,
reversing three paths with different colors from the three vertex of
$f$, removing the half edges and orienting the edges of $f$ in both
directions.

\begin{figure}[!h]
\center
\includegraphics[scale=0.5]{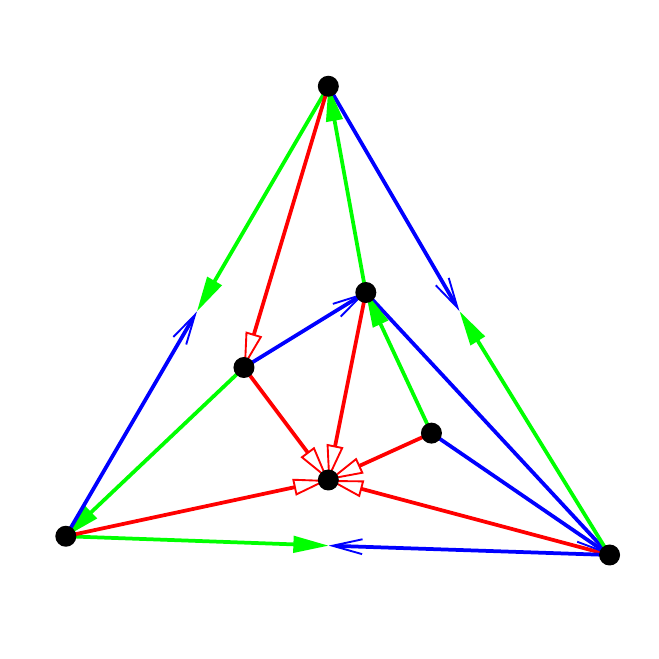}
\includegraphics[scale=0.5]{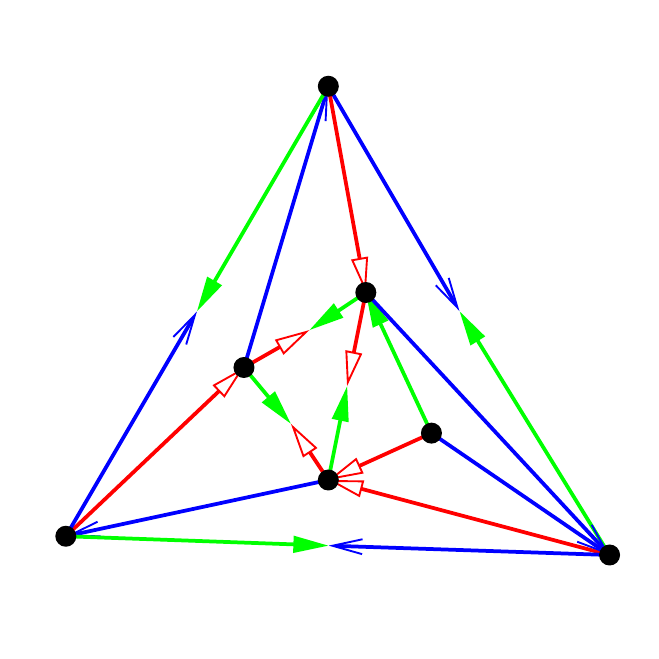}
\caption{Examples of orientation}
\label{fig:bipolarpoles}
\end{figure}

We have defined spherical Schnyder woods with poles in the primal since
in the context of Schnyder woods, the objects that are considered
seem to be most interesting when, for triangulation, only type 1 edges are
used.  Here we will not study much further these natural classes of
\{1,2\}-{\sc edge} angle labellings of planar maps.

\section{Generalized Schnyder woods}
\label{sec:generalSW}

Any map (on any orientable surface) admits a trivial {\sc edge} angle
labeling: the one with all angles labeled $i$ (and thus all edges,
vertices and faces are of type 0).  A natural non-trivial case, that
is also symmetric for the duality, is to consider {\sc edge},
$\mathbb{N}^*$-{\sc vertex}, $\mathbb{N}^*$-{\sc face} angle labellings
of general maps.  In planar Schnyder woods only type 1 and type 2
edges are used. Here we allow type 0 edges because they seem
unavoidable for some maps (see discussion below).  This suggests the
following definition of Schnyder woods in higher genus.

First, the generalization of the Schnyder property is the following (see Figure~\ref{fig:369}):

\begin{definition}[Generalized Schnyder property]
\label{def:generalschnyderproperty}
Given a map $G$ on a genus $g\geq 1$ orientable surface, a vertex $v$
and an orientation and coloring of the edges incident to $v$ with the
colors $0$, $1$, $2$, we say that  $v$ satisfies the
\emph{generalized Schnyder property}, if
$v$ satisfies the following local property for $k\geq 1$:

\begin{itemize}
\item Vertex $v$ has out-degree $3k$.
\item The edges $e_0(v),\ldots, e_{3k-1}(v)$ leaving $v$ in
  counterclockwise order are such that $e_j(v)$ has color
  $(j \bmod 3)$.
\item Each edge entering $v$ in color $i$ enters $v$ in a
  counterclockwise sector from $e_{j}(v)$ to $e_{j+1}(v)$ (where $j$
  and $j+1$ are understood modulo $3k$) with
  $(j\bmod 3) \neq (i\bmod 3) \neq ((j+1) \bmod 3)$.
\end{itemize}
\end{definition}

Then, the generalization of Schnyder woods is the
following (where the three types of edges depicted on
Figure~\ref{fig:edgelabeling} are allowed):

\begin{definition}[Generalized Schnyder wood]
\label{def:highgenus}
  Given a map $G$ on a genus $g\geq 1$ orientable surface, a
  \emph{generalized Schnyder wood} of $G$ is an orientation and
  coloring of the edges of $G$ with the colors $0$, $1$, $2$, where
  every edge is oriented in one direction or in two opposite
  directions (each direction having a distinct color and being
  outgoing, or each direction having the same color and being
  incoming), satisfying the following conditions:
 
\begin{itemize}
\item Every vertex satisfies the generalized Schnyder property (see
  Definition~\ref{def:generalschnyderproperty}).
\item There is no face whose boundary  is a
monochromatic cycle.
\end{itemize}
\end{definition}

When there is no ambiguity we may omit the word ``generalized'' in
``generalized Schnyder wood'' or ``generalized Schnyder property''.

See Figure~\ref{fig:tore-primal} for two examples of Schnyder woods in
the torus (where the torus is represented
by a square in the plane whose opposite sides are pairwise
identified).

\begin{figure}[h!]
\center
\includegraphics[scale=0.4]{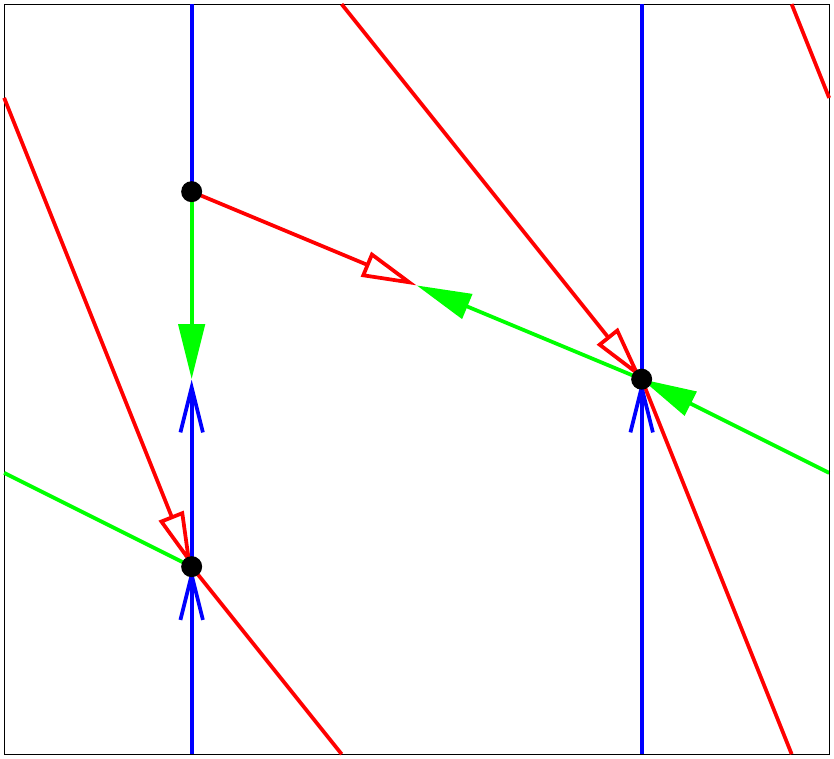}
\hspace{1cm}
\includegraphics[scale=0.4]{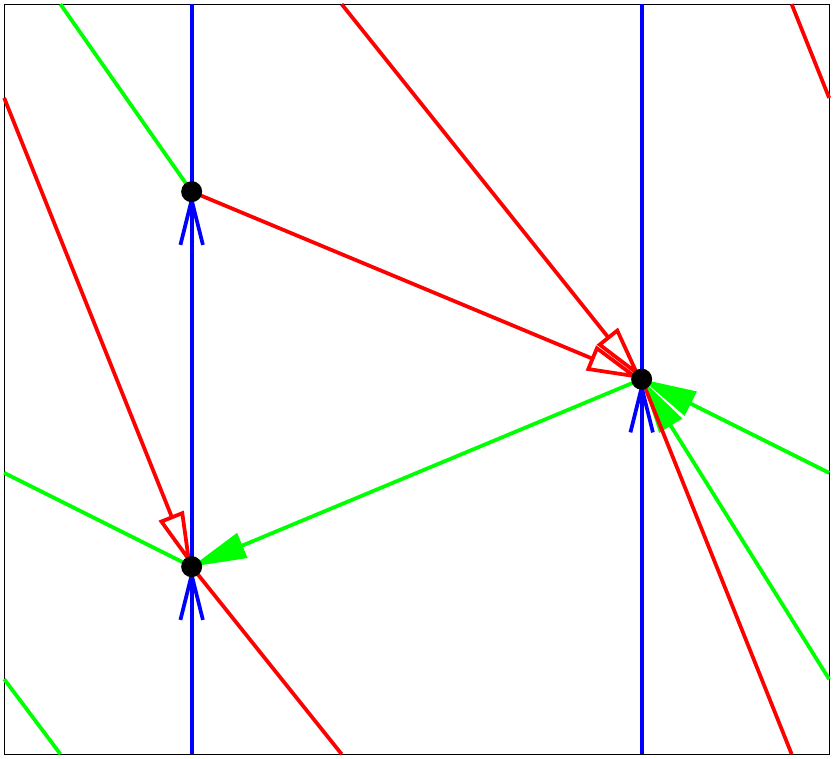}
\caption{A Schnyder wood of a toroidal map and of a toroidal
  triangulation.}
\label{fig:tore-primal}
\end{figure}

Figure~\ref{fig:doubletorus} is an example of a Schnyder wood on
a triangulation of the double torus. The double torus is represented
by an octagon. The sides of the octagon are
identified according to their labels. All the vertices of the
triangulation have outdegree 3 except two vertices that have outdegree
6, which are circled. Each of the latter appears twice in the
representation.

\begin{figure}[!h]
\center
\includegraphics[scale=0.3]{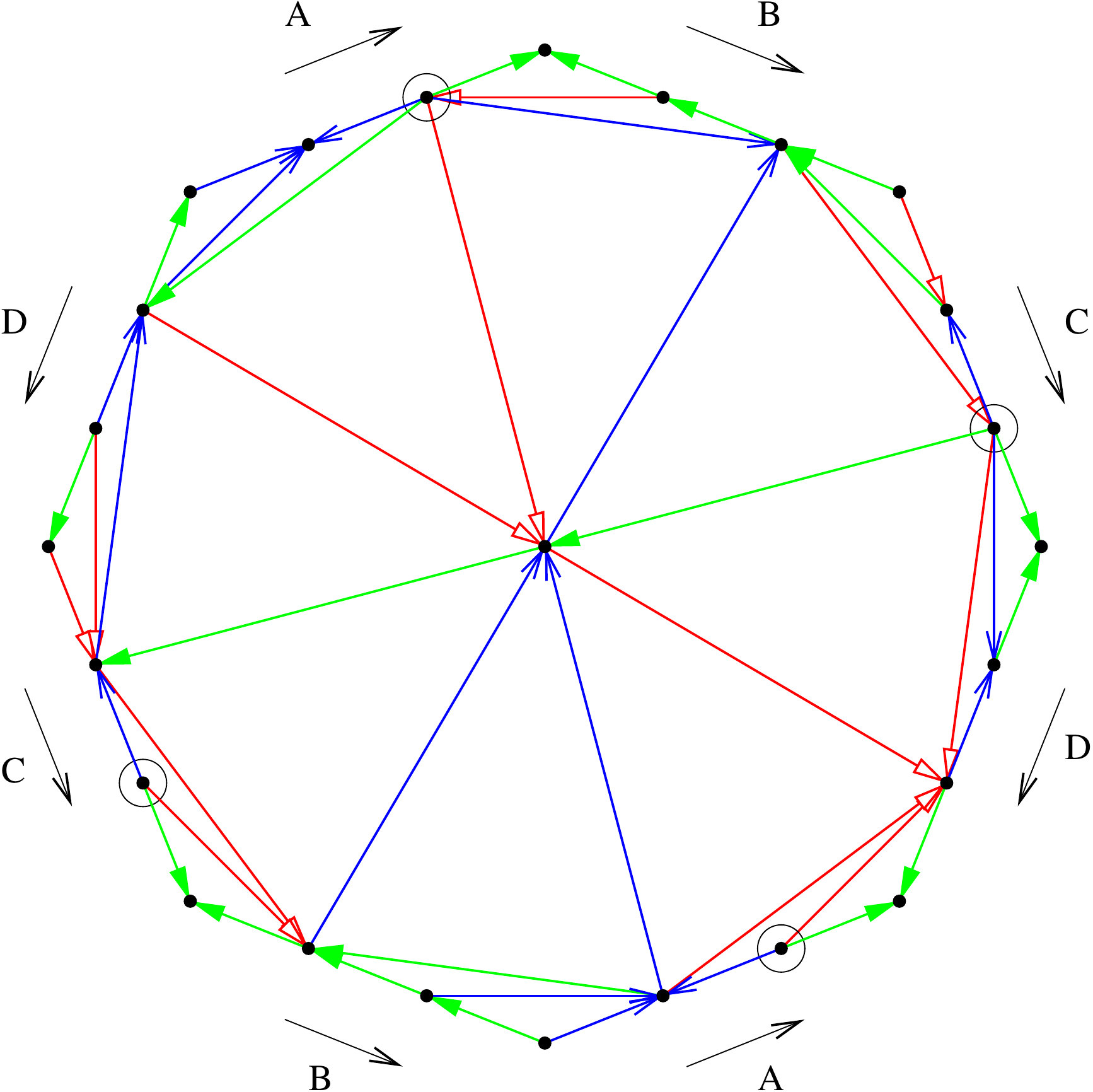}
\caption{A Schnyder wood of a triangulation of the double torus.}
\label{fig:doubletorus}
\end{figure}

The correspondence of Figure~\ref{fig:edgelabeling} immediately gives
the following bijection whose proof is omitted.

\begin{proposition}
\label{prop:bijtoregen}
If $G$ is a map on a genus $g\geq 1$ orientable surface, then the
generalized Schnyder woods of $G$ are in bijection with the {\sc
  edge}, $\mathbb{N}^*$-{\sc vertex}, $\mathbb{N}^*$-{\sc face} angle
labelings of $G$.
\end{proposition}

Note that in the examples of Figures~\ref{fig:tore-primal}
and~\ref{fig:doubletorus}, type 0 edges do not appear. However, 
for $g\geq 2$, there are some maps with vertex degrees and face
degrees at most $5$ that admits generalized Schnyder woods. For these
maps, type $0$ edges are unavoidable.
Figure~\ref{fig:annoying-maps-small} gives an example of such a map
for $g=2$, with a Schnyder wood that has two edges of type 0.

\begin{figure}[!h]
\center
\includegraphics[scale=0.3]{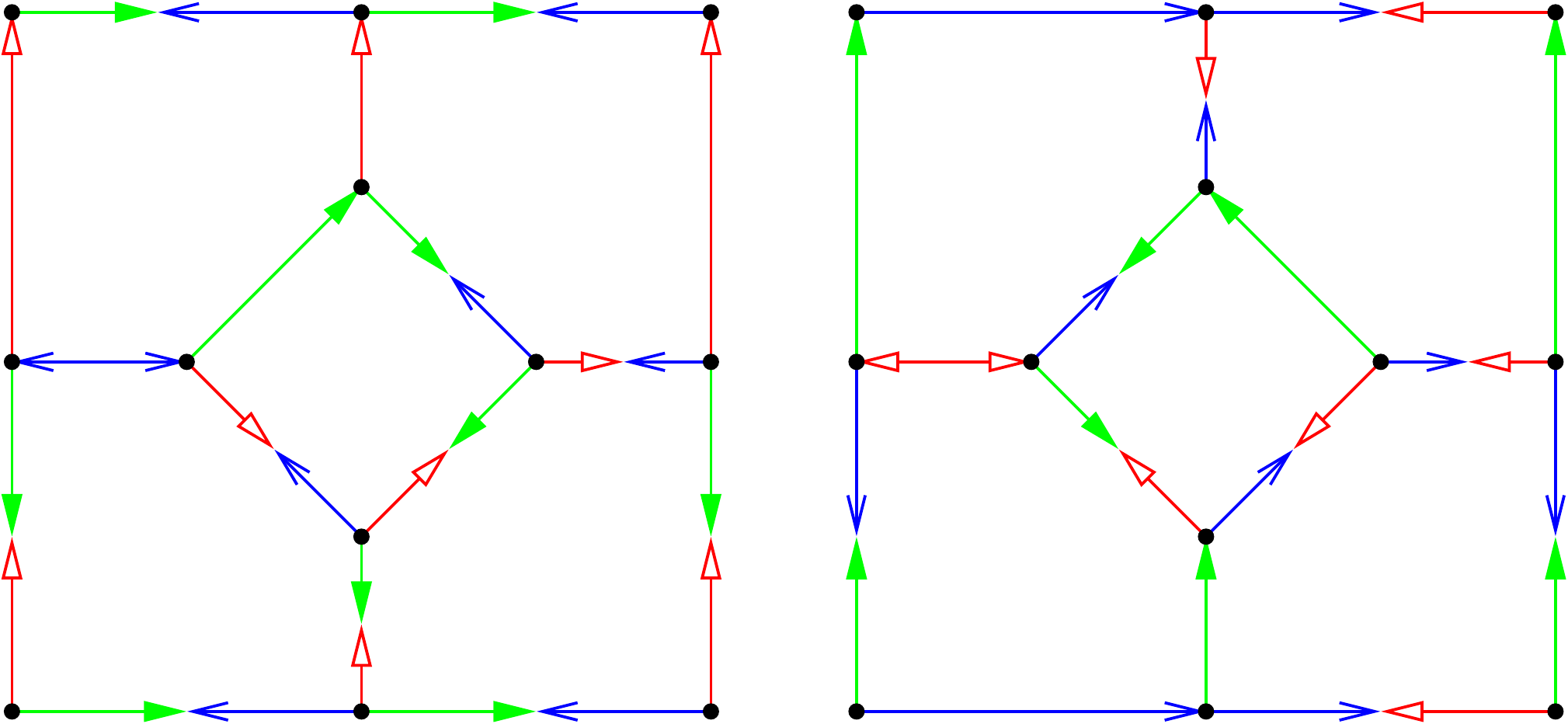}
\caption{A Schnyder wood of a double-toroidal map that has two edges
  of type 0.  Here, the two parts are toroidal and the two central
  faces are identified (by preserving the colors) to obtain a
  double-toroidal map.}
\label{fig:annoying-maps-small}
\end{figure}

\section{Schnyder woods and duality}
\label{sec:duality}
Given a map $G$ on an genus $g\geq 1$ orientable surface, the
\emph{dual} of a Schnyder wood of $G$ is the orientation and coloring
of the edges of $G^*$ obtained by the rules represented on
Figure~\ref{fig:edgelabelingdual} where the correspondence of
Figure~\ref{fig:edgelabeling} is still valid in the dual map.

\begin{figure}[!h]
\center
\begin{tabular}{ccc}
\includegraphics[scale=0.5]{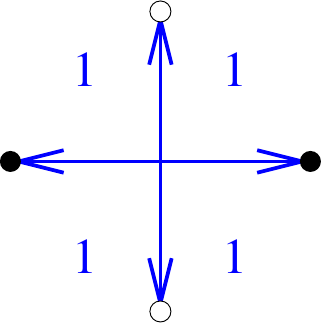} \ \ \  &  \ \ \
\includegraphics[scale=0.5]{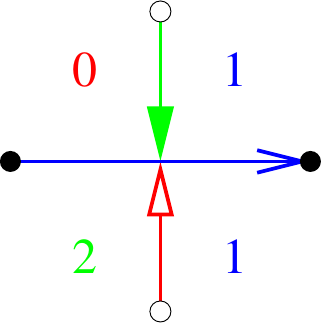} \ \ \  &  \ \ \
\includegraphics[scale=0.5]{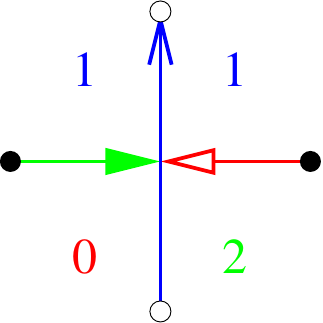} \\
Type 0  \ \ \  &  \ \ \ Type 1  \ \ \  &  \ \ \ Type 2  \\
\end{tabular}
\caption{Rules for the dual Schnyder wood and the corresponding angle labeling.}
\label{fig:edgelabelingdual}
\end{figure}

Note that a Schnyder wood of $G$ is an {\sc edge}, $\mathbb{N}^*$-{\sc
  vertex}, $\mathbb{N}^*$-{\sc face} angle labeling of $G$ (by
Proposition~\ref{prop:bijtoregen}), thus an {\sc edge},
$\mathbb{N}^*$-{\sc vertex}, $\mathbb{N}^*$-{\sc face} angle labeling
of $G^*$ (by symmetry of the definitions {\sc vertex} and {\sc
  face}), thus a Schnyder wood of $G^*$. So we have the following:

  \begin{proposition}
\label{th:bijdual}
Given a map $G$ on a genus $g\geq 1$ orientable surface, there is a
bijection between the \emph{Schnyder woods} of $G$ and of
its dual.
  \end{proposition}

An example of a Schnyder wood  of a toroidal map and its dual is given
on Figure~\ref{fig:example-dual-tore}.

\begin{figure}[!h]
\center
\includegraphics[scale=0.4]{tore-primal}
\hspace{1cm}
\includegraphics[scale=0.4]{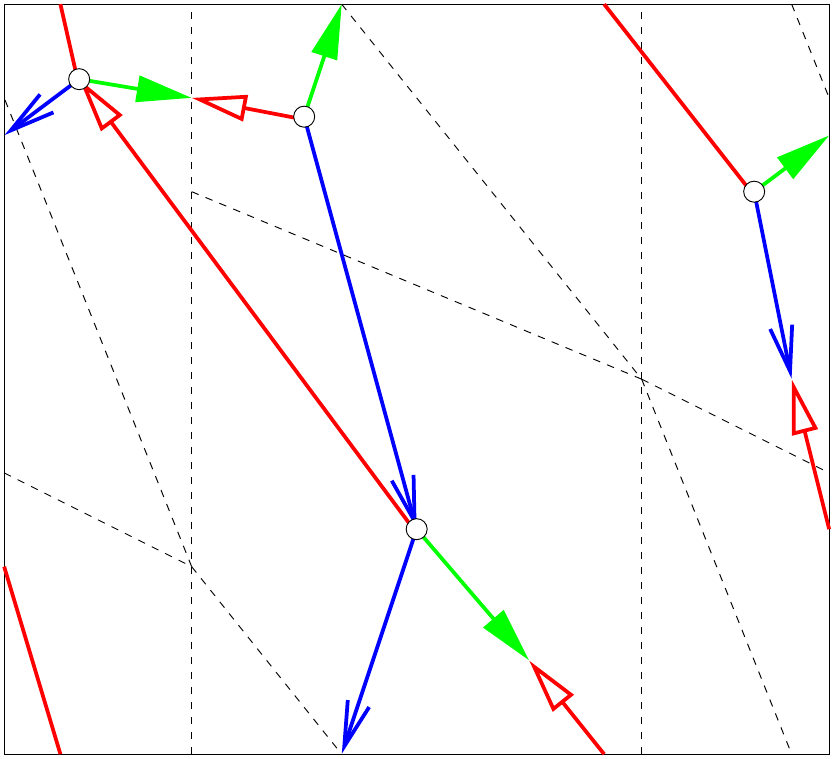}
\caption{Schnyder woods of a primal and dual toroidal map.}
\label{fig:example-dual-tore}
\end{figure}

\section{Schnyder woods in the universal cover}
\label{sec:universalcover}

We refer to~\cite{Massey} for the general theory of universal covers.
The universal cover of the torus (resp. an orientable surface of genus
$g\geq 2$) is a surjective mapping $p$ from the plane (resp. the open
unit disk) to the surface that is locally a homeomorphism.  If the
torus is represented by a hexagon in the plane whose opposite
sides are pairwise identified, then the universal cover of the torus
is obtained by replicating the hexagon to tile the plane.
Figure~\ref{fig:universalcover} shows how to obtain the universal
cover of the double torus.  The key property is that a closed curve on
the surface corresponds to a closed curve in the universal cover if
and only if it is contractible.

\begin{figure}[!h]
\center
\includegraphics[scale=1.2]{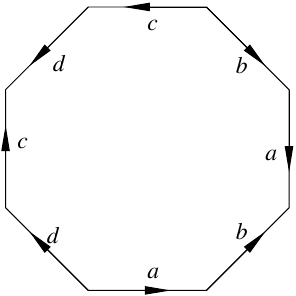}
\hspace{2em}
\includegraphics[scale=0.75]{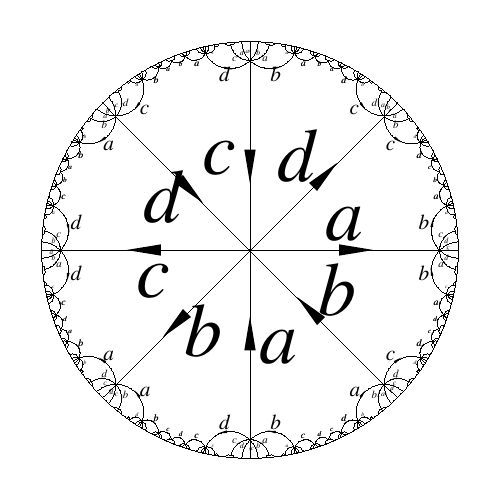}
\caption{Canonical representation and universal cover of the double torus
  (source: Yann Ollivier
  \texttt{http://www.yann-ollivier.org/maths/primer.php}).}
\label{fig:universalcover}
\end{figure}

Universal covers can be used to represent a map on an orientable surface
as an infinite planar map. Any property of the map can be lifted to
its universal cover, as long as it is defined locally. Thus universal
covers are an interesting tool for the study of Schnyder woods
since all the definitions we have given so far are purely local.

Consider a map $G$ on a genus $g\geq 1$ orientable surface.  Let
$G^\infty$ be the infinite planar map drawn on the universal cover and
defined by $p^{-1}(G)$. Note that $G$ does not have contractible loops or
homotopic multiple edges if and only if $G^\infty$ is
simple.

We need the following general lemma concerning universal covers:

\begin{lemma}
\label{lem:finitecc}
Suppose that for a finite set of vertices $X$ of $G^\infty$, the graph
$G^\infty\setminus X$ is not connected. Then $G^\infty\setminus X$ has
a finite connected component.
\end{lemma}

\begin{proof}
  Suppose the lemma is false and $G^\infty\setminus X$ is not
  connected and has no finite component. Then it has a face bounded by
  an infinite number of vertices. As the vertices of $G^\infty$ have
  bounded degree, by putting vertices $X$ back there is still a face
  bounded by an infinite number of vertices. The corresponding face in
  $G$ is not homeomorphic to an open disk, a contradiction with $G$
  being a map.
\end{proof}

A graph is \emph{$k$-connected} if it has at least $k+1$
vertices and if it stays connected after removing any $k-1$ vertices.
Extending the notion of essentially 2-connectedness defined in
\cite{MR98} for the toroidal case, we say that $G$ is
\emph{essentially $k$-connected} if $G^\infty$ is $k$-connected. Note
that the notion of being essentially $k$-connected is different from
 being $k$-connected. There is no implication in any direction and
being \emph{essentially $k$-connected} depends on the mapping (see
Figure~\ref{fig:essential-2} and Figure~\ref{fig:essential-1}). Note that a map is always essentially
$1$-connected.

\begin{figure}[!h]
\center
\includegraphics[scale=0.4]{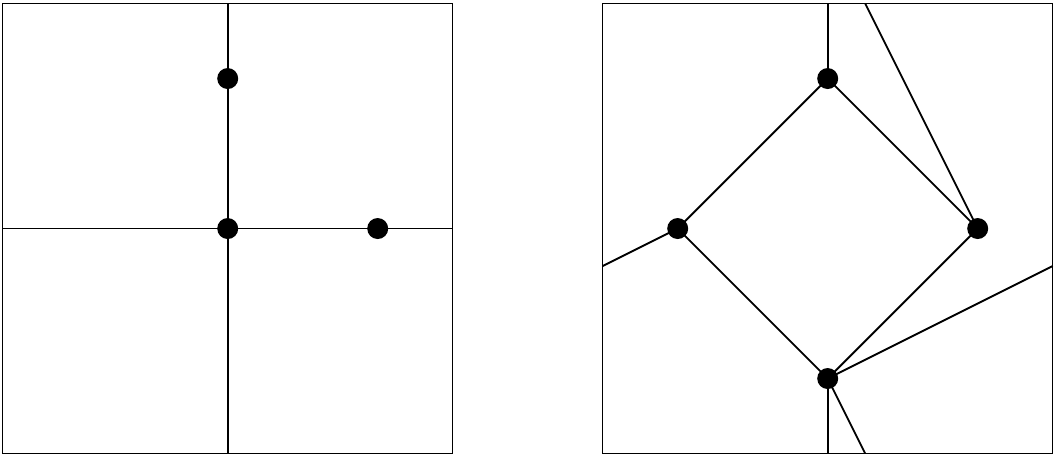}
\caption{An essentially 2-connected map that is not 2-connected and an
  essentially 3-connected map that is not 3-connected.}
\label{fig:essential-2}
\end{figure}

\begin{figure}[!h]

\center
\includegraphics[scale=0.4]{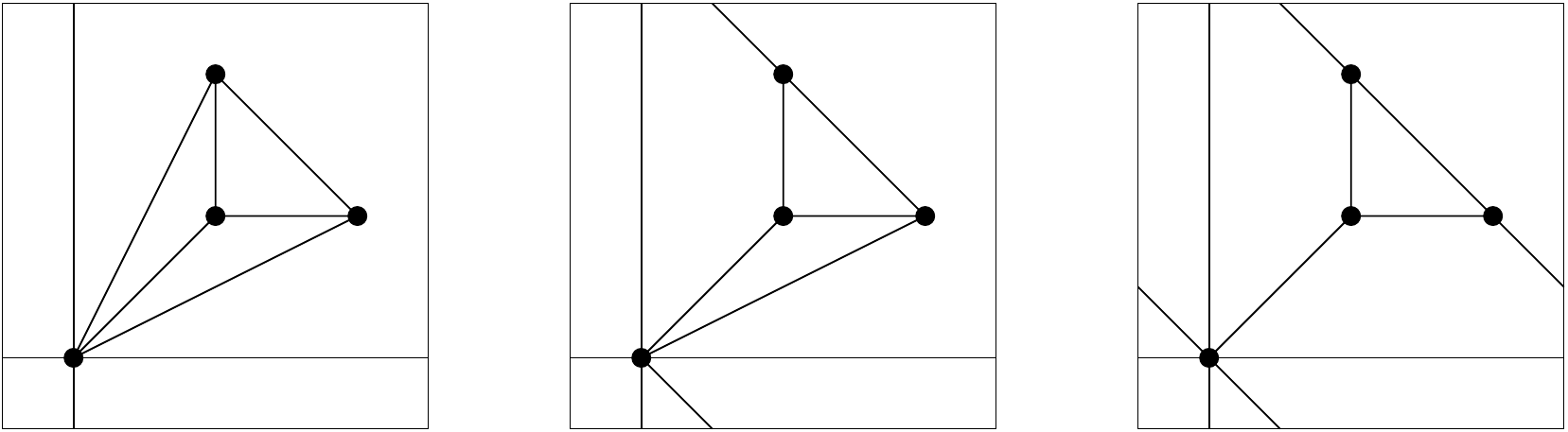}
\caption{Three different mappings of the same underlying 3-connected
  graph, which are respectively: essentially 1-connected (not
  essentially 2-connected), essentially 2-connected (not essentially
  3-connected), essentially 3-connected (not essentially 4-connected).}
\label{fig:essential-1}
\end{figure}

Suppose now that $G$ is given with a Schnyder wood (i.e. an {\sc
  edge}, $\mathbb{N}^*$-{\sc vertex}, $\mathbb{N}^*$-{\sc face} angle
labeling by Proposition~\ref{prop:bijtoregen}).  Consider the orientation
and coloring of the edges of $G^\infty$ corresponding to the Schnyder
wood of $G$.

Let $G^\infty_i$ be the directed graph induced by the edges of color
$i$ of $G^\infty$. This definition includes edges that are
half-colored $i$, and in this case, the edges get only the direction
corresponding to color $i$.  The graph $(G^\infty_i)^{-1}$ is the
graph obtained from $G^\infty_i$ by reversing all its edges.  The
graph $G^\infty_i\cup (G^\infty_{i-1})^{-1}\cup (G^\infty_{i+1})^{-1}$
is obtained from the graph $G$ by orienting edges in one or two
directions depending on whether this orientation is present in
$G^\infty_i$, $(G^\infty_{i-1})^{-1}$ or $(G^\infty_{i+1})^{-1}$.
Similarly to what happens for planar Schnyder woods (see
Lemma~\ref{lem:nodirectedcycleplan}), we have the following important
property:

\begin{lemma}
  \label{lem:nodirectedcycle}
The graph $G^\infty_i\cup (G^\infty_{i-1})^{-1}\cup
(G^\infty_{i+1})^{-1}$ contains no directed cycle.
\end{lemma}

\begin{proof}
   Suppose there is a directed cycle in $G^\infty_i\cup
  (G^\infty_{i-1})^{-1}\cup (G^\infty_{i+1})^{-1}$.  Let $C$ be such a
  cycle containing the minimum number of faces in the map $D$ with
  border $C$. Suppose by symmetry that $C$ turns around $D$
  counterclockwisely. Every vertex of $D$ has at least one outgoing
  edge of color $i+1$ in $D$.  So there is a cycle of color $(i+1)$ in
  $D$ and this cycle is $C$ by minimality of $C$.  Every vertex of $D$
  has at least one outgoing edge of color $i$ in $D$. So, again by
  minimality of $C$, the cycle $C$ is a cycle of color $i$.  Thus all
  the edges of $C$ are oriented in color $i$ counterclockwisely and in
  color $i+1$ clockwisely.

  By the definition of Schnyder woods, there is no face the boundary
  of which is a monochromatic cycle, so $D$ is not a face.  Let $vx$
  be an edge in the interior of $D$ that is outgoing for $v$. The
  vertex $v$ can be either in the interior of $D$ or in $C$ (if $v$
  has more than three outgoing arcs). In both cases, $v$ has
  necessarily an edge $e_i$ of color $i$ and an edge $e_{i+1}$ of
  color $i+1$, leaving $v$ and in the interior of $D$.  Consider
  $W_i(v)$ (resp. $W_{i+1}(v)$) a monochromatic walk starting from
  $e_i$ (resp. $e_{i+1}$), obtained by following outgoing edges of
  color $i$ (resp. $i+1$).  By minimality of $C$ those walks are not
  contained in $D$. We hence have that $W_i(v)\setminus v$ and
  $W_{i+1}(v)\setminus v$ intersect $C$. Thus each of these walks
  contains a non-empty subpath from $v$ to $C$. The union of these two
  paths, plus a part of $C$ contradicts the minimality of $C$.
\end{proof}

Let $v$ be a vertex of $G^\infty$. For each color $i$, vertex $v$ is
the starting vertex of some walks of color $i$, we denote the union of
these walks by $P_i(v)$. Every vertex has at least one outgoing edge
of color $i$ and the set $P_i(v)$ is obtained by following all these
edges of color $i$ starting from $v$.  The analogous of
Lemma~\ref{lem:nocommon-plan} is:

\begin{lemma}
  \label{lem:nocommongeneral}
  For every vertex $v$ and $i,j\in\{0,1,2\}$, $i\neq j$, the two
  graphs $P_{i}(v)$ and $P_{j}(v)$ have $v$ has only common
  vertex.
\end{lemma}

\begin{proof}
  If $P_{i}(v)$ and $P_{j}(v)$ intersect on two vertices, then
  $G^\infty_{i}\cup (G_{j}^{^\infty})^{-1}$ contains a cycle,
  contradicting Lemma~\ref{lem:nodirectedcycle}.
\end{proof}

Now  we can prove the following:

\begin{lemma}
\label{lem:conjessentially}
If a map $G$ on a genus $g\geq 1$ orientable surface admits a
generalized Schnyder wood, then $G$ is essentially 3-connected.
\end{lemma}

\begin{proof}
  Suppose by contradiction that there exist two vertices $x,y$ of
  $G^\infty$ such that $G'=G^\infty\setminus\{x,y\}$ is not
  connected. Then, by Lemma~\ref{lem:finitecc}, the graph $G'$ has a
  finite connected component $R$. Let $v$ be a vertex of $R$. By
  Lemma~\ref{lem:nodirectedcycle}, for $0\leq i \leq 2$, the graph
  $P_{i}(v)$ does not lie in $R$ so it intersects either $x$ or
  $y$. So for two distinct colors $i,j$, the two graphs
  $P_{i}(v)$ and $P_{j}(v)$ intersect in a vertex distinct from $v$,
  a contradiction to Lemma~\ref{lem:nocommongeneral}.
\end{proof}

\section{Conjectures on the existence of  Schnyder woods}
\label{sec:conjectureexistence}

Proving that every triangulation on a genus $g\geq 1$ orientable
surface admits a 1-{\sc edge} angle labeling would imply the following
theorem of Bar\'at and Thomassen~\cite{BT06}:

\begin{theorem}[\cite{BT06}]
\label{th:barat}
A simple triangulation on a genus $g\geq 1$ orientable surface admits
an orientation of its edges such that every vertex has outdegree
divisible by $3$.
\end{theorem}

Recently, Theorem~\ref{th:barat} has been improved
by Albar, Gonçalves and Knauer~\cite{AGK14}:

\begin{theorem}[\cite{AGK14}]
\label{th:AGK}
A simple triangulation on a genus $g\geq 1$ orientable
surface admits an orientation of its edges such that every vertex has
outdegree at least $3$, and divisible by $3$.
\end{theorem}

Note that Theorems~\ref{th:barat} and~\ref{th:AGK} are proved only in
the case of simple triangulations (i.e. no loops and no multiple
edges). We believe them to be true also for non-simple triangulations
without contractible loops nor homotopic multiple edges.

Theorem~\ref{th:AGK} suggests the existence of 1-{\sc edge} angle
labelings with no sinks, i.e. 1-{\sc edge}, $\mathbb{N}^*$-{\sc
  vertex} angle labelings. One can easily check that in a
triangulation, a 1-{\sc edge} angle labeling is also 1-{\sc
  face}. Thus one can hope that a triangulation on a genus $g\geq 1$
orientable surface admits a 1-{\sc edge}, $\mathbb{N}^*$-{\sc vertex},
1-{\sc face} angle labeling. Note that a 1-{\sc edge}, 1-{\sc face}
angle labeling of a map implies that faces have size three. So we
propose the following conjecture, whose ``only if'' part follows from the
previous sentence:

\begin{conjecture}
\label{conjecture}
  A map on a genus $g\geq 1$ orientable surface admits a
  1-{\sc edge}, $\mathbb{N}^*$-{\sc vertex}, 1-{\sc face} angle
  labeling if and only if it is a triangulation.
\end{conjecture}

If true, Conjecture~\ref{conjecture} would strengthen
Theorem~\ref{th:AGK} in two ways. First, it considers more
triangulations (not only simple ones). Second, it requires the
coloring property of Figure~\ref{fig:369} around vertices.

How about general maps? We propose the following conjecture, whose
``only if'' part is implied by Proposition~\ref{prop:bijtoregen}
and Lemma~\ref{lem:conjessentially}:

\begin{conjecture}
\label{conjecture2}
A map on a genus $g\geq 1$ orientable surface admits an {\sc edge},
$\mathbb{N}^*$-{\sc vertex}, $\mathbb{N}^*$-{\sc face} angle labeling
if and only if it is essentially 3-connected.
\end{conjecture}

Note that the graph of Figure~\ref{fig:annoying-maps-small} is
essentially 3-connected and has no $\{1,2\}$-{\sc edge} angle
labeling. So Conjecture~\ref{conjecture2} is false without edges of
type 0. This explains why we have introduced edges of type 0 that do
not appear in the planar case.

Conjecture~\ref{conjecture2} implies Conjecture~\ref{conjecture} since
for a triangulation every face would be of type 1, and thus every edge
would be of type 1.  The case $g=1$ of Conjecture~\ref{conjecture2} is
proved in this manuscript (see Part~\ref{part:torus}) whereas both
conjectures are open for $g\geq 2$.

\chapter{Characterization of Schnyder orientations}
\label{sec:characterization}

\section{A bit of homology}
\label{sec:homology}

In the next sections, we need a bit of surface homology of general
maps, which we discuss now. For a deeper introduction to homology we
refer to~\cite{Gib10}.

For the sake of generality, in this subsection we consider that maps
may have loops or multiple edges.  Consider a map $G=(V,E)$ on an
orientable surface of genus $g$, given with an arbitrary orientation
of its edges. This fixed arbitrary orientation is implicit in all the
paper and is used to handle flows.  A \emph{flow} $\phi$ on $G$ is a
vector in $\mb{Z}^{|E|}$. For any $e\in E$, we denote by $\phi_e$ the
coordinate $e$ of $\phi$.

A \emph{walk} $W$ of $G$ is a sequence of edges with a direction of
traversal such that the ending point of an edge is the starting
point of the next edge.  A walk is \emph{closed} if the start and end
vertices coincide. A walk has a \emph{characteristic flow} $\phi(W)$
defined by:
$$\phi(W)_e:=\#\text{times }W\text{ traverses } e \text{ forward} - \#\text{times
}\\
W\text{ traverses } e \text{ backward}$$

This definition naturally extends to sets of walks.  From now on we
consider that a set of walks and its characteristic flow are the same
object and by abuse of notation we can write $W$ instead of
$\phi(W)$. We do the same for \emph{oriented subgraphs},
i.e. subgraphs that can be seen as a set of walks.

A \emph{facial walk} is a closed walk bounding a face.  Let $\mc{F}$
be the set of counterclockwise facial walks and let
$\mb{F}={<}\phi(\mc{F}){>}$ the subgroup of $\mb{Z}^E$ generated by
$\mc{F}$.  Two flows $\phi, \phi'$ are \emph{homologous} if
$\phi -\phi' \in \mb{F}$.  They are \emph{reversely homologous} if
$\phi +\phi' \in \mb{F}$. They are \emph{weakly homologous} if they
are homologous or reversely homologous.  We say that a flow $\phi$ is
$0$-homologous if it is homologous to the zero flow, i.e.
$\phi \in \mb{F}$.

Let $\mc{W}$ be the set of \emph{closed} walks and let
$\mb{W}={<}\phi(\mc{W}){>}$ the subgroup of $\mb{Z}^E$ generated by
$\mc{W}$.  The group $H(G)=\mb{W}/\mb{F}$ is the \emph{first homology
  group} of $G$. It is well-known that $H(G)\cong\mb{Z}^{2g}$ only
depends on the genus of the map.  

A set $\{B_1,\ldots,B_{2g}\}$ of (closed) walks of $G$ is said to be a
\emph{homology-basis} if the equivalence classes of their
characteristic vectors $\{[\phi(B_1)],\ldots,[\phi(B_{2g})]\}$ generate
$H(G)$.  Then for any closed walk $W$ of $G$, we have
$W=\sum_{F\in\mc{F}}\lambda_FF+\sum_{1\leq i\leq 2g}\mu_iB_i$ for some
$\lambda\in\mathbb{Z}^{f},\mu\in\mathbb{Z}^{2g}$. Moreover one of the
$\lambda_F$ can be set to zero (and then all the other coefficients
are unique).

For any map, there exists a set of cycles that forms a homology-basis
and it is computationally easy to build. A possible way to do this is
by considering a spanning tree $T$ of $G$, and a spanning tree $T^*$
of $G^*$ that contains no edges dual to $T$.  By Euler's formula,
there are exactly $2g$ edges in $G$ that are not in $T$ nor dual to
edges of $T^*$. Each of these $2g$ edges forms a unique cycle with
$T$. It is not hard to see that this set of cycles, given with any
direction of traversals, forms a homology-basis. Moreover, note that
the intersection of any pair of these  cycles is either a single
vertex or a common path.

The edges of the dual map $G^*$ of $G$ are oriented such that the dual
edge $e^*$ of an edge $e$ of $G$ goes from the face on the right of $e$ to
the face on the left of $e$. 
Let $\mc{F}^*$ be the set of counterclockwise facial walks of $G^*$.
Consider $\{B^*_1,\ldots,B^*_{2g}\}$ a set of cycles
of $G^*$ that
form a homology-basis. Let $p$ be a flow of $G$ and $d$ a flow
of $G^*$. We define the following: $$\beta(p,d)=\sum_{e\in G}p_e
d_{e^*}$$ Note that $\beta$ is a bilinear function.

\begin{lemma}\label{lm:homologous}
Given  two flows $\phi,\phi'$ of $G$, the following properties are
equivalent to each other:

\begin{enumerate}
\item The two flows $\phi, \phi'$ are homologous.
\item For any closed walk $W$ of $G^*$ we have
  $\beta(\phi,W)=\beta(\phi',W)$.
\item For any $F\in \mc{F^*}$, we have $\beta(\phi,F)=\beta(\phi',F)$, and,
  for any $1\le i\le 2g$, we have $\beta(\phi,B^*_i)=\beta(\phi',B^*_i)$.
\end{enumerate}
\end{lemma}

\begin{proof}
  $(1. \Longrightarrow 3.)$ Suppose that $\phi, \phi'$ are
  homologous. Then we have
  $\phi-\phi'=\sum_{F\in\mc{F}}\lambda_FF$ for some
  $\lambda\in\mathbb{Z}^f$. It is easy to see that, for any
  closed walk $W$ of $G^*$, a facial walk $F\in\mc{F}$ satisfies
  $\beta(F,W)=0$, so $\beta(\phi,W)=\beta(\phi',W)$ by linearity of
  $\beta$.
 
  $(3. \Longrightarrow 2.)$ Suppose that for any $F\in \mc{F^*}$, we
  have $\beta(\phi,F)=\beta(\phi',F)$, and, for any $1\le i\le 2g$, we
  have $\beta(\phi,B^*_i)=\beta(\phi',B^*_i)$. Let $W$ be any closed
  walk of $G^*$. We have
  $W=\sum_{F\in\mc{F}^*}\lambda_FF+\sum_{1\leq i\leq 2g}\mu_iB^*_i$
  for some $\lambda\in\mathbb{Z}^{f},\mu\in\mathbb{Z}^{2g}$. Then by
  linearity of $\beta$ we have   $\beta(\phi,W)=\beta(\phi',W)$.

  $(2. \Longrightarrow 1.)$ Suppose
  $\beta(\phi,W)=\beta(\phi',W)$ for any closed walk $W$ of $G^*$. Let
  $z=\phi-\phi'$.  Thus $\beta(z,W)=0$ for any closed walk $W$ of
  $G^*$.  We label the faces of $G$ with elements of $\mathbb{Z}$ as
  follows. Choose an arbitrary face $F_0$ and label it $0$. Then,
  consider any face $F$ of $G$ and a path $P_{F}$ of $G^*$ from $F_0$
  to $F$. Label $F$ with $\ell_F=\beta(z,P_{F})$. Note that the label
  of $F$ is independent from the choice of $P_{F}$. Indeed, for any
  two paths $P_1,P_2$ from $F_0$ to $F$, we have $P_1-P_2$ is a closed
  walk, so $\beta(z,P_1-P_2)=0$ and thus $\beta(z,P_1)=\beta(z,P_2)$.
  Let us show that $z=\sum_{F\in\mc{F}} \ell_F \phi(F)$.

\begin{alignat*}{3}
  \sum_{F\in\mc{F}} \ell_F \phi(F) &=\sum_{e\in G}
  \left(\ell_{F_2}-\ell_{F_1}\right) \phi(e) &
  \text{(face $F_2$ is on the left of $e$ and $F_1$ on the right)}\\
  &=\sum_{e\in G}
  \left(\beta(z,P_{F_2})-\beta(z,P_{F_1})\right)
  \phi(e) &\text{(definition of $\ell_F$)}\\
  &=\sum_{e\in G} \beta(z,P_{F_2}-P_{F_1}) \phi(e)  &\text{(linearity of $\beta$)}\\
  &=\sum_{e\in G} \beta(z,e^*) \phi(e) &\text{($P_{F_1}+e^*-P_{F_2}$ is a closed walk)}\\
  & =\sum_{e\in G} \left(\sum_{e'\in G} z_{e'} \phi(e^*)_{e'^*}\right)
  \phi(e) &\text{(definition of $\beta$)}\\
  &=\sum_{e\in G} z_e \phi(e)\\
  &=z
\end{alignat*}

So $z\in  \mb{F}$ and thus $\phi, \phi'$ are homologous.
\end{proof}

\section{General characterization}
\label{sec:gencharacterization}

By a result of De Fraysseix and Ossona de Mendez~\cite{FO01}, there is
a bijection between orientations of the internal edges of a planar
triangulation where every inner vertex has outdegree $3$ and Schnyder
woods. Thus, any orientation with the proper outdegree corresponds to
a Schnyder wood. This is not true in higher genus as already in the
torus, there exist orientations that do not correspond to any Schnyder
wood (see Figure~\ref{fig:orientation}).  In this section, we
characterize orientations that correspond to Schnyder angle labelings.

\begin{figure}[!h]
\center
\includegraphics[scale=0.5]{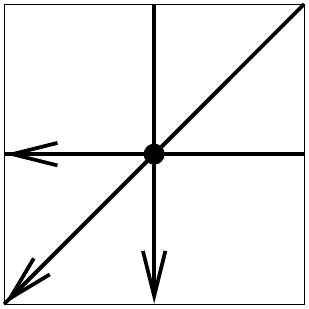} 
\ \ \ \ 
\includegraphics[scale=0.5]{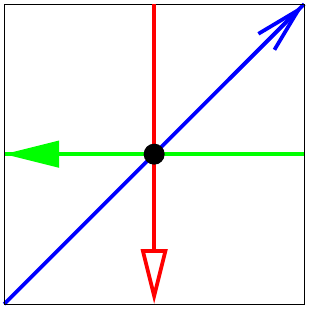}
\caption{Two different orientations of a toroidal triangulation. Only
  the one on the right corresponds to a Schnyder wood.}
\label{fig:orientation}
\end{figure}

Consider a map $G$ on an orientable surface of genus $g$.  The mapping
of Figure~\ref{fig:edgelabeling} shows how a Schnyder labeling of $G$
can be mapped to an orientation of the edges with edges oriented in
one direction or in two opposite directions.  These edges can be
defined more naturally in the primal-dual-completion of $G$.

The \emph{primal-dual-completion} $\pdc{G}$ is the map obtained from
simultaneously embedding $G$ and $G^*$ such that vertices of $G^*$ are
embedded inside faces of $G$ and vice-versa. Moreover, each edge
crosses its dual edge in exactly one point in its interior, which also
becomes a vertex of $\pdc{G}$.  Hence, $\pdc{G}$ is a bipartite graph
with one part consisting of \emph{primal-vertices} and
\emph{dual-vertices} and the other part consisting of
\emph{edge-vertices} (of degree $4$). Each face of $\pdc{G}$ is a
quadrangle incident to one primal-vertex, one dual-vertex and two
edge-vertices. Actually, the faces of $\pdc{G}$ are in correspondence
with the angles of $G$. This means that angle labelings of $G$
correspond to face labelings of $\pdc{G}$.

Given $\alpha:V\to \mb{N}$, an orientation of $G$ is an
\emph{$\alpha$-orientation}~\cite{Fel04} if for every vertex $v\in V$ its
outdegree $d^+(v)$ equals $\alpha(v)$.  We call an orientation of
$\pdc{G}$ a \emph{$\bmod_3$-orientation} if it is an
$\alpha$-orientation for a function $\alpha$ satisfying
: $$\alpha(v)= \begin{cases}
  0\bmod 3 & \text{if }v\text{ is a primal- or dual-vertex}, \\
  1\bmod 3& \text{if }v\text{ is an edge-vertex.}\\
\end{cases}
$$

Note that an Schnyder labeling of $G$ corresponds to a
$\bmod_3$-orientation of $\pdc{G}$, by the mapping of
Figure~\ref{fig:pdc}, where the three types of edges are
represented. Indeed, type 0 corresponds to an edge-vertex of outdegree
$4$. Type 1 and type 2 both correspond to an edge-vertex of outdegree
$1$; in type 1 (resp. type 2) the outgoing edge goes to a
primal-vertex (resp. dual-vertex). In all cases we have
$d^+(v) = 1\bmod 3$ if $v$ is an edge-vertex. By
Lemma~\ref{lem:EDGElabeling}, the labeling is also {\sc vertex} and
{\sc face}. Thus $d^+(v) = 0\bmod 3$ if $v$ is a primal- or
dual-vertex.

\begin{figure}[!h]
\center
\begin{tabular}{ccc}
\includegraphics[scale=0.5]{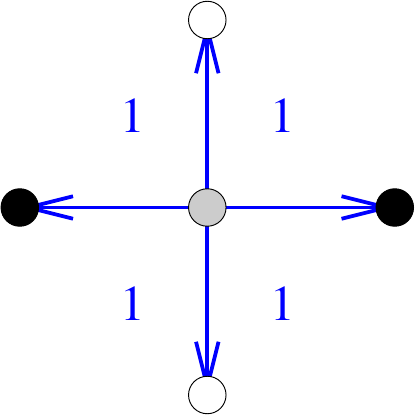} \ \ \  &  \ \ \
\includegraphics[scale=0.5]{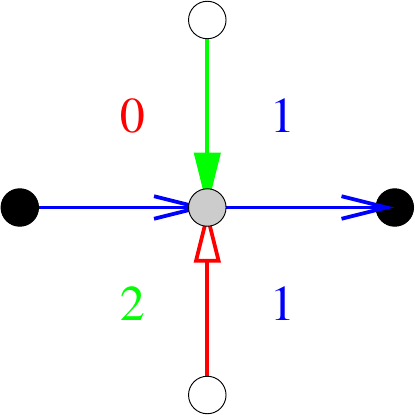} \ \ \  &  \ \ \
\includegraphics[scale=0.5]{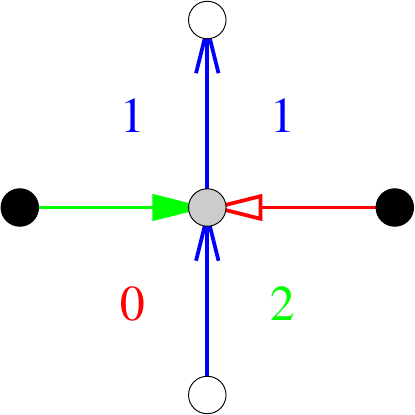} \\
Type 0  \ \ \  &  \ \ \ Type 1  \ \ \  &  \ \ \ Type 2  \\
\end{tabular}
\caption{How to map a Schnyder labeling to a
  $\bmod_3$-orientation of the primal-dual completion. Primal-vertices
  are black, dual-vertices are white and edge-vertices are gray. This
  serves as a convention for the other figures.}
\label{fig:pdc}
\end{figure}

Figure~\ref{fig:pdc-example-tore} represent the primal-dual completion
of the toroidal map of Figure~\ref{fig:example-dual-tore} with a
Schnyder labeling and the corresponding orientation and coloring of
its edges. Note that it corresponds to a superposition of the primal
and dual Schnyder woods of Figure~\ref{fig:example-dual-tore}.

\begin{figure}[!h]
\center
\includegraphics[scale=0.5]{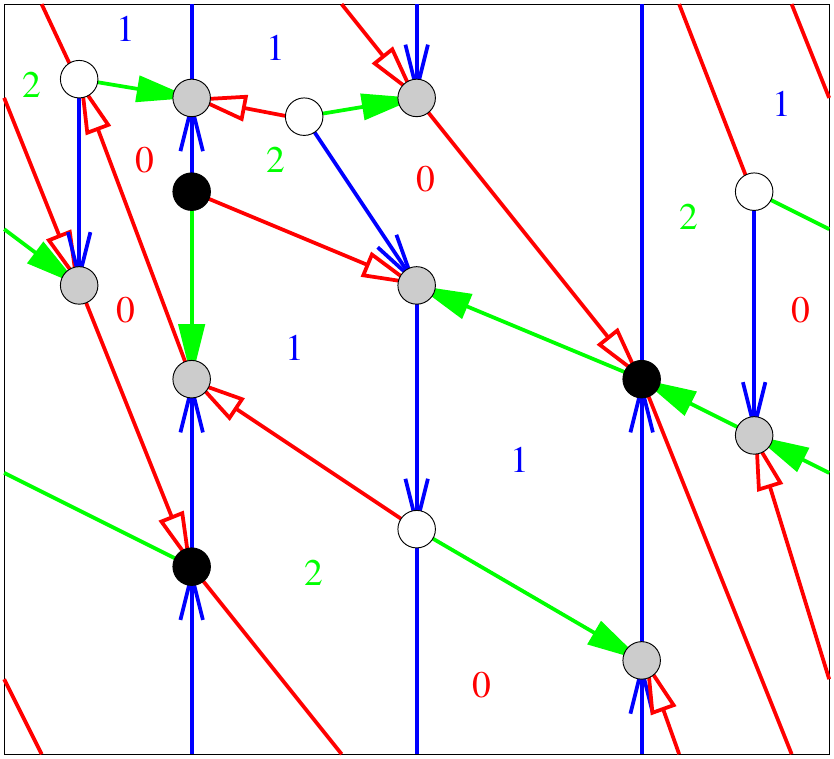}
\caption{Primal-dual completion of a toroidal map given with a
  Schnyder labeling and the corresponding orientation and coloring.}
\label{fig:pdc-example-tore}
\end{figure}

As mentioned earlier, De Fraysseix and Ossona de
Mendez~\cite{FO01} give a bijection between internal 3-orientations
and Schnyder woods of planar triangulations. Felsner~\cite{Fel04}
generalizes this result for planar Schnyder woods and orientations of
the primal-dual completion having prescribed out-degrees: 
$3$ for primal- or dual-vertices and $1$ for edges-vertices. The
situation is more complicated in higher genus
(see Figure~\ref{fig:orientation}). It is not enough to prescribe
outdegrees in order to characterize orientations corresponding to
Schnyder  labelings.

We call an orientation of $\pdc{G}$ corresponding to a Schnyder
labeling of $G$ a \emph{Schnyder orientation}.  Note that with our
definition of generalized Schnyder wood, a Schnyder orientation
corresponds to a Schnyder wood if and only if $d^+(v)>0$ for any
primal- or dual-vertex $v$ (i.e. vertex or face of type 0 are allowed in
Schnyder labelings/orientations but not in Schnyder woods).  In this
section we characterize which orientations of $\pdc{G}$ are Schnyder
orientations.

Consider an orientation of the primal-dual completion $\pdc{G}$.  Let
$\Out=\{(u,v)\in E(\pdc{G})\mid v\text{ is an edge-vertex}\}$, i.e.
the set of edges of $\pdc{G}$ which are going from a primal- or
dual-vertex to an edge-vertex. We call these edges \emph{out-edges}.
For $\phi$ a flow of the dual of the primal-dual completion
$\pdc{G}^*$, we define $\delta(\phi)=\beta(\Out,\phi)$.  More
intuitively, if $W$ is a walk of $\pdc{G}^*$, then:
$$
\begin{array}{ll}
  \delta(W)  = &  \ \ \#\text{out-edges crossing }W\text{
    from left to right}\\
  & -\#\text{out-edges crossing }W\text{ from right to
    left}.  
\end{array}
$$

The bilinearity of $\beta$ implies the linearity of $\delta$.

The following lemma gives a necessary and sufficient condition
for an orientation to be a Schnyder orientation.

\begin{lemma}
 \label{lem:charforall}
 An orientation of $\pdc{G}$ is a Schnyder orientation if and only if
 any closed walk $W$ of $\pdc{G}^*$ satisfies
 $\delta (W) = 0 \bmod 3$.
\end{lemma}

\begin{proof} $(\Longrightarrow)$ Consider an {\sc edge} angle
  labeling $\ell$ of $G$ and the corresponding Schnyder orientation
  (see Figure~\ref{fig:pdc}).
  Figure~\ref{fig:figdelta} illustrates how $\delta$ counts the
  variation of the label when going from one face of $\pdc{G}$ to
  another face of $\pdc{G}$ . The represented cases
  correspond to a walk $W$ of $\pdc{G}^*$ consisting of just one
  edge. If the edge of $\pdc{G}$ crossed by $W$ is not an out-edge,
  then the two labels in the face are the same and $\delta(W) =0$. If
  the edge crossed by $W$ is an out-edge, then the labels differ by
  one. If $W$ is going counterclockwise around a primal- or
  dual-vertex, then the label increases by $1 \bmod 3$ and
  $\delta(W)=1$. If $W$ is going clockwise around a primal- or
  dual-vertex then the label decreases by $1 \bmod 3$ and
  $\delta(W)=-1$. One can check that this is consistent with all the
  edges depicted in Figure~\ref{fig:pdc}.  Thus for any walk $W$ of
  $\pdc{G}^*$ from a face $F$ to a face $F'$, the value of $\delta(W)
  \bmod 3$ is equal to $\ell(F')-\ell(F)\bmod 3$.  Thus if $W$ is a
  closed walk then $\delta(W)= 0 \bmod 3$.

\begin{figure}[!h]
\center
\begin{tabular}{cc}
\includegraphics[scale=0.5]{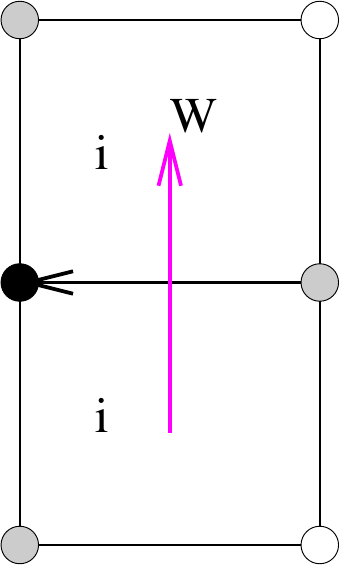} \ \ \  &  \ \ \
\includegraphics[scale=0.5]{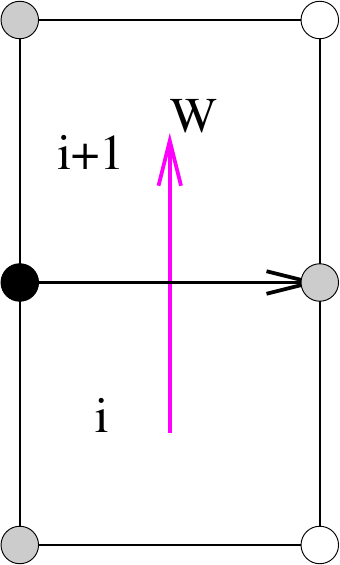} \\
$\delta(W) =0$ \ \ \  &  \ \ \ $\delta(W)=1$ \\
\end{tabular}
\caption{How $\delta$ counts the variation of the labels.}
\label{fig:figdelta}
\end{figure}

$(\Longleftarrow)$ Consider an orientation of $\pdc{G}$ such that
any closed
walk $W$ of $\pdc{G}^*$ satisfies $\delta (W) = 0 \bmod 3$.  Pick any
face $F_0$ of $\pdc{G}$ and label it $0$. Consider any face $F$ of
$\pdc{G}$ and a path $P$ of $\pdc{G}^*$ from $F_0$ to $F$. Label $F$
with the value $\delta (P)\bmod 3$. Note that the label of $F$ is
independent from the choice of $P$ as for any two paths $P_1, P_2$
going from $F_0$ to $F$, we have $\delta (P_1) =\delta (P_2) \bmod 3$
since $\delta (P_1 - P_2) = 0 \bmod 3$ as $P_1- P_2$ is a closed walk.

Consider an edge-vertex $v$ of $\pdc{G}$ and a walk $W$ of $\pdc{G}^*$
going \cw around $v$. By assumption  $\delta (W) = 0 \bmod
3$ and $d(v)=4$ so $d^+(v)=1\bmod 3$.
 One can check (see Figure~\ref{fig:pdc}) that around an
edge-vertex $v$ of outdegree $4$, all the labels are the same and thus
$v$ corresponds to an edge of $G$ of type 0. One can also check that
around an edge-vertex $v$ of outdegree $1$, the labels are in clockwise
order, $i-1$, $i$, $i$, $i+1$ for some $i$ in $\{0,1,2\}$ where the
two faces with the same label are incident to the outgoing edge of
$v$. Thus $v$ corresponds to an edge of $G$ of type 1 or 2 depending
on the fact that the outgoing edge reaches a primal- or a
dual-vertex. So the obtained labeling of the faces of $\pdc{G}$
corresponds to an {\sc edge} angle labeling of $G$ and the considered
orientation is a Schnyder orientation.
\end{proof}

We now study properties of $\delta$ w.r.t.~homology in order to
simplify the condition of Lemma~\ref{lem:charforall} that concerns any
closed walk of $\pdc{G}^*$.  

Let
$\hat{\mc{F}^*}$ be the set of counterclockwise facial walks of
$\pdc{G}^*$.

\begin{lemma}
\label{lem:facedelta0}
In a $\bmod_3$-orientation of $\pdc{G}$, any $F\in\hat{\mc{F}^*}$
satisfies $\delta(F)=0\bmod 3$.
\end{lemma}
\begin{proof}
  If $F$ corresponds to an edge-vertex $v$ of $\pdc{G}$, then $v$ has
  degree exactly $4$ and outdegree $1$ or $4$ by definition of
  $\bmod_3$-orientations. So there are exactly $0$ or $3$ out-edges
  crossing $F$ from right to left, and $\delta(F)=0\bmod 3$.

 If $F$ corresponds to a primal- or dual-vertex $v$, then $v$ has
 outdegree $0 \bmod 3$ by definition of $\bmod_3$-orientations. So
 there are exactly $0 \bmod 3$ out-edges crossing $F$ from left to
 right, and $\delta(F)=0\bmod 3$.
\end{proof}

\begin{lemma}
\label{lem:basedelta}
In a $\bmod_3$-orientation of $\pdc{G}$, if $\{B_1,\ldots,B_{2g}\}$ is a
set of cycles
of $\pdc{G}^*$ that forms a homology-basis, then for any closed walk
$W$ of $\pdc{G}^*$ homologous to
$\mu_1 B_1 + \cdots + \mu_{2g} B_{2g}$, $\mu\in\mathbb{Z}^{2g}$ we have
$\delta(W)= \mu_1 \delta(B_1) + \cdots + \mu_{2g} \delta(B_{2g}) \bmod
3$.
\end{lemma}

\begin{proof} We have
  $W=\sum_{F\in\hat{\mc{F}^*}}\lambda_FF+\sum_{1\leq i\leq
    2g}\mu_iB_i$
  for some $\lambda\in\mathbb{Z}^{f}$, $\mu\in\mathbb{Z}^{2g}$.  Then by linearity of $\delta$
  and Lemma~\ref{lem:facedelta0}, the lemma follows.
\end{proof}

Lemma~\ref{lem:basedelta} can be used to simplify the condition of
Lemma~\ref{lem:charforall} and show that if $\{B_1,\ldots,B_{2g}\}$ is a
set of cycles
of $\pdc{G}^*$ that forms a homology-basis, then an
orientation of $\pdc{G}$ is a Schnyder orientation if and only if it
is a $\bmod_3$-orientation such that $\delta(B_{i})= 0 \bmod 3$, for
all $1\le i\le 2g$. Now we define a simpler function $\gamma$ that is used
to formulate a similar characterization result but with $\gamma$ (see
Theorem~\ref{th:characterizationgamma}).

Consider a (not necessarily directed) cycle $C$ of $G$ or $G^*$
together with a direction of traversal. We associate to $C$ its
corresponding cycle in $\pdc{G}$ denoted by $\pdc{C}$.  We define
$\gamma(C)$ by:
$$\gamma (C) = \#\ \text{edges of $\pdc{G}$ leaving $\pdc{C}$ on its right} -
\#\  \text{edges of $\pdc{G}$ leaving $\pdc{C}$ on its left}$$

Since it considers cycles of $\pdc{G}$ instead of walks of
$\pdc{G}^*$, it is easier to deal with parameter $\gamma$ rather than
parameter $\delta$.  However $\gamma$ does not enjoy the same property
w.r.t.~homology as $\delta$.  For homology we have to consider walks
as flows, but two walks going several time through a given vertex may
have the same characteristic flow but different $\gamma$.  This
explains why $\delta$ is defined first. Now we adapt the results for
$\gamma$.

The value of $\gamma$ is related to $\delta$ by the next lemmas. Let
$C$ be a cycle of $G$ or $G^*$ with a direction of traversal. Let
$W_L(C)$ be the closed walk of $\pdc{G}^*$ just on the left of $C$ and
going in the same direction as $C$ (i.e. $W_L(C)$ is composed of the
dual edges of the edges of $\pdc{G}$ incident to the left of
$\pdc{C}$).  Note that since the faces of $\pdc{G}^*$ have exactly one
incident vertex that is a primal-vertex, walk $W_L(C)$ is in fact a
cycle of $\pdc{G}^*$. Similarly, let $W_R(C)$ be the cycle of
$\pdc{G}^*$ just on the right of $C$.

\begin{lemma}
\label{lem:gammaequaldelta}
Consider an orientation of $\pdc{G}$ and 
a cycle $C$ of $G$,
% or $G^*$,
 then $\gamma(C) = \delta (W_L(C)) + \delta (W_R(C))$.
\end{lemma}

\begin{proof}
We consider the different cases that can
  occur. An edge that is entering a primal-vertex of $\pdc{C}$, is not
  counting in either $\gamma(C),\delta (W_L(C)), \delta (W_R(C))$.  An
  edge that is leaving a primal-vertex of $\pdc{C}$ from its right
  side (resp. left side) is counting $+1$ (resp. $-1$) for $\gamma(C)$
  and $\delta (W_R(C))$ (resp.  $\delta (W_L(C))$).

  For edges incident to edge-vertices of $\pdc{C}$ both sides have to
  be considered at the same time. Let $v$ be an edge-vertex of
  $\pdc{C}$. Vertex $v$ is of degree $4$ so it has exactly two edges
  incident to $\pdc{C}$ and not on $C$. One of these edge, $e_L$, is
  on the left side of $\pdc{C}$ and dual to an edge of $W_L(C)$. The
  other edge, $e_R$, is on the right side of $\pdc{C}$ and dual to an
  edge of $W_R(C)$. If $e_L$ and $e_R$ are both incoming edges for
  $v$, then $e_R$ (resp. $e_L$) is counting $-1$ (resp. $+1$) for
  $\delta (W_R(C))$ (resp.  $\delta (W_L(C))$) and not counting for
  $\gamma(C)$. If $e_L$ and $e_R$ are both outgoing edges for $v$,
  then $e_R$ and $e_L$ are not counting for both $\delta (W_R(C))$,
  $\delta (W_L(C))$ and sums to zero for $\gamma(C)$. If $e_L$ is
  incoming and $e_R$ is outgoing for $v$, then $e_R$ (resp. $e_L$) is
  counting $0$ (resp. $+1$) for $\delta (W_R(C))$ (resp.
  $\delta (W_L(C))$), and counting $+1$ (resp. $0$) for
  $\gamma(C)$. The last case, $e_L$ is outgoing and $e_R$ is incoming,
  is symmetric and one can see that in the four cases we have that
  $e_L$ and $e_R$ count the same for $\gamma(C)$ and
  $\delta (W_L(C)) + \delta (W_R(C))$.  Thus finally
  $\gamma(C)=\delta (W_L(C)) + \delta (W_R(C))$.
\end{proof}

\begin{lemma}
\label{lem:deltagammaequ0mod3}
In a $\bmod_3$-orientation of $G$, a cycle $C$ of $G$ satisfies
$$\delta (W_L(C))=0 \bmod 3 \ \ \text{and}\ \ \delta (W_R(C))=0 \bmod 3 \iff
\gamma(C) = 0 \bmod 3
$$ 
\end{lemma}

\begin{proof}
($\Longrightarrow$) Clear by Lemma~\ref{lem:gammaequaldelta}.

($\Longleftarrow$) Suppose that $\gamma(C)=0 \bmod 3$.  Let $x_L$
(resp. $y_L$) be the number of edges of $\pdc{G}$ that are dual to edges
of $W_L(C)$, that are outgoing for a primal-vertex of $\pdc{C}$
(resp. incoming for an edge-vertex of $\pdc{C}$). Similarly, let $x_R$
(resp. $y_R$) be the number of edges of $\pdc{G}$ that are dual to edges
of $W_R(C)$, that are outgoing for a primal-vertex of $\pdc{C}$
(resp. incoming for an edge-vertex of $\pdc{C}$).  So
$\delta (W_L(C))=y_L-x_L$ and $\delta (W_R(C))=x_R-y_R$. So by
Lemma~\ref{lem:gammaequaldelta},
$ \gamma(C)= \delta (W_L(C))+\delta (W_R(C))=(y_L+x_R)-(x_L+y_R) =0
\bmod 3$.

Let $k$ be the number of vertices of $C$. So $\pdc{C}$ has $k$
primal-vertices, $k$ edge-vertices and $2k$ edges.  Edge-vertices have
degree $4$, one incident edge on each side of  $\pdc{C}$ and outdegree
$1\bmod 3$. So their total number of outgoing edges not in $\pdc{C}$ is
$2k-(y_L+y_R)$ and their total number of outgoing edges in $\pdc{C}$ is
$k-2k+(y_L+y_R) \bmod 3$.  Primal-vertices have outdegree $(0\bmod 3)$
so their total number of outgoing edges on $\pdc{C}$ is
$-(x_L+x_R) \bmod 3$. So in total $2k=(y_L+y_R)-k-(x_L+x_R) \bmod 3$.
So $(y_L+y_R)-(x_L+x_R)=0 \bmod 3$. By combining this with plus
(resp. minus) $(y_L+x_R)-(x_L+y_R) =0 \bmod 3$, one obtains that
$2\delta (W_L(C))=2(y_L-x_L)=0 \bmod 3$ (resp.
$2\delta (W_R(C))=2(x_R-y_R)=0 \bmod 3$). Since $\delta (W_L(C))$ and
$\delta (W_R(C))$ are integer one obtains $\delta (W_L(C))=0 \bmod 3$
and $\delta (W_R(C))=0 \bmod 3$.
\end{proof}

Finally we have the following characterization theorem concerning
Schnyder orientations:

\begin{theorem}
\label{th:characterizationgamma}
Consider a map $G$ on an orientable surface of genus $g$. Let
$\{B_1,\ldots,B_{2g}\}$ be a set of cycles of $G$ that forms a
homology-basis.  An orientation of $\pdc{G}$ is a Schnyder orientation
if and only if it is a $\bmod_3$-orientation such that
$\gamma(B_{i})= 0 \bmod 3$, for all $1\le i\le 2g$.
\end{theorem}

\begin{proof}
  $(\Longrightarrow)$ Consider an {\sc edge} angle
  labeling $\ell$ of $G$ and the corresponding Schnyder orientation
  (see Figure~\ref{fig:pdc}). Type~0 edges correspond to edge-vertices
  of outdegree $4$, while type~1 and~2 edges correspond to
  edge-vertices of outdegree $1$. Thus $d^+(v)= 1\bmod 3$ if $v$ is an
  edge-vertex.
By Lemma~\ref{lem:EDGElabeling}, the labeling
  is {\sc vertex} and {\sc face}. Thus $d^+(v) = 0\bmod 3$ if $v$ is a
  primal- or dual-vertex. So the orientation is a
  $\bmod_3$-orientation. By Lemma~\ref{lem:charforall}, we have
  $\delta (W) = 0 \bmod 3$ for any closed walk $W$ of $\pdc{G}^*$. So
  we have $\delta(W_L(B_1)), \ldots, \delta(W_L(B_{2g}))$,
  $\delta(W_R(B_1)), \ldots, \delta(W_R(B_{2g}))$ are all equal to
  $0 \bmod 3$.  Thus, by Lemma~\ref{lem:deltagammaequ0mod3}, we have
  $\gamma(B_{i}) = 0 \bmod 3$, for all $1\le i\le 2g$.

  $(\Longleftarrow)$ Consider a $\bmod_3$-orientation of $G$ such that
  $\gamma(B_{i}) = 0 \bmod 3$, for all $1\le i\le 2g$.  By
  Lemma~\ref{lem:deltagammaequ0mod3}, we have 
  $\delta(W_L(B_{i}))= 0 \bmod 3$ for all $1\le i\le 2g$. Moreover
  $\{W_L(B_1),\ldots,W_L(B_{2g})\})$ forms a homology-basis.  So
  by Lemma~\ref{lem:basedelta}, $\delta (W) = 0 \bmod 3$ for any
  closed walk $W$ of $\pdc{G}^*$. So the orientation is a Schnyder
  orientation by Lemma~\ref{lem:charforall}.
\end{proof}

The condition of Theorem~\ref{th:characterizationgamma} is easy to
check: choose $2g$ cycles that form a homology-basis and check
whether $\gamma$ equals $0\bmod 3$ for each of them.

When restricted to triangulations and to edges of type $1$ only, the
definition of $\gamma$ can be simplified. Consider a triangulation $G$
on an orientable surface of genus $g$ and an orientation of the edges
of $G$.
Figure~\ref{fig:trimap} shows how
 to transform the orientation of $G$ into an orientation of
 $\pdc{G}$. Note that all the edge-vertices have outdegree exactly
 $1$. Furthermore, all the dual-vertices only have outgoing edges and
 since we are considering triangulations they have outdegree exactly
 $3$.

 \begin{figure}[!h]
 \center
 \begin{tabular}{cc}
 \includegraphics[scale=0.5]{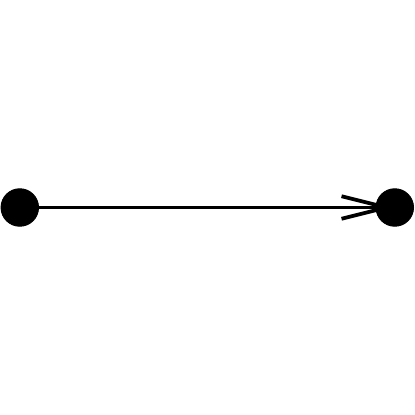} \ \ \  &  \ \ \
 \includegraphics[scale=0.5]{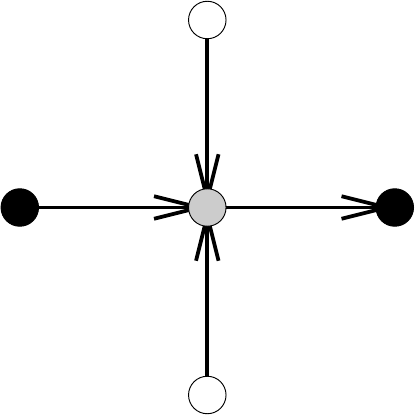} \\
 $G$ \ \ \  &  \ \ \ $\pdc{G}$ \\
 \end{tabular}
 \caption{How to transform an orientation of a triangulation $G$ into
   an orientation of $\pdc{G}$.}
 \label{fig:trimap}
 \end{figure}

Then the definition of $\gamma$ can be
simplified by the following:
$$\gamma (C) = \#\ \text{edges of ${G}$ leaving ${C}$ on its right} -
\#\  \text{edges of ${G}$ leaving ${C}$ on its left}$$

Note that comparing to the general definition of $\gamma$, just the
symbol \ $\pdc{}$\ \ on $\pdc{C}$ have been removed.

The orientation of the toroidal triangulation on the left of
Figure~\ref{fig:orientation} is an example of a 3-orientation of a
toroidal triangulation where some non-contractible cycles have value
$\gamma$ not equal to $0 \bmod 3$.  The value of $\gamma$ for the
three loops is $2, 0$ and $-2$. This explains why this orientation does
not correspond to a Schnyder wood. On the contrary, on the right of
the figure, the three loops have $\gamma$ equal to $0$ and we have a
Schnyder wood.

\chapter{Structure of Schnyder orientations}
\label{sec:structure}
\section{Transformations between Schnyder orientations}
\label{sec:transformations}

We investigate the structure of the set of Schnyder orientations of a
given graph. For that purpose we need some definitions that are given
on a general map $G$ and then applied to $\pdc{G}$.
 
Consider a map $G$ on an orientable surface of genus $g$.  Given two
orientations $D$ and $D'$ of $G$, let $D\setminus D'$ denote the
subgraph of $D$ induced by the edges that are not oriented as in $D'$.

An oriented subgraph $T$ of $G$ is \emph{partitionable} if its edge
set can be partitioned into three sets $T_0$, $T_1$, $T_2$ such that
all the $T_i$ are pairwise homologous, i.e. $T_i-T_j\in\mb{F}$ for
$i,j\in\{0,1,2\}$.  An oriented subgraph $T$ of $G$ is called a
\emph{topological Tutte-orientation} if $\beta(T,W)=0\bmod 3$ for
every closed walk $W$ in $G^*$ (more intuitively, the number of edges
crossing $W$ from left to right minus the number of those crossing $W$
from right to left is divisible by three).

The name ``topological Tutte-orientation'' comes from the fact that an
oriented graph $T$ is called a \emph{Tutte-orientation} if the
difference of outdegree and indegree is divisible by three,
i.e. $d^+(v)-d^-(v)= 0\bmod 3$, for every vertex $v$.  So a
topological Tutte-orientation is a Tutte orientation, since the latter
requires the condition of the topological Tutte orientation only for
the walks $W$ of $G^*$ going around a vertex $v$ of $G$.

The notions of partitionable and topological
Tutte-orientation are  equivalent:

\begin{lemma}\label{lm:topologicalTutte}
  An oriented subgraph of $G$ is partitionable if and only if it is a
  topological Tutte-orientation.
\end{lemma}
\begin{proof}
  $(\Longrightarrow)$ If $T$ is partitionable, then by definition it
  is the disjoint union of three homologous edge sets $T_0$, $T_1$,
  and $T_2$. Hence by Lemma~\ref{lm:homologous}, $\beta(T_0,W) =
  \beta(T_1,W) = \beta(T_2,W)$ for any closed walk $W$ of $G^*$.  By
  linearity of $\beta$ this implies that $\beta(T,W)= 0 \mod 3$ for
  any closed walk $W$ of $G^*$. So $T$ is a topological
  Tutte-orientation.

  $(\Longleftarrow)$ Let $T$ be a topological Tutte-orientation of
  $G$, i.e. $\beta(T,W)= 0 \mod 3$ for any closed walk $W$ of $G^*$.
  In the following, $T$-faces are the faces of $T$ considered as
  an embedded graph. Note that $T$-faces are not necessarily disks.
  Let us introduce a $\{0,1,2\}$-labeling of the $T$-faces. Label an
  arbitrary $T$-face $F_0$ by $0$. For any $T$-face $F$, find a path
  $P$ of $G^*$ from $F_0$ to $F$. Label $F$ with $\beta(T,P) \bmod
  3$. Note that the label of $F$ is independent from the choice of $P$
  by our assumption on closed walks.  For $0\leq i \leq 2$, let $T_i$
  be the set of edges of $T$ which two incident $T$-faces are labeled
  $i-1$ and $i+1$.  Note that an edge of $T_i$ has label $i-1$ on its
  left and label $i+1$ on its right.  The sets $T_i$ form a partition
  of the edges of $T$.  Let $\mc{F}_i$ be the counterclockwise facial
  walks of $G$ that are in a $T$-face labeled $i$.  We have
  $\phi(T_{i+1})-\phi(T_{i-1})=\sum_{F\in \mc{F}_i}\phi(F)$, so the
  $T_i$ are homologous.
\end{proof}

Let us refine the notion of partitionable.  Denote by $\mathcal{E}$
the set of \emph{oriented Eulerian subgraphs} of ${G}$ (i.e. the
oriented subgraphs of ${G}$ where each vertex has the same in- and
out-degree).  Consider a partitionable oriented subgraph $T$ of $G$,
with edge set partition $T_0$, $T_1$, $T_2$ having the same homology.
We say that $T$ is \emph{Eulerian-partitionable} if $T_i\in\mc{E}$ for
all $0\leq i \leq 2$. Note that if $T$ is Eulerian-partitionable then
it is Eulerian.  Note that an oriented subgraph $T$ of $G$ that is
$0$-homologous is also Eulerian and thus Eulerian-partitionable (with
the partition $T,\emptyset,\emptyset$).  Thus $0$-homologous is a
refinement of Eulerian-partitionable.

We now investigate the structure of Schnyder orientations.  For that
purpose, consider a map $G$ on an orientable surface of genus $g$ and
apply the above definitions and results to orientations of $\pdc{G}$.

Let $D,D'$ be two orientations of $\pdc{G}$ such that $D$ is a
Schnyder orientation and $T=D\setminus D'$. Let $\Out=\{(u,v)\in
E(D)\mid v\text{ is an edge-vertex}\}$. Similarly, let
$\Out'=\{(u,v)\in E(D')\mid v\text{ is an edge-vertex}\}$. Note that
an edge of $T$ is either in $\Out$ or in $\Out'$, so
$\phi(T)=\phi(\Out)- \phi(\Out')$. The three following lemmas give
necessary and sufficient conditions on $T$ for $D'$ being a Schnyder
orientation.

\begin{lemma}\label{thm:EDGEtransform}
  $D'$ is a Schnyder orientation if and only if $T$ is
  partitionable.
\end{lemma}

\begin{proof}
  Let $D'$ is a Schnyder orientation.  By Lemma~\ref{lem:charforall},
  this is equivalent to the fact that for any closed walk $W$ of
  $\pdc{G}^*$, we have $\beta (\Out',W) = 0 \bmod 3$. Since
  $\beta(\Out,W)= 0 \bmod 3$, this is equivalent to the fact that for
  any closed walk $W$ of $\pdc{G}^*$, we have
  $\beta (T,W) = 0 \bmod 3$. Finally, by
  Lemma~\ref{lm:topologicalTutte} this is equivalent to $T$ being
  partitionable.
\end{proof}

\begin{lemma}\label{th:EDGEeulerian}
  $D'$ is a Schnyder orientation having the same outdegrees as $D$ if
  and only if $T$ is Eulerian-partitionable.
\end{lemma}

\begin{proof}
  $(\Longrightarrow)$ Suppose $D'$ is a Schnyder orientation having
  the same outdegrees as $D$.  Lemma~\ref{thm:EDGEtransform} implies
  that $T$ is partitionable into $T_0$, $T_1$, $T_2$ having the same
  homology. By Lemma~\ref{lm:homologous}, for each closed walk $W$ of
  $\pdc{G}^*$, we have $\beta(T_0,W)=\beta(T_1,W)=\beta(T_2,W)$.
  Since $D,D'$ have the same outdegrees, we have that $T$ is
  Eulerian. Consider a vertex $v$ of $\pdc{G}$ and a walk $W_v$ of
  $\pdc{G}^*$ going counterclockwise around $v$. For any oriented
  subgraph $H$ of $\pdc{G}^*$, we have
  $d_H^+(v)-d_H^-(v)=\beta(H,W_v)$, where $d_H^+(v)$ and $d_H^-(v)$
  denote the outdegree and indegree of $v$ restricted to $H$,
  respectively. Since $T$ is Eulerian, we have $\beta(T,W_v)=0$. Since
  $\beta(T_0,W_v)=\beta(T_1,W_v)=\beta(T_2,W_v)$ and
  $\sum \beta(T_i,W_v)=\beta(T,W_v)=0$, we obtain that
  $\beta(T_0,W_v)=\beta(T_1,W_v)=\beta(T_2,W_v)=0$. So each $T_i$ is
  Eulerian.

  $(\Longleftarrow)$ Suppose $T$ is Eulerian-partitionable.  Then
  Lemma~\ref{thm:EDGEtransform} implies that $D'$ is a Schnyder
  orientation. Since $T$ is Eulerian, the two orientations $D,D'$ have
  the same outdegrees.
\end{proof}

Consider $\{B_1,\ldots,B_{2g}\}$ a set of cycles of ${G}$
that forms a homology-basis. For $\Gamma \in \mathbb{Z}^{2g}$,
we say that an orientation of $\pdc{G}$ is of \emph{type} $\Gamma$ if
$\gamma(B_{i})=\Gamma_i$ for all $1\le i\le 2g$.

\begin{lemma}\label{th:EDGE0}
  $D'$ is a Schnyder orientation having the same outdegrees and the
  same type as $D$ (for the considered basis) if and only if $T$ is
  $0$-homologous (i.e. $D,D'$ are homologous).
\end{lemma}

\begin{proof}
  $(\Longrightarrow)$ Suppose $D'$ is a Schnyder orientation having
  the same outdegrees and the same type as $D$. Then,
  Lemma~\ref{th:EDGEeulerian} implies that $T$ is
  Eulerian-partitionable and thus Eulerian. So for any
  $F\in \hat{\mc{F^*}}$, we have $\beta(T,F)=0$. Moreover, for
  $1\le i\le 2g$, consider the region $R_i$ between $W_L(B_i)$ and
  $W_R(B_i)$ containing $B_i$. Since $T$ is Eulerian, it is going in and
  out of $R_i$ the same number of time. So
  $\beta(T,W_L(B_i)-W_R(B_i))=0$.  Since $D,D'$ have the same type, we
  have $\gamma_D(B_i)=\gamma_{D'}(B_i)$. So by
  Lemma~\ref{lem:gammaequaldelta},
  $\delta_D(W_L(B_i))+\delta_D(W_R(B_i))=\delta_{D'}(W_L(B_i))+\delta_{D'}(W_R(B_i))$. Thus
  $\beta(T,W_L(B_i)+W_R(B_i))=\beta(\Out-\Out',W_L(B_i)+W_R(B_i))=\delta_D(W_L(B_i))+\delta_D(W_R(B_i))-\delta_{D'}(W_L(B_i))-\delta_{D'}(W_R(B_i))=0$. By
  combining this with the previous equality, we obtain
  $\beta(T,W_L(B_i))=\beta(T,W_R(B_i))=0$ for all $1\le i\le 2g$. Thus
  by Lemma~\ref{lm:homologous}, we have that $T$ is 0-homologous.

  $(\Longleftarrow)$ Suppose that $T$ is $0$-homologous. Then $T$ is in
  particular Eulerian-partitionable (with the partition
  $T,\emptyset,\emptyset$). So Lemma~\ref{th:EDGEeulerian} implies
  that $D'$ is a Schnyder orientation with the same outdegrees as $D$.
  Since $T$ is $0$-homologous, by Lemma~\ref{lm:homologous}, for all
  $1\le i\le 2g$, we have
  $\beta(T,W_L(B_i))=\beta(T,W_R(B_i))=0$. Thus
  $\delta_{D}(W_L(B_i))=\beta(\Out,W_L(B_i)) =
  \beta(\Out',W_L(B_i))=\delta_{D'}(W_L(B_i))$
  and
  $\delta_{D}(W_R(B_i))=\beta(\Out,W_R(B_i)) =
  \beta(\Out',W_R(B_i))=\delta_{D'}(W_R(B_i))$.
  So by Lemma~\ref{lem:gammaequaldelta},
  $\gamma_{D}(B_i)=\delta_{D}(W_L(B_i))+\delta_{D}(W_R(B_i))=\delta_{D'}(W_L(B_i))+\delta_{D'}(W_R(B_i))=\gamma_{D'}(B_i)$.
  So $D,D'$ have the same type.
\end{proof}

Lemma~\ref{th:EDGE0} implies that when you consider Schnyder
orientations having the same outdegrees the property that they have
the same type does not depend on the choice of the basis since being
homologous does not depend on the basis. So we have the following:

  \begin{lemma}
    \label{lem:typebasis}
    If two Schnyder orientations have the same outdegrees and the same
    type (for the considered basis), then they have the same type for
    any basis.
  \end{lemma}

  Lemma~\ref{thm:EDGEtransform},~\ref{th:EDGEeulerian}
  and~\ref{th:EDGE0} are summarized in the following theorem (where by
  Lemma~\ref{lem:typebasis} we do not have to assume a particular
  choice of a basis for the third item):

\begin{theorem}\label{th:transformations}
  Let $G$ be a map on an orientable surface and $D,D'$ orientations of
  $\pdc{G}$ such that $D$ is a Schnyder orientation and $T=D\setminus
  D'$. We have the following:
  \begin{itemize}
  \item $D'$ is a Schnyder orientation if and only if $T$ is
    partitionable.
\item $D'$ is a Schnyder orientation having the same outdegrees as $D$
  if and only if $T$ is Eulerian-partitionable.
\item $D'$ is a Schnyder orientation having the same outdegrees and
  the same type as $D$ if and only if $T$ is $0$-homologous
  (i.e. $D,D'$ are homologous).
  \end{itemize}
\end{theorem}

We show in the next section that the set of Schnyder orientations that are
homologous (see third item of Theorem~\ref{th:transformations})
carries a structure of distributive lattice.

\section{The distributive lattice of homologous orientations}
\label{sec:lattice}

Consider a partial order $\leq$ on a set $S$.  Given two elements
$x,y$ of $S$, let $m(x,y)$ (resp. $M(x,y)$) be the set of elements $z$
of $S$ such that $z\leq x$ and $z\leq y$ (resp. $z\geq x$ and
$z\geq y$).  If $m(x,y)$ (resp. $M(x,y)$) is not empty and admits a
unique maximal (resp. minimal) element, we say that $x$ and $y$ admit
a \emph{meet} (resp. a \emph{join}), noted $x\vee y$ (resp.
$x\wedge y$).  Then $(S,\leq)$ is a \emph{lattice} if any pair of
elements of $S$ admits a meet and a join. Thus in particular a lattice
has a unique minimal (resp. maximal) element.  A lattice is
\emph{distributive} if the two operators $\vee$ and $\wedge$ are
distributive on each other.

 For the sake of generality, in this subsection we consider that maps
 may have loops or multiple edges.  Consider a map $G$ on an
 orientable surface and a given orientation $D_0$ of $G$. Let
 $O(G,D_0)$ be the set of all the orientations of $G$ that are
 homologous to $D_0$.  In this section we prove that $O(G,D_0)$ forms
 a distributive lattice and show some additional properties that are
 very general and concerns not only Schnyder orientations. This
 generalizes results for the plane obtained by Ossona de
 Mendez~\cite{Oss94} and Felsner~\cite{Fel04}.  The distributive
 lattice structure can also be derived from a result of
 Propp~\cite{Pro93} interpreted on the dual map, see the discussion
 below Theorem~\ref{th:lattice}.

In order to define an order on $O(G,D_0)$, fix an arbitrary face $f_0$
of $G$ and let $F_0$ be its counterclockwise facial walk.  Let
$\mc{F}'=\mc{F}\setminus \{F_0\}$ (where $\mc{F}$ is the set of
counterclockwise facial walks of $G$ as defined earlier). Note that
$\phi(F_0)=-\sum_{F\in\mc{F}'}\phi(F)$. Since the characteristic flows
of $\mc{F}'$ are linearly independent, any oriented subgraph of $G$
has at most one representation as a combination of characteristic
flows of $\mc{F}'$. Moreover the $0$-homologous oriented subgraphs of
$G$ are precisely the oriented subgraph that have such a
representation.  We say that a $0$-homologous oriented subgraph $T$ of
$G$ is \emph{counterclockwise} (resp. \emph{clockwise}) if its
characteristic flow can be written as a combination with positive
(resp. negative) coefficients of characteristic flows of $\mc{F}'$,
i.e. $\phi(T)=\sum_{F\in\mc{F}'}\lambda_F\phi(F)$, with
$\lambda\in\mathbb{N}^{|\mc{F}'|}$
(resp. $-\lambda\in\mathbb{N}^{|\mc{F}'|}$).  Given two orientations
$D,D'$, of $G$ we set $D\leq_{f_0} D'$ if and only if $D\setminus D'$
is counterclockwise.  Then we have the following theorem.

\begin{theorem}[\cite{Pro93}]
\label{th:lattice}
Let $G$ be a map on an orientable surface given with a particular
orientation $D_0$ and a particular face $f_0$.  Let $O(G,D_0)$ the set
of all the orientations of $G$ that are homologous to $D_0$.  We have
$(O(G,D_0),\leq_{f_0})$ is a distributive lattice.
\end{theorem}

We attribute Theorem~\ref{th:lattice} to Propp even if it is not
presented in this form in~\cite{Pro93}. Here we do not introduce
Propp's formalism, but provide a new proof of Theorem~\ref{th:lattice}
(as a consequence of the forthcoming
Proposition~\ref{th:lattice}). This allows us to introduce notions
used later in the study of this lattice.

To prove Theorem~\ref{th:lattice}, we need to define the elementary
flips that generates the lattice.  We start by reducing the graph
${G}$. We call an edge of ${G}$ \emph{rigid with respect to
  $O(G,D_0)$} if it has the same orientation in all elements of
$O(G,D_0)$. Rigid edges do not play a role for the structure of
$O(G,D_0)$. We delete them from ${G}$ and call the obtained embedded
graph $\widetilde{G}$.  Note that this graph is embedded but it is not
necessarily a map, as some faces may not be homeomorphic to open
disks. Note that if all the edges are rigid, i.e. $|O(G,D_0)|=1$, then
$\widetilde{G}$ has no edges.

\begin{lemma}
\label{lem:non-rigid}
Given an edge $e$ of $G$, the following are equivalent:
\begin{enumerate}
\item $e$ is non-rigid
\item  $e$ is contained in a
$0$-homologous oriented subgraph of $D_0$
\item  $e$ is contained in a
$0$-homologous oriented subgraph of any element of $O(G,D_0)$
\end{enumerate}
\end{lemma}

\begin{proof}
  $(1 \Longrightarrow 3)$ Let $D\in O(G,D_0)$. If $e$ is non-rigid,
  then it has a different orientation in two elements $D',D''$ of
  $O(G,D_0)$.  Then we can assume by symmetry that $e$ has a different
  orientation in $D$ and $D'$ (otherwise in $D$ and $D''$ by
  symmetry). Since $D,D'$ are homologous to $D_0$, they are also
  homologous to each other. So $T=D\setminus D'$ is a $0$-homologous
  oriented subgraph of $D$ that contains $e$.

 $(3 \Longrightarrow 2)$ Trivial since $D_0\in O(G,D_0)$

  $(2 \Longrightarrow 1)$ If an edge $e$ is contained in a $0$-homologous
  oriented subgraph $T$ of $D_0$. Then let $D$ be the element of
  $O(G,D_0)$ such that $T=D_0\setminus D$. Clearly $e$ is oriented
  differently in $D$ and $D_0$, thus it is non-rigid.
\end{proof}

By Lemma~\ref{lem:non-rigid}, one can build $\widetilde{G}$ by keeping
only the edges that are contained in a $0$-homologous oriented
subgraph of $D_0$.  Note that this implies that all the edges of
$\widetilde{G}$ are incident to two distinct faces of $\widetilde{G}$.
Denote by $\widetilde{\mathcal{F}}$ the set of oriented subgraphs of
$\widetilde{G}$ corresponding to the boundaries of faces of
$\widetilde{G}$ considered counterclockwise.  Note that any
$\widetilde{F}\in \widetilde{\mathcal{F}}$ is $0$-homologous and so
its characteristic flow has a unique way to be written as a
combination of characteristic flows of $\mc{F}'$. Moreover this
combination can be written
$\phi(\widetilde{F})=\sum_{F\in X_{\widetilde{F}}}\phi(F)$, for
$X_{\widetilde{F}}\subseteq\mc{F}'$. Let $\widetilde{f}_0$ be the face
of $\widetilde{G}$ containing $f_0$ and $\widetilde{F}_0$ be the
element of $\widetilde{\mathcal{F}}$ corresponding to the boundary of
$\widetilde{f}_0$.  Let
$\widetilde{\mathcal{F}}'=\widetilde{\mathcal{F}}\setminus
\{\widetilde{F}_0\}$.
The elements of $\widetilde{\mathcal{F}}'$ are precisely the
elementary flips which suffice to generate the entire distributive
lattice $(O(G,D_0),\leq_{f_0})$.

We prove two technical lemmas concerning $\widetilde{\mathcal{F}}'$:

\begin{lemma}
\label{cl:partition}
Let $D\in O(G,D_0)$ and $T$ be a non-empty $0$-homologous oriented
subgraph of $D$.  Then there exist edge-disjoint  $0$-homologous oriented subgraphs
$T_1,\ldots,T_k$ of $D$ such that $\phi(T)=\sum_{1\leq i \leq k}
\phi(T_i)$, and, for $1\leq i \leq k$, there exists
$\widetilde{X_i}\subseteq \widetilde{{\mathcal{F}}}'$ and
$\epsilon_i\in\{-1,1\}$ such that
$\phi(T_i)=\epsilon_i\sum_{\widetilde{F}\in \widetilde{X_i}}
\phi(\widetilde{F})$.
\end{lemma}

\begin{proof}
  Since $T$ is $0$-homologous, we have
  $\phi(T)=\sum_{F\in\mathcal{F'}}\lambda_F\phi(F)$, for $\lambda\in
  \mathbb{Z}^{|\mathcal{F'}|}$. Let $\lambda_{f_0}=0$.  Thus we have
  $\phi(T)=\sum_{F\in\mathcal{F}}\lambda_F\phi(F)$.  Let
  $\lambda_{\min}=\min_{F\in\mathcal{F}}\lambda_F$ and
  $\lambda_{\max}=\max_{F\in\mathcal{F}}\lambda_F$. Note that we may
  have $\lambda_{\min}$ or $\lambda_{\max} =0$ but not both since $T$
  is non-empty.  For $1\leq i\leq \lambda_{\max}$, let
  $X_{i}=\{F\in\mathcal{F'}\,|\,\lambda_F\geq i\}$ and $\epsilon_i=1$.
  Let $X_0=\emptyset$ and $\epsilon_0=1$. For $\lambda_{\min}\leq
  i\leq -1$, let $X_{i}=\{F\in\mathcal{F'}\,|\,\lambda_F\leq i\}$ and
  $\epsilon_i=-1$.  For $\lambda_{\min}\leq i \leq \lambda_{\max}$,
  let $T_i$ be the oriented subgraph such that
  $\phi(T_i)=\epsilon_i\sum_{F\in X_i}\phi(F)$.  Then we have
  $\phi(T)=\sum_{\lambda_{\min}\leq i \leq \lambda_{\max}} \phi(T_i)$.

  Since $T$ is an oriented subgraph, we have
  $\phi(T)\in\{-1,0,1\}^{|E(G)|}$. Thus for any edge of ${G}$, incident to
  faces $F_1$ and $F_2$, we have
  $(\lambda_{F_1}-\lambda_{F_2})\in\{-1,0,1\}$. So, for $1\leq i\leq
  \lambda_{\max}$, the oriented graph $T_i$ is the border between
  the faces with $\lambda$ value equal to $i$ and $i-1$.
  Symmetrically, for $\lambda_{\min}\leq i\leq -1$, the oriented graph
  $T_i$ is the border between the faces with $\lambda$ value equal
  to $i$ and $i+1$.  So all the $T_i$ are edge disjoint and are
  oriented subgraphs of $D$.

  Let $\widetilde{X_i}=\{\widetilde{F}\in\widetilde{\mathcal{F}}'\,|
  \,\phi(\widetilde{F})=\sum_{F\in X'} \phi(F) \textrm{ for some }
  X'\subseteq X_i\}$.  Since $T_i$ is $0$-homologous, the edges of
  $T_i$ can be reversed in $D$ to obtain another element of
  $O(G,D_0)$. Thus there is no rigid edge in $T_i$.  Thus
  $\phi(T_i)=\epsilon_i\sum_{F\in
    X_i}\phi(F)=\epsilon_i\sum_{\widetilde{F}\in\widetilde{X_i}}\phi(\widetilde{F})$.
\end{proof}

\begin{lemma}
\label{cl:facenonrigid}
Let $D\in O(G,D_0)$ and $T$ be a non-empty $0$-homologous oriented
subgraph of $D$ such that there exists $\widetilde{X}\subseteq
\widetilde{{\mathcal{F}}}'$ and $\epsilon\in\{-1,1\}$ satisfying
$\phi(T)=\epsilon\sum_{\widetilde{F}\in \widetilde{X}}
\phi(\widetilde{F})$. Then there exists $\widetilde{F}\in\widetilde{X}$ such
that $\epsilon\,\phi(\widetilde{F})$ corresponds to an oriented
subgraph of $D$.
\end{lemma}

\begin{proof}
  The proof is done by induction on $|\widetilde{X}|$.  Assume that
  $\epsilon =1$ (the case $\epsilon =-1$ is proved similarly).

  If $|\widetilde{X}|=1$, then the conclusion is clear since
  $\phi(T)=\sum_{\widetilde{F}\in\widetilde{X}}\phi(\widetilde{F})$.
  We now assume that $|\widetilde{X}|> 1$.  Suppose by contradiction
  that for any $\widetilde{F}\in\widetilde{X}$ we do not have the
  conclusion, i.e $\phi(\widetilde{F})_e\neq\phi(T)_e$ for some
  $e \in \widetilde{F}$. Let $\widetilde{F_1}\in\widetilde{X}$ and
  $e \in \widetilde{F_1}$ such that
  $\phi(\widetilde{F_1})_e\neq\phi(T)_e$.  Since $\widetilde{F_1}$ is
  counterclockwise, we have $\widetilde{F_1}$ on the left of $e$.  Let
  $\widetilde{F_2}\in\widetilde{{\mathcal{F}}} $ that is on the right
  of $e$. Note that $\phi(\widetilde{F_1})_e=-\phi(\widetilde{F_2})_e$
  and for any other face $\widetilde{F}\in \widetilde{\mathcal{F}}$,
  we have $\phi(\widetilde{F})_e=0$.  Since
  $\phi(T)=\sum_{\widetilde{F}\in\widetilde{X}}\phi(\widetilde{F})$,
  we have $\widetilde{F_2}\in\widetilde{X}$ and $\phi(T)_e=0$.  By
  possibly swapping the role of $\widetilde{F_1}$ and
  $\widetilde{F_2}$, we can assume that
  $\phi(D)_e=\phi(\widetilde{F_1})_e$ (i.e. $e$ is oriented the same way
  in $\widetilde{F_1}$ and  $D$). Since $e$ is not rigid, there
  exists an orientation $D'$ in $O(G,D_0)$ such that
  $\phi(D)_e=-\phi(D')_e$.

  Let $T'$ be the non-empty $0$-homologous oriented subgraph of $D$
  such that $T'=D\setminus D'$. Lemma~\ref{cl:partition} implies that
  there exists edge-disjoint  $0$-homologous oriented subgraphs $T_1,\ldots,T_k$ of
  $D$ such that $\phi(T)=\sum_{1\leq i \leq k} \phi(T_i)$, and, for
  $1\leq i \leq k$, there exists $\widetilde{X_i}\subseteq
  \widetilde{{\mathcal{F}}}'$ and $\epsilon_i\in\{-1,1\}$ such that
  $\phi(T_i)=\epsilon_i\sum_{\widetilde{F}\in \widetilde{X_i}}
  \phi(\widetilde{F})$.  Since $T'$ is the disjoint union of
  $T_1,\ldots,T_k$, there exists $1\leq i\leq k$, such that $e$ is an
  edge of $T_i$. Assume by symmetry that $e$ is an edge of
  $T_1$. Since $\phi(T_1)_e=\phi(D)_e=\phi(\widetilde{F_1})_e$, we
  have $\epsilon_1=1$, $\widetilde{F_1}\in \widetilde{X_1}$ and
  $\widetilde{F_2}\notin \widetilde{X_1}$.

  Let $\widetilde{Y}=\widetilde{X}\cap \widetilde{X_1}$. Thus
  $\widetilde{F_1}\in \widetilde{Y}$ and $\widetilde{F_2}\notin
  \widetilde{Y}$. So $|\widetilde{Y}|<|\widetilde{X}|$. Let
  $T_{\widetilde{Y}}$ be the oriented subgraph of $G$ such that
  $T_{\widetilde{Y}}=\sum_{\widetilde{F}\in \widetilde{Y}}
  \phi(\widetilde{F})$.  Note that the edges of $T$ (resp. $T_1$) are
  those incident to exactly one face of $\widetilde{X}$
  (resp. $\widetilde{X_1}$). Similarly every edge of
  $T_{\widetilde{Y}}$ is incident to exactly one face of
  $\widetilde{Y}=\widetilde{X}\cap \widetilde{X_1}$, i.e. it has one
  incident face in $\widetilde{Y}=\widetilde{X}\cap \widetilde{X_1}$
  and the other incident face not in $\widetilde{X}$ or not in
  $\widetilde{X_1}$.  In the first case this edge is in $T$, otherwise
  it is in $T_1$. So every edge of $T_{\widetilde{Y}}$ is an edge of
  $T\cup T_1$. Hence $T_{\widetilde{Y}}$ is an oriented subgraph of
  $D$. So we can apply the induction hypothesis on
  $T_{\widetilde{Y}}$. This implies that there exists
  $\widetilde{F}\in\widetilde{Y}$ such that $\widetilde{F}$ is an
  oriented subgraph of $D$. Since
  $\widetilde{Y}\subseteq\widetilde{X}$, this is a contradiction to
  our assumption.
\end{proof}

We need the following
characterization of distributive lattice from~\cite{Fel09}:

\begin{theorem}[\cite{Fel09}]
\label{th:hasse}
  An oriented graph $\mathcal{H}=(V,E)$ is the Hasse diagram of a
  distributive lattice if and only if it is connected, acyclic, and
  admits an edge-labeling $c$ of the edges such that:
\begin{itemize}
\item if $(u,v), (u,w)\in E$ then 
    \begin{itemize}
      \item[(U1)] $c(u,v)\neq c(u,w)$ and
      \item[(U2)] there is $z\in V$ such that $(v,z), (w,z)\in E$,
        $c(u,v)=c(w,z)$, and $c(u,w)=c(v,z)$.
    \end{itemize}
\item if $(v,z), (w,z)\in E$ then
    \begin{itemize}
      \item[(L1)] $c(v,z)\neq c(w,z)$ and
      \item[(L2)] there is $u\in V$ such that $(u,v), (u,w)\in E$,
        $c(u,v)=c(w,z)$, and $c(u,w)=c(v,z)$.
    \end{itemize}
\end{itemize}
\end{theorem}

We define the directed graph $\mathcal{H}$ with vertex set
$O(G,D_0)$.  There is an oriented edge from $D_1$ to $D_2$ in
$\mathcal{H}$ (with $D_1\leq_{f_0}D_2$) if and only if
$D_1\setminus D_2\in \widetilde{\mathcal{F}}'$.  We define the label
of that edge as $c(D_1,D_2)=D_1\setminus D_2$. We show that
$\mathcal{H}$ fulfills all the conditions of Theorem~\ref{th:hasse},
and thus obtain the following:

\begin{proposition}
  \label{lem:fulfillshasse}
  $\mathcal{H}$ is the Hasse diagram of a distributive lattice.
\end{proposition}

\begin{proof}
  The characteristic flows of elements of $\widetilde{\mathcal{F}}'$
  form an independent set, hence the digraph $\mathcal{H}$ is acyclic.
  By definition all outgoing and all incoming edges of a vertex of
  $\mathcal{H}$ have different labels, i.e. the labeling $c$ satisfies
  (U1) and (L1).  If $(D_u,D_v)$ and $(D_u,D_w)$ belong to
  $\mathcal{H}$, then $T_v=D_u\setminus D_v$ and
  $T_w=D_u\setminus D_w$ are both elements of
  $\widetilde{\mathcal{F}}'$, so they must be edge disjoint.  Thus,
  the orientation $D_z$ obtained from reversing the edges of $T_w$ in
  $D_v$ or equivalently $T_v$ in $D_w$ is in $O(G,D_0)$. This gives
  (U2). The same reasoning gives (L2).  It remains to show that
  $\mathcal{H}$ is connected.

Given a $0$-homologous oriented subgraph $T$ of $G$,
such that $T=\sum_{F\in\mc{F}'}\lambda_F\phi(F)$, we define
$s(T)=\sum_{F\in\mathcal{F'}}|\lambda_F|$.

Let $D,D'$ be two orientations of $O(G,D_0)$, and $T=D\setminus
D'$. We prove by induction on $s(T)$ that $D,D'$ are connected in
$\mathcal{H}$.  This is clear if $s(T)=0$ as then $D=D'$. So we now
assume that $s(T)\neq 0$ and so that $D,D'$ are distinct.
Lemma~\ref{cl:partition} implies that there exists edge-disjoint $0$-homologous 
oriented subgraphs $T_1,\ldots,T_k$ of $D$ such that
$\phi(T)=\sum_{1\leq i \leq k} \phi(T_i)$, and, for $1\leq i \leq k$,
there exists $\widetilde{X_i}\subseteq \widetilde{{\mathcal{F}}}'$ and
$\epsilon_i\in\{-1,1\}$ such that
$\phi(T_i)=\epsilon_i\sum_{\widetilde{F}\in \widetilde{X_i}}
\phi(\widetilde{F})$.  Lemma~\ref{cl:facenonrigid} applied to $T_1$
implies that there exists $\widetilde{F_1}\in\widetilde{X_1}$ such
that $\epsilon_1\,\phi(\widetilde{F_1})$ corresponds to an oriented subgraph
of $D$.  Let $T'$ be the oriented subgraph such that
$\phi(T)=\epsilon_1\phi(\widetilde{F_1})+\phi(T')$. Thus:
\begin{eqnarray*}
\phi(T') &=&\phi(T)-\epsilon_1\phi(\widetilde{F_1})\\
&=&\sum_{1\leq i\leq
  k}\phi(T_i)-\epsilon_1\phi(\widetilde{F_1})\\
&=&\sum_{\widetilde{F}\in(\widetilde{X_1}\setminus\{\widetilde{F_1}\})}
\epsilon_1 \phi(\widetilde{F})+\sum_{2\leq i\leq k}\sum_{\widetilde{F}\in
  \widetilde{X_i}} \epsilon_i\phi(\widetilde{F})\\
&=&\sum_{\widetilde{F}\in(\widetilde{X_1}\setminus\{\widetilde{F_1}\})}
\sum_{F\in{X_{\widetilde{F}}}} \epsilon_1\phi({F})+\sum_{2\leq i\leq
  k}\sum_{\widetilde{F}\in
  \widetilde{X_i}}\sum_{F\in{X_{\widetilde{F}}}} \epsilon_i\phi({F})
\end{eqnarray*}
So $T'$ is $0$-homologous.  Let $D''$ be such that $\epsilon_1
\widetilde{F_1}=D\setminus D''$.  So we have $D'' \in O(G,D_0)$ and
there is an edge between $D$ and $D''$ in $\mathcal{H}$. Moreover
$T'=D''\setminus D'$ and $s(T')=
s(T)-|{X_{\widetilde{F_1}}}|<s(T)$. So the induction hypothesis on
$D'',D'$ implies that they are connected in $\mathcal{H}$. So $D,D'$
are also connected in $\mathcal{H}$.
\end{proof}

Note that Proposition~\ref{lem:fulfillshasse} gives a  proof of
Theorem~\ref{th:lattice} independent from
Propp~\cite{Pro93}. 

We continue to investigate further the set
$O(G,D_0)$.

\begin{lemma}
\label{lem:necessary}
For every element $ \widetilde{F}\in \widetilde{\mathcal{F}}$, there
exists $D$ in $O(G,D_0)$ such that $\widetilde{F}$ is an oriented subgraph
of $D$. 
\end{lemma}

\begin{proof}
  Let $ \widetilde{F}\in \widetilde{\mathcal{F}}$. Let $D$ be an
  element of $O(G,D_0)$ that maximizes the number of edges of
  $\widetilde{F}$ that have the same orientation in $\widetilde{F}$
  and $D$ (i.e.  $D$ maximizes the number of edges oriented
  counterclockwise on the boundary of the face of $\widetilde{G}$
  corresponding to $\widetilde{F}$).  Suppose by contradiction that
  there is an edge $e$ of $\widetilde{F}$ that does not have the same
  orientation in $\widetilde{F}$ and $D$. Edge $e$ is in
  $\widetilde{G}$ so it is non-rigid.  Let $D'\in O(G,D_0)$ such that
  $e$ is oriented differently in $D$ and $D'$. Let $T=D\setminus D'$.
  By Lemma~\ref{cl:partition}, there exist edge-disjoint  $0$-homologous oriented
  subgraphs $T_1,\ldots,T_k$ of $D$ such that
  $\phi(T)=\sum_{1\leq i \leq k} \phi(T_i)$, and, for
  $1\leq i \leq k$, there exists
  $\widetilde{X}_i\subseteq \widetilde{{\mathcal{F}}}'$ and
  $\epsilon_i\in\{-1,1\}$ such that
  $\phi(T_i)=\epsilon_i\sum_{\widetilde{F}'\in \widetilde{X}_i}
  \phi(\widetilde{F}')$.
  W.l.o.g., we can assume that $e$ is an edge of $T_1$.  Let $D''$ be
  the element of $O(G,D_0)$ such that $T_1=D\setminus D''$.  The
  oriented subgraph $T_1$ intersects $\widetilde{F}$ only on edges of
  $D$ oriented clockwise on the border of $\widetilde{F}$. So $D''$
  contains strictly more edges oriented counterclockwise on the border
  of the face $\widetilde{F}$ than $D$, a contradiction.  So all the
  edges of $\widetilde{F}$ have the same orientation in $D$.  So
  $\widetilde{F}$ is a $0$-homologous oriented subgraph of $D$.
\end{proof}

By Lemma~\ref{lem:necessary}, for every element
$ \widetilde{F}\in \widetilde{\mathcal{F}}'$ there exists $D$ in
$O(G,D_0)$ such that $\widetilde{F}$ is an oriented subgraph of $D$. Thus
there exists $D'$ such that $\widetilde{F}=D\setminus D'$ and $D,D'$
are linked in $\mathcal{H}$.  Thus 
$\widetilde{\mathcal{F}}'$ is a minimal set that generates the
lattice.

A distributive lattice has a unique maximal (resp. minimal) element.
Let $D_{\max}$ (resp. $D_{\min}$) be the maximal (resp. minimal)
element of $(O(G,D_0),\leq_{f_0})$.

\begin{lemma}
\label{lem:maxtilde}
  $\widetilde{F}_0$ (resp. $-\widetilde{F}_0$) is an oriented subgraph
  of $D_{\max}$ (resp. $D_{\min}$). 
\end{lemma}

\begin{proof}
  By Lemma~\ref{lem:necessary}, there exists $D$ in $O(G,D_0)$ such
  that $\widetilde{F_0}$ is an oriented subgraph of $D$. Let
  $T=D\setminus D_{\max}$. Since $D\leq_{f_0} D_{\max}$, the
  characteristic flow of $T$ can be written as a combination with
  positive coefficients of characteristic flows of
  $\widetilde{\mathcal{F}}'$, i.e.
  $\phi(T)=\sum_{\widetilde{F}\in
    \widetilde{\mathcal{F}}'}\lambda_F\phi(\widetilde{F})$
  with $\lambda\in\mathbb{N}^{|\mc{F}'|}$. So $T$ is disjoint from
  $\widetilde{F}_0$.  Thus $\widetilde{F}_0$ is an oriented subgraph
  of $D_{\max}$. The proof is similar for $D_{\min}$.
  \end{proof}

\begin{lemma}
\label{prop:maximal}
$D_{\max}$ (resp. $D_{\min}$) contains no counterclockwise
(resp. clockwise) non-empty $0$-homologous oriented subgraph w.r.t.~$f_0$.
\end{lemma}

\begin{proof}
  Suppose by contradiction that $D_{\max}$ contains a counterclockwise
  non-empty $0$-homologous oriented subgraph $T$. Then there exists
  $D\in O(G,D_0)$ distinct from $D_{\max}$ such that
  $T=D_{\max}\setminus D$. We have $D_{\max}\leq_{f_0} D$ by
  definition of $\leq_{f_0}$, a contradiction to the maximality of
  $D_{\max}$.
\end{proof}

Note that in the definition of counterclockwise (resp. clockwise)
non-empty $0$-homologous oriented subgraph, used in
Lemma~\ref{prop:maximal}, the sum is taken over elements of $\mc{F}'$
and thus does not use $F_0$. In particular, $D_{\max}$
(resp. $D_{\min}$) may contain regions whose boundary is oriented
counterclockwise (resp. clockwise) according to the region but then
such a region contains $f_0$.

There is a generic known method~\cite{Meu} (see
also~\cite[p.23]{UecPHD}) to compute in linear time a minimal
$\alpha$-orientation of a planar map as soon as an
$\alpha$-orientation is given. This method also works on oriented
surfaces and can be applied to obtain the minimal orientation
$D_{\min}$ in linear time. We explain the method briefly below.

It is  simpler to explain how to compute the minimal orientation $D_{\min}$
homologous to $D_0$ in a dual setting.  The first observation to make
is that two orientations $D_1, D_2$ of $G$ are homologous if and only
if there dual orientations $D^*_1, D^*_2$ of $G^*$ are equivalent up
to reversing some directed cuts. Furthermore $D_1\le_{f_0} D_2$ if and
only if $D^*_1$ can be obtained from $D^*_2$ by reversing directed
cuts oriented from the part containing $f_0$.  Let us compute
$D^*_{\min}$ which is the only orientation of $G^*$, obtained from
$D^*_0$ by reversing directed cuts, and without any directed cut
oriented from the part containing $f_0$.  For this, consider the
orientation $D^*_0$ of $G^*=(V^*,E^*)$ and compute the set
$X\subseteq V^*$ of vertices of $G^*$ that have an oriented path toward
$f_0$. Then $(X,V^*\setminus X)$ is a directed cut oriented from the
part containing $f_0$ that one can reverse. Then update the set of
vertices that can reach $f_0$ and go on until $X=V^*$.  It is not
difficult to see that this can be done in linear time. Thus we obtain
the minimal orientation in linear time.

We conclude this section by applying Theorem~\ref{th:lattice} to
Schnyder orientations:

\begin{theorem}
\label{cor:lattice}
  Let $G$ be a map on an orientable surface given with a particular
  Schnyder orientation $D_0$ of $\pdc{G}$ and a particular face $f_0$ of
  $\pdc{G}$.  Let $S(\pdc{G},D_0)$ be the set of all the Schnyder
  orientations of $\pdc{G}$ that have the same outdegrees and same
  type as $D_0$.  We have that $(S(\pdc{G},D_0),\leq_{f_0})$ is a
  distributive lattice.
\end{theorem}

\begin{proof}
  By the third item of Theorem~\ref{th:transformations}, we have
  $S(\pdc{G},D_0)=O(\pdc{G},D_0)$. Then the conclusion holds by
  Theorem~\ref{th:lattice}.
\end{proof}

\part{Properties in  the toroidal case}
\label{part:torus}

\chapter{Definitions and properties}
\label{sec:torus}

\section{Toroidal Schnyder woods}

According to Euler's formula, the torus is certainly the most
beautiful oriented surface since $g=1$ and Euler's formula sums to
zero, i.e.  $n-m+f=0$.
Thus when generalizing Schnyder woods, there is no need of
vertices satisfying several times the Schnyder property (see
Figure~\ref{fig:369}), nor edges of type $0$ (two incoming edges with
the same color, see Figure~\ref{fig:edgelabeling}) and the general
definition of Schnyder woods (see Definition~\ref{def:highgenus}) can
be simplified in this case. 

Indeed an {\sc edge}, $\mathbb{N}^*$-{\sc vertex}, $\mathbb{N}^*$-{\sc
  face} angle labeling of a toroidal map is in fact a $\{1,2\}$-{\sc
  edge}, $1$-{\sc vertex}, $1$-{\sc face} angle labeling.  One gets
this by counting label changes on the angles around vertices, faces
and edges (like in the proof of Proposition~\ref{prop:bijbipolar}).
Around vertices and faces there are at least $3$ changes. Around an
edge there are at most $3$ changes (see
Figure~\ref{fig:edgelabeling}).  So $3n+3f\leq 3m$. By Euler's formula
we have equality $3n+3f = 3m$. So around all vertices, faces and edges
there are exactly $3$ changes. This implies that vertices and faces
are of type $1$ and edges of type $1$ or $2$.

The definition of Schnyder woods in the toroidal case is thus the
following:

\begin{definition}[Toroidal Schnyder wood]
  Given a toroidal map $G$, a \emph{toroidal Schnyder wood} of $G$ is
  an orientation and coloring of the edges of $G$ with the colors $0$,
  $1$, $2$, where every edge $e$ is oriented in one direction or in
  two opposite directions (each direction having a distinct color and
  being outgoing), satisfying the following conditions:
 
\begin{itemize}
\item Every vertex satisfies the Schnyder property (see
  Definition~\ref{def:schnyderproperty}).
\item There is no face whose boundary is a monochromatic cycle.
\end{itemize}
\end{definition}

When there is no ambiguity we may omit the word ``toroidal'' in
``toroidal Schnyder wood''.

See Figure~\ref{fig:tore-primal} for two examples of toroidal Schnyder
woods. 

Like previously, there is a bijection between toroidal Schnyder woods
and particular angle labelings.  
 
\begin{proposition}
\label{prop:bijtore}
If $G$ is a toroidal map, then the toroidal Schnyder woods of $G$ are
in bijection with the \{1,2\}-{\sc edge}, 1-{\sc vertex}, 1-{\sc face}
angle labelings of $G$.
\end{proposition}

The proof of Proposition~\ref{prop:bijtore} is omitted. It is very
similar to the proof of Proposition~\ref{prop:bijbipolar} and the
counting argument is given at the beginning of this section.  Like for
spherical Schnyder woods, there are no type 2 edges for triangulations
and thus the angle labeling is 1-{\sc edge}, 1-{\sc vertex}, 1-{\sc
  face} in the case of toroidal triangulations.

Let $G$ be a toroidal map given with a Schnyder wood.  The graph $G_i$
denotes the directed graph induced by the edges of color $i$, including
edges that are half-colored $i$, and in this case, the edges get only
the direction corresponding to color $i$.  Each graph $G_i$ has exactly $n$
edges.  Note that $G_i$ is not necessarily connected.
Figure~\ref{fig:notconnected} is an example of a Schnyder wood that
has one color whose corresponding subgraph is not connected.

 \begin{figure}[!h]
 \center
 \includegraphics[scale=0.5]{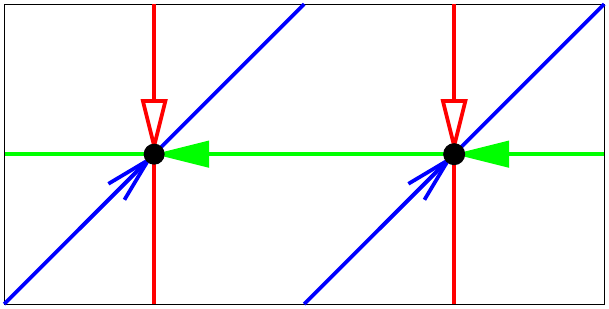}
 \caption{Example of a Schnyder wood that has one color whose
   corresponding subgraph is not connected.}
 \label{fig:notconnected}
 \end{figure}

 Each connected component of $G_i$ has exactly one outgoing edge for
 each of its vertices. Thus it has exactly one directed cycle that is
 a \emph{monochromatic cycle} of color $i$, or \emph{$i$-cycle} for
 short. Note that monochromatic cycles can contain edges oriented in
 two directions with different colors, but the \emph{orientation} of a
 $i$-cycle is the orientation given by the (half-)edges of color $i$.

 The situation is very different from the plane since the $G_i$ are
 not necessarily connected, contain cycles and we do not have a
 decomposition into three trees like on Figures~\ref{fig:ex-tree}
 or~\ref{fig:3c-tree}. Nevertheless the behavior of the monochromatic
 cycles  give some interesting global structure.

\section{Monochromatic cycles}
\label{sec:monochrocycles}

We often use the following lemma that is folklore:

 \begin{lemma}[Folklore]
\label{lem:intersect2}
Let $C,C'$ be two non contractible cycle of a toroidal map.  If $C,C'$
are not weakly homologous, then their intersection is non empty.
\end{lemma}

Let $G$ be a toroidal map given with a Schnyder wood.
We say that two monochromatic cycles $C_i,C_j$ of different colors are
\emph{reversal} if one is obtained from the other by reversing all the
edges, i.e. $C_i=C_j^{-1}$. We say that two monochromatic cycles are
\emph{crossing} if they intersect but are not reversal. We define the
\emph{right side} of a $i$-cycle $C_i$, as the right side while
``walking'' along the directed cycle by following the orientation
given by the edges colored $i$.

\begin{lemma}
\label{lem:allhomotopic}
For a given color $i\in\{0,1,2\}$, all $i$-cycles are
non-contractible, non intersecting and weakly homologous.
\end{lemma}

\begin{proof}
  By Lemma~\ref{lem:nodirectedcycle}, all $i$-cycles are
  non-contractible.  If there exist two such distinct $i$-cycles that
  are intersecting, then there is a vertex that has two outgoing edge
  of color $i$, a contradiction to the Schnyder property. So the
  $i$-cycles are non intersecting.  Then, by
  Lemma~\ref{lem:intersect2}, they are weakly homologous.
\end{proof}

Note that there are Schnyder woods whose $i$-cycles are
reversely homologous as shown on the example of
Figure~\ref{fig:gamma0-weakly} for which the green color has two
 cycles going in opposite way.

\begin{figure}[!h]
\center
\includegraphics[scale=0.5]{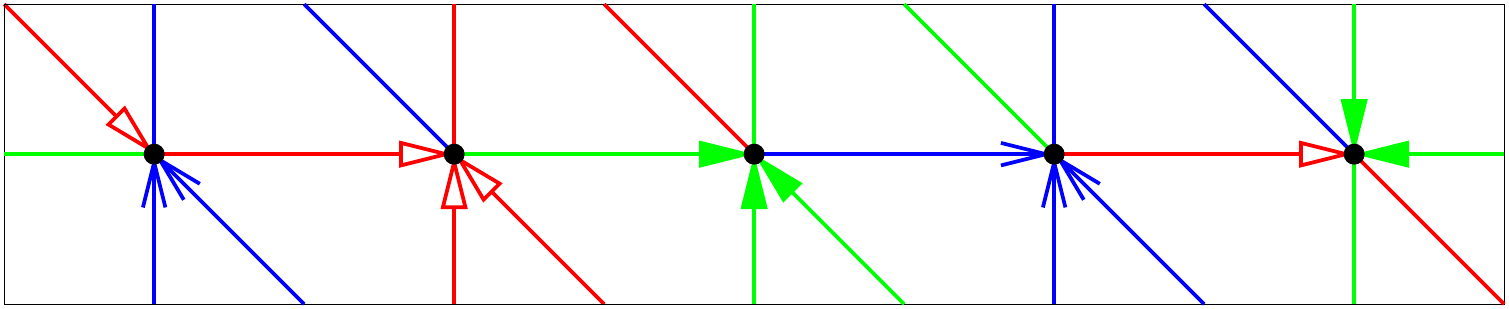} 
\caption{Example of a Schnyder wood with some green cycles that are
  reversely homologous.}
\label{fig:gamma0-weakly}
\end{figure}

\begin{lemma}
\label{lem:twonothomotopic}
If two monochromatic cycles are crossing then 
 they are not weakly homologous.
\end{lemma}

\begin{proof}
Suppose that there exist two
  monochromatic cycles $C_{i-1}$ and $C_{i+1}$, of color $i-1$ and
  $i+1$, that are crossing and weakly homologous.  By
  Lemma~\ref{lem:nodirectedcycle}, the cycles $C_{i-1}$ and
  $C_{i+1}$ are non-contractible.  Since $C_{i-1}\neq C_{i+1}^{-1}$
  and $C_{i-1}\cap C_{i+1}\neq \emptyset$, the cycle $C_{i+1}$ is
  leaving $C_{i-1}$. It is leaving $C_{i-1}$ on its right side by the
  Schnyder property.  Since $C_{i-1}$ and $C_{i+1}$ are weakly
  homologous, the cycle $C_{i+1}$ is entering $C_{i-1}$ at least once
  from its right side, a contradiction to the Schnyder property.
\end{proof}

\begin{lemma}
\label{lem:tworeversal}
If two monochromatic cycles are intersecting and weakly homologous,
then they are reversal.
\end{lemma}

\begin{proof}
  By definition of crossing, two monochromatic cycles that are
  intersecting, weakly homologous and not reversal contradict
  Lemma~\ref{lem:twonothomotopic}.
\end{proof}

For $i\in\{0,1,2\}$, we define the \emph{$i$-winding number}
$\omega_i$ of the Schnyder wood as the number of times a $(i-1)$-cycle
crosses a $(i+1)$-cycle.  Thus by Lemma~\ref{lem:twonothomotopic}, we
have $\omega_i=0$ if the $(i-1)$- and $(i+1)$-cycles are weakly
homologous to each other (some may be reversal of each other or not).
The \emph{winding number} of the Schnyder wood is the maximum of the
$\omega_i$. For example the winding numbers of the Schnyder wood of
Figure~\ref{fig:notconnected} are $\omega_0=2$, $\omega_1=1$,
$\omega_2=1$ and those of Figure~\ref{fig:gamma0-weakly} are
$\omega_0=\omega_1=\omega_2=0$. 

\chapter{Crossing Schnyder woods}
\section{Crossing properties}
\label{sec:crossing}
We define properties concerning the way the monochromatic cycles may
cross each other. A Schnyder wood of a toroidal map is said to be:

\begin{itemize}
\item {\emph{half-crossing}:} if there exists a pair $i,j$
 of different colors, such that there exist a $i$-cycle 
crossing a $j$-cycle.
\item {\emph{intersecting}:} if, for $i\in\{0,1,2\}$, every $i$-cycle intersects at least one
  $(i-1)$-cycle and at least one $(i+1)$-cycle.

\item {\emph{crossing}:} if, for $i\in\{0,1,2\}$, every $i$-cycle crosses at least one
  $(i-1)$-cycle and at least one $(i+1)$-cycle.
\end{itemize}

The three properties are such that: crossing $\implies$ intersecting
$\implies$ half-crossing. The first implication is clear, the second
is proved in the following:

\begin{proposition}
  An intersecting Schnyder wood is half-crossing.
\end{proposition}

\begin{proof}
  Consider an $i$-cycle $C_i$ that intersects at least one
  $(i-1)$-cycle $C_{i-1}$ and at least one $(i+1)$-cycle $C_{i+1}$. If
  $C_{i-1},C_i,C_{i+1}$ are all weakly homologous, then
  $C_{i-1}=C_i^{-1}=C_{i+1}$ by Lemma~\ref{lem:tworeversal}. This
  contradicts the definition of Schnyder woods since there is no edges
  with three colors. So at least two of $C_{i-1},C_i,C_{i+1}$ are not
  weakly homologous and thus intersecting by
  Lemma~\ref{lem:intersect2}. These two cycles are then crossing.
\end{proof}

Figures~\ref{fig:cross-0} gives example of Schnyder woods of a
toroidal map that are respectively not half-crossing, half-crossing
(and not intersecting), intersecting (and not
crossing), crossing.

\begin{figure}[!h]
\center
\begin{tabular}{cc}
\includegraphics[scale=0.4]{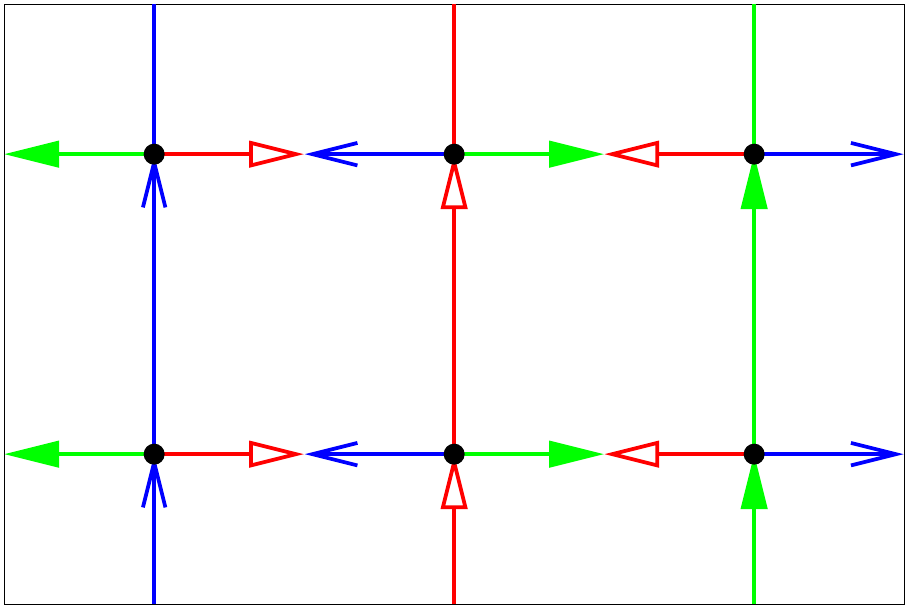}  \ & \
\includegraphics[scale=0.4]{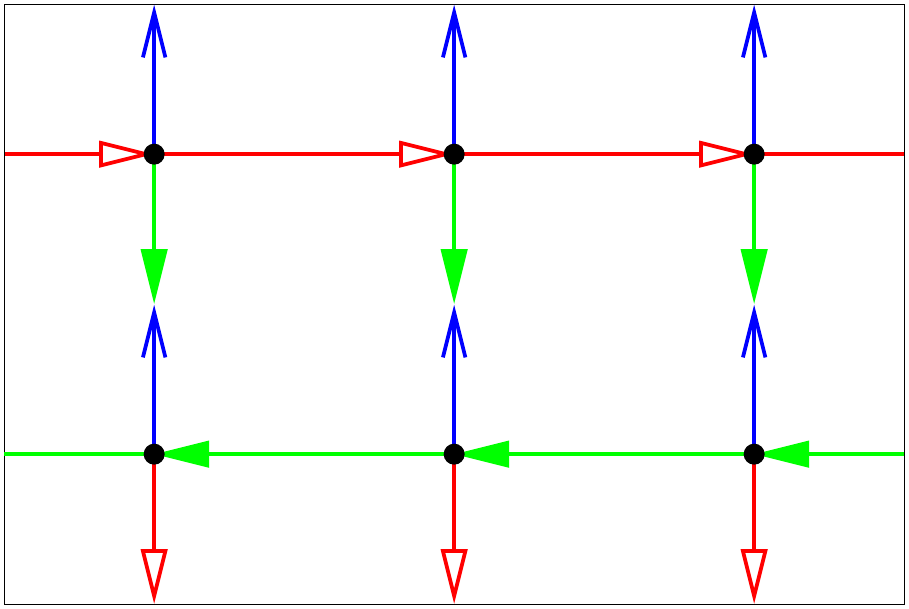}  \\
Not half-crossing\ &\  Half-crossing \\
& \\
\includegraphics[scale=0.4]{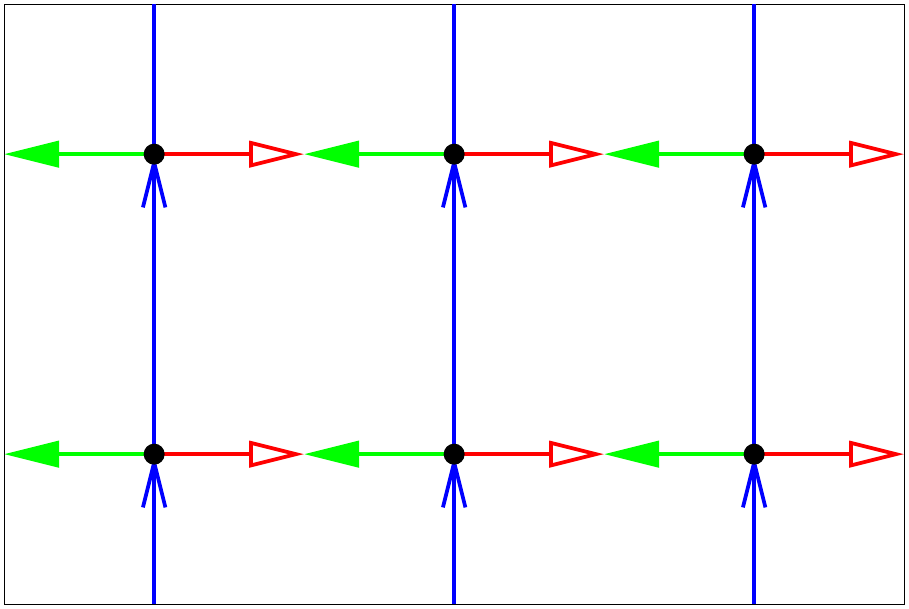} \ & \
\includegraphics[scale=0.4]{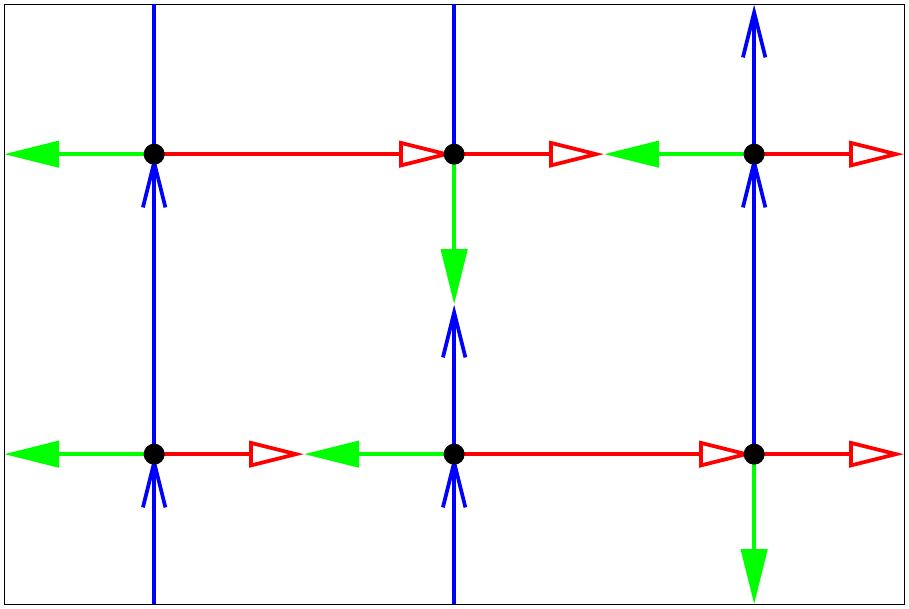}  \\
Intersecting \ &\   Crossing  \\
\end{tabular}
\caption{Example of Schnyder woods of a toroidal map that are
  respectively not half-crossing, half-crossing, intersecting,
  crossing.}
\label{fig:cross-0}
\end{figure}

For triangulation the situation is a bit simpler since "crossing
$\iff$ intersecting". 

\begin{proposition}
\label{prop:crossing-trieq}
  An intersecting Schnyder wood of a toroidal triangulation is
  crossing.
\end{proposition}

\begin{proof}
  Consider an $i$-cycle $C_i$ and a $j$-cycle $C_j$ with $i\neq j$. By
  assumption, cycle $C_i$ intersects at least one $j$-cycle $C'_j$.
  There is no edges oriented in two directions so $C_i,C'_j$ are not
  reversal of each other. Thus there are not weakly homologous by
  Lemma~\ref{lem:tworeversal}. By Lemma~\ref{lem:allhomotopic}, all
  the $j$-cycle are weakly homologous. So $C_i,C_j$ are not weakly
  homologous, thus crossing by Lemma~\ref{lem:intersect2}.
\end{proof}

Figure~\ref{fig:cross-0-tri} gives examples of Schnyder woods of a
toroidal triangulation that are respectively not half-crossing,
half-crossing (and not crossing), crossing.

\begin{figure}[!h]
\center
\begin{tabular}{cc}
\includegraphics[scale=0.4]{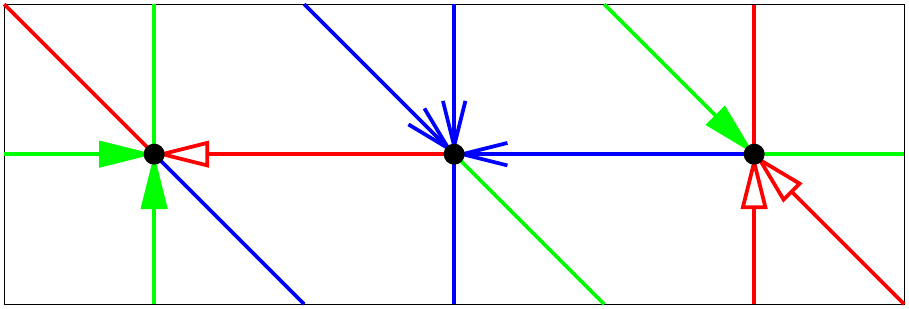} \ & \
                                             \includegraphics[scale=0.4]{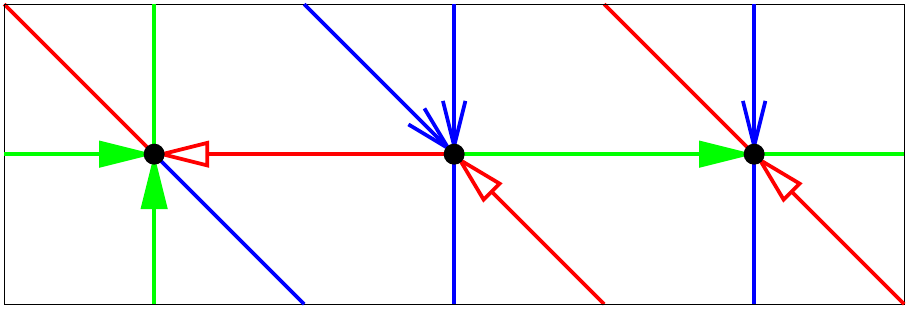}  \\
  Not half-crossing \ &\  Half-crossing \\
\end{tabular}

\vspace{1em}

\begin{tabular}{c}
  \includegraphics[scale=0.4]{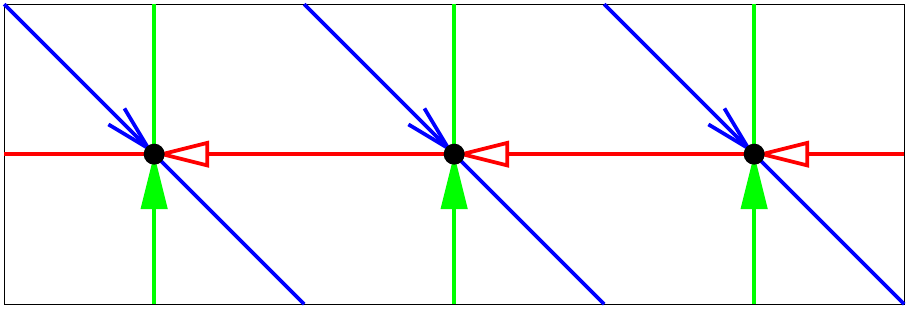} 
  \\
   Crossing
\end{tabular}
\caption{Examples of Schnyder woods of a toroidal triangulation that
  are respectively not half-crossing, half-crossing, crossing.}
\label{fig:cross-0-tri}
\end{figure}

The half-crossing property forces the monochromatic cycles to be all
homologous and not only weakly homologous as in
Lemma~\ref{lem:allhomotopic}:

\begin{lemma}
\label{lem:halfcyclehomologous}
In a half-crossing Schnyder wood, for $i\in\{0,1,2\}$, all the $i$-cycles are homologous.
\end{lemma}
\begin{proof}
  Consider a half-crossing Schnyder wood.  Let $C, C'$ be two
  monochromatic cycles of different colors that are crossing. Let $i\in\{0,1,2\}$.  By
  Lemma~\ref{lem:allhomotopic}, all $i$-cycles are weakly
  homologous. Cycles $C,C'$ are not weakly homologous, so  at least one
  of them is not weakly homologous to the $i$-cycles. Assume,
  w.l.o.g. that $C$ is not weakly homologous to the $i$-cycles. Then
  by   Lemma~\ref{lem:intersect2}, cycle $C$ crosses all the
  $i$-cycles. Thus all the $i$-cycles are entering and leaving $C$
  from the same sides. So they are all  homologous to each other.
\end{proof}

Consider a toroidal map $G$ given with a  Schnyder wood.
Let $\mathcal C_i$ be the set of $i$-cycles of $G$. Let
$(\mathcal C_i)^{-1}$ denote the set of cycles obtained by reversing
all the cycles of $\mathcal C_i$.

The following theorem gives a global structure of the monochromatic
cycles in the different cases that may occur
 (see Figure~\ref{fig:case-cross}):

\begin{theorem}
\label{lem:type-cross}
Let $G$ be a toroidal map given with a Schnyder wood.  Then, for
$i\in\{0,1,2\}$, all $i$-cycles are non-contractible, non intersecting
and weakly homologous.  Moreover, if the Schnyder wood is

\begin{itemize}
\item \emph{Not half-crossing:}

Then all monochromatic cycles are weakly homologous (even if they
have different colors).
\item \emph{Half-crossing:}

Then, for $i\in\{0,1,2\}$, all $i$-cycles are homologous. 
Moreover there exists a color $i$ such
that for any pair of monochromatic cycles $C_i,C_j$ of colors $i,j$,
with $i\neq j$, the two cycles $C_i$ and $C_j$ are crossing (we say
that the Schnyder wood is \emph{$i$-crossing} if we want to
specify the color $i$ in this case).
Moreover, if the Schnyder wood is:

\begin{itemize}
\item \emph{Not crossing:} 

Then the $(i-1)$-cycles and the
  $(i+1)$-cycles are reversely homologous. The Schnyder wood is not
  $(i-1)$-crossing and not $(i+1)$-crossing. Moreover, if the
  Schnyder wood is intersecting then
  $\mathcal C_{i-1}=(\mathcal C_{i+1})^{-1}$.

\item \emph{Crossing:} 

Then for every pair of two monochromatic cycles
  $C_i,C_j$ of different colors $i,j$, the two cycles $C_i$ and $C_j$
  are crossing. The Schnyder wood is $i$-crossing for all
  $i \in \{0,1,2\}$.
\end{itemize}
\end{itemize}
\end{theorem}

\begin{proof}
  By Lemma~\ref{lem:allhomotopic}, for a given color $i\in\{0,1,2\}$,
  all $i$-cycles are non-contractible, non intersecting and weakly
  homologous.   Now we consider the two cases whether the Schnyder wood is
  half-crossing or not:

  \begin{itemize}

  \item \emph{The Schnyder wood is not half-crossing:} 

    Consider two cycles $C,C'$ of distinct colors.  Then by definition
    of half-crossing, $C,C'$ are either not intersecting or
    reversal. If they are not intersecting, they are weakly homologous
    by Lemma~\ref{lem:intersect2}, if they are reversal also by
    definition. So $C,C'$ are weakly homologous.

  \item \emph{The Schnyder wood is half-crossing:}

    Let $C, C'$ be two monochromatic cycles of different colors that
    are crossing. For $i\in\{0,1,2\}$, by
    Lemma~\ref{lem:allhomotopic}, the $i$-cycles are
    homologous. Cycles $C,C'$ are not weakly homologous, so for all
    $i\in\{0,1,2\}$, at least one of them is not weakly homologous to
    the $i$-cycles. There is three value for $i$ and two cycles
    $C,C'$. So there exists $i\in\{0,1,2\}$, such that at least one of
    $C,C'$ is not weakly homologous to all the $(i-1)$-cycles and
    $(i+1)$-cycles. Assume w.l.o.g. that $C$ satisfies the property.
    Then by Lemma~\ref{lem:intersect2}, cycle $C$ intersects, and thus
    crosses, all the $(i-1)$- and $(i+1)$-cycles.  Then $C$ is of
    color $i$ and all $i$-cycles are homologous to $C$ so they all
    crosses the $(i-1)$- and $(i+1)$-cycles and the Schnyder wood is
    $i$-crossing.  

    Now we consider the two cases whether the Schnyder wood is
    crossing or not (we still consider $i$ as above such that the
    Schnyder wood is $i$-crossing):

  \begin{itemize}
  \item \emph{The Schnyder wood is not crossing:} 

    Since the Schnyder wood is not crossing there exists two cycles of
    different colors that are not crossing. None of these cycles are
    of color $i$ since the Schnyder wood is $i$-crossing.  So there
    exists a $(i-1)$-cycle $C_{i-1}$ and a $(i+1)$-cycle $C_{i+1}$
    that are not crossing. If they are intersecting, then they are
    reversal, and so weakly homologous by definition of crossing. If
    they are not intersecting, then they are also weakly homologous by
    Lemma~\ref{lem:intersect2}. So in both cases, $C_{i-1}, C_{i+1}$
    are weakly homologous and not crossing. They both crosses a
    $i$-cycle $C_i$, and by the Schnyder property $C_{i-1}$ is
    crossing $C_i$ from left to right, and $C_{i+1}$ is crossing $C_i$
    from right to left. So $C_{i-1}, C_{i+1}$ are reversely
    homologous, so the $(i-1)$-cycles and the $(i+1)$-cycles are
    reversely homologous and by Lemma~\ref{lem:tworeversal}. Moreover
    the Schnyder wood is not $(i-1)$-crossing or $(i+1)$-crossing
    since $C_{i-1}, C_{i+1}$ are not crossing..

  Suppose moreover that the Schnyder wood is intersecting.  Let
  $C'_{i-1}$ be any $(i-1)$-cycle.  By assumption, $C'_{i-1}$
  intersects a $(i+1)$-cycle $C'_{i+1}$. Cycles $C'_{i-1}, C'_{i+1}$
  are reversely homologous. So by Lemma~\ref{lem:tworeversal},
  $C'_{i-1}$ and $C'_{i+1}$ are reversal. Thus
  $\mathcal C_{i-1}\subseteq(\mathcal C_{i+1})^{-1}$ and by symmetry
  $\mathcal C_{i-1}=(\mathcal C_{i+1})^{-1}$.

 \item \emph{The Schnyder wood is  crossing:}

   Since the Schnyder wood is $i$-crossing we just have to prove that
   $(i-1)$-cycles and $(i+1)$-cycles crosses each other.  Consider a
   $(i-1)$-cycle $C_{i-1}$ and a $(i+1)$-cycle $C_{i+1}$. Cycle
   $C_{i-1}$ crosses at least one $(i+1)$-cycle $C'_{i+1}$. All the
   $(i+1)$-cycles are homologous to each other so $C_{i-1},C_{i+1}$
   are also crossing.
  \end{itemize}
  \end{itemize}
\end{proof}

Figure~\ref{fig:case-cross} illustrates the different case of
Theorem~\ref{lem:type-cross}. 
 
\begin{figure}[!h]
\center
\begin{tabular}{ccc}
\includegraphics[scale=0.3]{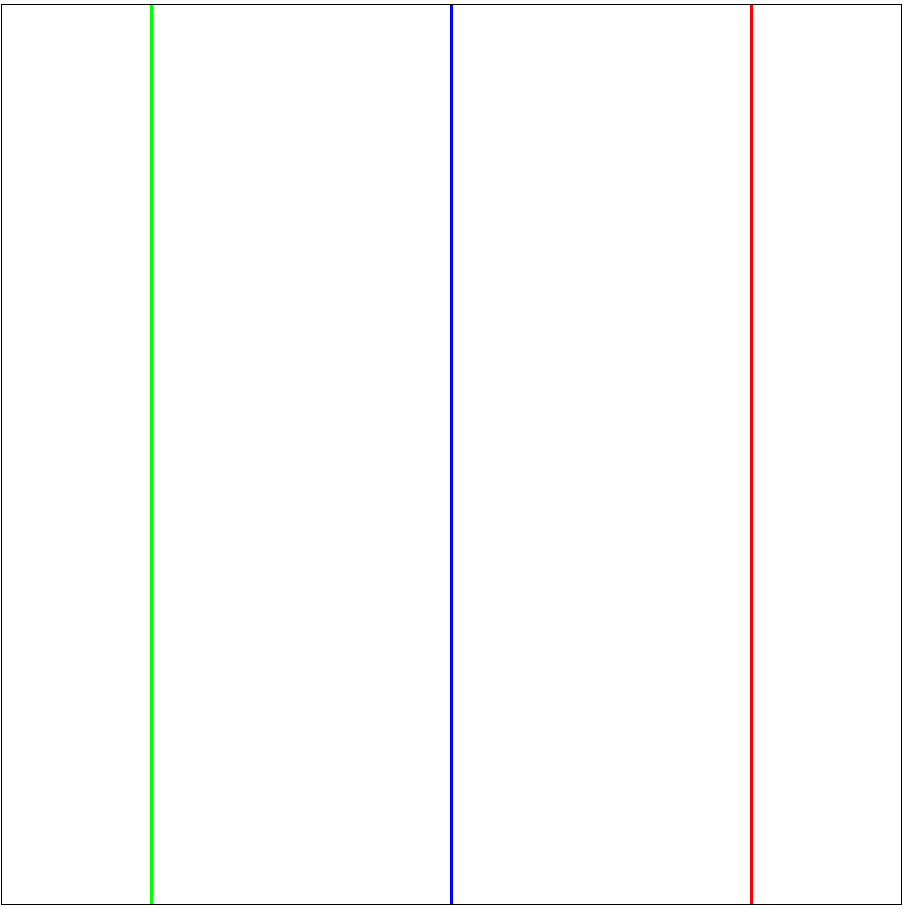}
& \hspace{2em} & 
\includegraphics[scale=0.3]{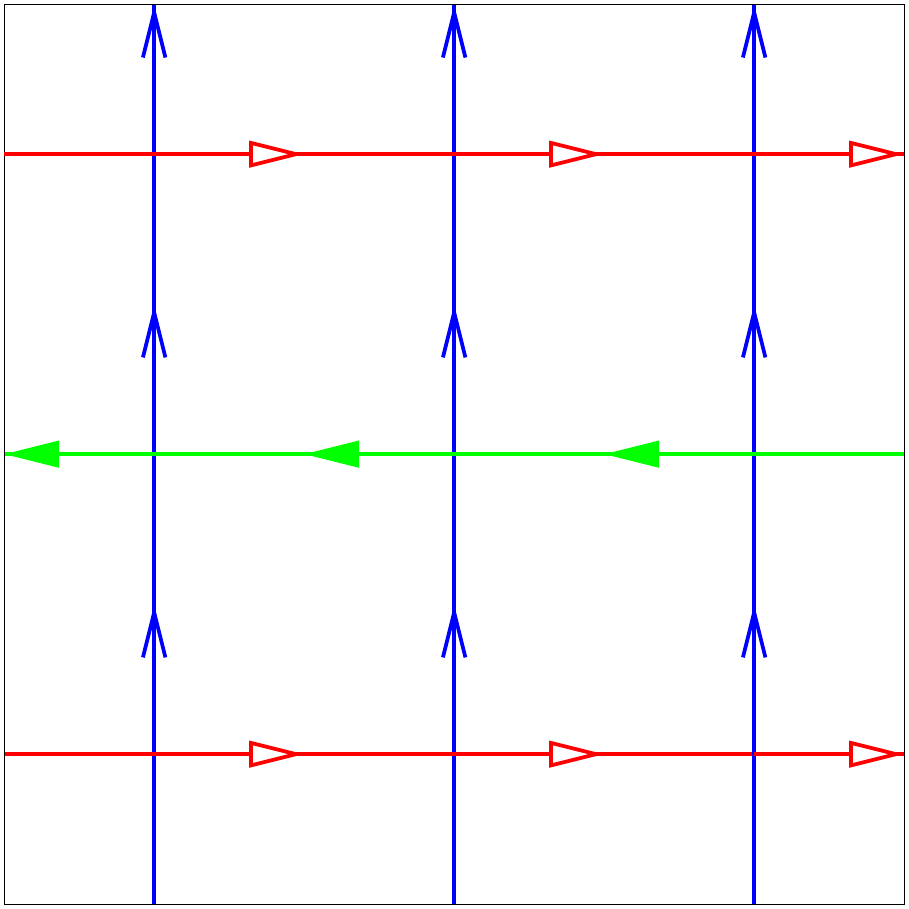}
\\
Not half-crossing & & Half-crossing\\
\\
\includegraphics[scale=0.3]{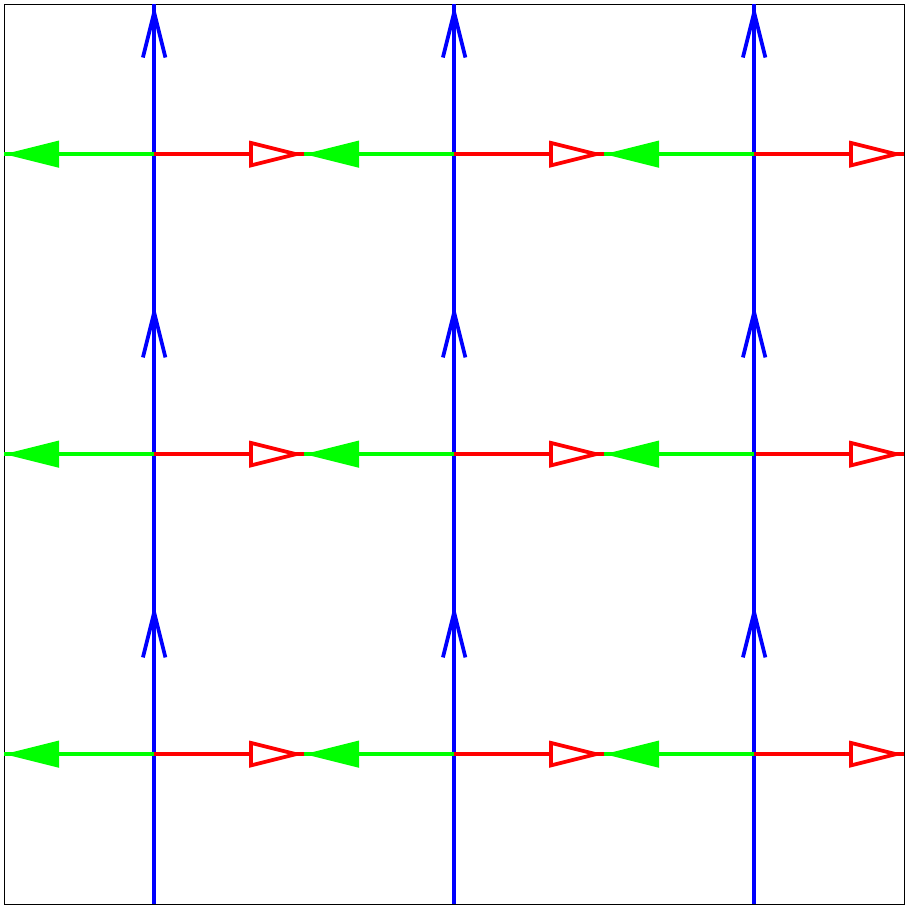} 
& \hspace{2em} & 
\includegraphics[scale=0.3]{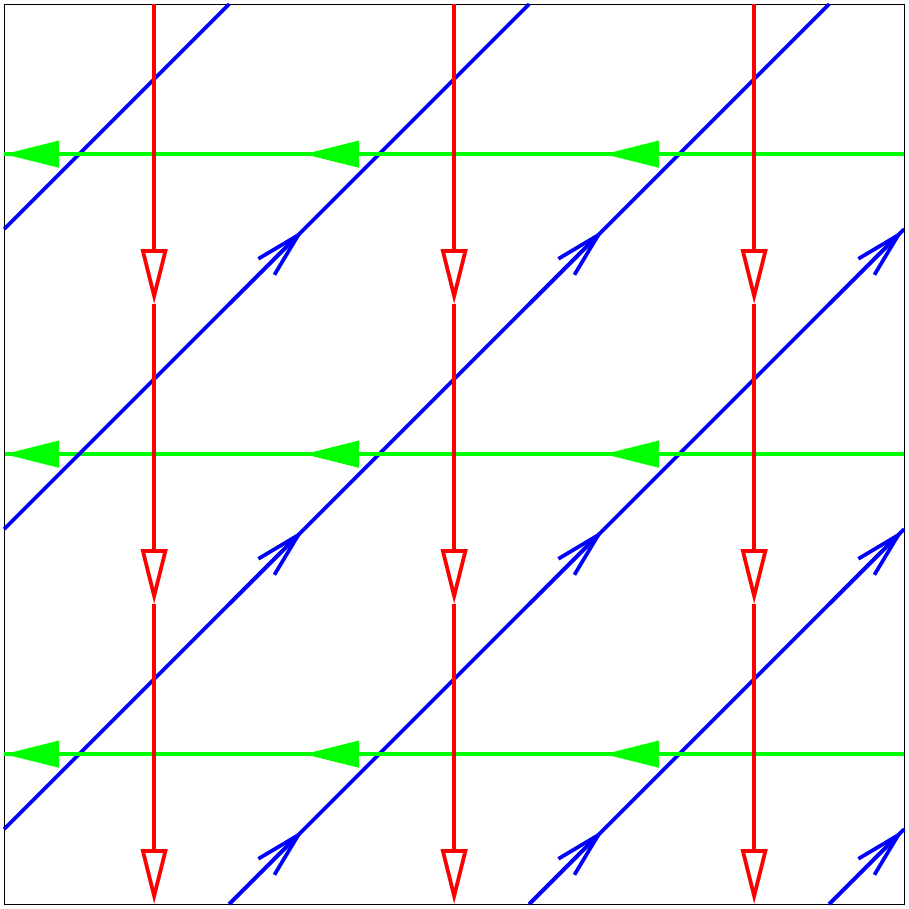} \\
Intersecting & & Crossing\\
\end{tabular}

\caption{The global structure of monochromatic cycles in Schnyder
  woods.}
\label{fig:case-cross}
\end{figure}

Note that a Schnyder wood that is not half-crossing has winding
numbers $\omega_0=\omega_1=\omega_2=0$, a $i$-crossing Schnyder wood
that is not crossing has $\omega_{i-1}=\omega_{i+1}=0$, and a crossing
Schnyder wood has $\omega_0\neq 0, \omega_1\neq 0, \omega_2\neq 0$.

Half-crossing Schnyder
woods are  particularly interesting  since they enable to
define regions in the universal cover, generalizing
Figure~\ref{fig:regions-planar}, that behave similarly as in
the planar case (see Section~\ref{sec:crossinguniversalcover}).

Half-crossing Schnyder woods of a given toroidal map have the nice
property to be all homologous to each other (see
Chapter~\ref{sec:balanced}), thus they all lie in the same
lattice. This lattice behave as a canonical lattice that can be used
to obtain interesting bijections (see Chapter~\ref{chap:encoding}).

We prove in Chapter~\ref{chap:generalexistence} that every essentially
3-connected toroidal map admits an intersecting Schnyder wood. In
fact, we prove that every essentially 3-connected toroidal map admits
a crossing Schnyder wood except for a very simple family of maps
depicted on Figure~\ref{fig:basic} (a non-contractible cycle on $n$
vertices, $n\geq 1$, plus $n$ homologous loops).

\begin{figure}[!h]
\center
\includegraphics[scale=0.5]{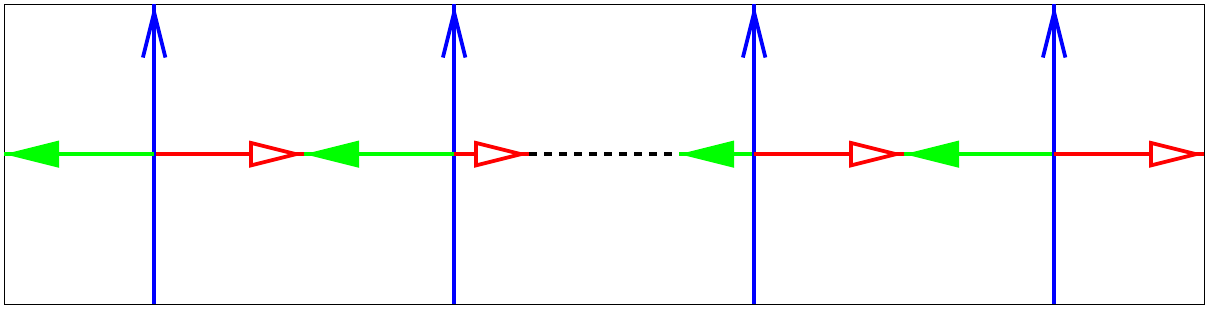}
\caption{The family of basic toroidal maps that admits intersecting
  Schnyder woods but no crossing Schnyder woods.}
\label{fig:basic}
\end{figure}

Intersecting Schnyder woods, can be used to embed a toroidal map on an
orthogonal surfaces and this is particularly interesting for graph
drawing purpose (see Chapter~\ref{sec:ortho}).

\section{Crossing  and duality}

Consider a toroidal map $G$ given with a Schnyder wood.  Let
$\overline{G}$ be a simultaneous drawing of $G$ and $G^*$ such that
only dual edges intersect (see example of
Figure~\ref{fig:example-superposition}).

\begin{figure}[!h]
\center
\includegraphics[scale=0.4]{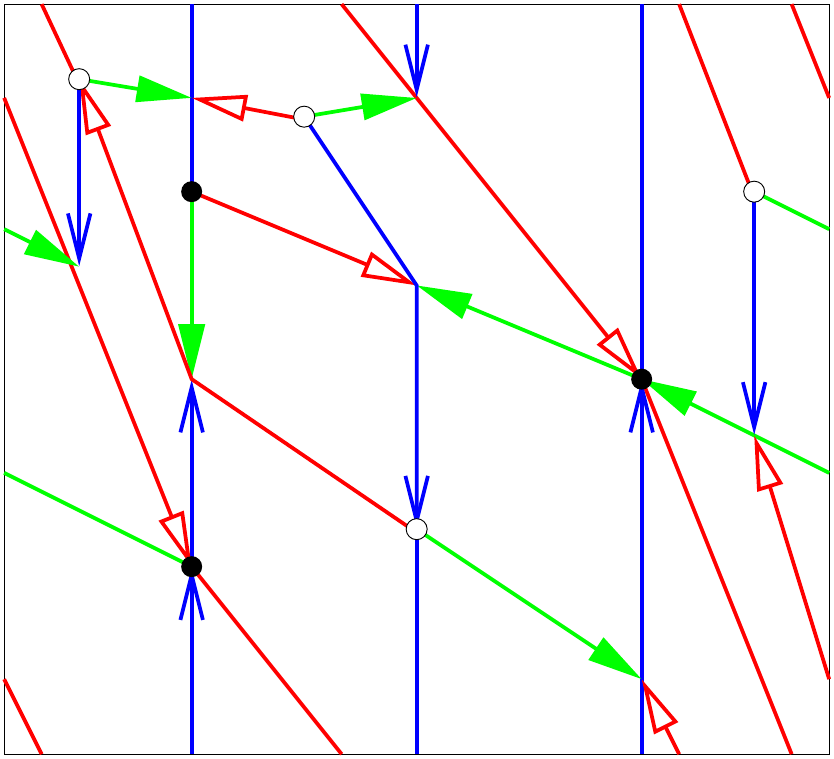}
\caption{Simultaneous drawing of the primal and dual Schnyder wood of
  Figure~\ref{fig:example-dual-tore}.}
\label{fig:example-superposition}
\end{figure}

The dual of a Schnyder wood is defined in such a way that an edge of
$G$ and an edge of $G^*$ of the same color never intersect in
$\overline{G}$ (see Figure~\ref{fig:edgelabelingdual}). Thus the
$i$-cycles of $G^*$ are weakly homologous to $i$-cycles of $G$.  When
the Schnyder wood is half-crossing, these cycles are in fact reversely
homologous:

\begin{lemma}
\label{lem:dualreverselyhomologous}
If the Schnyder wood is half-crossing, then, for $i\in\{0,1,2\}$, the
$i$-cycles of $G^*$ are reversely homologous to the $i$-cycles of $G$.
\end{lemma}

\begin{proof}
  By Theorem~\ref{lem:type-cross}, there is a color $i$ such that the
  Schnyder wood is $i$-crossing.  Let $C_i,C_{i+1}$ be a $i$- and
  $(i+1)$-cycle of $G$.  The two cycles $C_i$ and $C_{i+1}$ are
  crossing. Let $j$ be a color in $\{0,1,2\}$. Let $C_j$ and $C^*_j$
  be $j$-cycles of the primal and the dual Schnyder wood
  respectively. Since $C_i,C_{i+1}$ are crossing, cycle $C_j$ is
  crossing at least one of $C_i,C_{i+1}$. Since $C^*_j$ is weakly
  homologous to $C_j$, thus it is crossing the same cycles of
  $\{C_i,C_{i+1}\}$ as $C_j$ but in opposite direction by the dual
  rules of Figure~\ref{fig:edgelabelingdual}.  So $C_j$ and $C^*_j$ are
  reversely homologous.
\end{proof}

By Lemma~\ref{lem:dualreverselyhomologous}, the property of being
$i$-crossing thus goes to duality, i.e. a Schnyder wood is
$i$-crossing if and only if the dual Schnyder wood is
$i$-crossing. Consequently the same is true for the half-crossing and
crossing properties by Theorem~\ref{lem:type-cross}.  The intersecting
property also goes to duality but this is more difficult to
prove. Finally, we have the following:

\begin{proposition}
\label{th:dualtorecross}
A Schnyder wood of a toroidal map is crossing (resp. $i$-crossing,
half-crossing, intersecting) if and only if the dual Schnyder wood is
crossing (resp. $i$-crossing, half-crossing, intersecting).
\end{proposition}

\begin{proof}
  The result is clear by Theorem~\ref{lem:type-cross} and
  Lemma~\ref{lem:dualreverselyhomologous} for crossing, $i$-crossing
  and half-crossing properties.

Consider now an intersecting Schnyder wood that is not crossing.  By
Theorem~\ref{lem:type-cross} and previous remarks, there is a color
$i$ such that the primal and dual Schnyder woods are
$i$-crossing. Suppose by contradiction that the dual Schnyder wood is
not intersecting.  Then there is a $j$-cycle $C^*$, with
$j\in \{i-1,i+1\}$, that is not intersecting a monochromatic cycle of
color in $\{i-1,i+1\}\setminus \{j\}$.  By symmetry we can assume that
$C^*$ is of color $i-1$.  By Lemma~\ref{lem:dualreverselyhomologous},
the $(i-1)$-cycles of $G^*$ are reversely homologous to the
$(i-1)$-cycles of $G$.  Let $C$ be the $(i-1)$-cycle of the primal
that is the first on the right side of $C^*$ in $\overline{G}$. By
Theorem~\ref{lem:type-cross}, cycle $C^{-1}$ is a $(i+1)$-cycle of
$G$. Let $R$ be the region delimited by $C^*$ and $C$ situated on the
right side of $C^*$.  Cycle $C^*$ is not a $(i+1)$-cycle so there is
at least one edge of color $i+1$ leaving a vertex of $C^*$. By the
Schnyder rule, this edge is entering the interior of the region
$R$. An edge of $G^*$ of color $i+1$ cannot intersect $C$ and cannot
enter $C^*$ from its right side. So in the interior of the region $R$
there is at least one $(i+1)$-cycle $C^{*}_{i+1}$ of $G^*$.  Cycle
$C^{*}_{i+1}$ is reversely homologous to $C^{-1}$ by
Lemma~\ref{lem:dualreverselyhomologous}. If $C^{*}_{i+1}$ is not a
$(i-1)$-cycle, then we can define $R'\subsetneq R$ the region
delimited by $C^{*}_{i+1}$ and $C$ situated on the left side of
$C^{*}_{i+1}$ and as before we can prove that there is a $(i-1)$-cycle
of $G^*$ in the interior of $R'$. So in any case, there is a
$(i-1)$-cycle $C^*_{i-1}$ of $G^*$ in the interior of $R$ and
$C^{*}_{i-1}$ is homologous to $C^{*}$. Let $R''\subsetneq R$ be the
region delimited by $C^*$ and $C^*_{i-1}$ situated on the right side
of $C^*$. Clearly $R''$ does not contain $C$. Thus by definition of
$C$, the region $R''$ does not contain any $(i-1)$-cycle of $G$. But
$R''$ is non empty and contains at least one vertex $v$ of $G$. The
path $P_{i-1}(v)$ cannot leave $R''$, a contradiction. So the dual
Schnyder wood is intersecting and we have the result.
  \end{proof}

\section{Crossing  in the universal cover}
\label{sec:crossinguniversalcover}

Consider a toroidal map $G$, given with a Schnyder wood. Consider the
corresponding orientation and coloring of the edges of the universal
cover $G^\infty$.

As monochromatic cycles are non-contractible by
Theorem~\ref{lem:type-cross}, a $i$-cycle corresponds to a family of
infinite directed monochromatic paths of $G^\infty$ (infinite in both
directions of the path). These paths are called \emph{monochromatic
  lines} of color $i$, or \emph{$i$-lines} for short. Since all the
$i$-cycles are non intersecting and weakly homologous, all the
$i$-lines are non intersecting either. The $i$-lines are kind of
``parallel'' in $G^\infty$. Moreover if the Schnyder wood is
half-crossing all the $i$-cycles are homologous, so all the $i$-lines
are oriented in the same direction.

The definition of $P_i(v)$ (defined in
Section~\ref{sec:universalcover}) when restricted to the toroidal case
is the following: For each color $i$, vertex $v$ is the starting
vertex of a infinite path of color $i$, denoted $P_i(v)$.  By
Lemma~\ref{lem:nocommongeneral}, for every vertex $v$, the three paths
$P_0(v)$, $P_1(v)$, $P_2(v)$ are disjoint on all their vertices except
on $v$. Thus they divide $G^\infty$ into three unbounded regions
$R_0(v)$, $R_1(v)$ and $R_2(v)$, where $R_i(v)$ denotes the region
delimited by the two paths $P_{i-1}(v)$ and $P_{i+1}(v)$ (see
Figure~\ref{fig:regions-tore}). Let
$R_i^{\circ}(v)=R_i(v)\setminus (P_{i-1}(v)\cup P_{i+1}(v))$.  As
$P_i(v)$ is infinite, it necessarily contains two vertices $u,u'$ of
$G^\infty$ that are copies of the same vertex of $G$. The subpath of
$P_i(v)$ between $u$ and $u'$ corresponds to a $i$-cycle of $G$ and
thus is part of a $i$-line of $G^\infty$.  Let $L_i(v)$ be the
$i$-line intersecting $P_i(v)$.
 
\begin{figure}[!h]
\center

\scalebox{0.5}{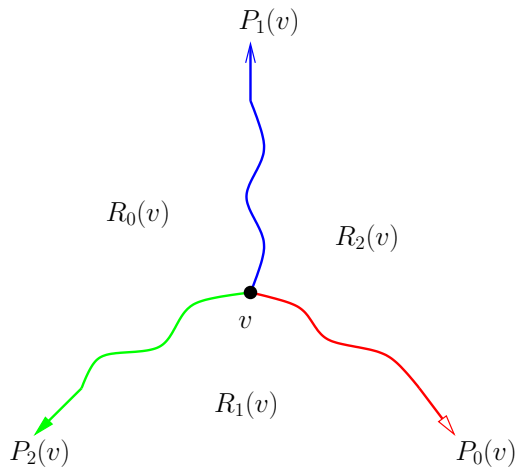}
\caption{Regions corresponding to a vertex}
\label{fig:regions-tore}
\end{figure}

When the Schnyder wood has the global property of being half-crossing,
then it is well structured and the regions $R_i(v)$ behave
similarly as in the planar case (see \cite{Fel01,Fel03}):

\begin{lemma}
  \label{lem:regionss} 
  If the Schnyder wood is half-crossing, then for all distinct
  vertices $u$, $v$, we have:
\begin{description}
\item[(i)] If $u\in R_i (v)$, then $R_i(u)\subseteq R_i(v)$.
\item[(ii)] If $u\in R_i^\circ (v)$, then $R_i(u)\subsetneq R_i(v)$.
\item[(iii)] There exists $i$ and $j$ with $R_i(u) \subsetneq R_i(v)$
  and $R_j(v)\subsetneq R_j (u)$.
\end{description}
\end{lemma}

\begin{proof}
  (i) Suppose that $u\in R_i (v)$. By Theorem~\ref{lem:type-cross},
  there exist $k\in\{0,1,2\}$ such that the Schnyder wood is
  $k$-crossing.  By symmetry we can assume that $k\in\{i-1,i\}$.  Then
  in $G$, $i$-cycles are not weakly homologous to $(i-1)$-cycles. Thus
  in $G^\infty$, every $i$-line crosses every $(i-1)$-line. Moreover a
  $i$-line crosses a $(i-1)$-line exactly once and from its right side
  to its left side.  Vertex $v$ is between two consecutive
  monochromatic $(i-1)$-lines $L_{i-1},L'_{i-1}$, with $L'_{i-1}$
  situated on the right of $L_{i-1}$.  Let $R$ be the region situated
  on the right of $L_{i-1}$, so $v\in R$.

  \begin{claim}
    \label{cl:leaves}
    For any vertex $w$ of $R$, the path $P_{i}(w)$ leaves the region $R$.
  \end{claim}

  \begin{proofclaim}
    The $i$-line $L_i(w)$ has to cross $L_{i-1}$ exactly once and from
    right to left, thus $P_{i}(w)$ leaves the region $R$.
  \end{proofclaim}

  The paths $P_{i-1}(v)$ and $P_{i+1}(v)$ cannot leave the region $R$
  by the Schnyder property.  Thus by Claim~\ref{cl:leaves}, we have
  $P_{i}(v)$ leaves the region $R$ and $R_i(v)\subseteq R$. So
  $u\in R$. The paths $P_{i-1}(u)$ and $P_{i+1}(u)$ cannot leave
  region $R_i(v)$. By Claim~\ref{cl:leaves}, the path $P_{i}(u)$
  leaves the region $R_i(v)$ and so $R_i(u)\subseteq R_i(v)$.

  (ii) Suppose that $u\in R_i^\circ (v)$. By (i), $R_i(u)\subseteq R_i(v)$. The paths $P_{i-1}(u)$ and
  $P_{i+1}(u)$ are contained in $R_i(v)$ and none of them can contain
  $v$ by the Schnyder property.  So all the faces of $R_i(v)$ incident
  to $v$ are not in $R_i(u)$ (and there is at least one such face).

  (iii) By symmetry, we prove that there exists $i$ with $R_i(u)
  \subsetneq R_i(v)$.  If $u\in R_i^\circ(v)$ for some color $i$, then
  $R_i(u) \subsetneq R_i(v)$ by (ii).  Suppose now that $u\in P_i(v)$
  for some $i$.  By Lemma~\ref{lem:nocommongeneral}, at least one of the two
  paths $P_{i-1}(u)$ and $P_{i+1}(u)$ does no contain $v$. Suppose by
  symmetry that $P_{i-1}(u)$ does not contain $v$. As $u\in
  P_i(v)\subseteq R_{i+1}(v)$, we have $R_{i+1}(u)\subseteq
  R_{i+1}(v)$ by (i), and as none of $P_{i-1}(u)$ and $P_{i}(u)$
  contains $v$, we have $R_{i+1}(u) \subsetneq R_{i+1}(v)$.
\end{proof}

Note that if the Schnyder wood is not half-crossing then the
conclusions of Lemma~\ref{lem:regionss} might be false.
Figure~\ref{fig:regions-nothalf} gives an example of a Schnyder wood
of a toroidal map for which the three items of
Lemma~\ref{lem:regionss} are not satisfied. The torus is represented
as a black rectangle. It is replicated one time above and one time
below to draw partially $G^\infty$. The regions $R_i(v)$ of the two
bold vertices are partially represented by dashed colors.  Each bold
vertex is in the strict interior of the red region of the other bold
vertex but these two vertices do not satisfy any conclusions of the
three items of Lemma~\ref{lem:regionss}.

\begin{figure}[!h]
\center
\includegraphics[scale=0.3]{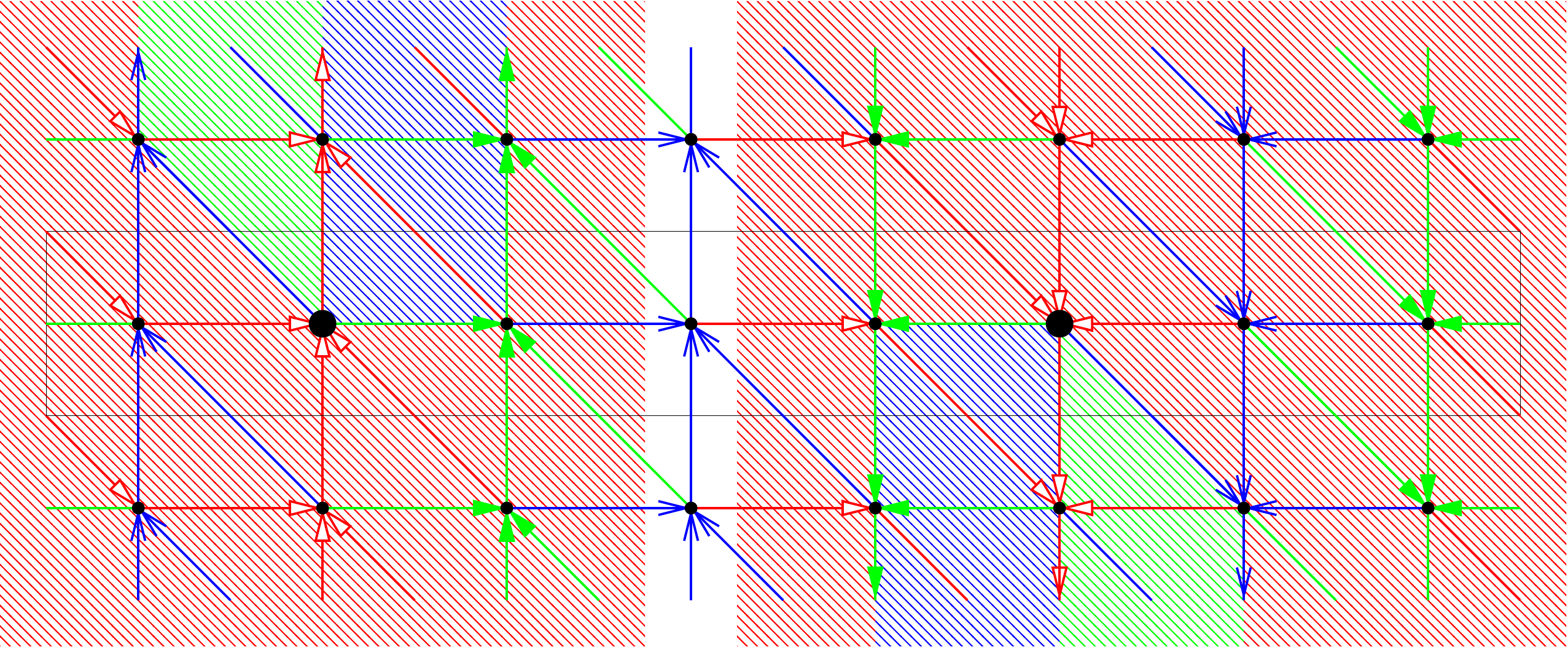} 
\caption{Example of a Schnyder wood that is not half-crossing and
  whose regions do not satisfy the conclusions of
  Lemma~\ref{lem:regionss}.}
\label{fig:regions-nothalf}
\end{figure}

\chapter{Balanced Schnyder woods}
\label{sec:balanced}

\section{Balanced  for maps}
\label{sec:HTCmaps}

In this section we study properties of $\gamma$ and use them to define
a particular set of Schnyder woods.

Consider a toroidal map $G$.
An $\alpha$-orientation of $\pdc{G}$ is called a
\emph{$3$-orientation} if:
 $$\alpha(v)= \begin{cases}
   3 & \text{if }v\text{ is a primal- or dual-vertex}, \\
   1 & \text{if }v\text{ is an edge-vertex.}\\
\end{cases}
$$

Next lemma shows that $\gamma$ behaves well w.r.t.~homologous cycles
in $3$-orientations of $\pdc{G}$: two homologous non-contractible
cycles have the same value $\gamma$. We formulate this in a more general
form:

\begin{lemma}
  \label{lem:gammahomologyequalbasis} 
  Consider a $3$-orientation of $\pdc{G}$, a non-contractible cycle
  $C$ of $G$ or $G^*$, given with a direction of traversal, and a
  homology-basis $\{B_1,B_2\}$ of $G$, such that $B_1,B_2$ are
  non-contractible cycles whose intersection is a single vertex or a
  common path.  If $\pdc{C}$ is homologous to
  $k_1 \pdc{B_1} + k_2\pdc{B_{2}}$, with $k_1,k_2 \in \mathbb Z$, then
  $\gamma(C)=k_1\,\gamma (B_1) + k_2\,\gamma (B_2)$.
\end{lemma}

\begin{proof}
  Let $v$ be a vertex in the intersection of $B_1,B_2$ such that, if
  this intersection is a common path, then $v$ is one of the
  extremities of this path and let $u$ be the other extremity.
  Consider a drawing of $(\pdc{G})^\infty$ obtained by replicating a
  flat representation of $\pdc{G}$ to tile the plane.  Let $v_0$ be a
  copy of $v$ in $(\pdc{G})^\infty$.  Consider the walk $W$ starting
  from $v_0$ and following $k_1$ times the edges corresponding to
  $B_1$ and then $k_2$ times the edges corresponding to $B_2$ (we are
  going backward if $k_i$ is negative). This walk ends at a copy $v_1$
  of $v$.  Since $C$ is non-contractible we have $k_1$ or $k_2$ not
  equal to $0$ and thus $v_1$ is distinct from $v_0$.  Let $W^\infty$
  be the infinite walk obtained by replicating $W$ (forward and
  backward) from $v_0$.  Note that their might be some repetition of
  vertices in $W^\infty$ if the intersection of $B_1,B_2$ is a
  path. But in that case, by the choice of $B_1,B_2$, we have that
  $W^\infty$ is almost a path, except maybe at all the transitions
  from ``$k_1 {B_1}$'' to ``$k_2B_{2}$'', or at all the transitions
  from ``$k_2 {B_2}$'' to ``$k_1B_{1}$'', where it can goes back and
  forth a path corresponding to the intersection of $B_1$ and
  $B_2$. The existence or not of such ``back and forth'' parts depends
  on the signs of $k_1, k_2$ and the way $B_1,B_2$ are going through
  their common path. Figure~\ref{fig:replicating} gives an example of
  this construction with $(k_1,k_2)=(1,1)$ and $(k_1,k_2)=(1,-1)$ when
  $B_1,B_2$ intersects on a path and are oriented the same way along
  this path as on Figure~\ref{fig:replicating0}.

\begin{figure}[!h]
\center
\includegraphics[scale=0.3]{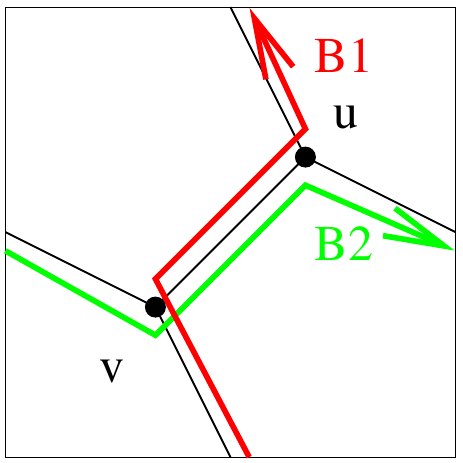}
\caption{Intersection of the basis.}
\label{fig:replicating0}
\end{figure}

\begin{figure}[!h]
\begin{tabular}{cc}
\includegraphics[scale=0.3]{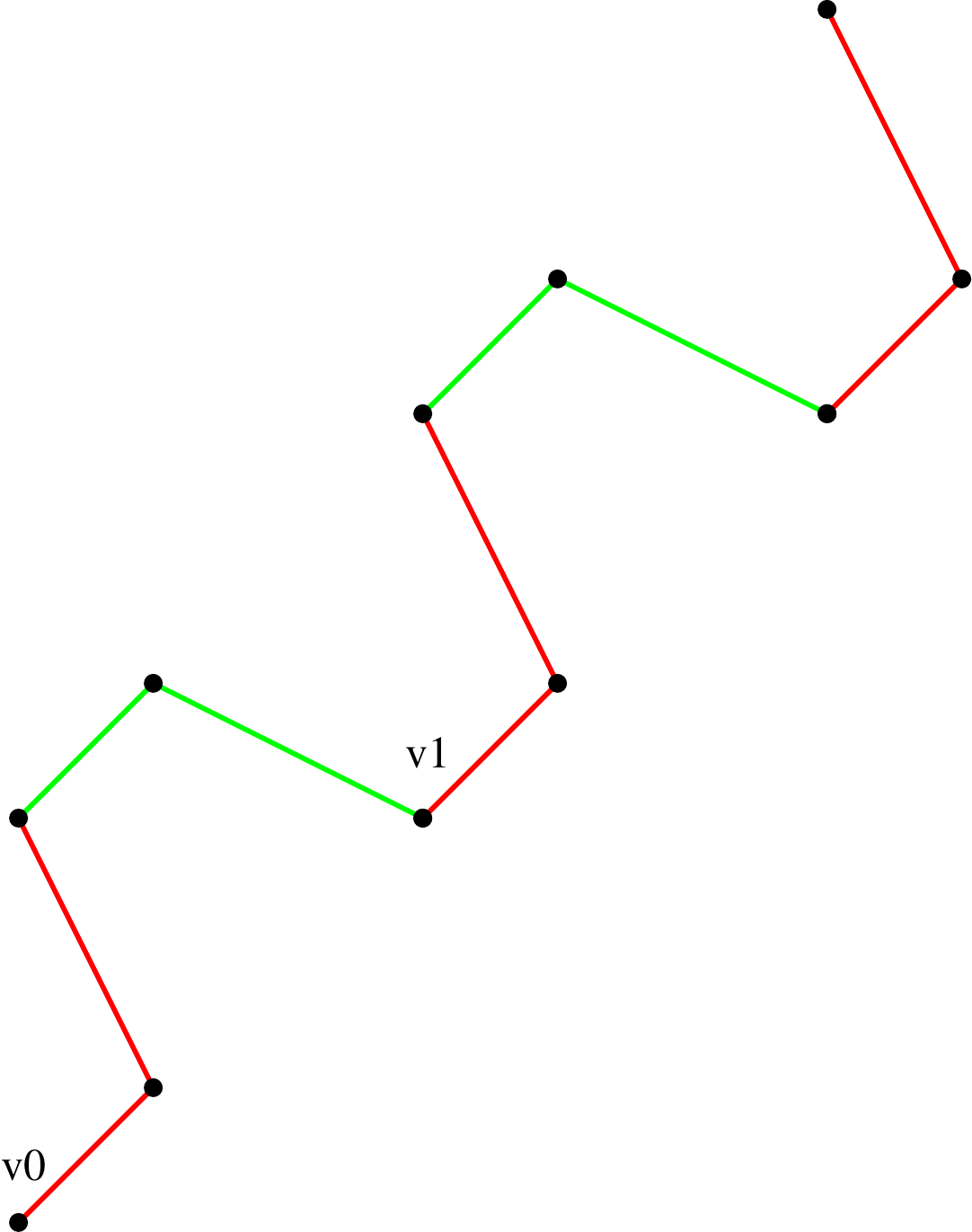} \ \ \ & \ \ \
\includegraphics[scale=0.3]{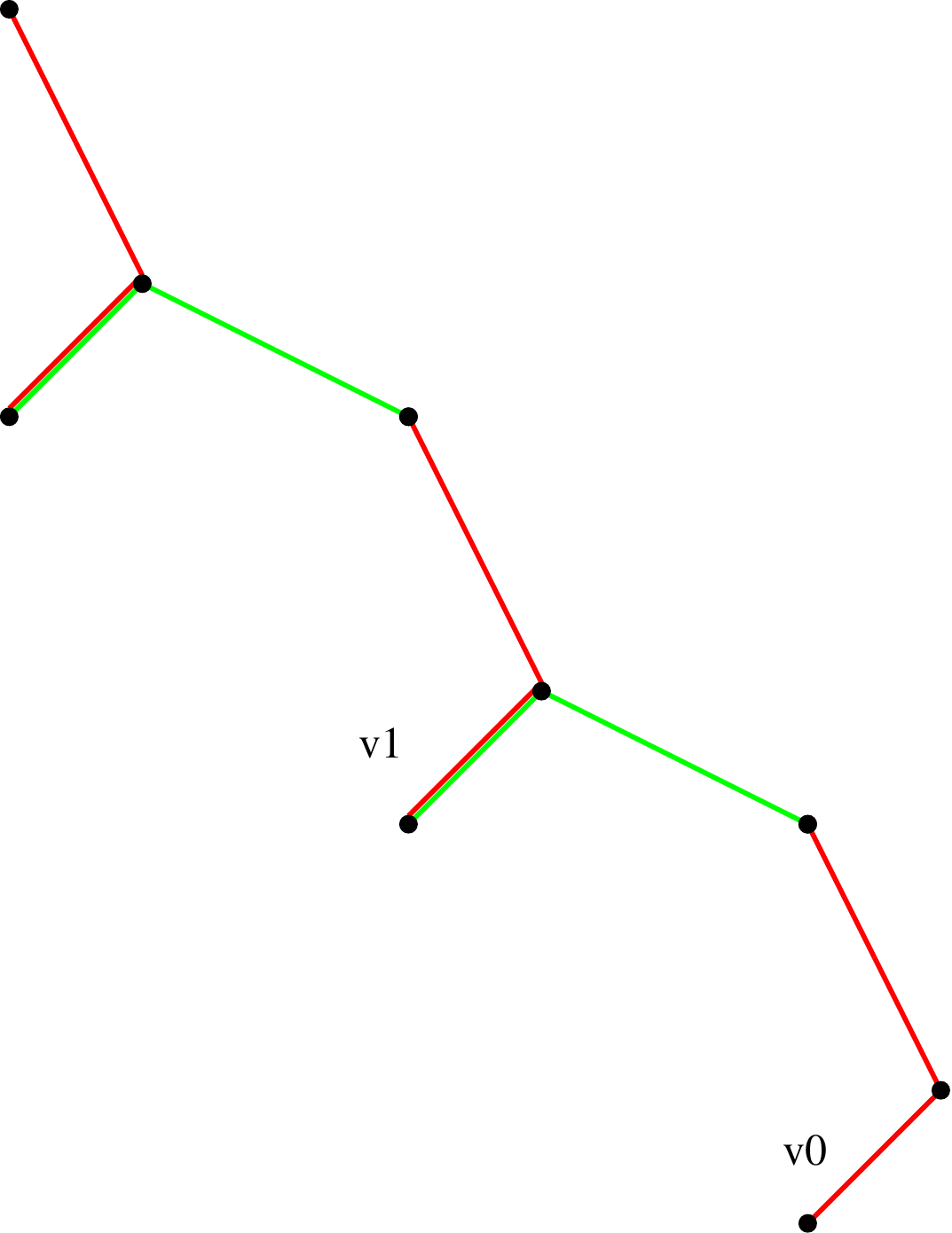} \\
$(k_1,k_2)=(1,1)$ \ \ \ & \ \ \ $(k_1,k_2)=(1,-1)$\\
\end{tabular}
\caption{Replicating ``$k_1 {B_1}$'' and ``$k_2B_{2}$'' in the universal cover.}
\label{fig:replicating}
\end{figure}
 
   We ``simplify'' $W^\infty$ by removing all the parts that consists
   of going back and forth along a path (if any) and call $B^\infty$
   the obtained walk that is now without repetition of vertices. By
   the choice of $v$, we have that $B^\infty$ goes through copies of
   $v$. If $v_0,v_1$ are no more a vertex along $B^\infty$, because of
   a simplification at the transition from ``$k_2 {B_2}$'' to
   ``$k_1B_{1}$'', then we replace $v_0$ and $v_1$ by the next copies
   of $v$ along $W^\infty$, i.e. at the transition from
   ``$k_1 {B_1}$'' to ``$k_2B_{2}$''.

Since $\pdc{C}$ is homologous to $k_1 \pdc{B_1} + k_{2}\pdc{B_{2}}$,
we can find an infinite path $C^\infty$, that corresponds to copies of
$C$ replicated, that does not intersect $B^\infty$ and situated on the
right side of $B^\infty$. Now we can find a copy $B'^\infty$ of
$B^\infty$, such that $C^\infty$ lies between $B^\infty$ and
$B'^\infty$ without intersecting them. Choose two copies $v'_0,v'_1$
of $v_0,v_1$ on $B'^\infty$ such that the vectors $v_0v_1$ and
$v'_0v'_1$  are equal.  

Let $R_0$ be the region bounded by $B^\infty,B'^\infty$.  Let $R_1$
(resp. $R_2$) be the subregion of $R_0$ delimited by $B^\infty$ and
$C^\infty$ (resp. by $C^\infty$ and $B'^\infty$).  We consider
$R_0,R_1,R_2$ as cylinders, where the segments $[v_0,v'_0],[v_1,v_1']$
(or part of them) are identified.  Let $B,B',C'$ be the cycles of
$R_0$ corresponding to $B^\infty, B'^\infty, C'^\infty$ respectively.
Let $x$ be the number of edges of $\pdc{G}$ leaving $B$ in $R_1$. Let
$y$ be the number of edges of $\pdc{G}$ leaving $B'$ in $R_2$. Let
$x'$ (resp. $y'$) be the number of edges of $\pdc{G}$ leaving $C'$ in
$R_2$ (resp. $R_1$).  We have $C'$ corresponds to exactly one copy of
$C$, so $\gamma(C)=x'-y'$.  Similarly, we have $B$ and $B'$ that
almost corresponds to $k_1$ copies of $B_1$ followed by $k_2$ copies
of $B_2$, except the fact that we may have removed a back and forth
part (if any).  In any case we have the following:

\begin{claim}
\label{cl:computegamma}
  $k_1\,\gamma
(B_1) + k_2\,\gamma (B_2)=x-y$
\end{claim}

 \begin{proofclaim}
   We prove the case where the common intersection of $B_1,B_2$ is a
   path (if the intersection is a single vertex, the proof is very
   similar and even simpler). We assume, w.l.o.g, by eventually
   reversing one of $B_1$ or $B_2$, that $B_1,B_2$ are oriented the
   same way along their intersection, so we are in the situation of
    Figure~\ref{fig:replicating0}.  

Figure~\ref{fig:computegamma1} shows
   how to compute $k_1\,\gamma (B_1) + k_2\,\gamma (B_2)+y-x$ when
   $(k_1,k_2)=(1,1)$. Then, one can check that
   each outgoing edge of $\pdc{G}$ is counted exactly the same number
   of time positively and negatively. So everything compensate and we
   obtain $k_1\,\gamma (B_1) + k_2\,\gamma (B_2)+y-x=0$.

\begin{figure}[!h]
\center
\begin{tabular}{cccccc}
\includegraphics[scale=0.27]{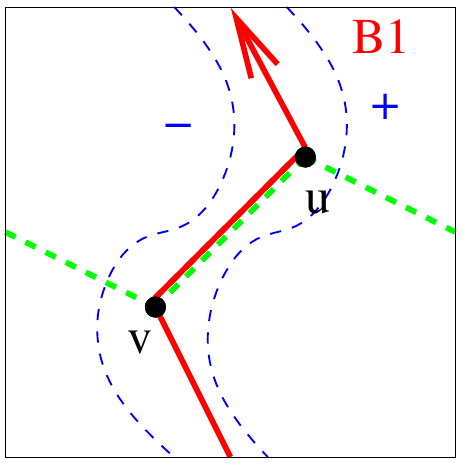}
&&\includegraphics[scale=0.27]{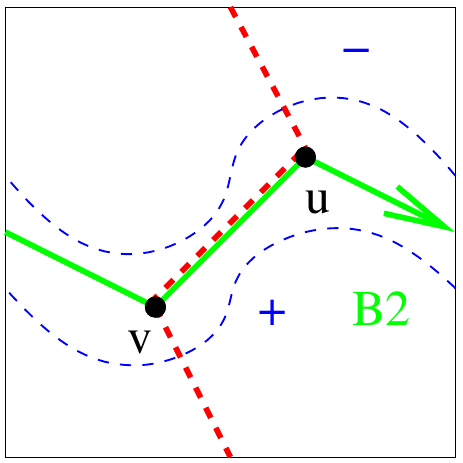}
&&\includegraphics[scale=0.27]{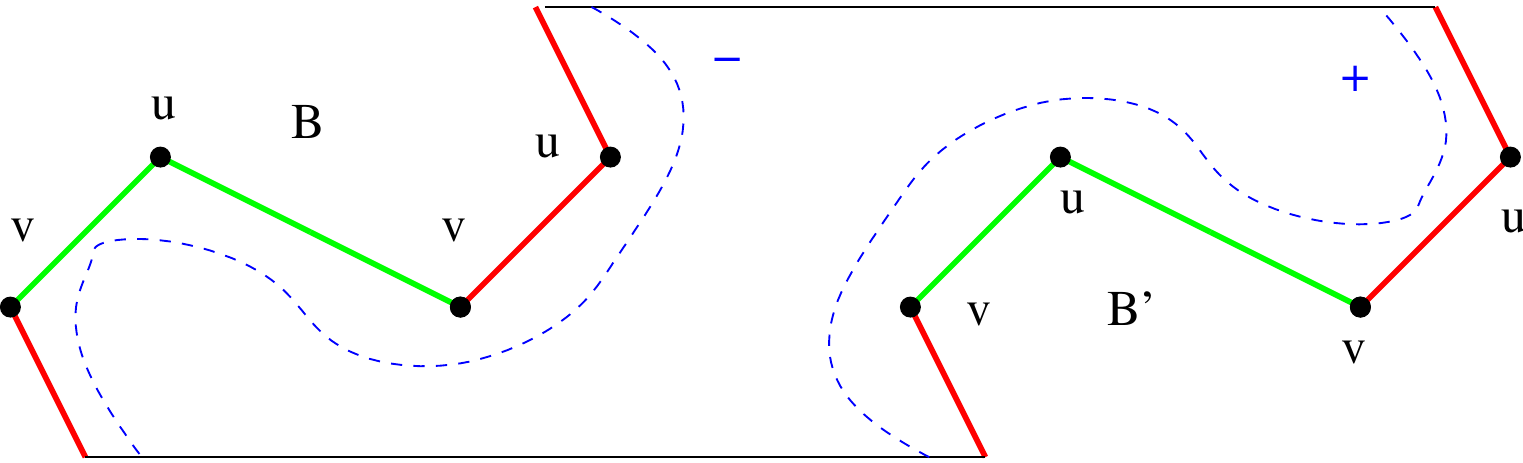} &\\
$\gamma(B_1)$&$+$ & $\gamma (B_2)$ &$+$& $(y-x)$ &$=0$
\end{tabular}
\caption{Case $(k_1,k_2)=(1,1)$.}
\label{fig:computegamma1}
\end{figure}

Figure~\ref{fig:computegamma2} shows how to compute
$k_1\,\gamma (B_1) + k_2\,\gamma (B_2)+y-x$ when
$(k_1,k_2)=(1,-1)$. As above, most of the things compensate but, in the
end, we
obtain $k_1\,\gamma (B_1) + k_2\,\gamma (B_2)+y-x=d^+(u)-d^+(v)$, as
depicted on the figure.
Since the number of outgoing edges around each vertex is
equal to $3$, we have again the conclusion
$k_1\,\gamma (B_1) + k_2\,\gamma (B_2)+y-x=0$.

\begin{figure}[!h]
\center
\begin{tabular}{ccccc}
\includegraphics[scale=0.3]{gammabasis-5}
&&\includegraphics[scale=0.3]{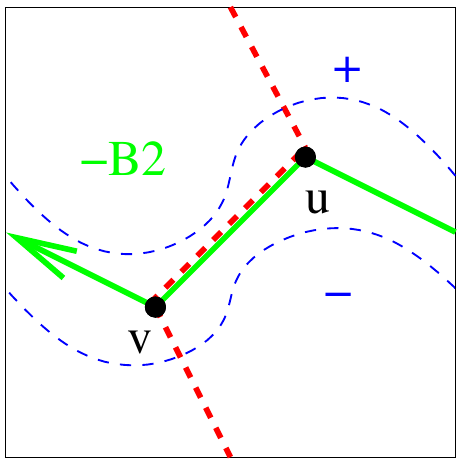}
&&\includegraphics[scale=0.3]{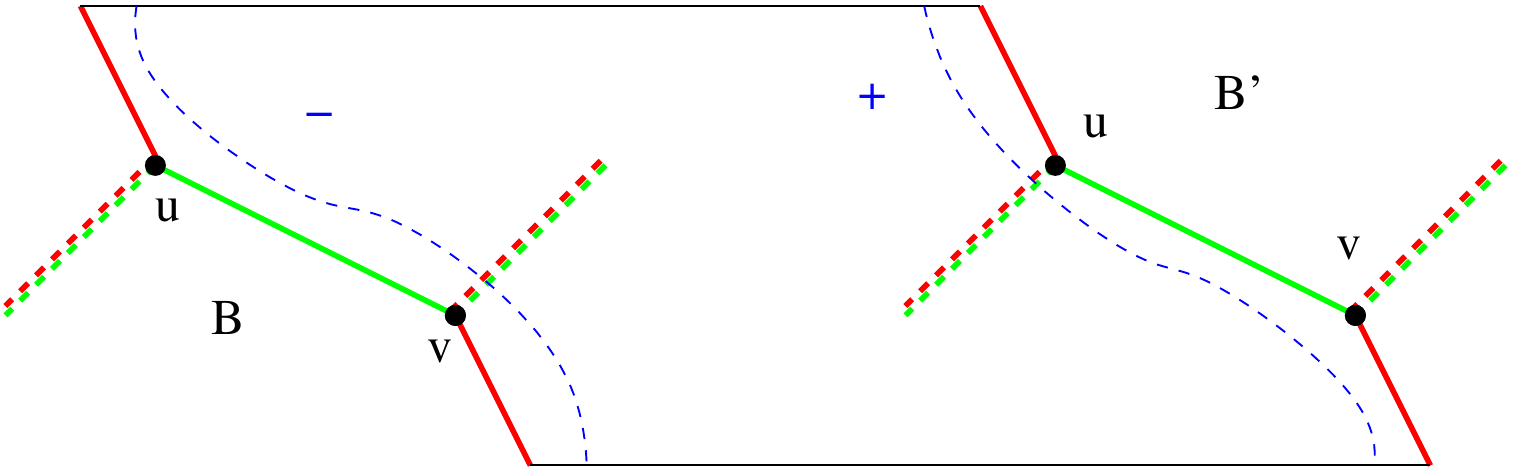}\\
$\gamma(B_1)$&$+$ & $(-\gamma (B_2))$ &$+$& $(y-x)$
\end{tabular}

\includegraphics[scale=0.3]{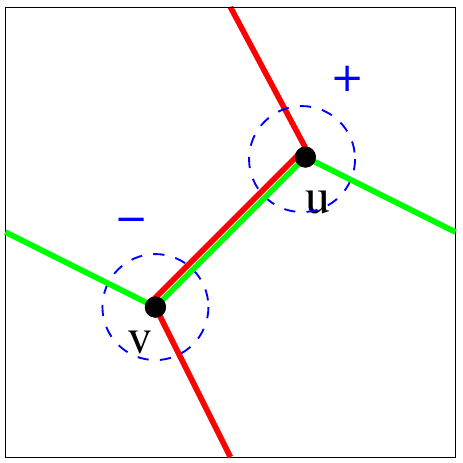}\\
$=d^+(u)-d^+(v)=0$
\caption{Case $(k_1,k_2)=(1,-1)$..}
\label{fig:computegamma2}
\end{figure}

One can easily be
convinced that when $|k_1|\geq 1$ and $|k_2|\geq 1$ then the same arguments
 apply. The only difference is that the red or green part of the
figures in the universal cover would be
longer (with  repetitions of $B_1$ and $B_2$). This parts being very
``clean'', they do not affect the way we compute the  equality.
Finally, if one of $k_1$ or $k_2$ is equal to zero, the analysis
is much  simpler and the conclusion holds.
 \end{proofclaim}

For $i\in\{1,2\}$, let $H_i$ be the
cylinder map made of all the vertices and edges of $\pdc{G}^\infty$
that are in the cylinder region $R_i$.
  Let $k$ (resp. $k'$) be the length of $B$ (resp. $C'$).  Let
  $n_1,m_1,f_1$ be respectively the number of vertices, edges and
  faces of $H_1$.  Let $n_{pd}$ be the number of primal- and
  dual-vertices of $H_1$ and $n_{e}$ be its number of edge-vertices.
  We have $n_1=n_{pd}+n_e$.  Since $H_1$ is a quadrangulation we have
  $4f_1=2m_1-2(k+k')$. Each edge of $H_1$ is incident to an
  edge-vertex. All the edge-vertices of $H_1$ have degree $4$ except
  those that are on $B$ and $C'$ that have degree $3$, so the total
  number of edges of $H_1$ is $m_1=4n_e-(k+k')$.  Finally combining
  these equalities with Euler's formula $n_1-m_1+f_1=0$, one obtain
  that $n_{pd}=n_{e}$. Moreover, since we are considering a
  $3$-orientation of $\pdc{G}$, we have that
  $m_1=3n_{pd}+n_e-(x'+y)=4n_e-(x'+y)$. So finally, by combining this
  with above formula, we obtain $k+k'=x'+y$. Similarly, by considering
  $H_2$, one obtain that $k+k'=x+y'$. Thus finally $x'+y=x+y'$ and
  thus $\gamma(C)=k_1\,\gamma (B_1) + k_2\,\gamma (B_2)$ by Claim~\ref{cl:computegamma}.
\end{proof}

Recall from Section~\ref{sec:transformations} that, given a
homology-basis of $G$ formed by a pair of cycles $B_1,B_2$ of $G$, the
type of a 3-orientation of $\pdc{G}$ is $\Gamma=(\gamma(B_1),\gamma(B_2))$.
Lemma~\ref{lem:gammahomologyequalbasis} implies the following:

\begin{lemma}
\label{lem:typegamma0foranybasis}
In a 3-orientation $D$ of $\pdc{G}$ the following are equivalent:
\begin{enumerate}
\item For two non-contractible not weakly
  homologous cycles $C,C'$ of ${G}$, we have
  $\gamma(C)=\gamma(C')=0$.
\item $D$ has type $(0,0)$ for a homology-basis of $G$  formed by a pair of
cycles.
\item $D$ has type $(0,0)$ for any homology-basis of $G$ formed by a
  pair of cycles.
\item For any non-contractible cycle $C$ of ${G}$ or $G^*$, we have
  $\gamma(C)=0$.
\end{enumerate}
\end{lemma}

\begin{proof}
  Clearly
  $4. \Longrightarrow 3. \Longrightarrow 2. \Longrightarrow 1.$. 

  We now prove that $1. \Longrightarrow 4.$ is implied by
  Lemma~\ref{lem:gammahomologyequalbasis}.  Consider two
  non-contractible not weakly homologous cycles $C,C'$ of ${G}$ such
  that $\gamma(C)=\gamma(C')=0$. Consider an homology-basis
  $\{B_1,B_2\}$ of $G$, such that $B_1,B_2$ are non-contractible
  cycles whose intersection is a single vertex or a path (see
  Section~\ref{sec:homology} for discussion on existence of such a
  basis). Let $k_1,k_2,k'_1,k'_2\in \mathbb Z$, such that $C$
  (resp. $C'$) is homologous to $k_1 B_1+k_2 B_2$ (resp.
  $k'_1 B_1+k'_2 B_2$).  Since $C$ is non-contractible we have
  $(k_1,k_2)\neq (0,0)$. By eventually exchanging $B_1,B_2$, we can
  assume, w.l.o.g., that $k_1\neq 0$.  By
  Lemma~\ref{lem:gammahomologyequalbasis}, we have
  $k_1 \gamma(B_1)+k_2 \gamma(B_2)=\gamma(C)=0=\gamma(C')=k'_1
  \gamma(B_1)+k'_2 \gamma(B_2)$.
  So $\gamma(B_1)=(-k_2/k_1) \gamma(B_2)$ and thus
  $(-k_2 k'_1/k_1 + k'_2)\gamma(B_2)=0$. So $k'_2=k_2 k'_1/k_1$ or
  $\gamma(B_2)=0$.  Suppose by contradiction, that
  $\gamma(B_2)\neq 0$.  Then
  $(k'_1,k'_2)= \frac{k'_1}{k_1} (k_1,k_2)$, and $C'$ is homologous to
  $\frac{k'_1}{k_1} C$.  Since $C$ and $C'$ are both non-contractible
  cycles, it is not possible that one is homologous to a multiple of
  the other, with a multiple different from $-1,1$. So $C, C'$ are
  weakly homologous, a contradiction.  So $\gamma(B_2)= 0$ and thus
  $\gamma(B_1)=0$.  Then by Lemma~\ref{lem:gammahomologyequalbasis},
  any non-contractible cycle of ${G}$ or $G^*$, have $\gamma$ equal to
  $0$.
\end{proof}

Consider a homology-basis of $G$ that is formed by a pair of cycles of
$G$.  We define the set $S_0(\pdc{G})$ of all the $3$-orientations of
$\pdc{G}$ that have type $\Gamma=(0,0)$ in the considered basis.  By
Lemma~\ref{lem:typegamma0foranybasis}, the definition of
$S_0(\pdc{G})$ does not depend on the choice of the basis, so its
element form a kind of canonical set.  We call the elements of
$S_0(\pdc{G})$ \emph{balanced orientations}.  A $3$-orientation of
$\pdc{G}$ is a balanced orientation if and only if it satisfies any
condition of Lemma~\ref{lem:typegamma0foranybasis}.  By
Theorem~\ref{th:transformations}, balanced orientations are homologous
to each other. By Theorem~\ref{cor:lattice}, they carry a structure of
distributive lattice.  This special lattice acts as a canonical
lattice and this is particularly important for bijection purpose (see
Chapter~\ref{chap:encoding}).  By
Theorem~\ref{th:characterizationgamma}, balanced orientations
corresponds to Schnyder woods of $G$ that we call \emph{balanced
  Schnyder woods}.  By Lemma~\ref{lem:typegamma0foranybasis}, a
Schnyder wood is balanced if and only if the dual Schnyder wood is
balanced.

Note  that there exists Schnyder woods that are not balanced. The
Schnyder wood of Figure~\ref{fig:gammanot0nothtc} is an example where
the horizontal cycle has $\gamma$ equal to $\pm 6$. Thus it is not balanced
by Lemma~\ref{lem:typegamma0foranybasis}.

\begin{figure}[!h]
\center
\includegraphics[scale=0.5]{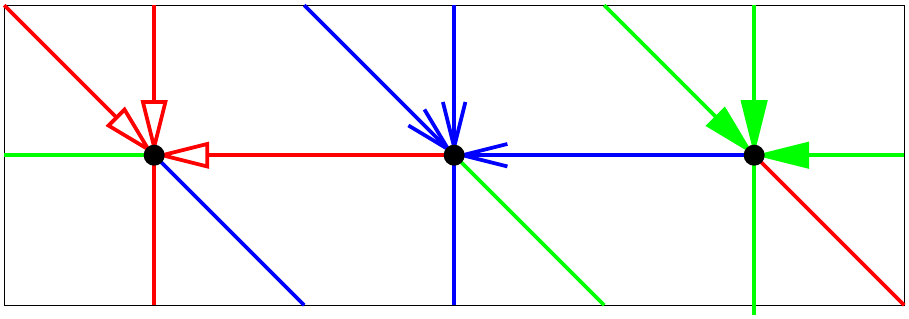} 
\caption{A Schnyder wood of a toroidal triangulation that is not balanced.}
\label{fig:gammanot0nothtc}
\end{figure}

One can remark that on Figure~\ref{fig:gammanot0nothtc}, the three
monochromatic cycles are all homologous to each other (i.e. going in
the same direction). In fact this is a general property that can be
proved:  a  Schnyder wood that is not balanced has three
monochromatic cycles of different colors that are all homologous to
each other.

Note that this property does not give a characterization of Schnyder
woods that are not balanced.  Figure~\ref{fig:htcthree} is an example
of a balanced Schnyder wood with three monochromatic cycles of
different colors that are homologous. The value of $\gamma$ for the
horizontal and a vertical cycle is $0$, thus it is balanced.

\begin{figure}[!h]
\center
\includegraphics[scale=0.3]{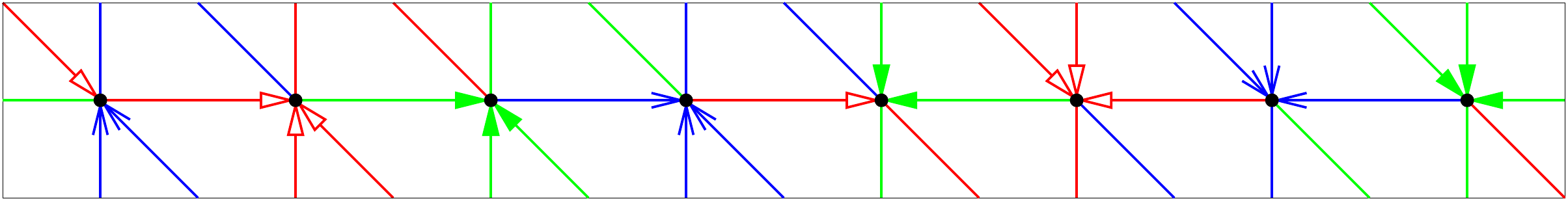} 
\caption{Example of a balanced Schnyder wood with three monochromatic
  cycles of different colors that are homologous.}
\label{fig:htcthree}
\end{figure}

\section{Balanced vs crossing}
\label{sec:balcross}
Consider a toroidal map $G$. Function $\gamma$ behaves well with
Schnyder woods of $G$ since monochromatic cycles have $\gamma$ equals
to zero:

\begin{lemma}
\label{lem:monocrhomegammaequal0}
  In a Schnyder wood of $G$,  a monochromatic cycle $C$ of
  $G$ satisfies $\gamma(C)=0$.
\end{lemma}

\begin{proof}
  We consider a $i$-cycle $C$ of $G$ and its corresponding cycle
  $\pdc{C}$ in $\pdc{G}$. Thus $\pdc{C}$ is alternating primal- and
  edge-vertices and in the direction of traversal of $C$, all the
  edges after a primal-vertex are outgoing and of color $i$. Let
  $u,v,w$ be three consecutive vertices of $C$, and $x,y$ be the
  edge-vertices between them in $\pdc{C}$, so $u,x,v,y,w$ appear in this
  order along $\pdc{C}$.  The edges $e=ux$ and $e'=vy$ are
  outgoing for $u$ and $v$ and of color $i$. We consider  three
  cases depending on the possible orientation and coloring of edge
  $xv$. These cases are represented on
  Figure~\ref{fig:monocrhomegammaequal0}.  One can check that in each
  case the computation of $\gamma$ on the part of $C$ between $e$ and
  $e'$ is equal to $0$. Thus in total we have $\gamma(C)=0$.
\end{proof}

\begin{figure}[!h]
\center
\begin{tabular}{ccc}
\includegraphics[scale=0.4]{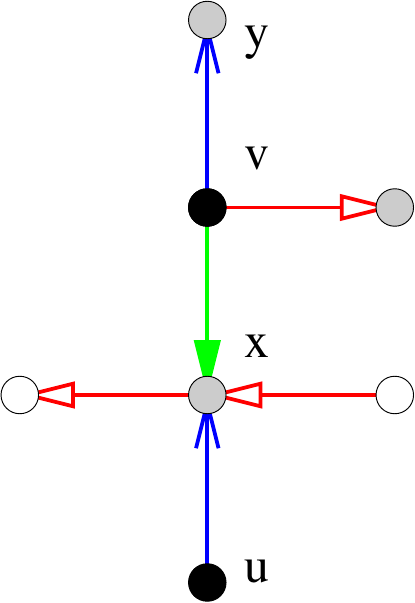} \ \ \  &  \ \ \
\includegraphics[scale=0.4]{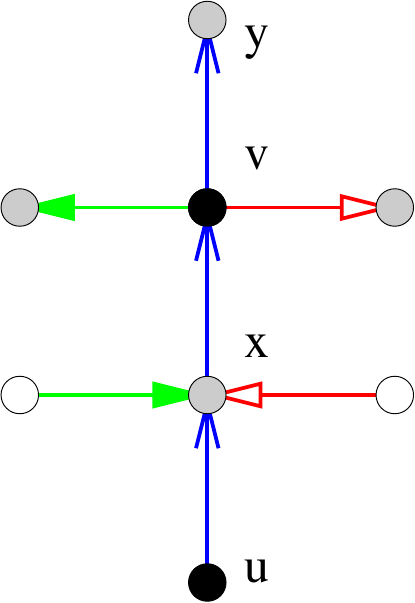} \ \ \  &  \ \ \
\includegraphics[scale=0.4]{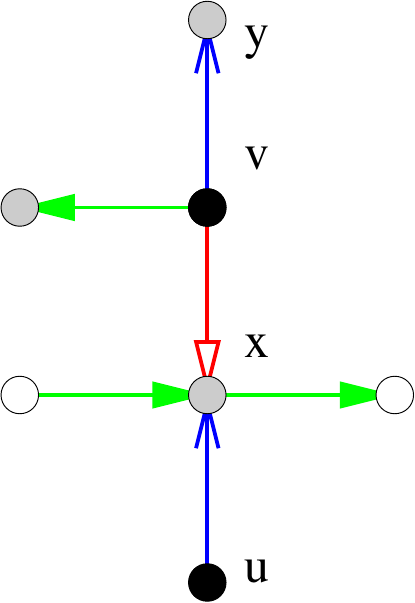} \\
\end{tabular}
\caption{The three cases that can occur when considering a part of a
  monochromatic cycle (vertical edges on the figure).}
\label{fig:monocrhomegammaequal0}
\end{figure}

Lemma~\ref{lem:monocrhomegammaequal0} implies the following:

\begin{proposition}
\label{lem:halftypegamma0}
an half-crossing Schnyder wood of $G$ is balanced.
\end{proposition}

\begin{proof}
  Suppose that the Schnyder wood is $i$-crossing by
  Theorem~\ref{lem:type-cross}. Let $C_i,C_{i-1}$ be two crossing
  monochromatic cycles of color $i$ and $i-1$ respectively. Cycles
  $C_i$ and $C_{i-1}$ are non-contractible, not weakly homologous and
  satisfy $\gamma(C_i)=\gamma(C_{i-1})=0$ by
  Lemma~\ref{lem:monocrhomegammaequal0}. So by
  Lemma~\ref{lem:typegamma0foranybasis}, the Schnyder wood is
  balanced.
\end{proof}

By Proposition~\ref{lem:halftypegamma0}, the set of balanced Schnyder
woods contains all the half-crossing Schnyder woods but it may also
contain some Schnyder woods that are not half-crossing.  The Schnyder
wood of Figure~\ref{fig:gamma0htc} is an example where $\gamma(C)=0$
for any non-contractible cycle $C$ of $G$, so it is balanced. It is
not half-crossing since all the monochromatic cycles are weakly
homologous (vertical on the figure).

\begin{figure}[!h]
\center
\includegraphics[scale=0.5]{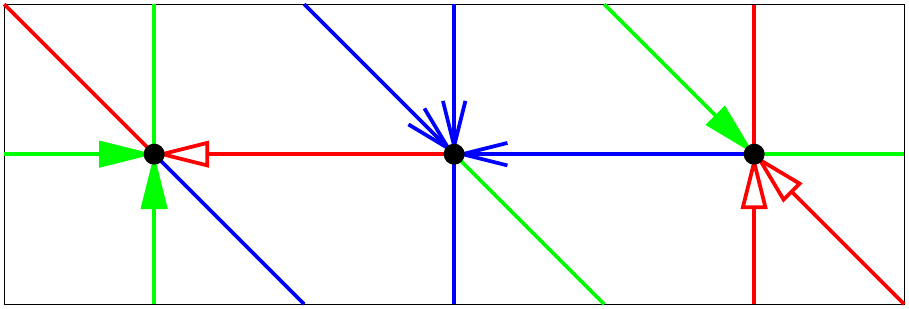} 
\caption{A Schnyder wood of a toroidal triangulation that is balanced and
  not half-crossing.}
\label{fig:gamma0htc}
\end{figure}

\section{Balanced for  triangulations}

Consider a toroidal triangulation $G$ given with an orientation.
Recall that the edges of the dual $G^*$ of $G$ are oriented such that
the dual $e^*$ of an edge $e$ of $G$ goes from the face on the right
of $e$ to the face on the left of $e$. (We insist that we do not
consider duality of Schnyder woods here but just a convention of
orientation of edges of  $G^*$.)

\begin{lemma}
\label{lem:incomingedges}
Consider an orientation of $G$ corresponding to a balanced Schnyder wood,
then the dual orientation of $G^*$ contains no oriented non-contractible
cycle.
\end{lemma}

\begin{proof}
  Suppose by contradiction that there exists an oriented
  non-contractible cycle $C$ is the dual orientation.  We compute
  $\gamma(C)$ to obtain a contradiction. Since $C$ is an oriented
  cycle in the dual, all the edges of $\pdc{G}$ that are incident to
  edge-vertices of $\pdc{C}$ and on the right of $\pdc{C}$ are
  outgoing and each of them counts $+1$ for $\gamma$. All the edges of
  $\pdc{G}$ that are incident to edge-vertices of $\pdc{C}$ and on the
  left of $\pdc{C}$ are ingoing and count $0$ for $\gamma$.  So each
  edge-vertex of $\pdc{C}$ counts $+1$ for $\gamma$. Since we are
  considering a triangulation, all the edges of $\pdc{G}$ incident to
  dual-vertices of $\pdc{C}$ are outgoing. Moreover each dual-vertex
  of $\pdc{C}$ is incident to exactly one such edge that is not on
  $\pdc{C}$.  Since $C$ is non-contractible, there is at least one
  edge incident to a dual-vertex of $\pdc{C}$ and on the right side of
  $\pdc{C}$. This edge counts $+1$ for $\gamma$.  Since $\pdc{C}$ has
  the same number $k$ of dual- and edge-vertices, we have
  $\gamma(C)\geq k+1-(k-1)=2$, a contradiction to
  Lemma~\ref{lem:typegamma0foranybasis}.
\end{proof}

Lemma~\ref{lem:incomingedges} gives the key property of balanced
Schnyder woods in triangulations that enables to generalize
\npss method~\cite{PS06} for the toroidal case, as done in
Chapter~\ref{chap:encoding}.

For the non-balanced Schnyder wood of
Figure~\ref{fig:gammanot0nothtc}, one can see that there is an
horizontal oriented non-contractible cycle in the dual, so it does not
satisfy the conclusion of Lemma~\ref{lem:incomingedges}. Note that the
conclusion of Lemma~\ref{lem:incomingedges} is not a characterization
of being balanced.  Figure~\ref{fig:gamma0glue} gives an example of a
Schnyder wood that is not balanced but satisfies the conclusion of
Lemma~\ref{lem:incomingedges} (its horizontal cycle has $\gamma$
equals $\pm 6$).

\begin{figure}[!h]
\center
\includegraphics[scale=0.4]{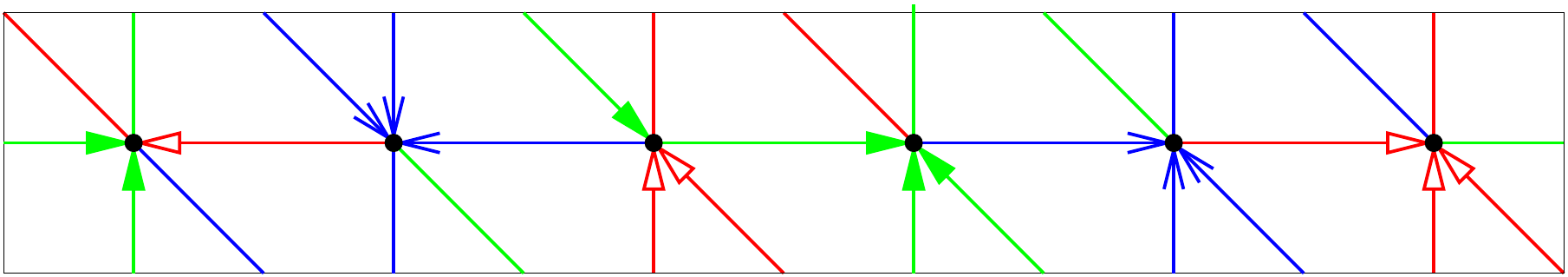}
\caption{A Schnyder wood that is not balanced but contains no oriented
  non-contractible cycle in the dual.}
\label{fig:gamma0glue}
\end{figure}

\section{Balanced lattice for triangulations}
\label{sec:flip} 

In this section, we push further the study of the lattice of balanced
Schnyder woods for triangulations. These properties are used in
Chapter~\ref{chap:encoding} to obtain a nice bijection.

Consider a toroidal triangulation $G$ and suppose that $G$ admits a
balanced orientation, so $S_0({G})$ is non-empty.  An
$\alpha$-orientation of $G$ is called a \emph{$3$-orientation} if
$\alpha(v)= 3$ for all vertices $v$ of $G$. Recall that $S_0(\pdc{G})$
denotes the set of all balanced orientations of $\pdc{G}$.  Since we
are considering a triangulation, the $3$-orientations of $G$ are in
bijection with the $3$-orientation of $\pdc{G}$.  Thus we simply
denote $S_0({G})$ the $3$-orientations of $G$ corresponding to
elements of $S_0(\pdc{G})$.  Let $f_0$ be any face of $G$. By
Theorem~\ref{cor:lattice}, we have that $(S_0({G}),\leq_{f_0})$ is a
distributive lattice.  Let $\widetilde{G}$ be the graph obtained from
$G$ by deleting all the rigid edges of $G$, i.e. edges that have the
same orientation in all elements of $S_0({G})$ (see
Section~\ref{sec:lattice}).

Suppose that $G$ is given with a 3-orientation.  For an edge $e_0$ of
$G$ we define the \emph{left walk} from $e_0$ as the sequence of edges
$W=(e_i)_{i\geq 0}$ of $G$ obtained by the following: if the edge
$e_i$ is entering a vertex $v$, then $e_{i+1}$ is chosen among the
three edges leaving $v$ as the edge that is on the left coming from
$e_i$ (i.e. the first one while going \cw around $v$). A closed
left-walk is left walk that is repeating periodically on itself,
i.e. a finite sequence of edges $W=(e_i)_{0\leq i\leq k-1}$ such that
its repetition is a left-walk.  We have the following lemma concerning
closed left-walks in balanced orientations:

\begin{lemma}
\label{lem:facialwalktilde}
In an orientation $D$ of $S_0(G)$, a closed left walk $W$ of $G$ is a
triangle with its interior on its left side.
\end{lemma}

\begin{proof}
  Consider a closed left walk $W=(e_i)_{0\leq i\leq k-1}$ of $D$, with
  $k\geq 1$. W.l.o.g., we may assume that all the $e_i$ are distinct,
  i.e. there is no strict subwalk of $W$ that is also a closed left
  walk.  Note that $W$ cannot cross itself otherwise it is not a left
  walk. However $W$ may have repeated vertices but in that case it
  intersects itself tangentially on the right side.

  Suppose by contradiction that there is an oriented subwalk $W'$ of
  $W$, that forms a cycle $C$ enclosing a region $R$ on its right side
  that is homeomorphic to an open disk.  Let $v$ be the starting and
  ending vertex of $W'$. Note that we do not consider that $W'$ is a
  strict subwalk of $W$, so we might have $W'=W$.  Consider the graph
  $G'$ obtained from $G$ by keeping all the vertices and edges that
  lie in the region $R$, including $W'$. Since $W$ can intersect
  itself only tangentially on the right side, we have that $G'$ is a
  plane map whose outer face boundary is $W'$ and whose interior is
  triangulated.  Let $k'$ be the length of $W'$. Let $n',m',f'$ be the
  number of vertices, edges and faces of $G'$.  By Euler's formula,
  $n'-m'+f'=2$. All the inner faces have size $3$ and the outer face
  has size $k'$, so $2m'=3(f'-1)+k'$. Since $W'$ is a subwalk of a
  left walk, all the vertices of $G'$, except $v$, have their outgoing
  edges in $G'$. Since $W'$ is following oriented edges of $G$, vertex $v$
  has at least one outgoing edge in $G'$.  Thus, as we are considering
  a $3$-orientation, we have $m'\geq 3(n'-1)+1$.  Combining these
  three equalities gives $k'\leq -1$, a contradiction.  So there is no
  oriented subwalk of $W$, that forms a cycle enclosing an open disk
  on its right side.

We now claim the following:

\begin{claim}
\label{cl:leftsidedisk}
  The left side of $W$ encloses a region  homeomorphic to an open disk
\end{claim}

\begin{proofclaim}
  We consider two cases depending on the fact that $W$ is a cycle
  (i.e. with no repetition of vertices) or not.

\begin{itemize}
\item \emph{$W$ is a cycle} 

  Suppose by contradiction that $W$ is a non-contractible cycle. We
  have $k$ is the length of $W$. Since we are considering a
  $3$-orientation, the total number of outgoing edges incident to
  vertices of $W$ is $3k$. There are $k$ such edges that are on
  $W$. Since $W$ is a left walk, there is no outgoing edge incident to
  its left side. So there are $2k$ outgoing edges incident to its
  right side. Thus $\gamma(W)=2k\neq 0$. So by
  Lemma~\ref{lem:typegamma0foranybasis}, the $3$-orientation $D$ is
  not balanced, a contradiction. Thus $W$ is a contractible cycle. As
  explained above, the contractible cycle $W$ does not enclose a
  region homeomorphic to an open disk on its right side. So $W$ encloses a
  region homeomorphic to an open disk on its left side, as claimed.

\item \emph{$W$ is not a cycle} 

  Since $W$ cannot cross itself nor intersect itself tangentially on
  the left side, it has to intersect tangentially on the right side.
  Such an intersection on a vertex $v$ is depicted on
  Figure~\ref{fig:righttangent}.(a). The edges of $W$ incident to $v$
  are noted as on the figure, $a,b,c,d$, where $W$ is going
  periodically trough $a,b,c,d$ in this order.  The (green) subwalk of
  $W$ from $a$ to $b$ does not enclose a region homeomorphic to an
  open disk on its right side. So we are not in the case depicted on
  Figure~\ref{fig:righttangent}.(b). Moreover if this (green) subwalk
  encloses a region homeomorphic to an open disk on its left side,
  then this region contains the (red) subwalk of $W$ from $c$ to $d$,
  see Figure~\ref{fig:righttangent}.(c). Since $W$ cannot cross
  itself, this (red) subwalk necessarily encloses a region
  homeomorphic to an open disk on its right side, a contradiction. So
  the (green) subwalk of $W$ starting from $a$ has to form a
  non-contractible curve before reaching $b$. Similarly for the (red)
  subwalk starting from $c$ and reaching $d$. Since $W$ is a left-walk
  and cannot cross itself, we are, w.l.o.g., in the situation of
  Figure~\ref{fig:righttangent}.(d) (with possibly more tangent
  intersections on the right side). In any case, the left side of $W$
  encloses a region homeomorphic to an open disk.

 \begin{figure}[!h]
 \center
 \begin{tabular}{cc}
 \includegraphics[scale=0.3]{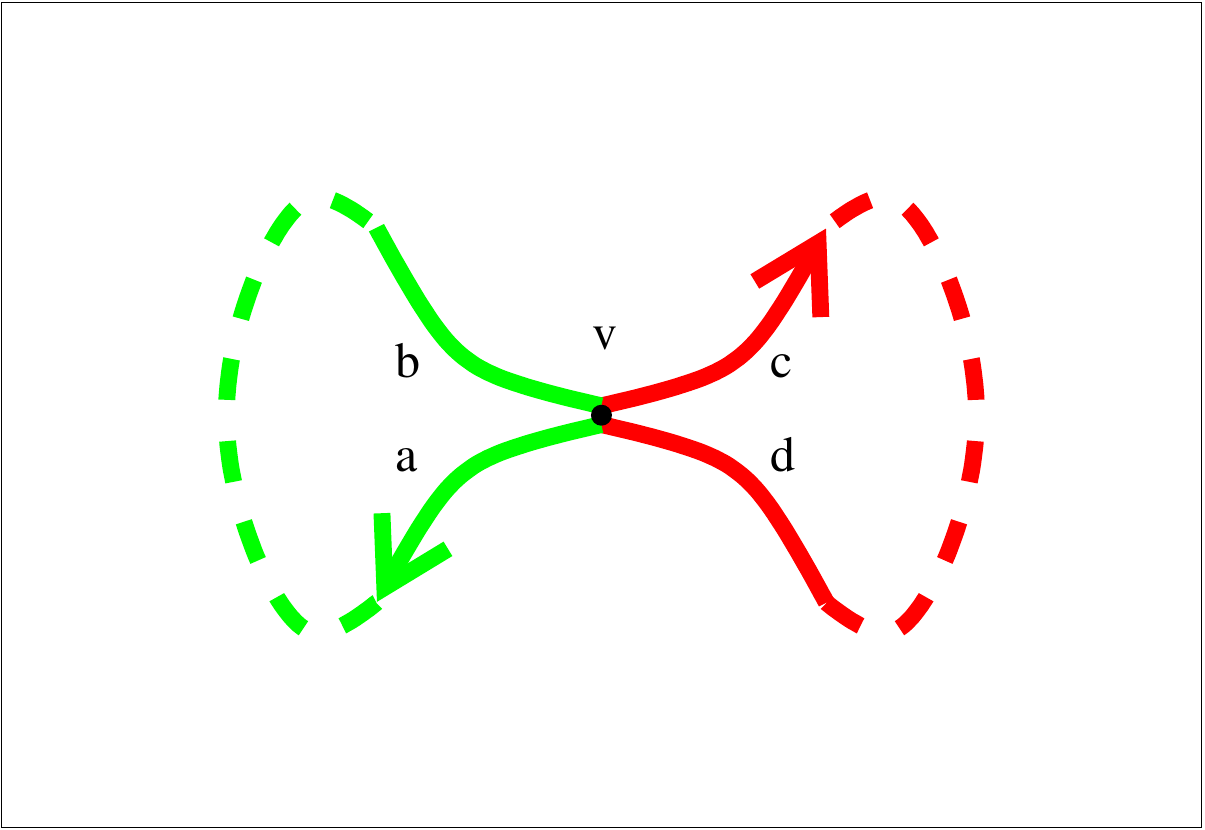} \ \ & \ \ 
 \includegraphics[scale=0.3]{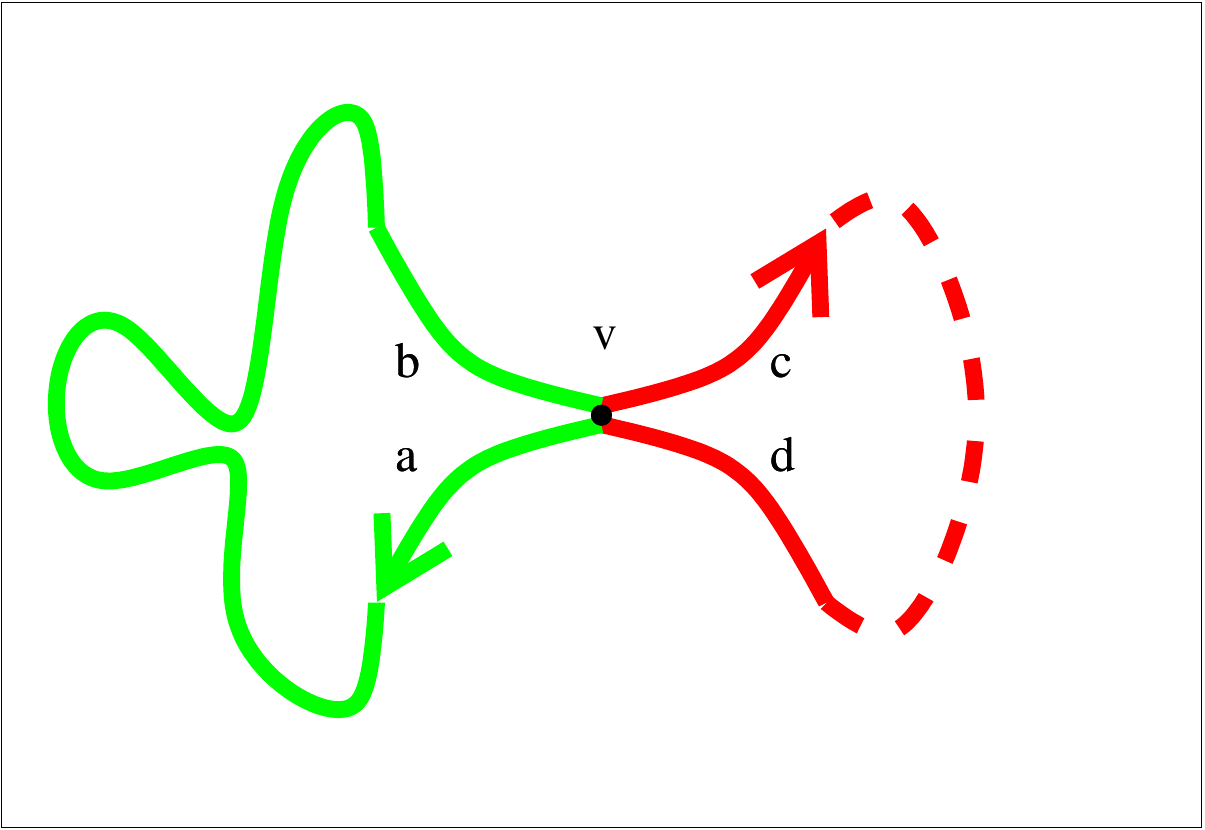}\\
(a) \ \ & \ \ (b)\\
& \\
 \includegraphics[scale=0.3]{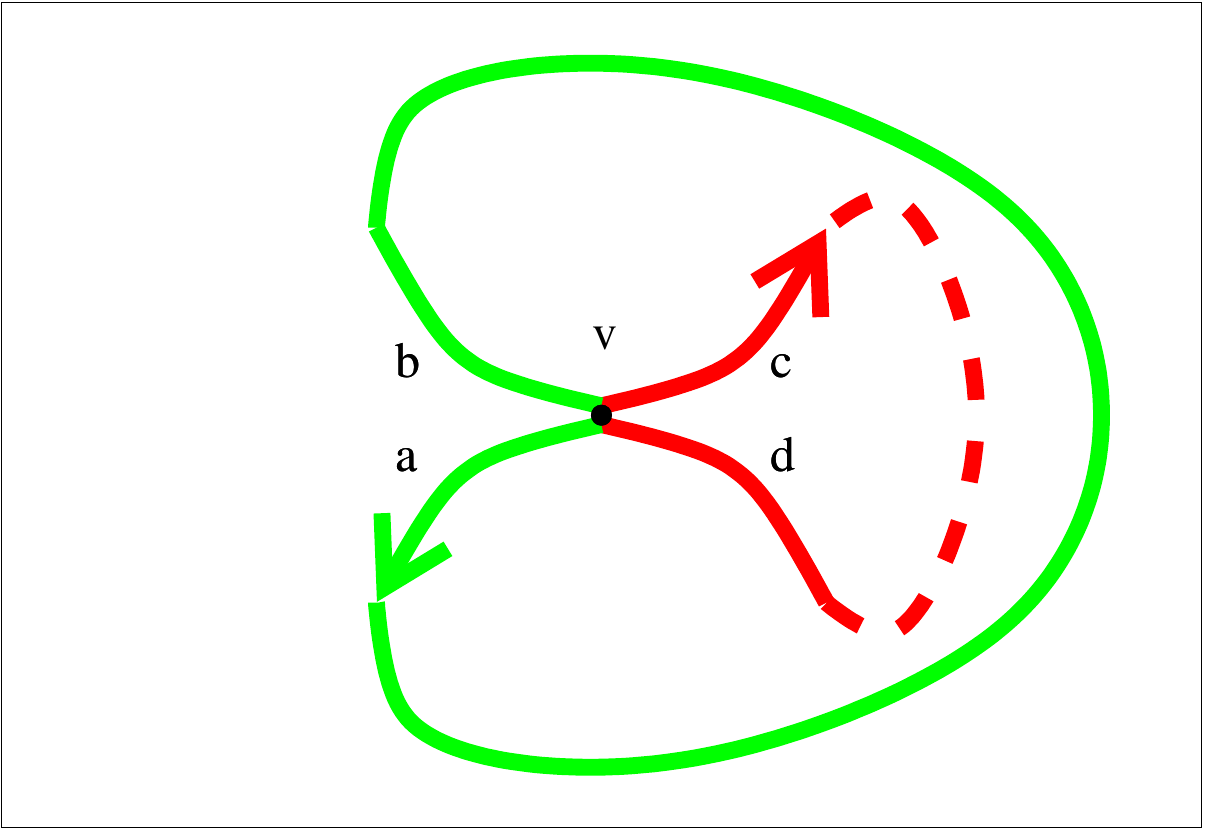}\ \ & \ \ 
 \includegraphics[scale=0.3]{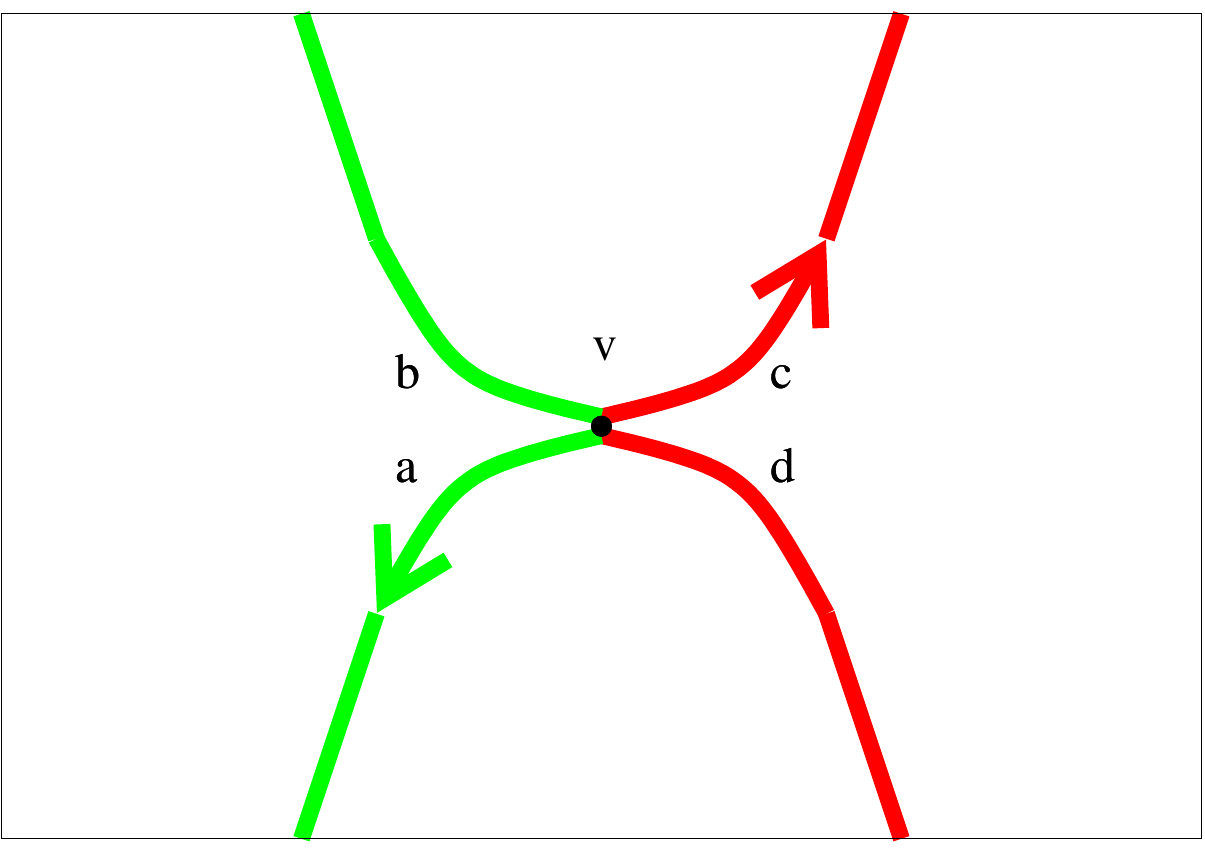}\\
(c) \ \ & \ \ (d)\\
 \end{tabular}
\caption{Case analysis for the proof of Claim~\ref{cl:leftsidedisk}.}
 \label{fig:righttangent}
 \end{figure}

\end{itemize}

\end{proofclaim}

By Claim~\ref{cl:leftsidedisk}, 
the left side of $W$ encloses a region $R$ homeomorphic to an open
disk.
Consider the graph $G'$ obtained from $G$ by keeping only the vertices
and edges that lie in the region $R$, including $W$. The vertices of
$W$ appearing several times are duplicated so that $G'$ is a plane
triangulation of a cycle.   Let
$n',m',f'$ be the number of vertices, edges and faces of $G'$.  By
Euler's formula, $n'-m'+f'=2$. All the inner faces have size $3$ and
the outer face has size $k$, so $2m'=3(f'-1)+k$.  All the inner
vertices have outdegree $3$ as we are considering a $3$-orientation of
$G$. All the inner edges of $G'$ incident to $W$ are oriented toward
$W$, and there are $k$ outer edges, so $m'=3(n'-k)+k$. Combining these
three equalities gives $k=3$, i.e. $W$ has size three and the lemma
holds.
\end{proof}

The boundary of a face of $\widetilde{G}$ may be composed of several
closed walks. Let us call \emph{quasi-contractible} the faces of
$\widetilde{G}$ that are homeomorphic to an open disk or to an open disk with
punctures.  Note that such a face may have several boundaries (if
there are some punctures) but exactly one of these boundaries encloses
the face. Let us call \emph{outer facial walk} this special
boundary. Then we have the following:

\begin{lemma}
\label{lem:contractibletilde}
All the faces of $\widetilde{G}$ are quasi-contractible and their
outer facial walk is a triangle.
\end{lemma}

\begin{proof}
  Suppose by contradiction that there is a face $\widetilde{f}$ of
  $\widetilde{G}$ that is not quasi-contractible or that is
  quasi-contractible but whose outer facial walk is not a
  triangle. Let $\widetilde{F}$ be the element of
  $\widetilde{\mathcal{F}}$ corresponding to the boundary of
  $\widetilde{f}$. By Lemma~\ref{lem:necessary}, there exists an
  orientation $D$ in $S_0(G)$ such that $\widetilde{F}$ is an oriented
  subgraph of $D$.

  All the faces of $G$ are triangles. Thus $\widetilde{f}$ is not a
  face of $G$ and contains in its interior at least one edge of $G$.
  Start from any such edge $e_0$ and consider the left walk
  $W=(e_i)_{i\geq 0}$ of $D$ from $e_0$.  Suppose that for $i\geq 0$,
  edge $e_i$ is entering a vertex $v$ that is on the border of
  $\widetilde{f}$.  Recall that by definition $\widetilde{F}$ is
  oriented \ccw according to its interior, so either $e_{i+1}$ is in
  the interior of $\widetilde{f}$ or $e_{i+1}$ is on the border of
  $\widetilde{f}$. Thus $W$ cannot leave $\widetilde{f}$.

  Since $G$ has a finite number of edges, some edges are used several
  times in $W$.  Consider a minimal subsequence
  $W'=e_k, \ldots, e_\ell$ such that no edge appears twice and
  $e_k=e_{\ell+1}$.  Thus $W$ ends periodically on $W'$ that is a
  closed left walk.  By Lemma~\ref{lem:facialwalktilde}, $W'$ is a
  triangle with its interior $R$ on its left side. Thus $W'$ is a
  $0$-homologous oriented subgraph of $D$. So all its edges are
  non-rigid by Lemma~\ref{lem:non-rigid}. So all the edges of $W'$ are
  part of the border of $\widetilde{f}$. Since $\widetilde{F}$ is
  oriented \ccw according to its interior, the region $R$ contains
  $\widetilde{f}$. So $\widetilde{f}$ is quasi-contractible and $W'$
  is its outer facial walk and a triangle, a contradiction.
\end{proof}

By Lemma~\ref{lem:contractibletilde}, every face of $\widetilde{G}$ is
quasi-contractible and its outer facial walk is a triangle.  So
$\widetilde{G}$ contains all the triangles of $G$ whose interiors are
maximal by inclusion, i.e. it contains all the edges that are not in
the interior of a separating triangle. In particular, $\widetilde{G}$
is non-empty and $|S_0(G)|\geq 2$.  The status (rigid or not) of an
edge lying inside a separating triangle is determined as in the planar
case: such an edge is rigid if and only if it is in the interior of a
separating triangle and incident to this triangle. Thus an edge of $G$
is rigid if and only if it is in the interior of a separating triangle
and incident to this triangle.

Since $(S_0(G),\leq_{f_0})$ is a distributive lattice, any element $D$
of $S_0(G)$ that is distinct from $D_{\max}$ and $D_{\min}$ contains
at least one neighbor above and at least one neighbor below in the
Hasse diagram of the lattice. Thus it has at least one face of
$\widetilde{G}$ oriented \ccw and at least one face of $\widetilde{G}$
oriented \cww. Thus by Lemma~\ref{lem:contractibletilde}, it contains
at least one triangle oriented \ccw and at least one triangle oriented
\cww. Next lemma shows that this property is also true for $D_{\max}$
and $D_{\min}$.

\begin{lemma}
\label{lem:trianglef0}
In $D_{\max}$ (resp. $D_{\min}$) there is a \ccw (resp. \cww) triangle
containing $f_0$, and a \cw (resp. \ccww) triangle not containing
$f_0$.
\end{lemma}

\begin{proof}
  By Lemma~\ref{lem:contractibletilde}, $\widetilde{f}_0$ is
  quasi-contractible and its outer facial walk is a triangle $T$.  By
  Lemma~\ref{lem:maxtilde}, $\widetilde{F}_0$ is an oriented subgraph
  of $D_{\max}$.  Thus $T$ is oriented \ccw and contains $f_0$.  The
  second part of the lemma is clear since $|S_0(G)|\geq 2$ so
  $D_{\max}$ has at least one neighbor below in the Hasse diagram of
  the lattice. Similarly for $D_{\min}$.
\end{proof}

Thus by above remarks and Lemma~\ref{lem:trianglef0}, all the balanced
Schnyder woods have at least one triangle oriented \ccw and at least
one triangle oriented \cww. Note that this property does not
characterize balanced Schnyder woods.  Figure~\ref{fig:gamma0glue}
gives an example of a Schnyder wood that is not balanced but satisfies
the property. Note also that not all Schnyder woods satisfy the
property.  Figure~\ref{fig:gammanot0nothtc} is an example of a
Schnyder wood that is not balanced and has no oriented triangle.

\chapter{Example of a balanced lattice}
\label{sec:example}

Figure~\ref{fig:lattice} illustrates the Hasse diagram of the balanced
Schnyder woods of the toroidal triangulation $G$ of
Figure~\ref{fig:gamma0htc}. Bold black edges are the edges of the
Hasse diagram $\mathcal{H}$. Each node of the diagram is a balanced
Schnyder wood of $G$. In every Schnyder wood, a face is dotted if its
boundary is directed. In the case of the special face $f_0$ the dot is
black. Otherwise, the dot is magenta if the boundary cycle is oriented
\ccw and cyan otherwise.  An edge in the Hasse diagram from $D$ to
$D'$ (with $D\leq D'$) corresponds to a face oriented \ccw in $D$
whose edges are reversed to form a face oriented \cw in $D'$, i.e. a
magenta dot is replaced by a cyan dot. The outdegree of a node is its
number of magenta dots and its indegree is its number of cyan dots.
All edges are not in the interior of a  triangle, so, by
previous section, they are not rigid. By Lemma~\ref{lem:necessary},
all the faces have a dot at least once. The special face is not
allowed to be flipped, it is oriented \ccw in the maximal Schnyder
wood and \cw in the minimal Schnyder wood by Lemma~\ref{lem:maxtilde}.
By Lemma~\ref{prop:maximal}, the maximal (resp. minimal) Schnyder wood
contains no other faces oriented \ccw (resp. \cww), indeed it contains
only cyan (resp. magenta) dots.  The words ``no'', ``half'', ``full''
correspond to Schnyder woods that are not half-crossing, half-crossing
(but not crossing), and crossing, respectively. By
Proposition~\ref{lem:halftypegamma0}, the diagram contains all the
half-crossing Schnyder woods of $G$. The minimal element is the
Schnyder wood of Figure~\ref{fig:gamma0htc}, and its neighbor is the
half-crossing Schnyder wood of Figure~\ref{fig:cross-0-tri}.

In the example the two crossing Schnyder woods lie in the ``middle''
of the lattice.  These Schnyder woods are of particular interests for
graph drawing purpose (see Chapter~\ref{sec:ortho}) whereas the
minimal balanced Schnyder wood (not crossing in this example) is
important for bijective encoding (see Chapter~\ref{chap:encoding}).
We show in next section how to find in linear time either a minimal
balanced Schnyder wood or a crossing Schnyder wood in the case of
toroidal triangulations.

\begin{figure}[!h]
\center
\includegraphics[scale=0.25]{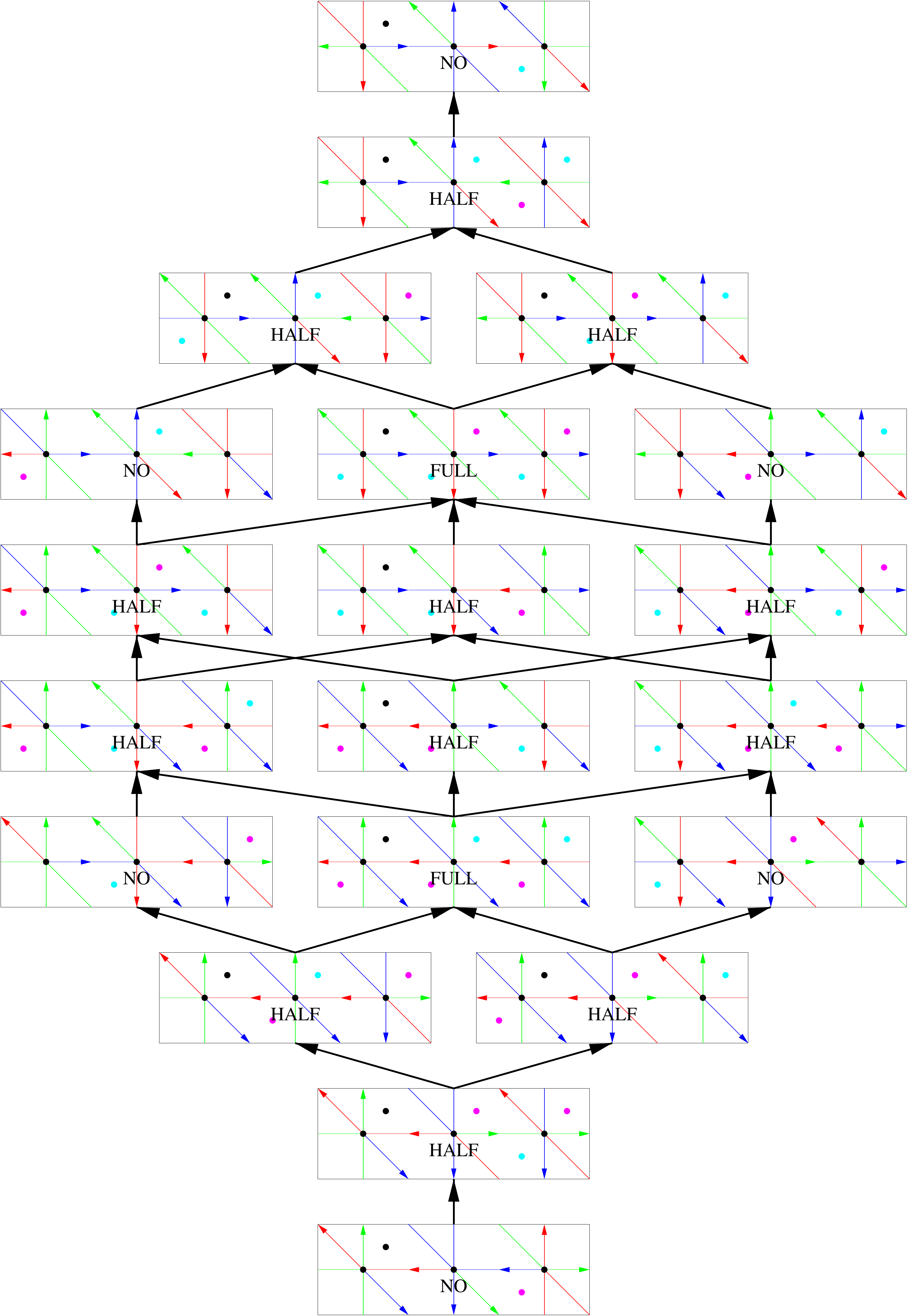} 
\caption{Example of the Hasse diagram of the distributive lattice of
  homologous orientations of a toroidal triangulation.}
\label{fig:lattice}
\end{figure}

The underlying toroidal triangulation of Figure~\ref{fig:lattice} has
only two Schnyder woods not depicted in
Figure~\ref{fig:lattice}. 
One of them two Schnyder wood is shown in Figure~\ref{fig:gammanot0nothtc}
and the other one is a 180\textdegree rotation of Figure~\ref{fig:gammanot0nothtc}. 
Each of these Schnyder wood is  alone in its lattice of
homologous orientations. 
All their edges are rigid. They have no
0-homologous oriented subgraph. 

Theorem~\ref{th:transformations} says that one can take the Schnyder
wood of Figure~\ref{fig:gammanot0nothtc}, reverse 3 or 6 vertical cycle
(such cycles form an Eulerian-partitionable oriented subgraph) to obtain
another Schnyder wood. Indeed, reversing any 3 of these cycles leads
to one of the Schnyder wood of Figure~\ref{fig:lattice} (for example
reversing the three loops leads to the crossing Schnyder wood of the
bottom part). Note that ${6 \choose 3}=20$ and there are exactly twenty
Schnyder woods on Figure~\ref{fig:lattice}. Reversing the 6 cycles leads
to the same picture pivoted by 180°.

Let's play  with this example and allow the special face to be
flipped. This consists in adding some extra edges between the nodes of
Figure~\ref{fig:lattice}. Then one obtain a kind of ``circular
lattice'' that is depicted on Figure~\ref{fig:lattice-rotate}. All
edges are going \cw between the nodes and corresponds to a \ccw face
that is flipped to become \cw. The $6$ bold magenta edges are the
extra edges that where added from Figure~\ref{fig:lattice}
(i.e. removing them gives the previous Hasse diagram). These edges
corresponds to flips of the special face.  By symmetry of the
underlying triangulation, any of the $6$ faces of the triangulation
would have  play the same role as  the special face. Thus the
$36$ edges of the diagram of Figure~\ref{fig:lattice-rotate} can be
partitioned into $6$ sets where each set is analogous to the magenta
edges. We do not know if there is a theory of such ``circular
lattice''. At least we believe that this point of view helps to
understand how things are going with homologous orientations.

\begin{figure}[!h]
\center
\includegraphics[scale=0.15]{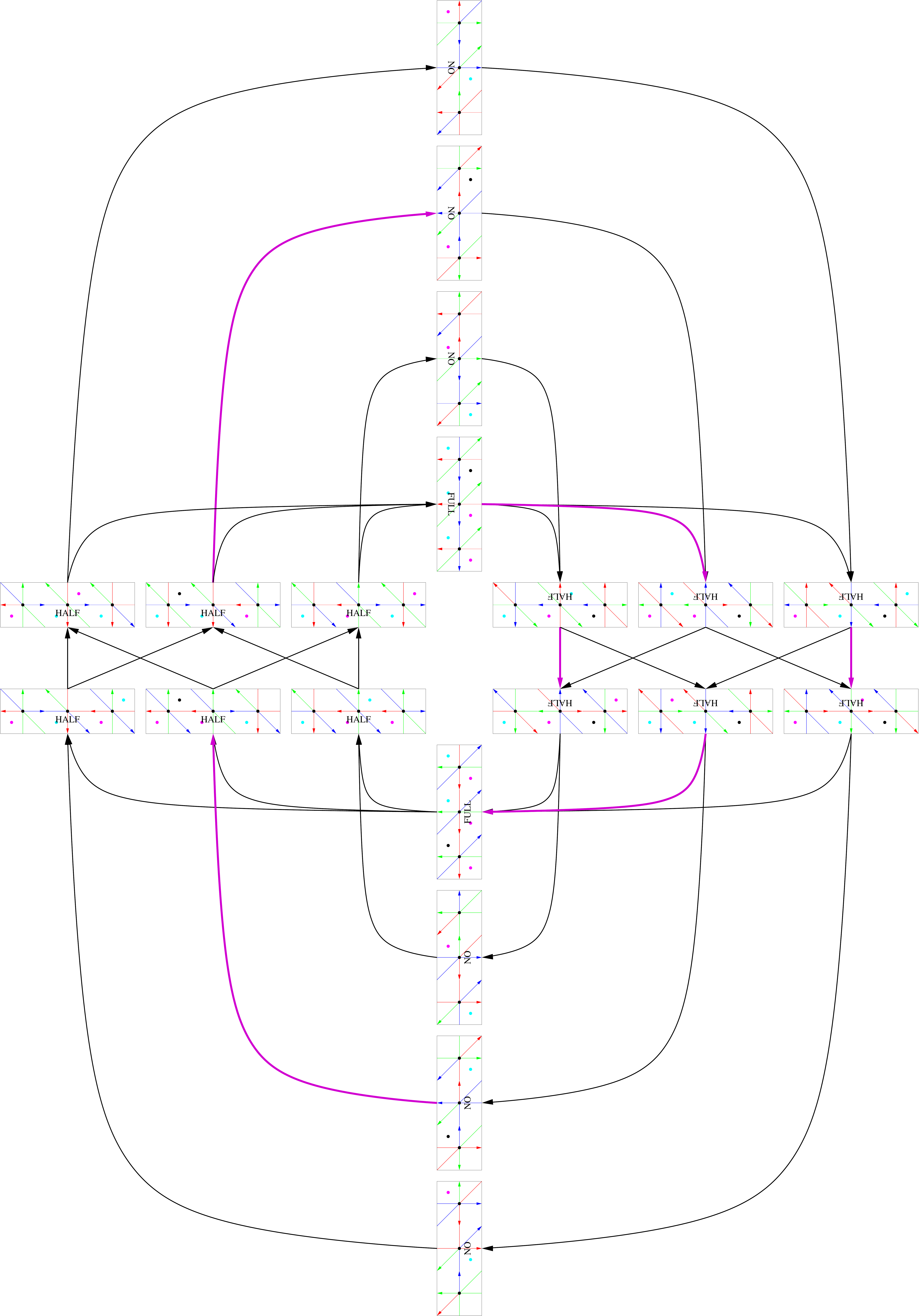} 
\caption{Diagram of the ``circular lattice'' obtained from
  Figure~\ref{fig:lattice} by allowing the special face to be flipped
  (magenta edges).}
\label{fig:lattice-rotate}
\end{figure}

\part{Proofs of existence in the toroidal case}
\label{chap:existence}

\chapter{Different proofs  for triangulations}

In the plane, the proof of existence of Schnyder wood is usually done
quite easily by using a so-called shelling order (or canonical order,
see~\cite{Kan96}). This method consists is starting from the outer
face and removing the vertices one by one. It leads to a simple linear
time algorithm to compute a Schnyder wood.  We do not see how to
generalize the shelling order method for the torus since the objects
that are then considered are too regular/homogeneous and there is no
special face (and thus no particular starting point) playing the role
of outer face.

We present in this chapter several proof of existence for toroidal
triangulation with some consequences. Theses proofs also help to
understand interesting structural properties of 3-orientations of
toroidal maps.  We propose to start with triangulations which are
simpler and the proofs are somewhat more intuitive. The case of general
toroidal maps (in fact essentially 3-connected toroidal maps) is
studied in next chapter.

\section{Contracting vertices}
\label{sec:contractionproof}

In this section we prove existence of balanced Schnyder wood for
triangulations by contracting edges until we obtain a triangulation
with just one vertex. The toroidal triangulation on one vertex is
represented on Figure~\ref{fig:orientation} with a balanced Schnyder
wood. Then the graph can be decontracted step by step to obtain a
balanced Schnyder wood of the original triangulation.

Given a toroidal triangulation $G$, the \emph{contraction} of a
non-loop-edge $e$ of $G$ is the operation consisting of continuously
contracting $e$ until merging its two ends. We note $G/e$ the obtained
map.  On Figure~\ref{fig:contraction-tri} the contraction of an edge
$e$ is represented.  Note that only one edge of each multiple edges
that is created is preserved (edge $e_{wx}$ and $e_{wy}$ on the
figure).

\begin{figure}[!h]
\center
\begin{tabular}{ccc} 
\includegraphics[scale=0.4]{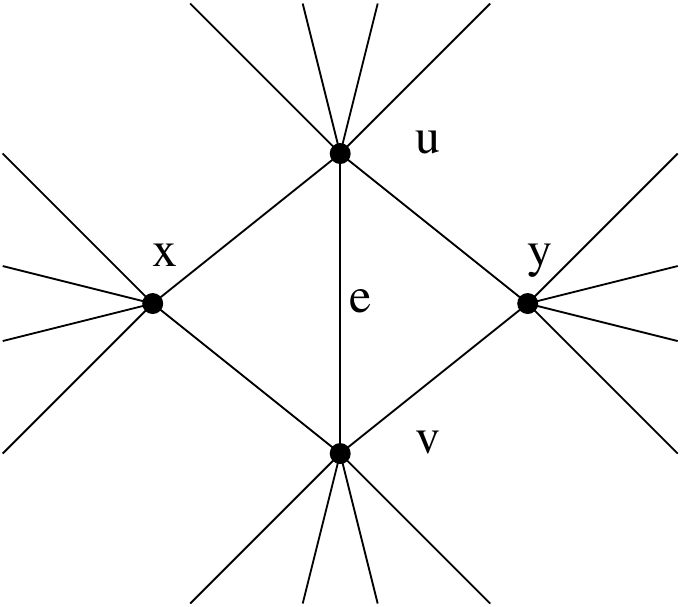}
& \hspace{3em} &
\includegraphics[scale=0.4]{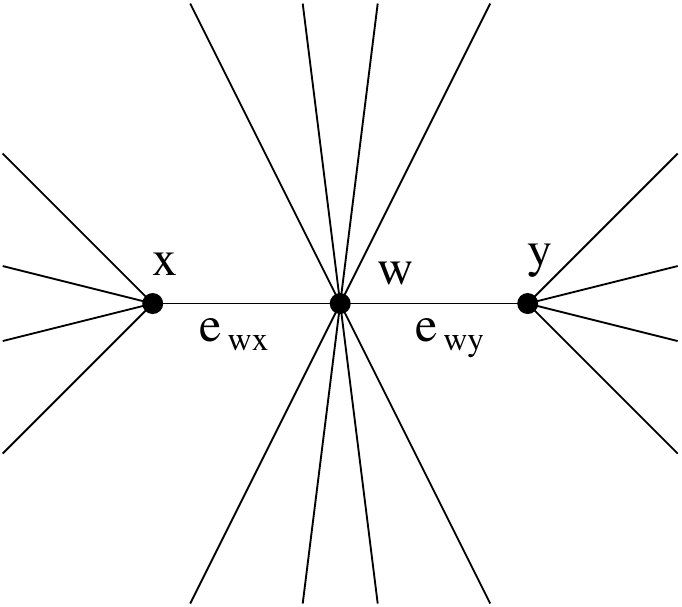}\\
& & \\
$G$& &$G/e$ \\
\end{tabular}
\caption{The contraction operation for triangulations}
\label{fig:contraction-tri}
\end{figure}

Note that the contraction operation is also defined when some vertices
are identified (see Figure~\ref{fig:contraction-loop-tri}): $x=u$ and
$y=v$ (the case represented on the figure), or $x=v$ and $y=u$
(corresponding to the symmetric case with a diagonal in the other
direction).

\begin{figure}[!h]
\center
\begin{tabular}{ccc} 
\includegraphics[scale=0.4]{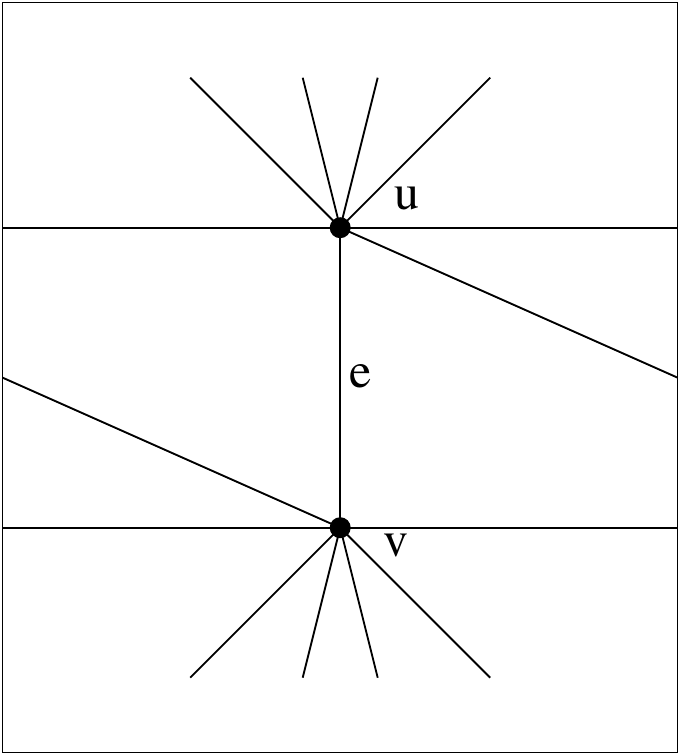}
& \hspace{3em} &
\includegraphics[scale=0.4]{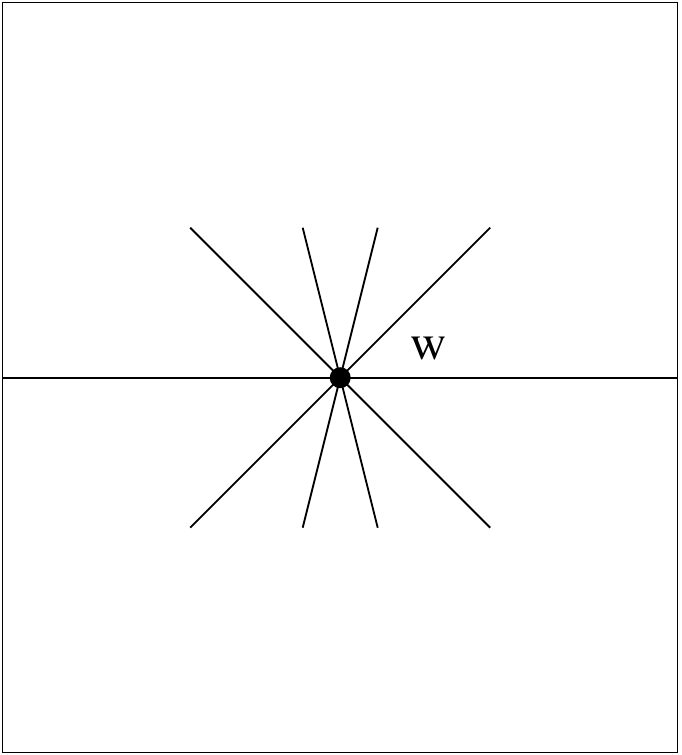}\\
& & \\
$G$& &$G/e$ \\
\end{tabular}
\caption{The contraction operation when some vertices are identified.}
\label{fig:contraction-loop-tri}
\end{figure}

The contraction operation enables to prove the following:

\begin{theorem}
\label{th:proofbycontractiontri}
A toroidal triangulation admits a balanced Schnyder wood.
\end{theorem}
\begin{proof}
  Suppose the theorem is false and let $G$ be a toroidal
  triangulation, with no contractible loop nor homotopic multiple
  edges, that does not admit a balanced Schnyder wood and that has the
  minimum number of vertices. The graph $G$ is not a single vertex,
  otherwise it is easy to find a balanced Schnyder wood of it (see
  Figure~\ref{fig:orientation}).  Then one can find an edge $e$ in $G$
  such that its contraction preserves the fact that $G$ is a toroidal
  triangulation with no contractible loop nor homotopic multiple edges
  (see~\cite{Moh96}). Let $G'=G/e$ be the toroidal map obtained after
  contracting $e$.
 
  By minimality, the graph $G'$ has a balanced Schnyder wood.  Note that
  since $G'$ is a toroidal triangulation, all edges are oriented in
  one direction only.  We now prove how to extend the Schnyder wood of
  $G'$ to a Schnyder wood of $G$.  Let $u,v$ be the two extremities of
  $e$ and $x,y$ the two vertices of $G$ such that the two faces
  incident to $e$ are $A=u,v,x$ and $B=v,u,y$ in clockwise order (see
  Figure~\ref{fig:contraction-tri}).  Note that $u$ and $v$ are
  distinct by definition of edge contraction but that $x$ and $y$ are
  not necessarily distinct, nor distinct from $u$ and $v$. Let $w$ be
  the vertex of $G'$ resulting from the contraction of $e$. Let
  $e_{wx}, e_{wy}$ be the two edges of $G'$ represented on
  Figure~\ref{fig:contraction-tri} (these edges are identified and
  form a loop on Figure~\ref{fig:contraction-loop-tri}).

  There are different cases corresponding to the different
  possibilities of orientation and coloring of the edges $e_{wx}$ and
  $e_{wy}$ in $G'$. There should be $6$ cases depending on if $e_{wx}$
  and $e_{wy}$ are both entering $w$, both leaving $w$ or one entering
  $w$ and one leaving $w$ (3 cases), multiplied by the coloring, both
  of the same color or not (2 cases). The case where $w$ has two edges
  leaving in the same color is not possible by the Schnyder
  property. So only $5$ cases remain represented on the left side
 of Figure~\ref{fig:decontracttri}.  In the
  particular case where some vertices are identified (see
  Figure~\ref{fig:contraction-loop-tri}), there is a loop edge in $G'$
  that is uniquely colored, as on  the left of
  Figure~\ref{fig:contraction-loop-tri-col}.
  
\begin{figure}[!h]
\center
  \begin{tabular}[]{ccccccccc}
1& \hspace{2em} & \includegraphics[scale=0.4]{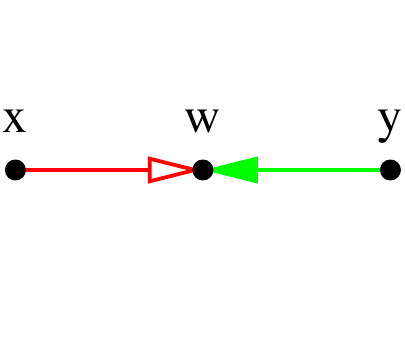} & \hspace{0em} &
\includegraphics[scale=0.4]{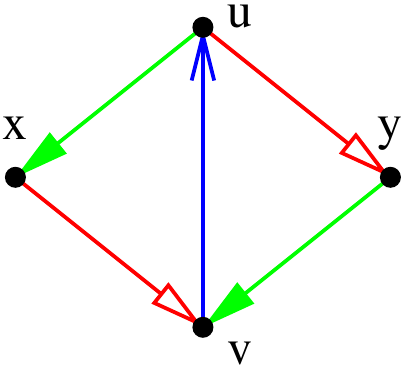} & \hspace{0em} &
\\

2 & & \includegraphics[scale=0.4]{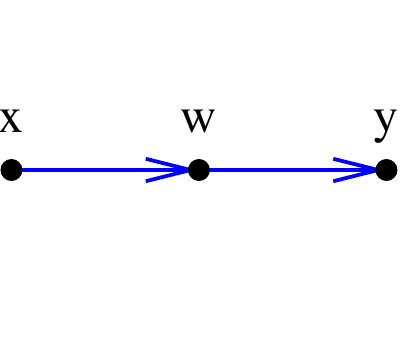} & \hspace{0em} &
\includegraphics[scale=0.4]{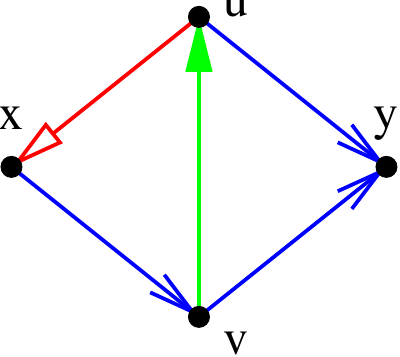} & \hspace{0em} &
\\
 
3 & &\includegraphics[scale=0.4]{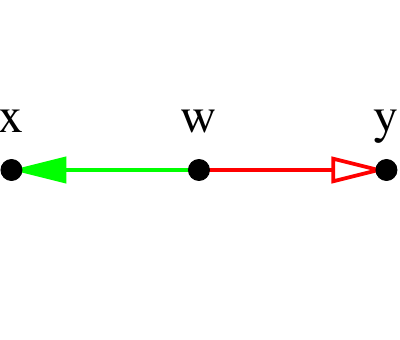} & \hspace{0em} &
\includegraphics[scale=0.4]{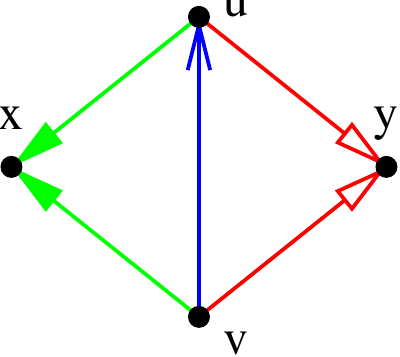} \\

4 & & \includegraphics[scale=0.4]{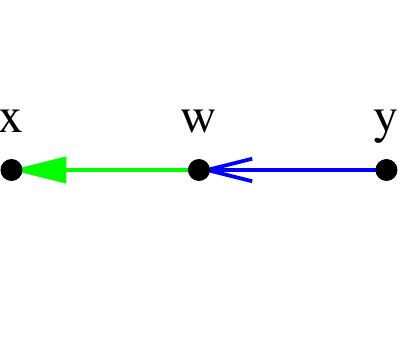} & \hspace{0em} &
\includegraphics[scale=0.4]{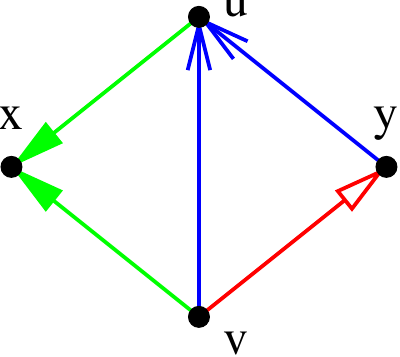} \\

5 & & \includegraphics[scale=0.4]{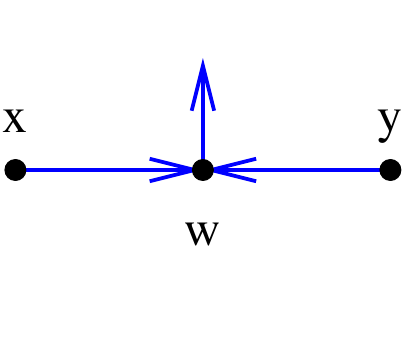} & \hspace{0em} &
\includegraphics[scale=0.4]{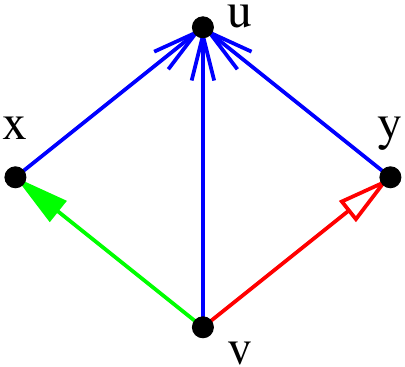} \\

 & & $G'$& &$G$ \\
  \end{tabular}

  \caption{Decontraction rules for triangulations.}
\label{fig:decontracttri}
\end{figure}

\begin{figure}[!h]
\center
\begin{tabular}{ccc} 
\includegraphics[scale=0.4]{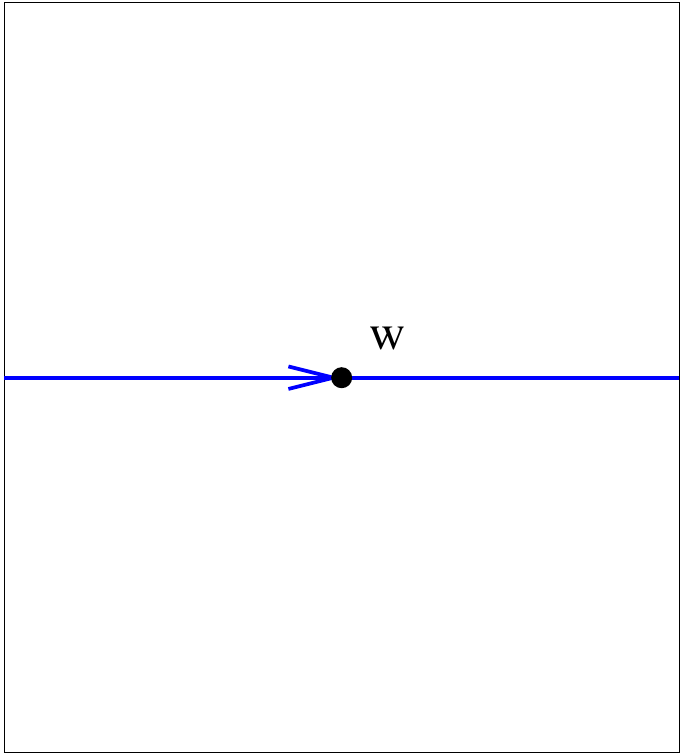}
& \hspace{3em} &
\includegraphics[scale=0.4]{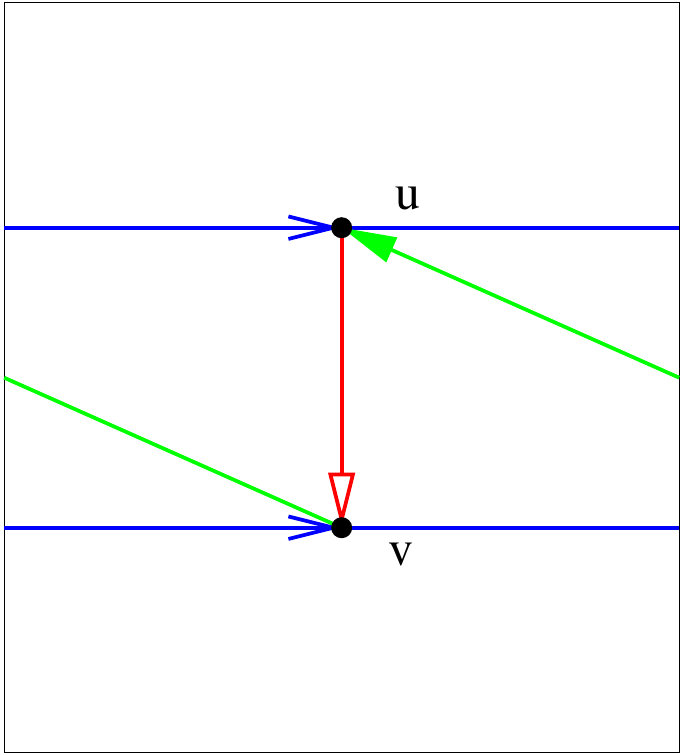}\\
& & \\
$G'$& &$G$ \\
\end{tabular}
\caption{The decontraction operation when some vertices are identified.}
\label{fig:contraction-loop-tri-col}
\end{figure}

In each case there might be several possibilities to orient and color
the edges of $G$ incident to the faces $A$ and $B$ in order to obtain
a Schnyder wood of $G$ from the Schnyder wood of $G'$ (all the edges
not incident to $A$ or $B$ are not modified).  These different
possibilities can easily be found by considering the angle labeling of
$G$ and $G'$ around vertices $u,v,w$.

For example, for case 1, the labeling of the angles in $G'$ are
depicted on the left of Figure~\ref{fig:angle-vierge-2}. Since outside
the contracted region, the edges are not modified, these labels have
to be maintained during the decontraction process as shown on the
middle of Figure~\ref{fig:angle-vierge-2}. In this case, there are
exactly three possible orientations and colorings of the edges of $G$
satisfying this labeling, one of these possibilities is represented on
right part of Figure~\ref{fig:angle-vierge-2}.

\begin{figure}[!h]
\center 
\begin{tabular}[]{ccc}
\includegraphics[scale=0.4]{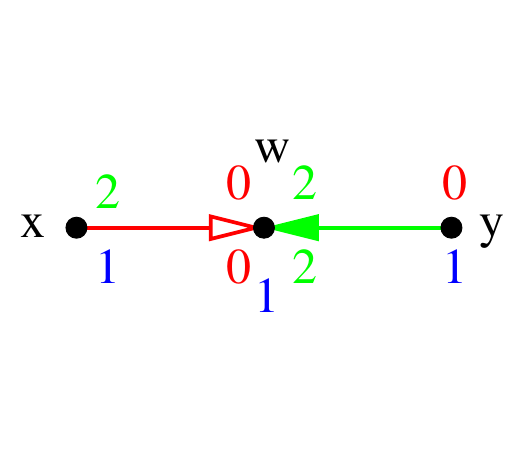}
\ \ &   \ \ 
                \includegraphics[scale=0.4]{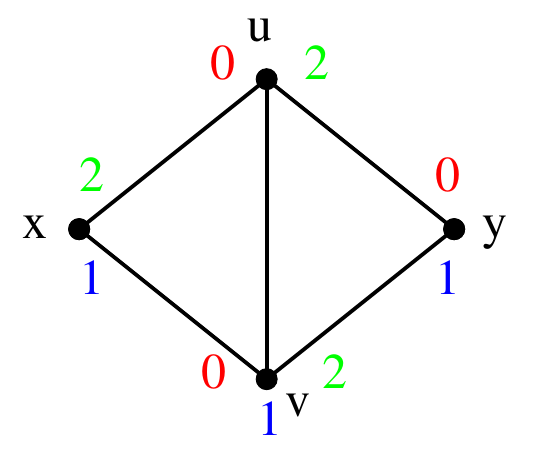}
\ \ &   \ \ 
\includegraphics[scale=0.4]{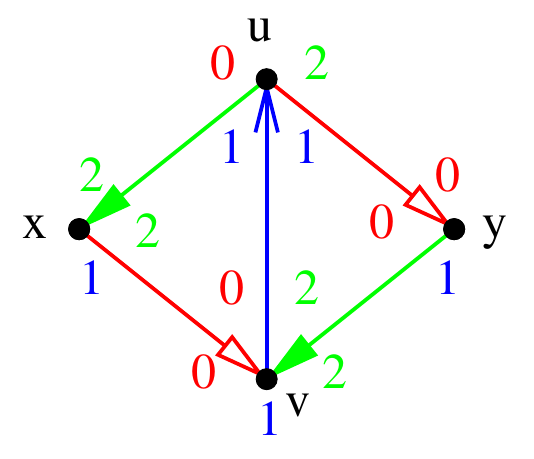}
\\
$G'$ \ \ & \ \ $G$& \ \ $G$\\
\\
\end{tabular}

  \caption{Angle labeling during decontraction for case 1 and the
    corresponding possible orientation and coloring of $G$.}
\label{fig:angle-vierge-2}
\end{figure}

For all the cases, there is always at least one possibility of
orientation and coloring of $G$ satisfying the Schnyder property and
one such possibility is represented on the right side of
Figures~\ref{fig:decontracttri}
and~\ref{fig:contraction-loop-tri-col}.  Since we are considering
edges that are oriented in one direction only, there is no face whose
boundary is a monochromatic cycle by the Schnyder property.  So after
decontraction we still have a Schnyder wood.  Moreover one can check
that the value $\gamma$ of any non-contractible cycle is not modified
during the decontraction process. For  this one can
consider all the possible ways that a cycle may go through the
contracted region and then check that the value of $\gamma$ is
the same in $G$ and $G'$.  Thus the resulting Schnyder wood is
balanced.
\end{proof}
 
Note that Theorem~\ref{th:proofbycontractiontri} confirms
that Conjecture~\ref{conjecture} is true when $g=1$.  The proof is
quite simple and can be turned into a linear time algorithm to find an
balanced Schnyder wood. The main difficulty lies in finding sequence of
edges that can be contracted in linear time.

An edge $e$ is said to be \emph{contractible} if it is not a loop and
if after contracting $e$ and identifying the borders of the two newly
created length two faces, one obtains a triangulation that is still
without contractible loop or homotopic multiple edges.

Consider an edge $e$ of $G$ with distinct ends $u,v$, and with
incident faces $uvx$ and $vuy$, such that these vertices appear in \cw
order around the corresponding face (so we are in the situation of
Figure~\ref{fig:contraction-tri}). The edge $e$ is contractible if and
only if, every walk enclosing an open disk containing a face other
than $uvx$ and $vuy$, goes through an edge distinct from $e$ at least
three times.  Equivalently $e$ is non-contractible if and only if it
belongs to a separating triangle or $u,v$ are both incident to a
loop-edge $\ell_u,\ell_v$, respectively, such that the walk of length
four $(\ell_u,e,\ell_v,e)$ encloses an open disk with at least three
faces (i.e. with at least one face distinct from $uvx$ and $vuy$). To
avoid the latter case, if vertex $u$ is incident to a loop-edge
$\ell_u$, we consider $e$ to be an edge that is consecutive to that
loop, so that we have $x=u$. In such a case, if there is a loop
$\ell_v$ incident to $v$ and a walk of length four of the form
$(\ell_u,e,\ell_v,e)$ enclosing a disk with at least three faces, then
there is also a separating triangle containing $e$. In the following
we show how to find such separating triangle, if there is one. If $u$
and $v$ have more common neighbors, than simply $x$ and $y$, consider
their second common neighbor going clockwise around $u$ from $e$ (the
first one being $x$, and the last being $y$) and call it $x'$. Call
$y'$ their second common neighbor going counterclockwise around $u$
from $e$. Now, either $uvx'$ or $uvy'$ is a separating triangle or
the edge $e$ is contractible. We consider these two cases below:

\begin{itemize}
\item If $e$ is contractible, then it is contracted and we apply the
  procedure recursively to obtain a Schnyder wood of the contracted
  graph.  Then we update the balanced Schnyder wood with the
  rules of Figures~\ref{fig:decontracttri}
  and~\ref{fig:contraction-loop-tri-col}. Note that this update is
  done in constant time.
\item If $uvx'$ (resp. $uvy'$) is a separating triangle, one can
  remove its interior, recursively obtain a toroidal Schnyder wood of
  the remaining toroidal triangulation, build a planar Schnyder wood
  of the planar triangulation inside $uvx'$ (resp. $uvy'$), and then
  superimpose the two (by eventually permuting the colors) to obtain a
  Schnyder wood of the whole graph. Note that computing a planar
  Schnyder wood can be done in linear time using a canonical ordering
  (see~\cite{Kan96}).
\end{itemize}

The difficulty here is to test whether $uvx'$ or $uvy'$ are triangles.
For that purpose, one first needs to compute a basis $(B_1,B_2)$ for
the homology.  Consider a spanning tree of the dual map $G^*$. The map
obtained from $G$ by removing those edges is unicellular, and removing
its treelike parts one obtains two cycles $(B_1,B_2)$ (intersecting on
a path with at least one vertex) that form a basis for the
homology. This can be computed in linear time for $G$ and then updated
when some edge is contracted or when the interior of some separating
triangle is removed. The updating takes constant time when some edge
is contracted, and it takes $O(n')$ time when removing $n'$ vertices
in the interior of some separating triangle. The overall cost of
constructing and maintaining the basis is thus linear in the size of
$G$.  Then a closed walk of length three $W$, given with an arbitrary
orientation, encloses a region homeomorphic to an open disk if and
only if $W$ crosses $B_i$ from right to left as many times as $W$
crosses $B_i$ from left to right, for every $i\in\{1,2\}$.  This test
can be done in constant time for $uvx'$ and $uvy'$ once the half edges
on the right and left sides of the cycles $B_i$ are marked. Marking
the half edges of $G$ and maintaining this marking while contracting
edges or while removing the interior of separating triangles can
clearly be done in linear time. We thus have that the total running
time to compute a balanced Schnyder wood of $G$ is linear.

Note that from a balanced Schnyder wood, one can easily obtain in
linear time the minimal balanced Schnyder wood w.r.t.~a special face
$f_0$ by applying the method discuss in Section~\ref{sec:lattice}.
Thus the total running time to compute a minimal balanced Schnyder
wood is linear.

The proof of Theorem~\ref{th:proofbycontractiontri} can easily be
generalized to essentially 3-connected toroidal maps. This is done in
Chapter~\ref{chap:generalexistence} in which we moreover do some
tedious extra efforts to obtain the existence of crossing Schnyder
woods and not only balanced Schnyder woods. Unfortunately
Chapter~\ref{chap:generalexistence} does not lead to a linear
algorithm.

We do not see how to generalize the proof of
Theorem~\ref{th:proofbycontractiontri} to higher genus
surfaces. Indeed it is not always possible to find rules to decontract
a vertex that has outdegree more than three. In the example of
Figure~\ref{fig:decontractgenus} where the two considered edges
incident to $w$ are outgoing, there is no way to orient and color the
decontracted graph in order to maintain a Schnyder wood by keeping he
orientation and coloring of edges that are not represented on the
right part of the figure.

\begin{figure}[!h]
\center
\includegraphics[scale=0.4]{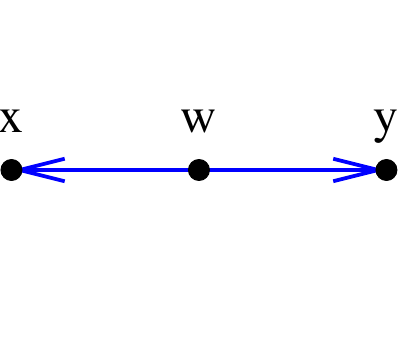}
 \ \ \ \ \ \
\includegraphics[scale=0.4]{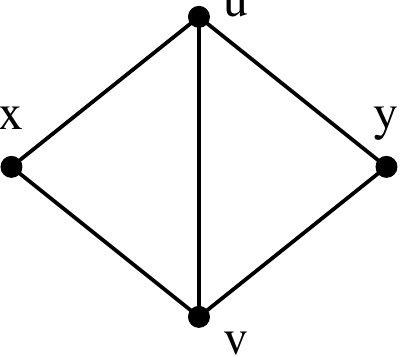}

  \caption{How to decontract a vertex of outdegree strictly  more than
    $3$?}
\label{fig:decontractgenus}
\end{figure}

\section{Reversing oriented subgraphs}
\label{sec:flipproof}

In this section we obtain a proof of existence of crossing Schnyder
woods by understanding some structural property on the set of
$3$-orientations of a toroidal triangulation.  In
Section~\ref{sec:flip-noncontract}, we flip non-contractible cycles to
move between the different lattices of homologous $3$-orientations in
order to go from any lattice into the balanced lattice. Indeed we
transform any $3$-orientation into a half-crossing Schnyder wood with
this method.  In Section~\ref{sec:flip-contract}, we flip 0-homologous
oriented subgraphs to move inside the balanced lattice from a
half-crossing Schnyder wood to a crossing Schnyder wood.  The
combination of the two sections gives a linear time algorithm to find
a crossing Schnyder wood for a toroidal triangulation.

\subsection{Non-contractible cycles}
\label{sec:flip-noncontract}

We first show several properties of 3-orientations of toroidal
triangulations. By Theorem~\ref{th:barat}, a simple toroidal triangulation admits a
3-orientation. This can be shown to be true also for non-simple
triangulations, for example using edge-contraction as explained in
Section~\ref{sec:contractionproof}.  

Consider a toroidal triangulation $G$ and a 3-orientation of $G$.  Let
$G^\infty$ be the universal cover of $G$.

\begin{lemma}
\label{lem:kmoins3}
A cycle $C$ of $G^\infty$ of length $k$ has exactly $k-3$ edges
leaving $C$ and directed toward  the interior of $C$.
\end{lemma}

\begin{proof}
  Let $x$ be the number of edges leaving $C$ and directed toward the
  interior of $C$.  Consider the cycle $C$ and its interior as a
  planar graph $C^{o}$. Euler's formula gives $n'-m'+f'=2$ where
  $n',m',f'$ are respectively the number of vertices, edges and faces
  of $C^{o}$.  Every inner vertex has exactly outdegree $3$, so
  $m'=3(n'-k)+k+x$. Every inner face has size three so $2m'=3(f'-1)+k$.
  Considering these equalities, one obtain that $x=k-3$ as claimed.
\end{proof}

As in Chapter~\ref{chap:introSW}, for an edge $e$ of $G$, we define
the \emph{middle walk from $e$} as the sequence of edges
$(e_i)_{i\geq 0}$ obtained by the following method.  Let $e_0=e$. If
the edge $e_i$ is entering a vertex $v$, then the edge $e_{i+1}$ is
chosen in the three edges leaving $v$ as the edge in the ``middle''
coming from $e_i$. The difference from Chapter~\ref{chap:introSW} is
that now a middle walk never ends since there is no outer vertex.

A directed cycle $M$ of $G$ is said to be a \emph{middle cycle} if
every vertex $v$ of $M$ has exactly one edge leaving $v$ on the left
of $M$ (and thus exactly one edge leaving $v$ on the right of $M$).
Note that a middle cycle $C$ satisfies $\gamma(C)=0$.  Note that if
$M$ is a middle cycle, and $e$ is an edge of $M$, then the middle walk
from $e$ consists of the sequence of edges of $M$ repeated
periodically.  Note that a middle cycle is non-contractible, otherwise
in $G^\infty$ it forms a contradiction to
Lemma~\ref{lem:kmoins3}. Similar arguments lead to:

   \begin{lemma}
\label{lem:middleequal}
Two middle cycles that are  weakly homologous are either
vertex-disjoint or equal.
   \end{lemma}

We have the following useful lemma concerning
middle walks and middle cycles:

\begin{lemma}
\label{lem:middlecycle}
A middle walk always ends on a middle cycle.
\end{lemma}

\begin{proof}
  Start from any edge $e_0$ of $G$ and consider the middle walk
  $W=(e_i)_{i\geq 0}$ from $e_0$.  The graph $G$ has a finite number
  of edges, so some edges will be used several times in $W$.  Consider a
  minimal subsequence $e_k, \ldots, e_\ell$ such that no edge appears
  twice and $e_k=e_{\ell+1}$.  Thus $W$ ends periodically on the
  sequence of edges $e_k, \ldots, e_\ell$. We prove that $e_k, \ldots,
  e_\ell$ is a middle cycle.

  Assume that $k=0$ for simplicity. Thus $e_0, \ldots, e_\ell$ is an
  Eulerian subgraph $E$. If $E$ is a cycle then it is a middle cycle
  and we are done. So we can consider that it visits some vertices
  several times.  Let $e_i, e_j$, with $0\leq i < j\leq \ell$, such
  that $e_i, e_j$ are both leaving the same vertex $v$. By definition
  of $\ell$, we have $e_i\neq e_j$. Let $A$ and
  $B$ be the two closed walks  $e_i,\ldots, e_{j-1}$ and $e_{j},\ldots, e_{i-1}$,
  respectively, where indices are modulo $\ell+1$.

  Consider a copy $v_0$ of $v$ in the universal cover $G^\infty$. 
  Define the walk $P$ obtained by starting from $v_0$ following the
  edges of $G^\infty$ corresponding to the edges of $A$, and then to the
  edges of $B$. Similarly, define the walk $Q$ obtained by starting from
  $v_0$ following the edges of $B$, and then the edges of $A$.  The two
  walks $P$ and $Q$ both start from $v_0$ and both end at the same
  vertex $v_1$ that is a copy of $v$. Note that $v_1$ and $v_0$
  may coincide. All the vertices that
  are visited on the interior of $P$ and $Q$ have exactly one edge
  leaving on the left and exactly one edge leaving on the right.  The
  two walks $P$ and $Q$ may intersect before they end at $v_1$ thus we
  define $P'$ and $Q'$ as the subwalks of $P$ and $Q$ starting from
  $v_0$, ending on the same vertex $u$ (possibly distinct from $v_1$ or
  not) and such that $P'$ and $Q'$ are not intersecting on their
  interior vertices. Then the union of $P'$ and $Q'$ forms a cycle $C$
  of $G^\infty$. All the vertices of $C$ except possibly $v_0$ and $u$,
  have exactly one edge leaving $C$ and directed toward the interior
  of $C$, a contradiction to Lemma~\ref{lem:kmoins3}.
\end{proof}

A consequence of Lemma~\ref{lem:middlecycle} is that any 3-orientation
of a toroidal triangulation has a middle cycle.  The 3-orientation of
the toroidal triangulation on the left of Figure~\ref{fig:orientation}
is an example where there is a unique middle cycle (the diagonal).

By Lemma~\ref{lem:middlecycle}, a middle walk $W$ always ends on a
middle cycle. Let us denote by $M_W$ this middle cycle and $P_W$ the
part of $W$ before $M_W$. Note that $P_W$ may be empty. We say that a
 walk is leaving a cycle $C$ if its starting edge is incident to
$C$ and leaving $C$.

Let us now prove the main lemma of this section.

\begin{lemma}
\label{lem:2middlecycle}
$G$ admits a 3-orientation with two middle cycles that are not
weakly homologous. 
\end{lemma}

\begin{proof}
  Suppose by contradiction that there is no 3-orientation of $G$ with
  two middle cycles that are not weakly homologous. We first prove  the
  following claim:

  \begin{claim}
\label{cl:middleequal}
There exists a 3-orientation of $G$ with a middle cycle $M$, a middle
walk $W$ leaving $M$ and $M_W=M$.
  \end{claim}

 \begin{proofclaim}
   Suppose by contradiction that there is no 3-orientation of $G$ with
   a middle cycle $M$, a middle walk $W$ leaving $M$ and $M_W=M$. We
   first prove the following:

\begin{sclaim}
\label{claim:middleinterior}
Any 3-orientation of $G$, middle cycle $M$ and middle walk $W$ leaving
$M$ are such that $M$ does not intersect the interior of $W$.
\end{sclaim}

\begin{proofsclaim}
  Assume by contradiction that $M$ intersects the interior of $W$.  By
  assumption, cycles $M_W$ and $M$ are weakly homologous and $M_W\neq
  M$. Thus by Lemma~\ref{lem:middleequal}, they are vertex-disjoint.
  So $M$ intersects the interior of $P_W$.  Assume by symmetry that
  $P_W$ is leaving $M$ on its left side.  If $P_W$ is entering $M$
  from its left side, in $G^\infty$, the edges of $P_W$ plus $M$ form
  a cycle contradicting Lemma~\ref{lem:kmoins3}. So $P_W$ is entering
  $M$ from its right side. Hence $M_W$ intersects the interior of
  $P_W$ on a vertex $v$. Let $e$ be the edge of $P_W$ leaving
  $v$. Then the middle cycle $M_W$ and the middle walk $W'$ starting from
  $e$ satisfies $M_{W'}=M_W$, contradicting the hypothesis. So $M$
  does not intersect the interior of $W$.
\end{proofsclaim}

Consider a 3-orientation, a middle cycle $M$ and a middle walk $W$
leaving $M$ such that the length of $P_W$ is maximized.  By assumption
$M_{W}$ is weakly homologous to $M$ and $M_W\neq M$. Assume by symmetry that $P_W$ is
leaving $M$ on its left side.
(\ref{claim:middleinterior}) implies that $M$ does not intersect the
interior of $W$. Let $v$ (resp. $e_0$) be the starting vertex
(resp. edge) of $W$.  Consider now the 3-orientation obtained by
reversing $M_W$.  Consider the middle walk $W'$ starting from $e_0$.
Walk $W'$ follows $P_W$, then arrives on $M_W$ and crosses it (since
$M_W$ has been reversed).  (\ref{claim:middleinterior}) implies that
$M$ does not intersect the interior of $W'$.  Similarly,
(\ref{claim:middleinterior}) applied to $M_W$ and $W'\setminus P_W$
(the walk obtained from $W'$ by removing the first edges corresponding
to $P_W$), implies that $M_W$ does not intersect the interior of
$W'\setminus P_W$. Thus, $M_{W'}$ is weakly homologous to $M_W$
and $M_{W'}$ is in the interior of the region between $M$ and $M_W$
on the right of $M$. Thus $P_{W'}$ strictly contains $P_W$
and is thus longer, a contradiction.
  \end{proofclaim}

  By Claim~\ref{cl:middleequal}, consider a 3-orientation of $G$ with
  a middle cycle $M$ and a middle walk $W$ leaving $M$ such that
  $M_W=M$. Note that $W$ is leaving $M$ from one side and entering it
  from the other side, otherwise $W$ and $M$ will contradicts
  Lemma~\ref{lem:kmoins3}.  Let $e_0$ be the starting edge of $W$. Let
  $v,u$ be the starting and ending point of $P_W$, respectively, where
  $u=v$ may occur.  Consider the 3-orientation obtained by reversing
  $M$. Let $Q$ be the directed path from $u$ to $v$ along $M$ ($Q$ is
  empty if $u=v$). Let $C$ be the directed cycle $P_W\cup Q$. We
  compute the value $\gamma$ of $C$. If $u\neq v$, then $C$ is almost
  everywhere a middle cycle, except at $u$ and $v$. At $u$, it has two
  edges leaving on its right side, and at $v$ it has two edges leaving
  on its left side. So we have $\gamma(C)=0$. If $u=v$, then $C$ is a
  middle cycle and $\gamma(C)=0$.  Thus, in any case
  $\gamma(C)=0$. Note that furthermore $\gamma(M)=0$ holds.  The two
  cycles $M,C$ are non-contractible and not weakly homologous, so any
  non-contractible cycle of $G$ has $\gamma$ equal to zero by
  Lemma~\ref{lem:typegamma0foranybasis}.

  Consider the middle walk $W'$ from $e_0$. By assumption $M_{W'}$ is
  weakly homologous to $M$.  The beginning $P_{W'}$ is the same as for
  $P_W$. As we have reversed the edges of $M$, when arriving on $u$,
  path $P_{W'}$ crosses $M$ and continues until reaching
  $M_{W'}$. Thus $M_{W'}$ intersects the interior of $P_{W'}$ at a
  vertex $v'$. Let $u'$ be the ending point of $P_{W'}$ (note that we
  may have $u'=v'$). Let $P'$ be the non-empty subpath of $P_{W'}$
  from $v'$ to $u'$. Let $Q'$ be the directed path from $u'$ to $v'$
  along $M_{W'}$ ($Q'$ is empty if $u'=v'$). Let $C'$ be the
  non-contractible directed cycle $P'\cup Q'$. We compute
  $\gamma(C')$. The cycle $C'$ is almost everywhere a middle cycle,
  except at $v'$. At $v'$, it has two edges leaving on its left or
  right side, depending on $M_{W'}$ crossing $P_{W'}$ from its left or
  right side. Thus, we have $\gamma(C')=\pm 2$, a contradiction.
\end{proof}

By Lemma~\ref{lem:2middlecycle}, for any toroidal triangulation, there
exists a 3-orientation with two middle cycles that are not weakly
homologous. By Lemma~\ref{lem:typegamma0foranybasis}, this implies that
any non-contractible cycle of $G$ has value $\gamma$ equal to
zero so the $3$-orientation corresponds to a balanced Schnyder wood. 

Note that $\gamma(C)=0$ for any non-contractible cycle $C$ does
not necessarily imply the existence of two middle cycle that are not
weakly homologous.  The 3-orientation of the toroidal triangulation of
Figure~\ref{fig:gamma0htc} is an example where $\gamma(C)=0$ for any
non-contractible cycle $C$ but all the middle cycle are weakly
homologous.

By combining Lemma~\ref{lem:2middlecycle} and
Theorem~\ref{th:characterizationgamma}, we obtain the following:

\begin{theorem}
\label{th:half-cross-tri}
A toroidal triangulation admits a half-crossing Schnyder wood.
\end{theorem}

\begin{proof} 
  By Lemma~\ref{lem:2middlecycle}, there exists a 3-orientation with
  two middle cycles that are not weakly homologous. By
  Lemma~\ref{lem:typegamma0foranybasis}, any non-contractible cycle of
  $G$ has value $\gamma$ equal to zero.  Thus the 3-orientation
  corresponds to a balanced Schnyder wood.  In a Schnyder wood, a
  middle cycle corresponds to a monochromatic cycle. Since there are
  two middle cycles that are not weakly homologous, there are two
  crossing monochromatic cycles and the Schnyder wood is
  half-crossing.
\end{proof}

The proof presented in this section helps to understand the structure
of $3$-orientations in toroidal triangulations. Indeed, it shows how
to transform any $3$-orientation into a half-crossing Schnyder wood by
reversing non-contractible cycles.  It can be transformed into a
linear algorithm to find an half-crossing Schnyder wood in a toroidal
triangulation. First one has to find a $3$-orientation and this can be
done in linear time by contracting edges as explained in
Section~\ref{sec:contractionproof}. Then one start a middle walk from
any edge. This walk ends on a middle cycle $C$. One picks any edge
that is leaving $C$ and start a new middle walk from this edge. Each
time the walk ends on a middle cycle $C'$ that is weakly homologous to
$C$, cycle $C'$ is reversed and thus the middle walk can continue
further until it crosses $C$. Then the obtained orientation
corresponds to a crossing Schnyder wood.

For a triangulation on a genus $g\geq 2$ orientable surface we already
know that there exists an orientation of the edges such that every
vertex has outdegree at least $3$, and divisible by $3$ by
Theorem~\ref{th:AGK}. Can one apply the same method as here,
i.e. reverse middle cycles, to obtain a Schnyder wood?  First, one has
to find the good notion of ``middle'' when dealing with special
vertices of degree $6,9,...$ We believe that one should consider that
``middle'' is obtained by choosing any edge that is leaving
$1 \bmod 3$ edges on the left (and thus $1 \bmod 3$ edges on the
right). Note that their might be several choices depending on the
outdegree of the vertex. Such a definition follows the notion of
monochromatic paths that we want to achieve (see
Figure~\ref{fig:369}). Then, unfortunately, the conclusion of
Lemma~\ref{lem:middlecycle} is false. Already when $g=2$, there might
exist ``middle'' walks whose ending part is not a cycle.
Figure~\ref{fig:middle8} gives an example of an orientation of the
edges of the triangulation of Figure~\ref{fig:doubletorus} where every
vertex has outdegree at least $3$, and divisible by $3$. The two
vertices that are circled have outdegree 6.  A ``middle'' periodic
walk is depicted in green, it is not a cycle and forms an ``eight''.
On this example one can try to reverse the edges of the eight or of
subcycles of middle walks in order to obtain a Schnyder wood but this
method is not working on this example even if the process is
repeated. Thus something more has to be found.

\begin{figure}[!h]
\center
\includegraphics[scale=0.3]{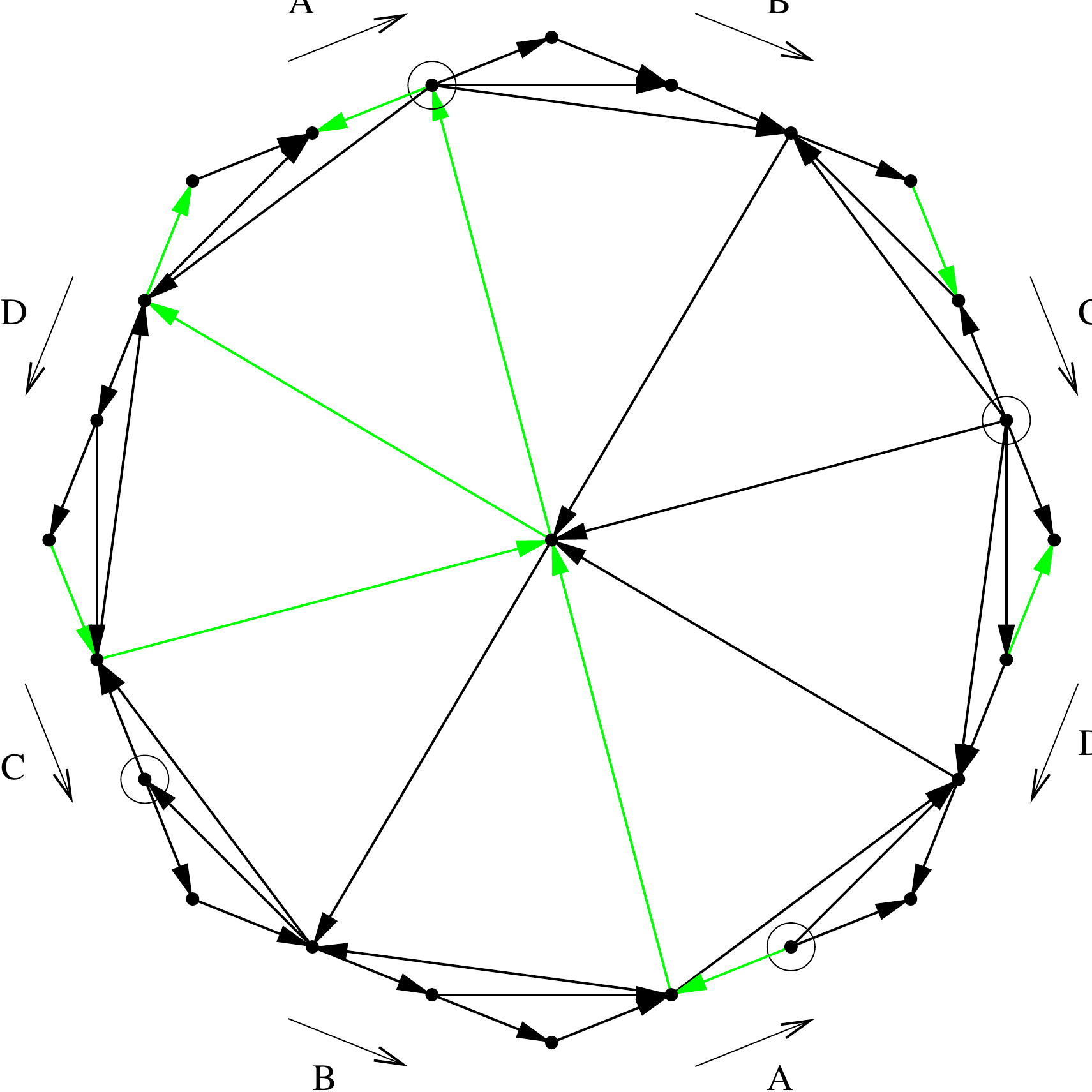}
\caption{Orientation of a triangulation of the double torus where a
  ``middle'' walk form an eight.}
\label{fig:middle8}
\end{figure}

Theorem~\ref{th:half-cross-tri} has some consequences concerning
realizers.  A nonempty family $\mathcal R$ of linear orders on the
vertex set $V$ of a simple graph $G$ is called a \emph{realizer} of
$G$ if for every edge $e$, and every vertex $x$ not in $e$, there is
some order $<_i\in \mathcal R$ so that $y<_ix$ for every $y\in e$. The
\emph{dimension} \cite{FT05} of $G$, is defined as the least positive
integer $t$ for which $G$ has a realizer of cardinality $t$. Realizers
are usually used on finite graphs, but here we allow $G$ to be an
infinite simple graph.

Schnyder woods where originally defined by Schnyder~\cite{Sch89} to
prove that a finite planar map $G$ has dimension at most $3$.  A
consequence of Theorem~\ref{th:half-cross-tri} is an analogous result
for the universal cover of a toroidal map:

\begin{theorem}
  The universal cover of a toroidal map has dimension
  at most three.
\end{theorem}

\begin{proof}
  Let $G$ be a toroidal map. By eventually adding edges to $G$ we may
  assume that $G$ is a toroidal triangulation.  By
  Theorem~\ref{th:half-cross-tri}, it admits a half-crossing Schnyder
  wood.  For $i\in\{0,1,2\}$, let $<_i$ be the order induced by the
  inclusion of the regions $R_i$ in $G^\infty$ (see
  Section~\ref{sec:universalcover}). That is $u<_{i}v$ if and only if
  $R_i(u)\subsetneq R_i(v)$. Let $<'_i$ be any linear extension of
  $<_i$ and consider $\mathcal R=\{<'_0,<'_1,<'_2\}$.  Let $e$ be any
  edge of $G^{\infty}$ and $v$ be any vertex of $G^{\infty}$ not in
  $e$. Edge $e$ is in a region $R_i(v)$ for some $i$, thus
  $R_i(u)\subseteq R_i(v)$ for every $u\in e$ by
  Lemma~\ref{lem:regionss}.(i).  As there is no edges oriented in two
  directions in a Schnyder wood of a toroidal triangulation, we have
  $R_i(u)\neq R_i(v)$ and so $u <_{i}v$.  Thus $\mathcal R$ is a
  realizer of $G^\infty$.
\end{proof}

\subsection{$0$-homologous oriented subgraphs}
\label{sec:flip-contract}
In this section we  transform a half-crossing Schnyder wood
of a toroidal triangulation into a crossing Schnyder wood by flipping
$0$-homologous oriented subgraphs and thus obtain the following theorem:

\begin{theorem}
\label{th:cross-tri}
  A toroidal triangulation admits a crossing Schnyder wood.
\end{theorem}

\begin{proof} 
  Consider a toroidal triangulation $G$ and suppose by contradiction
  that $G$ does not admit a crossing Schnyder wood.  By
  Theorem~\ref{th:half-cross-tri}, $G$ admits a half-crossing Schnyder
  wood.  Consider the set $\mathcal H$ of half-crossing Schnyder woods
  of $G$ with a coloring such that they are $i$-crossing all for the
  same color $i$.  By assumption, Schnyder woods of $\mathcal H$ are
  not crossing, thus, by Theorem~\ref{lem:type-cross}, all the
  $(i-1)$-cycles and the $(i+1)$-cycles are reversely homologous.
  Consider the subset $\mathcal H'$ of $\mathcal H$ that minimizes the
  winding number (defined in Section~\ref{sec:monochrocycles}).  Let
  $\omega$ be the winding number of elements of $\mathcal H'$ (so
  $\omega_i=0$ and $\omega_{i-1}=\omega_{i+1}=\omega$).  Among all
  the element of $\mathcal H'$, consider $H$ the half-crossing
  Schnyder wood that minimizes the number of $(i+1)$-cycles.

  Consider a vertex $v$ that is in the intersection of a $i$-cycle $C$
  and a $(i+1)$-cycle $C'$. Then we are in one of the three cases
  depicted on the left side of Figure~\ref{fig:hal-cross-slide} (the
  color blue plays the role of color $i$ in the figures). If the
  winding number of $H$ is $1$, then we are in first case. If the
  winding number is strictly more than $1$, then $C'$ is going either
  ``down'' (second case), or ``up'' (third case) according to $C$
  being ``vertical''.

  Consider the path $P_{i-1}(v)$ obtained by following edges of color
  $i-1$ from $v$. Except on it starting point $v$, the path $P_i(v)$
  cannot intersect $C'$, otherwise it contradicts the fact that there
  exists a $(i-1)$-cycle reversely homologous to $C'$.  Consider the
  subpath $P$ of $P_i(v)$ starting from $v$ and ending on the
  $\omega^{th}$ intersection with $C$.  Path $P$, cycle $C'$ and a
  part of $C$ define a region $R$ (depicted in gray on
  Figure~\ref{fig:hal-cross-slide}) whose border forms a
  $0$-homologous oriented subgraph $T$ oriented \ccw according to $R$.
  Consider the Schnyder wood $H'$ obtained by reversing all the edges
  of $T$. The different cases are represented on the right side of
  Figure~\ref{fig:hal-cross-slide}.  The colors of the edges on $T$
  get $+1$. The colors of the edges in the interior of $R$ get
  $-1$. The colors of the edges that are not in $R$ are not modified.

\begin{figure}[!h]
\center
\begin{tabular}{cc}
\includegraphics[scale=0.3]{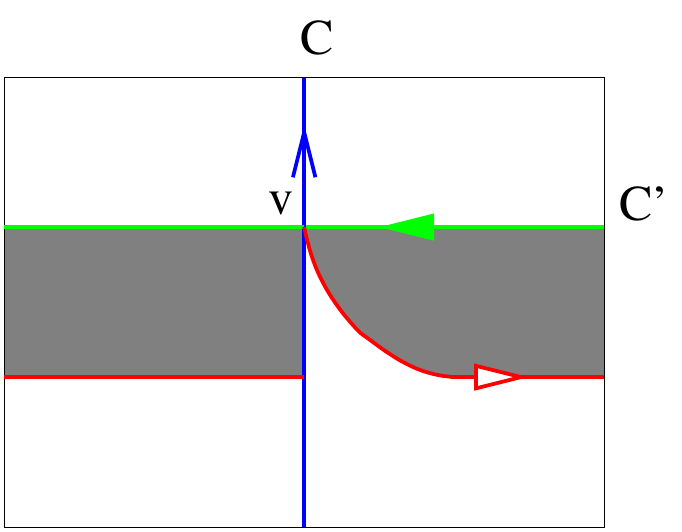} \ &\ 
\includegraphics[scale=0.3]{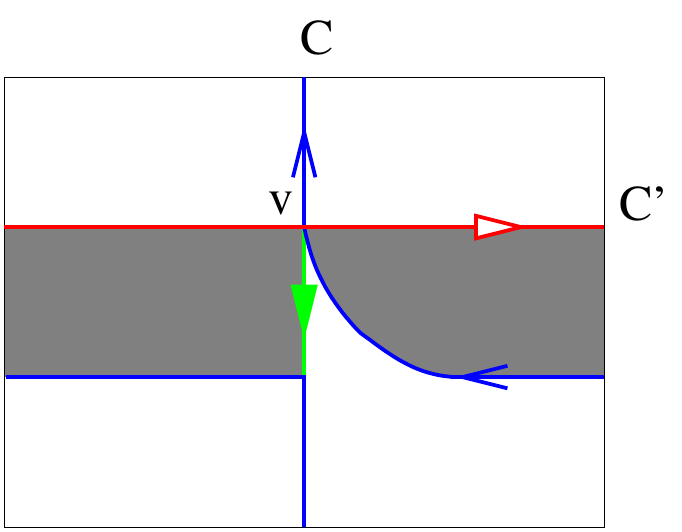} \\
 & \\
\includegraphics[scale=0.3]{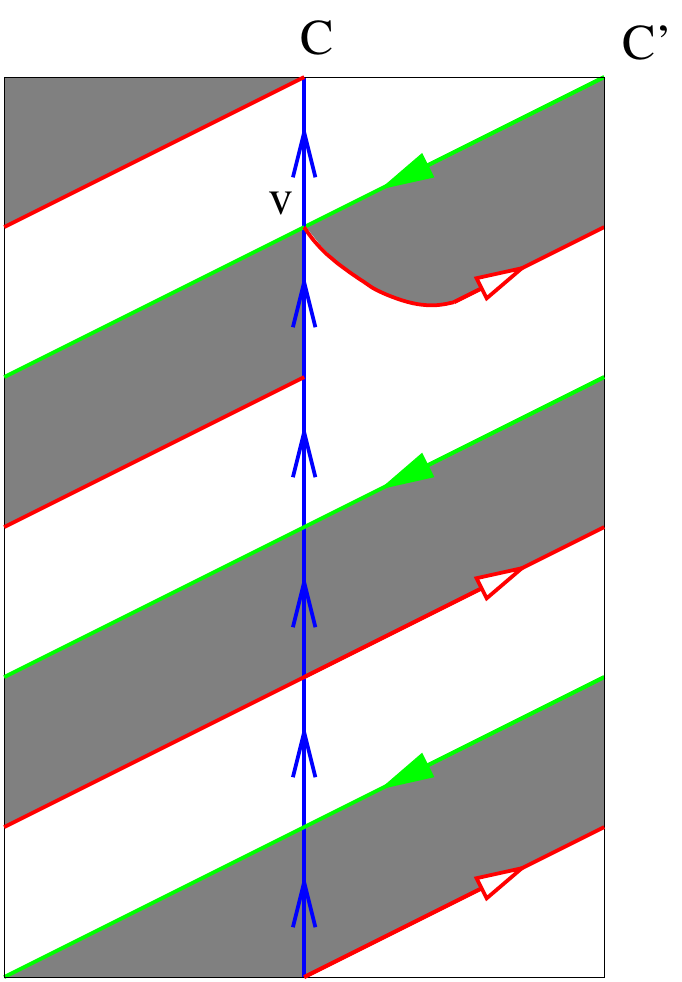}  \ &\ 
\includegraphics[scale=0.3]{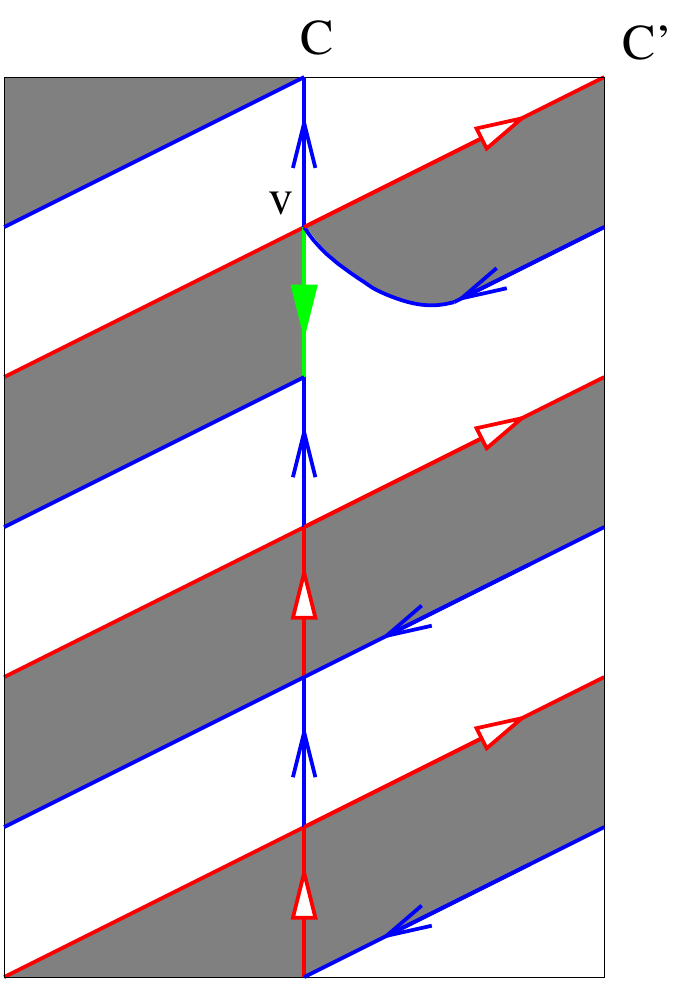} \\
 & \\
\includegraphics[scale=0.3]{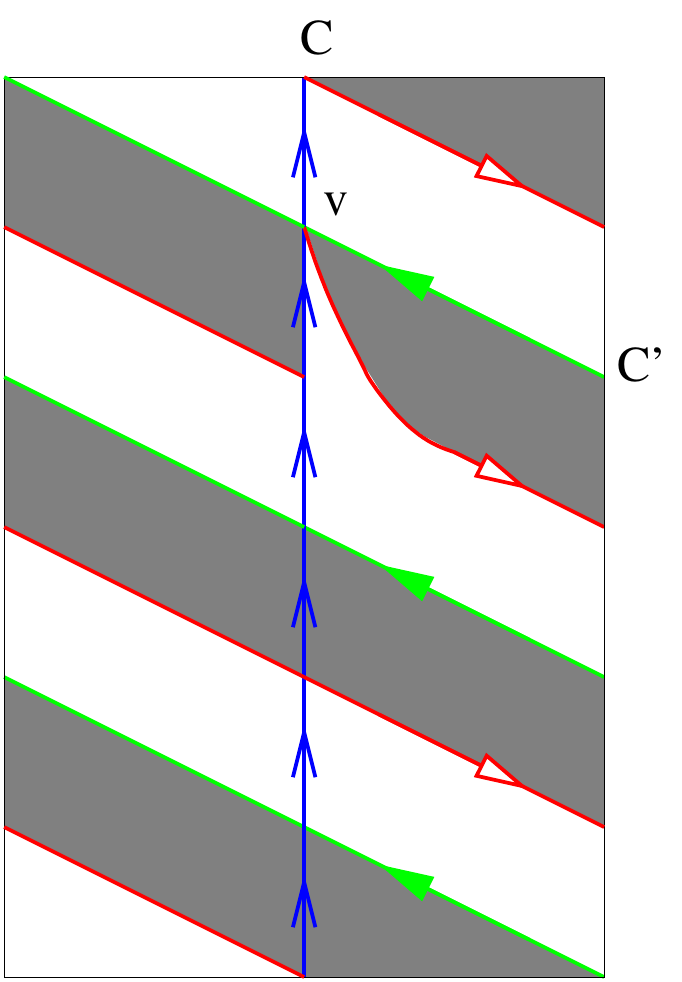}  \ &\ 
\includegraphics[scale=0.3]{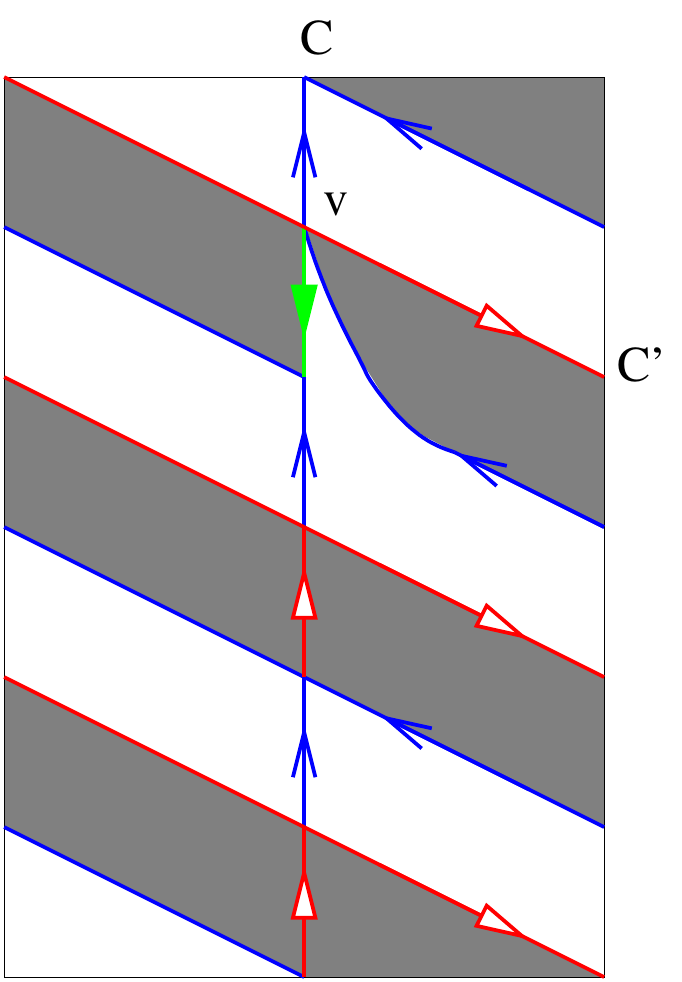} \\
 & \\
H  \ &\  H' \\
\end{tabular}
\caption{Left: original Schnyder wood, right: Schnyder wood
  obtained after reversing the edges on the border of the gray
  region.}
\label{fig:hal-cross-slide}
\end{figure}

Note that in each case of Figure~\ref{fig:hal-cross-slide}, there is a
$i$-cycle intersecting the $(i-1)$-cycle corresponding to the reversal
of $C'$. Thus the obtained Schnyder wood is half-crossing. It is not
crossing by assumption. Moreover its winding number is $1$. Thus, by
the choice of $H$,  the winding number of $H$ is also $1$ and we are in
fact in the first case of Figure~\ref{fig:hal-cross-slide}.

Schnyder wood $H'$ is either $i$-crossing or $(i-1)$-crossing.
Suppose by contradiction that it is $(i-1)$-crossing. By
Theorem~\ref{lem:type-cross}, the $i$- and $(i+1)$-cycles are
reversely homologous. Note that there is a $i$-cycle crossing $C$ from
right to left. Thus the $(i+1)$-cycles should cross $C$ from left to
right. But this is impossible by the Schnyder property since $C$ is
made of edges of color $i$ going ``up'' and edges of color $i+1$ going
``down''. So the obtained Schnyder wood is $i$-crossing and all the
$(i-1)$-cycles and $(i+1)$-cycles are reversely homologous.

We claim that a $(i+1)$-cycle cannot enter inside the region $R$. It
cannot enter inside $R$ by crossing $C'$ since $C'$ is now a
$(i-1)$-cycle and it should be reversely homologous to it.  It cannot
enter by the rest of the border of $R$ by the Schnyder property. So
finally there is no $(i+1)$-cycle intersecting $R$. Recall that edges
that are not in $R$ are not modified from $H$ to $H'$, so $H'$ has
strictly less $(i+1)$-cycles than $H$ and winding number $1$,
contradicting the choice of $H$.
\end{proof}

The proof of Theorem~\ref{th:cross-tri} can be turned into a linear time
algorithm to find a crossing Schnyder wood in a toroidal
triangulation. First one has to apply the method of
Section~\ref{sec:flip-noncontract} to obtain a half-crossing Schnyder
wood in linear time. One should test if the obtained Schnyder wood is
crossing. This can be done in linear time by following three
monochromatic path from a vertex and check if their ending periodic
monochromatic cycles intersect or not.  If the half-crossing Schnyder
wood has winding number strictly greater than $1$, then one apply the
transformation described on Figure~\ref{fig:hal-cross-slide} once to
get a half-crossing Schnyder wood with winding number $1$. One should
check that the obtained Schnyder wood is crossing or not. If it is
crossing we are done. Suppose it is not crossing but
$i$-crossing. Then one should consider the set $\mathcal C$ of all the
$(i+1)$-cycles, and apply the transformations described on
Figure~\ref{fig:hal-cross-slide} by picking one by one a cycle in $C$.
After each step the elements of $\mathcal C$ that are no more
$(i+1)$-cycles in the current Schnyder wood are removed from
$\mathcal C$. When $C$ is empty we have a crossing Schnyder wood. All
the considered transformations are done one disjoint part of the
triangulation and thus in total this can be done in linear time.

We do not know if this method can be generalized to
essentially-3-connected toroidal maps.  In
Section~\ref{sec:flip-noncontract} we are considering middle walks and
here monochromatic paths to defined some oriented cycles that have to
be flipped. For essentially 3-connected toroidal maps maps, we have to
work in the primal-dual-completion where some more complex object have
to be considered since oriented cycles have to be defined on primal-,
dual- and edge-vertices. Then it is not only a ``middle'' or
``monochromatic'' property that should be considered and most of the
intuition is lost.

Nevertheless in Chapter~\ref{chap:generalexistence}, we add
additional requirement to the contraction method presented in
Section~\ref{sec:contractionproof} to prove existence of crossing
Schnyder woods for essentially-3-connected toroidal maps. One drawback
of Chapter~\ref{chap:generalexistence} is that it does not lead to a
linear algorithm. So if one really wants to keep linear complexity,
then the above method is the best that we are able to do in linear
time and one might be interesting into generalizing it to essentially
3-connected toroidal maps.

The existence of crossing Schnyder woods for toroidal triangulations implies the
following theorem.

\begin{theorem}
\label{cor:cyclesedge-disjoint}
A toroidal triangulation contains three non-contractible and not weakly
homologous cycles that are pairwise edge-disjoint.
\end{theorem}

\begin{proof}
  One just has to apply Theorem~\ref{th:cross-tri} to obtain a
  crossing Schnyder wood and then, for each color $i$, choose
  an arbitrarily  $i$-cycle. These cycles are edge-disjoint as there is
  no bioriented edges for triangulations. Moreover they are not
  homologous by Lemma~\ref{lem:twonothomotopic}.
\end{proof}

Note that, in the case of simple toroidal triangulations, there exists
a stronger form of Theorem~\ref{cor:cyclesedge-disjoint} (see
Theorem~\ref{th:fij} in next section) .

  \section{Gluing two planar Schnyder woods}
\label{sec:gluingproof}

In this section we present a proof of existence of Schnyder woods that
consists in cutting a simple toroidal triangulation into two planar
\emph{near-triangulations} (i.e. whose inner faces have size three),
find planar Schnyder woods of these maps and glue them on their
boundaries to obtain a toroidal Schnyder wood of the original
triangulation.

Fijavz~\cite{Fij} proved a useful result concerning existence of
particular non homologous cycles in toroidal triangulations with no
loop and no multiple edges. (Recall that in this manuscript we are less
restrictive as we allow non-contractible loops and multiple edges that
are not homotopic.)

\begin{theorem}[\cite{Fij}]
\label{th:fij} 
A simple toroidal triangulation contains three non-contractible and
not weakly homologous cycles that all intersect on one vertex and that are
pairwise disjoint otherwise.
\end{theorem}

Note that Theorem~\ref{th:fij} is not true for all toroidal
triangulations as shown by the example on
Figure~\ref{fig:not-connected}, whereas the weak version,
Theorem~\ref{cor:cyclesedge-disjoint}, is true for all triangulations.

\begin{figure}[!h]
\center
\includegraphics[scale=0.5]{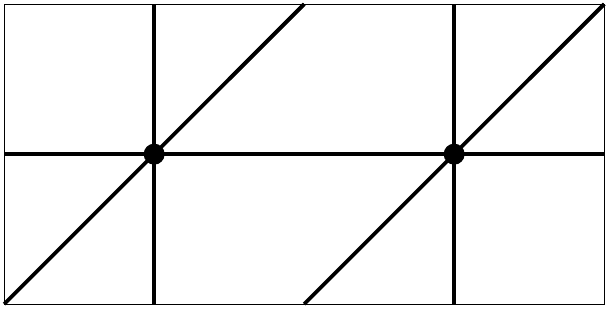}
\caption{A toroidal triangulation that does not contain three non-contractible and not weakly homologous cycles that all intersect on one
  vertex and that are pairwise disjoint otherwise.}
\label{fig:not-connected}
\end{figure}

We use Theorem~\ref{th:fij} to cut a simple triangulation and past
planar Schnyder woods in the two created regions. First, we need the
following lemma:

\begin{lemma}
  \label{lem:internally} If $G$ is a connected planar near-triangulation
  whose outer face is a cycle, and with three vertices
  $x_0,x_1,x_2$ on its outer face such that the three outer paths
  between the $x_i$ are chordless, then $G'$ is internally 3-connected
  for vertices $x_i$.
\end{lemma}

\begin{proof}
Let $G'$ be the graph obtained from $G$ by adding a vertex $z$ adjacent
to the three vertices $x_i$. We have to prove that $G'$ is 3-connected.
Let $S$ be a separator of $G'$ of minimum size and suppose by
contradiction that $1\leq |S|\leq 2$. Let $G''=G'\setminus S$.

For $v\in S$, the vertices of $N_{G'}(v)\setminus S$ should appear in
several connected components of $G''$, otherwise $S\setminus \{v\}$ is
also a separator of $G'$.  Since $G$ is a planar near-triangulation, the
neighbors of an inner vertex $v$ of $G$ form a cycle, thus there are at
least two vertex-disjoint paths between any two vertices of $N(v)$ in
$G'\setminus \{v\}$. So $S$ contains no inner vertex of $G$.
Similarly, the three neighbors of $z$ in $G'$ belong to a cycle of
$G'\setminus\{z\}$ (the outer face of $G$), so $S$ does not
face $z$.  So $S$ contains only vertices that are on the
outer face of $G$.

Let $v\in S$. Vertex $v$ is on the outer face of $G$ so its neighbors
in $G$ form a path $P$ where the two extremities of $P$ are the two
neighbors of $v$ on the outer face of $G$. Again if $v$ is one of the
$x_i$, then its neighbors in $G'$ belong to a cycle containing
$z,x_{i-1},x_{i+1}$, so $v$ is not one of the $x_i$.  So $S$ contains
an inner vertex $u$ of $P$. Vertex $u$ is also on the outer face of
$G$ so $uv$ is a chord of the outer cycle of $G$. As the three outer
paths between the $x_i$ are chordless, we have that $u,v$ lie on two
different outer paths between pairs of $x_i$. But then all the
vertices of $P\setminus \{u\}$ are in the same components of $G''$
because of $z$, a contradiction.
\end{proof}

Now we prove the existence of some particular crossing Schnyder woods
for simple toroidal triangulations.

\begin{theorem}
\label{th:schnydersimple}
A simple toroidal triangulation admits a crossing Schnyder wood with three
monochromatic cycles of different colors all intersecting on one
vertex and that are pairwise disjoint otherwise.
\end{theorem}

\begin{proof} Let $G$ be a simple toroidal triangulation.  By
  Theorem~\ref{th:fij}, let $C_0,C_1,C_2$ be three non-contractible
  and not weakly homologous cycles of $G$ that all intersect on one vertex $x$
  and that are pairwise disjoint otherwise. By eventually shortening
  the cycles $C_i$, we can assume that the three cycles $C_i$ are
  chordless.  By symmetry, we can assume that the six edges $e_i,e'_i$
  of the cycles $C_i$ incident to $x$ appear around $x$ in the
  counterclockwise order $e_0,e'_2,e_1,e'_0,e_2,e'_1$ (see
  Figure~\ref{fig:collage_notation}).  The cycles $C_i$ divide $G$
  into two regions, denoted $R_1,R_2$ such that $R_1$ is the region
  situated in the counterclockwise sector between $e_0$ and $e'_2$ of
  $x$ and $R_2$ is the region situated in the counterclockwise sector
  between $e'_2$ and $e_1$ of $x$. Let $G_i$ be the subgraph of $G$
  contained in the region $R_i$ (including the three cycles $C_i$).

\begin{figure}[!h]
\center
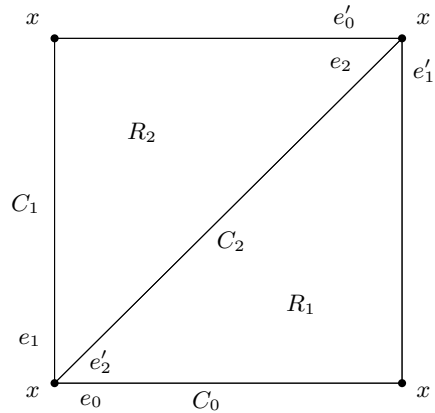
\caption{Notations of the proof of Theorem~\ref{th:schnydersimple}.}
\label{fig:collage_notation}
\end{figure}

Let $G'_1$ (resp. $G'_2$) be the graph obtained from $G_1$
(resp. $G_2$) by replacing $x$ by three vertices $x_0, x_1, x_2$, such
that $x_i$ is incident to the edges in the counterclockwise sector
between $e_{i+1}$ and $e'_i$ (resp. $e'_i$ and $e_{i-1}$) (see
Figure~\ref{fig:collage}).  The two graphs $G'_1$ and $G'_2$ are
planar near-triangulations and the $C_i$ are chordless, so by
Lemma~\ref{lem:internally}, they are internally 3-connected planar
maps for vertices $x_i$.  The vertices $x_0,x_1,x_2$ appear in
counterclockwise order on the outer face of $G'_1$ and $G'_2$.  By
Theorem~\ref{th:schnyder}, the two graphs $G'_i$ admit planar Schnyder
woods rooted at $x_0,x_1,x_2$.  Orient and color the edges of $G$ that
intersect the interior of $R_i$ by giving them the same orientation
and coloring as in a planar Schnyder wood of $G'_i$.  Orient and color
the cycle $C_i$ in color $i$ such that it is entering $x$ by edge
$e'_i$ and leaving $x$ by edge $e_i$. We claim that the orientation
and coloring that is obtained is a  Schnyder wood of $G$ (see
Figure~\ref{fig:collage}).

\begin{figure}[!h]
\center
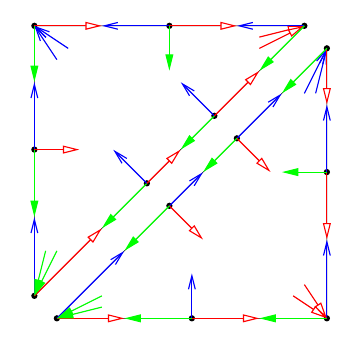
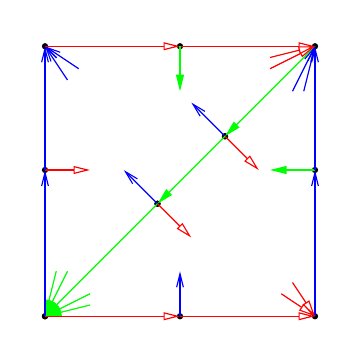
\caption{Gluing two planar Schnyder woods into a toroidal one to prove
  Theorem~\ref{th:schnydersimple}.}
\label{fig:collage}
\end{figure}

Clearly, any interior vertex of the region $R_i$ satisfies the
Schnyder property.  Let us show that it is also satisfied for any
vertex $v$ of a cycle $C_i$ distinct from $x$.  In a Schnyder wood of
$G'_1$, the cycle $C_i$ is oriented in two direction, from $x_{i-1}$
to $x_{i}$ in color $i$ and from $x_{i}$ to $x_{i-1}$ in color
$i-1$. Thus the edge leaving $v$ in color $i+1$ is an inner edge of
$G'_1$ and vertex $v$ has no edges entering in color
$i+1$. Symmetrically, in $G'_2$ the edge leaving $v$ in color $i-1$ is
an inner edge of $G'_2$ and vertex $v$ has no edges entering in color
$i-1$. Then one can paste $G'_1$ and $G'_2$ along $C_i$, orient $C_i$
in color $i$ and see that $v$ satisfies the Schnyder property. The
definition of  $G'_i$, and the orientation of the cycles is done so
that $x$ satisfies also the Schnyder property.  The cycles $C_i$ form
three monochromatic cycles of different colors that are all
intersecting on one vertex and that are pairwise disjoint otherwise
and thus the Schnyder wood is crossing.
\end{proof}

The proof of Theorem~\ref{th:schnydersimple} can be transformed into a
polynomial algorithm but we do not know its exact complexity. The (non
published) proof of Theorem~\ref{th:fij} is itself based on a result
of Robertson and Seymour~\cite{RS86} on disjoint paths.

We do not see how to adapt the method presented here if the map is non
simple, not a triangulation, or in higher genus. For non simple
triangulation, Theorem~\ref{th:fij} is false as shown by
Figure~\ref{fig:not-connected}. For essentially 3-connected toroidal
map, even if one find a way to cut the map, we do not know how to
deal with vertices on the border of the resulting planar maps since
they might have degree $2$ and then the planar map is not internally
3-connected. In higher genus we have no idea of what could be a
generalization of Theorem~\ref{th:fij}.

A Schnyder wood is \emph{point-crossing} if it is crossing and there
are three monochromatic cycles of different colors all intersecting in
one vertex (not necessarily disjoint otherwise).  
Theorem~\ref{th:schnydersimple} shows the existence of Schnyder woods
that are point-crossing and that have winding number $1$ for simple triangulation. Thus it is
particularly interesting since it is the only proof of such result
that we know.  Winding number $1$ and crossing property is not always
possible for non-simple toroidal triangulations (see for example the graph
of Figure~\ref{fig:not-connected}) but on the contrary the existence of
point-crossing Schnyder woods might be generalizable to all
triangulations?

Note that in the Schnyder wood obtained by
Theorem~\ref{th:schnydersimple}, we do not know if there are several
monochromatic cycles of one color or not.  We wonder whether
Theorem~\ref{th:schnydersimple} can be modified as follow: Does a
simple toroidal triangulation admits a Schnyder wood such that each
color induces a connected subgraph (so there is
just one monochromatic cycle per color)? Can one additionally require
that it is point-crossing and/or has winding number $1$?

\chapter{Existence for essentially 3-connected toroidal maps}
\label{chap:generalexistence}

In this chapter we generalize the proof of
Section~\ref{sec:contractionproof} to essentially 3-connected toroidal
maps, moreover we show that at each decontraction step the crossing
property can be preserved. Thus we obtain the existence of crossing
Schnyder woods for almost all essentially 3-connected toroidal maps,
except  a very particular family of toroidal maps for which only
intersecting Schnyder woods exists.

  \section{Contraction of essentially 3-connected toroidal maps}
\label{sec:existence3connected}

Given a map $G$ embedded on a surface.  The \emph{angle map}
\cite{Rose89} of $G$ is a map $A(G)$ on this surface whose vertices
are the vertices of $G$ plus the vertices of $G^*$ (i.e. the faces of
$G$), and whose edges are the angles of $G$, each angle being incident
with the corresponding vertex and face of $G$. Note that if $G$ has no
contractible loop nor homotopic multiple edges, then every face of $G$
has degree at least $3$ in $A(G)$.

Mohar and Rosenstiehl~\cite{MR98} proved that a map $G$ is essentially
2-connected if and only if the angle map $A$ of $G$ has no pair of
(multiple) edges bounding a disk (as every face in an angle map is a
quadrangle, such disk would contain some vertices of $G$).  The following
lemma naturally extends this characterization to essentially
3-connected toroidal maps.

\begin{lemma}
\label{lem:walk}
A toroidal map $G$ is essentially 3-connected if and only if the angle
map $A(G)$ has no walk of length at most four bounding a disk which is
not a face.
\end{lemma}

\begin{proof}
  ($\Longrightarrow$) Any walk in $A(G)$ of length at most 4 bounding
  a disk which is not a face lifts to a cycle of length at most 4
  bounding a disk which is not a face in $A(G^\infty)$. Thus such a
  walk implies the existence of a small separator in $G^\infty$,
  contradicting its 3-connectedness.

  ($\Longleftarrow$) According to~\cite{MR98}, if $G$ is essentially
  2-connected, $A(G)$ has no walk of length 2 bounding a disk.  If $G$
  is essentially 2-connected but not essentially 3-connected, then
  $A(G^\infty)$ has a cycle $C$ of length 4 bounding a disk which is
  not a face. Let $W$ be the walk of $A(G)$ corresponding to $C$.
  Walk $W$ contains a subwalk bounding a
  disk. Since $A(G)$ is bipartite, this subwalk has even length.
  Since $A(G)$ is essentially 2-connected, it has no such walk of
  length 2. Thus $W$ bounds a disk. Finally, this disk is not a single
  face since otherwise $C$ would bound a single face in $A(G^\infty)$.
\end{proof}

Given a toroidal map $G$, the \emph{contraction} of a non-loop-edge
$e$ of $G$ is the operation consisting of continuously contracting $e$
until merging its two ends. We note $G/e$ the obtained map.  On
Figure~\ref{fig:contraction} the contraction of an edge $e$ is
represented. We consider three different cases corresponding to
whether the faces adjacent to $e$ have size three or not.  Note that
only one edge of each set of multiple edges that is created is
preserved.  The contraction operation is also define when $t=u$ and
$y=v$ (the case represented on Figure~\ref{fig:contraction-loop}), or
$x=v$ and $z=u$ (corresponding to the symmetric case with a diagonal
in the other direction).

\begin{figure}[!h]
\center
\begin{tabular}{ccc}  
$G$& &$G/e$ \\
& & \\
\includegraphics[scale=0.4]{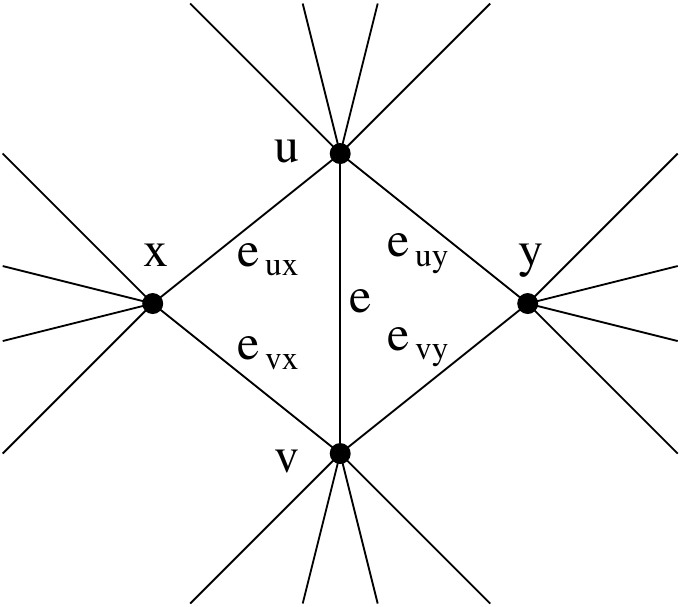}
& \hspace{3em} &
\includegraphics[scale=0.4]{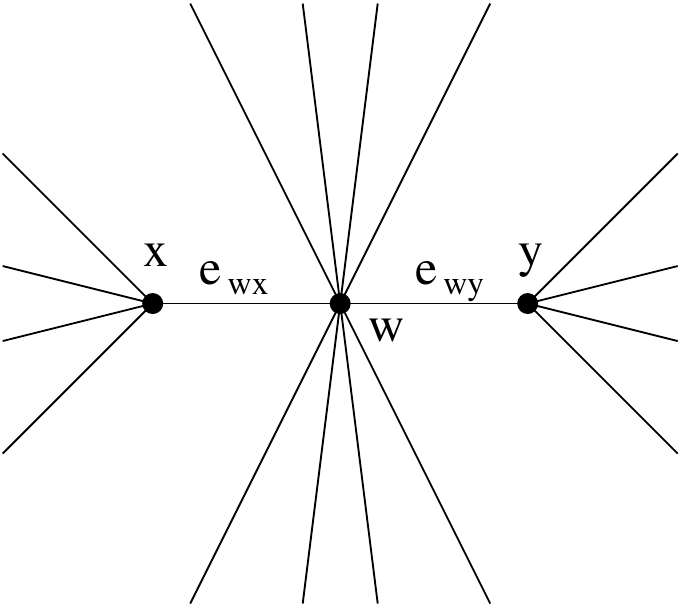}\\
(a) & & (a') \\
& & \\
\includegraphics[scale=0.4]{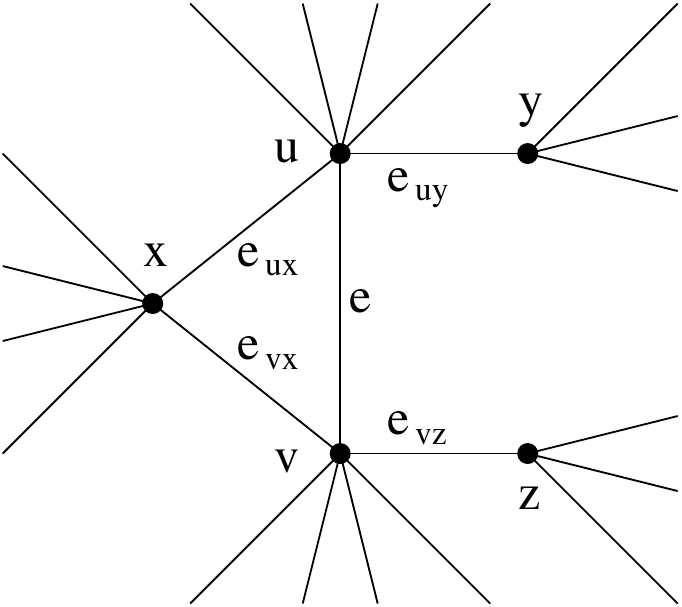}
& \hspace{3em} &
\includegraphics[scale=0.4]{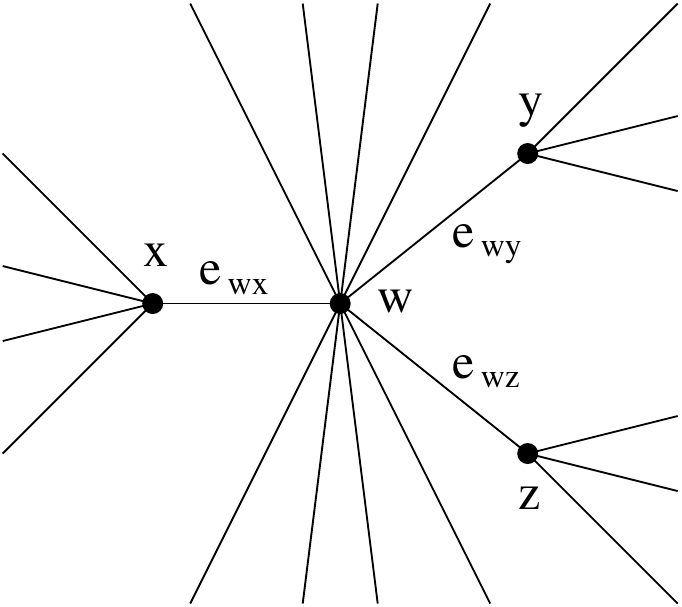}\\
(b) & & (b') \\
& & \\
\includegraphics[scale=0.4]{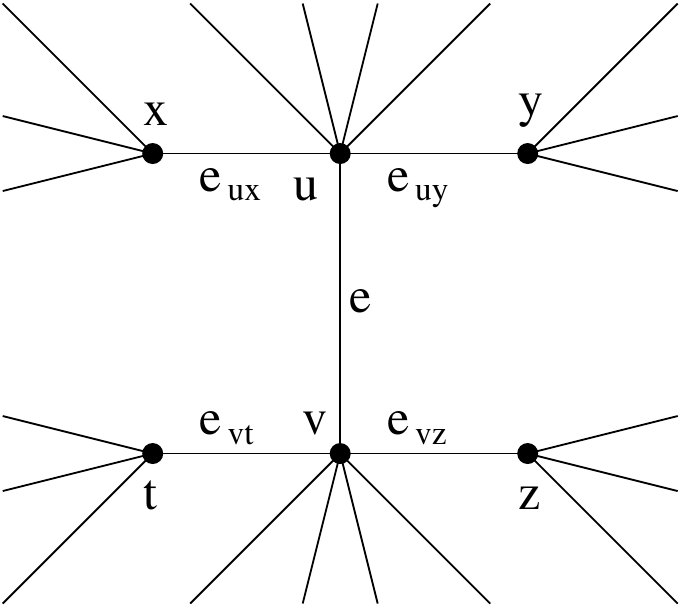}
& \hspace{3em} &
\includegraphics[scale=0.4]{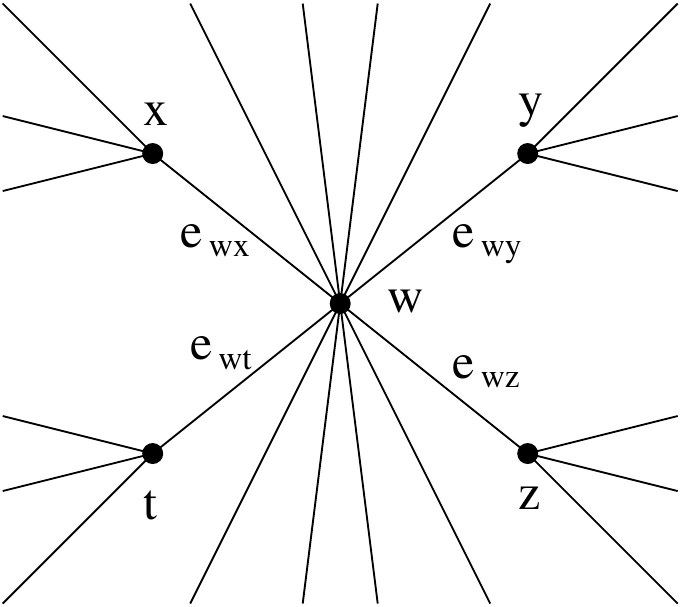}\\
(c) & & (c') \\
\end{tabular}
\caption{The contraction operation}
\label{fig:contraction}
\end{figure}

\begin{figure}[!h]
\center
\begin{tabular}{ccc} 
\includegraphics[scale=0.4]{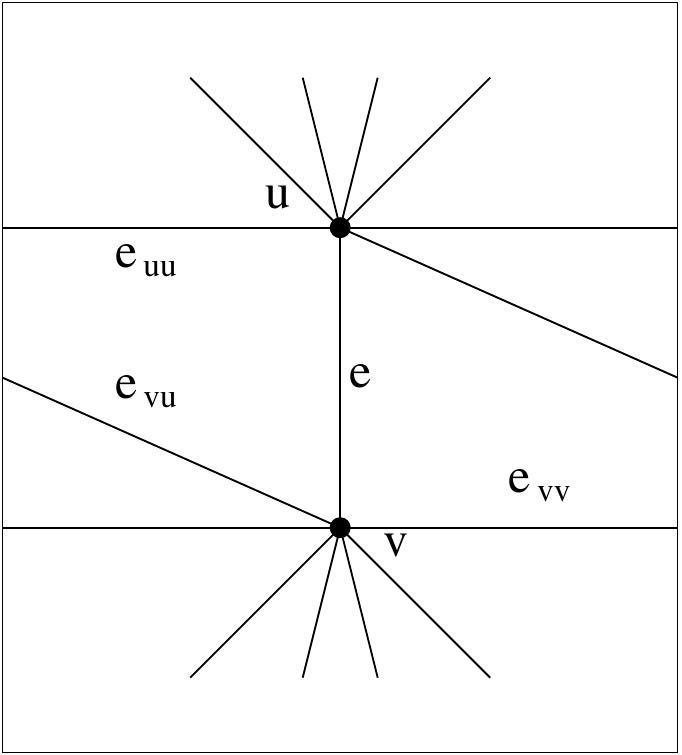}
& \hspace{3em} & 
\includegraphics[scale=0.4]{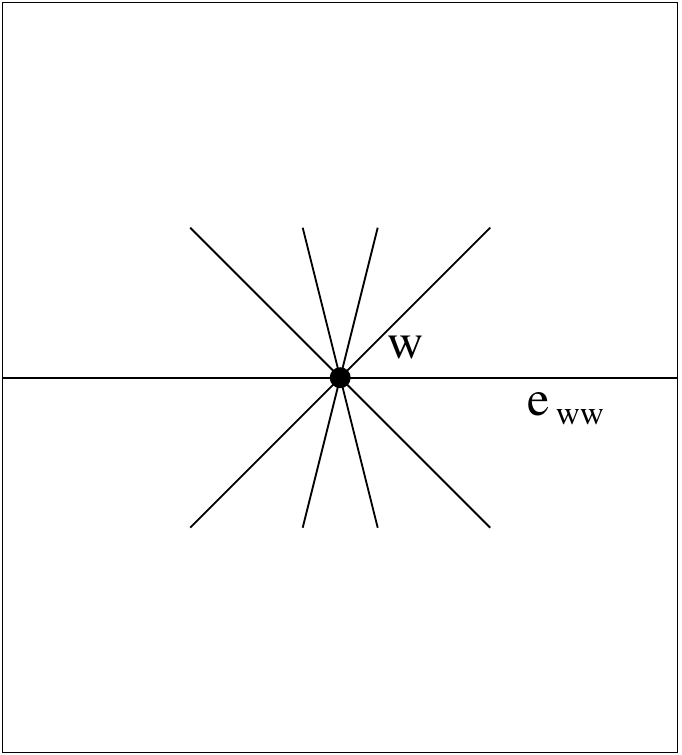}\\
(d) & & (d') \\
& & \\
\includegraphics[scale=0.4]{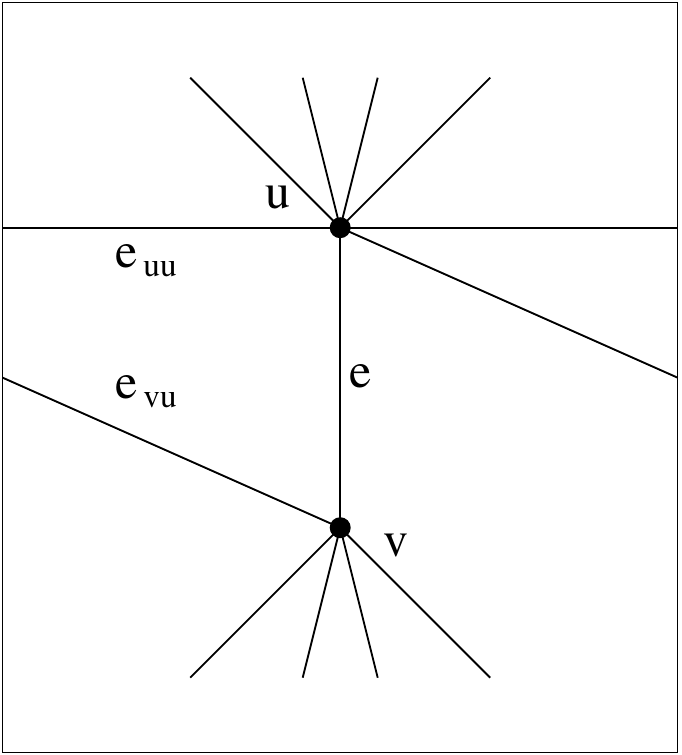}
& \hspace{3em} &
\includegraphics[scale=0.4]{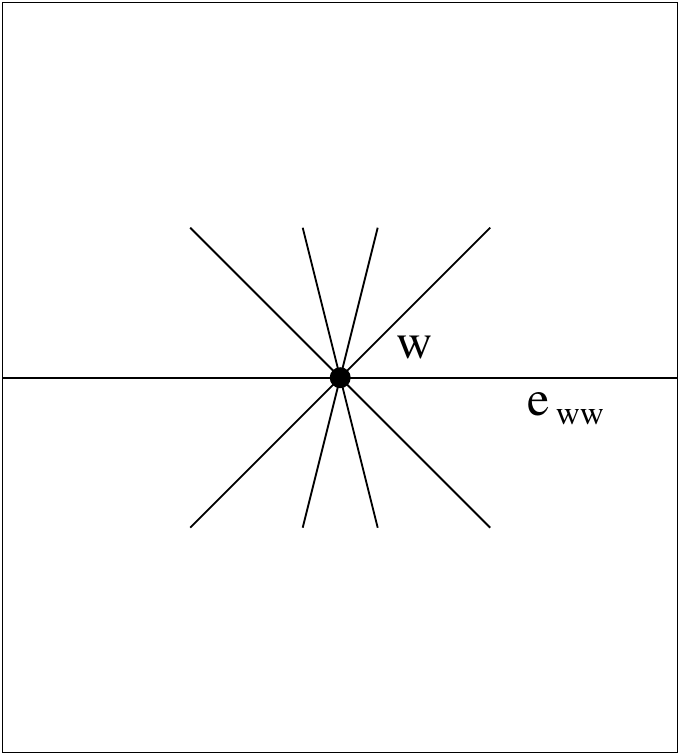}\\  
(e) & & (e') \\
& & \\
\includegraphics[scale=0.4]{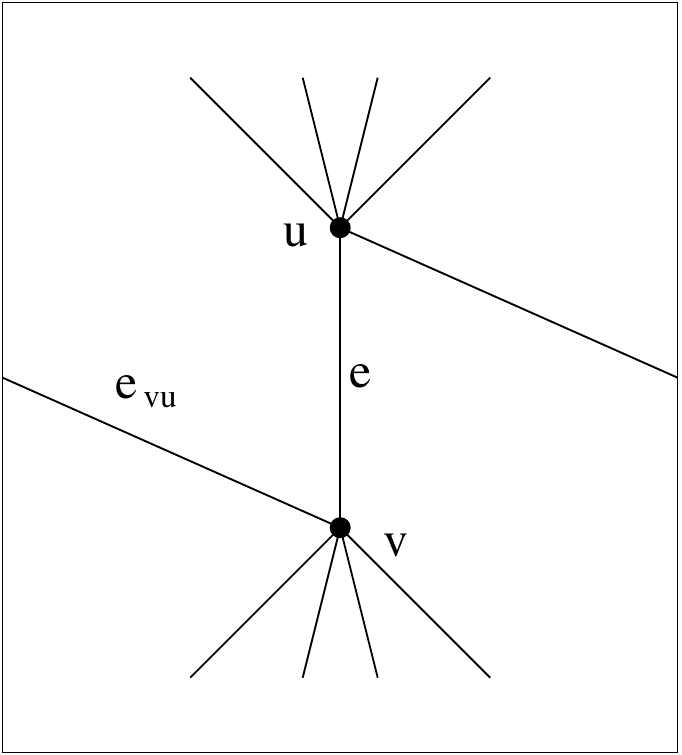}
& \hspace{3em} &
\includegraphics[scale=0.4]{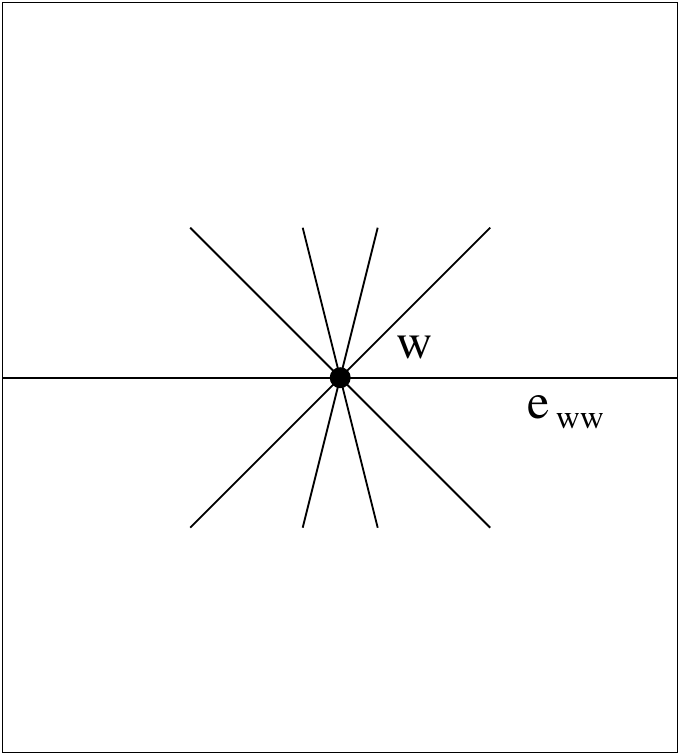}\\
(f) & & (f') \\
\end{tabular}
\caption{The contraction operation when some vertices are identified. }
\label{fig:contraction-loop}
\end{figure}

A non-loop edge $e$ of an essentially 3-connected toroidal map is
\emph{contractible} if the contraction of $e$ keeps the map
essentially 3-connected. We have the following lemma:

\begin{lemma}
\label{lem:contractibleedges}
  An essentially 3-connected toroidal map that is not reduced to a
  single vertex has a contractible edge.
\end{lemma}

\begin{proof}
  Let $G$ be an essentially 3-connected toroidal map with at least 2
  vertices. Note that for any non-loop edge $e$, the map $A(G/e)$ has no
  walk of length $2$ bounding a disk which is not a face, otherwise,
  $A(G)$ contains a walk of length at most $4$ bounding a disk which
  is not a face and thus, by Lemma~\ref{lem:walk}, $G$ is not
  essentially 3-connected.

  Suppose by contradiction that contracting any non-loop edge $e$ of
  $G$ yields a non essentially 3-connected map $G/e$. By
  Lemma~\ref{lem:walk},  the angle map $A(G)$ has no walk
  of length at most four bounding a disk which is not a face. For any
  non-loop edge $e$, let $W_4(e)$ be a 4-walk of $A(G/e)$ bounding a
  disk, which is maximal in terms of the faces it contains.  Among all
  the non-loop edges, let $e$ be the one such that the number of faces
  in $W_4(e)$ is minimum.  Let $W_4(e)=(v_1,f_1,v_2,f_2)$ and assume
  that the endpoints of $e$, say $a$ and $b$, are contracted into
  $v_2$ (see Figure~\ref{fig:essentially}).  Note that by
  maximality of $W_4(e)$, $v_1$ and $v_2$ do not have any common
  neighbor $f$ out of $W_4(e)$, such that $(v_1,f,v_2,f_1)$ bounds a
  disk.

\begin{figure}[h!]
\center
\begin{tabular}{ccc}
\includegraphics[scale=0.4]{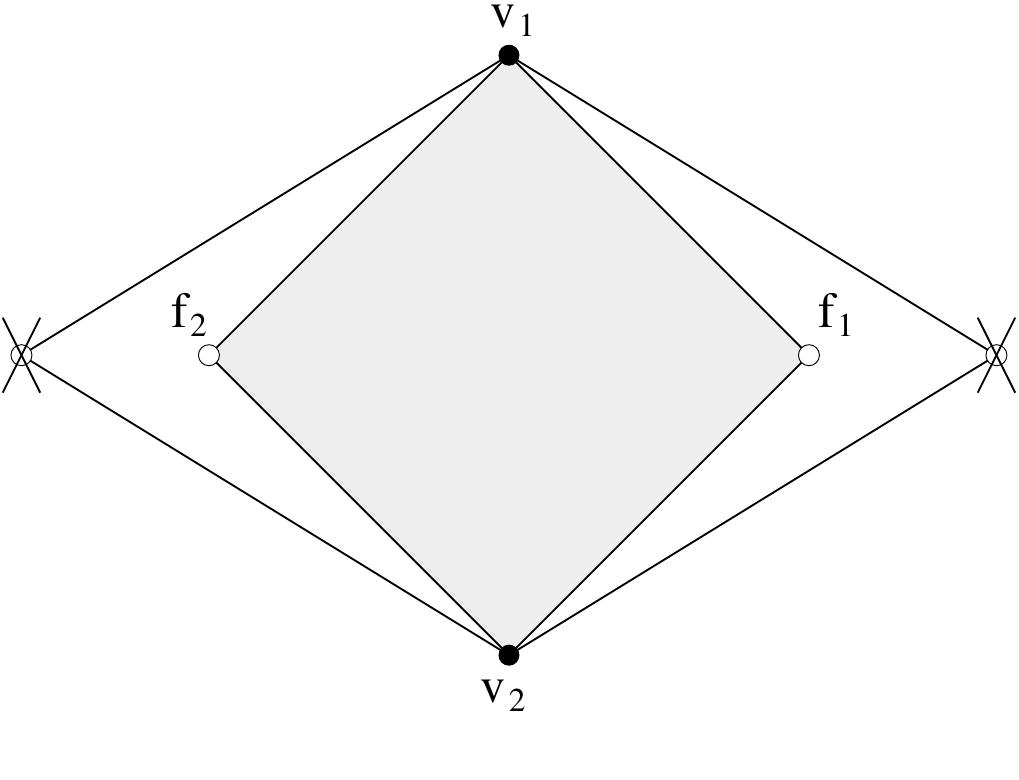}
&  &
\includegraphics[scale=0.4]{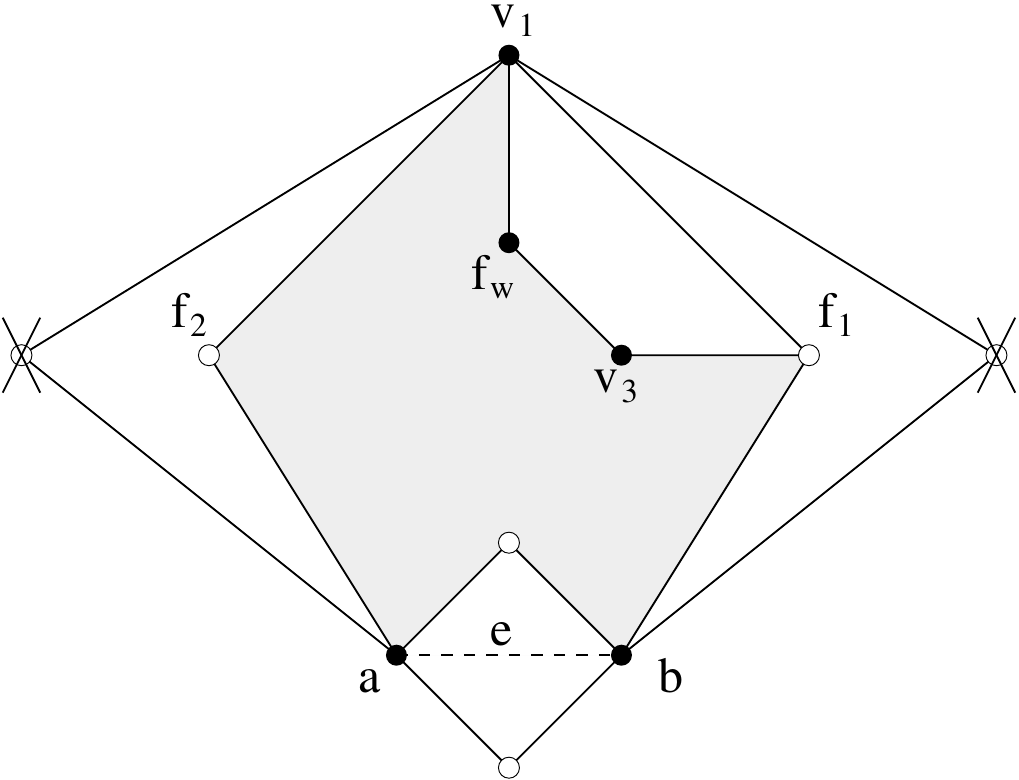}\\
$G/e$ & & $G$ \\
\end{tabular}
\caption{Notations of the proof of Lemma~\ref{lem:contractibleedges}}
\label{fig:essentially}
\end{figure}

Assume one of $f_1$ or $f_2$ has a neighbor inside $W_4(e)$. By
symmetry, assume $v_3$ is a vertex inside $W_4(e)$ such that there is
a face $F=(v_1,f_1,v_3,f_w)$ in $A(G/e)$, with eventually $f_w=f_2$.
Consider now the contraction of the edge $v_1v_3$.  Let
$P(v_1,v_3)=(v_1,f_x,v_y,f_z,v_3)$ be the path from $v_1$ to $v_3$
corresponding to $W_4(v_1v_3)$ and $P(a,b)=(a,f_2,v_1,f_1,b)$ the path
corresponding to $W_4(e)$.  Suppose that all the faces of
$W_4(v_1v_3)$ are in $W_4(e)$, then with $F$, $W_4(e)$ contains more
faces than $W_4(v_1v_3)$, a contradiction to the choice of $e$. So in
$A(G)$,  path $P(v_1,v_3)$  crosses  $P(a,b)$.

If $v_y=a$ or $v_y=b$. Then $v_1,v_y,f_z,v_3$ are in $W_4(e)$, so
$f_x$ is out and $(v_1,f_x,v_2,f_1)$ bounds a disk, a contradiction.
So $v_y\neq a$ and $v_y\neq b$ and thus $f_z=f_1$ or $f_z= f_2$.  If
$f_z=f_1$, then the cycle $(v_1,f_x,v_y,f_1)$ bounds a face by
Lemma~\ref{lem:walk}.  This implies that $W_4(v_1v_3)$ bounds a face,
a contradiction. So $f_z\neq f_1$ and thus $f_z= f_2$.  If
$f_w = f_2$, then similarly, the cycle $(v_1,f_x,v_y,f_2)$ bounds a
face by Lemma~\ref{lem:walk}.  This implies that $W_4(v_1v_3)$ bounds
a face, a contradiction. So $f_w\neq f_2$. Then $(v_1,f_2,v_3,f_w)$
bounds a face by Lemma~\ref{lem:walk} and then $f_w$ has degree $2$ in
$A(G)$, a contradiction.

   Assume now that none of $f_1$ or $f_2$ has a neighbor inside
  $W_4(e)$. Let $f'_1$, $f'_2$, $f_3$ and $f'_3$ be vertices of $A(G)$
  such that $(v_1,f_1,b,f'_1)$, $(v_1,f_2,a,f'_2)$ and
  $(a,f_3,b,f'_3)$ are faces (see
  Figure~\ref{fig:essentially2}). Suppose $f'_1 = f'_2 = f'_3$. Then
  in $A(G/e)$, the face $f'_1$ is deleted (among the 2 homologous
  multiple edges between $v_1,v_2$ that are created, only one is kept
  in $G/e$). Then $W_4(e)$ bounds a face, a contradiction. Thus there
  exists some $i$ such that $f'_i \neq f'_{i+1}$. Assume that $i=1$
  (resp. $i=2$ or 3), and let $v_3$ and $f''$ be such that there is a
  face $(v_1,f'_1,v_3,f'')$ in $A(G)$ (resp. $(a,f'_2,v_3,f'')$ or
  $(b,f'_3,v_3,f'')$). As above considering the contraction of the
  edge $v_1v_3$ (resp. $av_3$ or $bv_3$) yields a contradiction.
\begin{figure}[h!]
\center
\includegraphics[scale=0.4]{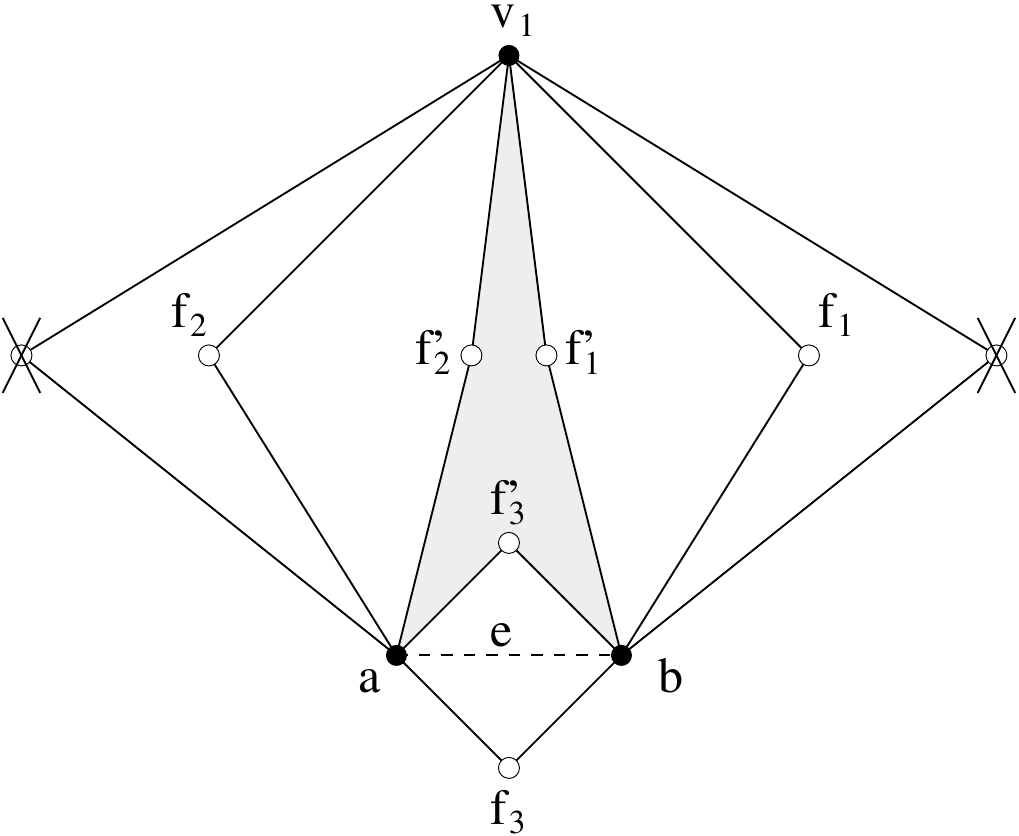}
\caption{Notations of the proof of Lemma~\ref{lem:contractibleedges}}
\label{fig:essentially2}
\end{figure}
\end{proof}

Lemma~\ref{lem:contractibleedges} shows that an essentially
3-connected toroidal map can be contracted step by step by keeping it
essentially 3-connected until obtaining a map with just one vertex.
The two essentially 3-connected toroidal maps on one vertex are
represented on Figure~\ref{fig:essential} with a Schnyder wood. The
map of Figure~\ref{fig:essential}.(a), the \emph{3-loops}, admits a
crossing Schnyder wood, and the map of Figure~\ref{fig:essential}.(b),
the \emph{2-loops}, admits an intersecting Schnyder wood.

\begin{figure}[!h] 
\center
\begin{tabular}{ccc}
\includegraphics[scale=0.5]{orientation-col}
& \hspace{4em} &
\includegraphics[scale=0.5]{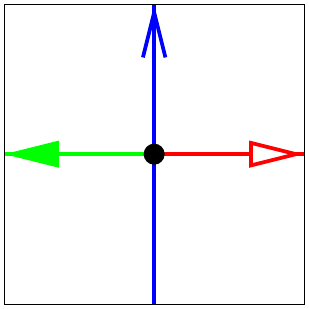}\\
(a) The 3-loops & &(b) The 2-loops 
\end{tabular}
\caption{The two essentially 3-connected toroidal maps on one vertex.}
\label{fig:essential}
\end{figure}

It would be convenient if one could contract any essentially
3-connected toroidal map until obtaining one of the two maps of
Figure~\ref{fig:essential} and then decontract the map to obtain an
intersecting Schnyder wood of the original map. Unfortunately, we are
not able to prove that the intersecting property can be preserved
during the decontraction process.

On the example of Figure~\ref{fig:3-connected}, it is not possible to
decontract the graph $G'$ (Figure~\ref{fig:3-connected}.(a)), and
extend its intersecting Schnyder wood to $G$
(Figure~\ref{fig:3-connected}.(b)) without modifying the edges that
are not incident to the contracted edge $e$.  Indeed, if we keep the
edges non incident to $e$ unchanged, there are only two possible ways
to extend the coloring in order to preserve the Schnyder property, but
none of them leads to an intersecting Schnyder wood.

\begin{figure}[h!]
\center
\begin{tabular}{ccc}
\includegraphics[scale=0.4]{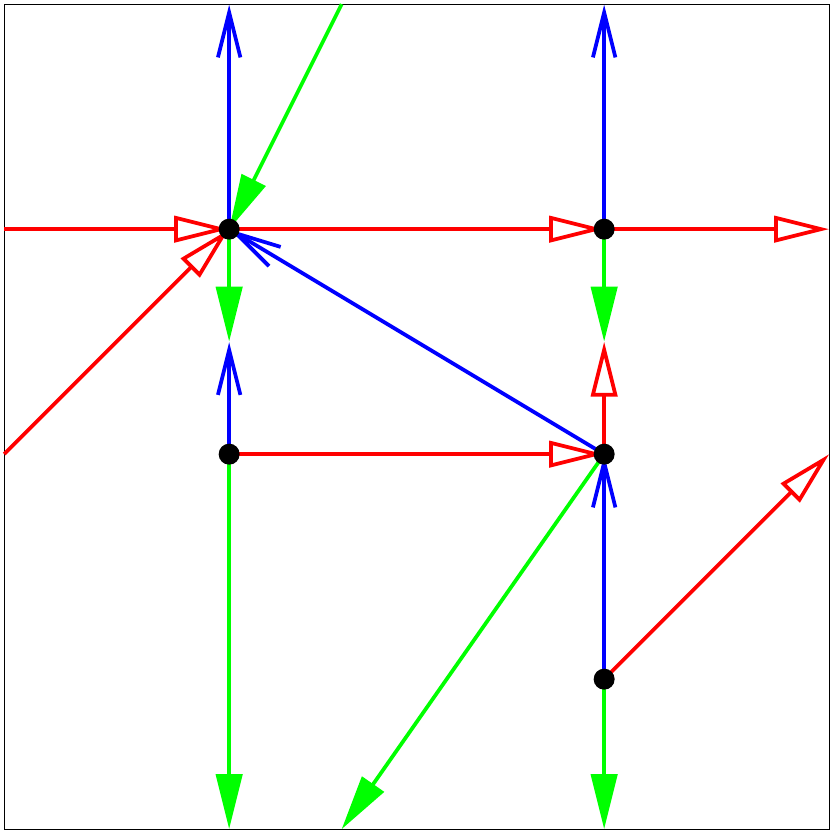}
& \hspace{1em} &
\includegraphics[scale=0.4]{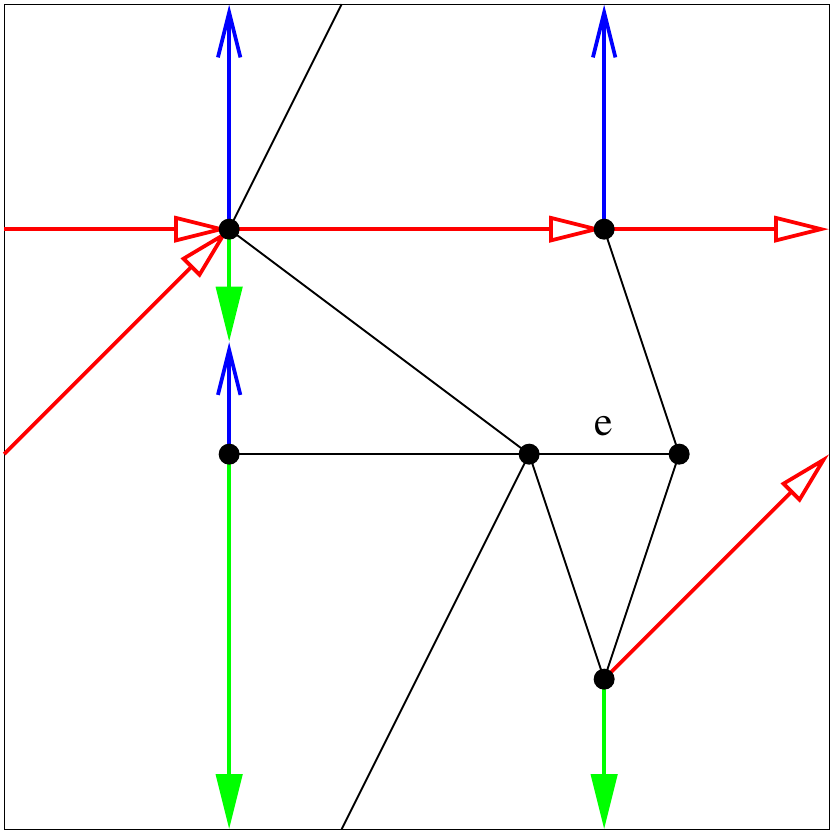}\\
(a)  & &(b)  \\
\end{tabular}
\caption{ (a) The map obtained by  contracting the edge
  $e$ of the map (b). It is not possible to color 
 and orient the black edges of
  (b) to obtain an intersecting Schnyder wood.
}
\label{fig:3-connected}
\end{figure}

Fortunately most essentially 3-connected toroidal maps admits crossing
Schnyder woods and not only intersecting Schnyder woods and we are
able to preserve the crossing property during the decontraction
process (see Section~\ref{sec:contractionlemma}).  

A toroidal map is \emph{basic} if it consists of a non-contractible
cycle on $n$ vertices, $n\geq 1$, plus $n$ homologous loops (see
Figure~\ref{fig:basic}). 

\begin{lemma}
\label{lem:basic}  
A basic toroidal map admits an intersecting
  Schnyder wood but no crossing Schnyder wood.
\end{lemma}

\begin{proof}
  Basic toroidal maps admits intersecting
  Schnyder woods as shown by
  Figure~\ref{fig:basic}. Suppose that a basic toroidal map $G$ on $n$
  vertices admits a crossing Schnyder wood. Consider one of the
  vertical loop $e$ and suppose by symmetry that it is oriented upward
  in color $1$. In a crossing Schnyder wood, all the monochromatic
  cycles of different colors are not weakly homologous, thus all the loops
  homologous to $e$ are also oriented upward in color $1$ and they are
  not bioriented. It remains just a cycle on $n$ vertices for edges
  of color $0$ and $2$. Thus the Schnyder wood is the one of
  Figure~\ref{fig:basic} and not crossing, a contradiction.
\end{proof}

To prove existence of crossing Schnyder woods for general essentially
3-connected toroidal maps, instead of contracting these maps to one of
the two maps of Figure~\ref{fig:essential}, we contract them to the
map of Figure~\ref{fig:essential}.(a), the 3-loops, or to the map of
Figure~\ref{fig:brique}, \emph{the brick}, that both admits crossing
Schnyder woods. (One can draw the universal cover of the brick to
understand its name.)

\begin{figure}[!h]
\center
\includegraphics[scale=0.5]{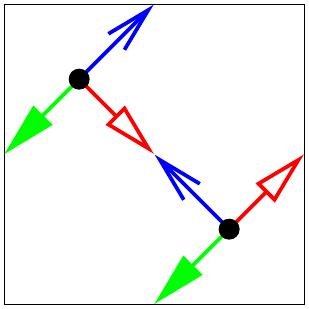}%\\
\caption{The brick, an essentially 3-connected toroidal map with two
  vertices, given with a crossing Schnyder wood.}
\label{fig:brique}
\end{figure}

Now we can state our main existence result:

\begin{theorem}
\label{th:existencebasic}
A toroidal map admits a crossing Schnyder wood if and only if it is
an essentially 3-connected non-basic toroidal map. Moreover, basic
toroidal maps admits intersecting Schnyder woods.
\end{theorem}

Note that Theorem~\ref{th:existencebasic} confirms
Conjecture~\ref{conjecture2} for $g=1$.  The proof of this theorem is
postponed to Section~\ref{sec:existencemainproof}. The main
difficulty is to show that one can decontract a crossing Schnyder wood
by preserving the crossing property. This is done in next section by
first relaxing the definition of crossing 
 and then proving a decontraction lemma. Before we prove the following
contraction lemma:

\begin{lemma}
\label{lem:contractionbrique}
A non-basic essentially 3-connected toroidal map can be contracted to
the 3-loops (Figure~\ref{fig:essential}.(a)) or to the brick
(Figure~\ref{fig:brique}) by keeping the toroidal map essentially 3-connected.
\end{lemma}

\begin{proof}
  Let us prove the lemma by induction on the number of edges of the
  map.  As the 3-loop and the brick are the only non-basic essentially
  3-connected toroidal maps with at most 3 edges, the lemma holds for
  the maps with at most 3 edges.  Consider now a non-basic essentially
  3-connected toroidal map $G$ with at least $4$ edges.  As $G$ has at
  least 2 vertices, it has at least one contractible edge by
  Lemma~\ref{lem:contractibleedges}.  If $G$ has a contractible edge
  $e$ which contraction yields a non-basic map $G'$, then by induction
  hypothesis on $G'$ we are done. Let us prove that such an edge
  always exists.  We assume by contradiction, that the contraction of
  any contractible edge $e$ yields a basic map $G'$. Let us denote
  $v_i$, with $1\le i \le n$, the vertices of $G'$ in such a way that
  $(v_1,v_2,\ldots,v_n)$ is a cycle of $G'$. We can assume that
  $v_1$ is the vertex resulting of  the contraction of $e$. Let
   $u$ and $v$ be the endpoints of $e$ in $G$.

   Suppose first that $u$ or $v$ is incident to a loop in $G$. By
   symmetry, we can assume that $v$ is incident to a loop and that $u$
   is in the cylinder between the loops around $v$ and $v_n$ (if $n=1$
   then $v_n=v$) and $u$ is the only vertex here. Since $G$ is
   non-basic and $u$ has at least 3 incident edges, two of them go to
   the same vertex but are not homotopic multiple edges. Since after
   the contraction of $e$ there is only one edge left in the cylinder,
   we can deduce that $u$ has at least two edges in common with
   $v$. On the other side since $G$ is essentially 3-connected $u$ has
   an edge $e'$ with $v_n$. This edge $e'$ is contractible since its
   contraction yields a map containing the basic graph on $n$
   vertices. But since this map has two edges linking $(uv_n)$ and
   $v$, it is non-basic. So $G$ has a contractible edge which
   contraction produces a non-basic map, contradicting our assumption.

   Suppose now that $u$ and $v$ do not have incident loop, we thus
   have that $G$ contains a (non-contractible) cycle $C$ of length $2$
   containing $e$. Let $e'$ be the other edge of $C$. Since $G$ is
   essentially 3-connected, both $u$ and $v$ have at least degree 3,
   and at least one of them has an incident edge on the left
   (resp. right) of $C$. If $n=1$, since $G$ has at least 4 edges
   there are 2 edges, say $f_1$ and $f_2$ between $u$ and $v$ and
   distinct from $e$ and $e'$ that are not homotopic multiple
   edges. The edges $f_1$ and $f_2$ remain distinct edges in $G'$. So
   $G'$ has one vertex and 3 edges, it is thus non-basic. Assume now
   that $n\ge 2$. In this case $u$ and $v$ are contained in a cylinder
   bordered by the loops at $v_2$ and at $v_n$ (with eventually
   $n=2$). We can assume that $u$ has at least one incident edge $f_1$
   on the left of $C$ to $v_n$, and that $v$ has at least one incident
   edge $f_2$ on the right of $C$ to $v_2$. In this case one can
   contract $f_1$ and note that the obtained map which contain at
   least 3 edges around $v$ ($e$, $e'$ and $f_2$) is essentially
   3-connected, and non-basic. So $G$ has a contractible edge which
   contraction produces a non-basic map, contradicting our assumption.
\end{proof}

\section{Special decontraction preserving crossing}
\label{sec:contractionlemma}

The goal of this section is to prove the following lemma:

\begin{lemma}
  \label{lem:contractype1}
  If $G$ is a toroidal map given with a non-loop edge $e$ whose
  extremities are of degree at least three and such that $G/ e$ admits
  a crossing Schnyder wood, then $G$ admits a crossing Schnyder wood.
\end{lemma}

The proof of Lemma~\ref{lem:contractype1} is long and technical. There
is a huge case analysis for the following reasons. One has to consider
the different kind of contraction depicted on
Figures~\ref{fig:contraction} and~\ref{fig:contraction-loop}. For each
of these cases, one as to consider the different ways that the edges
of $G/ e$ that are labeled on the figures can be oriented and colored
in a crossing Schnyder wood of $G/e$ see Figures~\ref{fig:case-b}
and~\ref{fig:case-e}.  For each of these cases, one as to show that
one can orient and color the edges of $G$ that are labeled on the
figures to extend the crossing Schnyder wood of $G$ to a crossing
Schnyder wood of $G/e$ (non-labeled edges keep the orientation and
coloring they have in $G/e$). This would be quite easy if one just
have to satisfy the Schnyder property like in
Section~\ref{sec:contractionproof}, but satisfying the crossing
property is much more complicated.

\begin{figure}[!h]
\center
  \begin{tabular}[]{cccccccc}
    a.1& \includegraphics[scale=0.4]{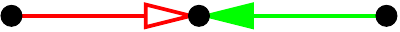} & \hspace{1em}
    & b.1 & \includegraphics[scale=0.2]{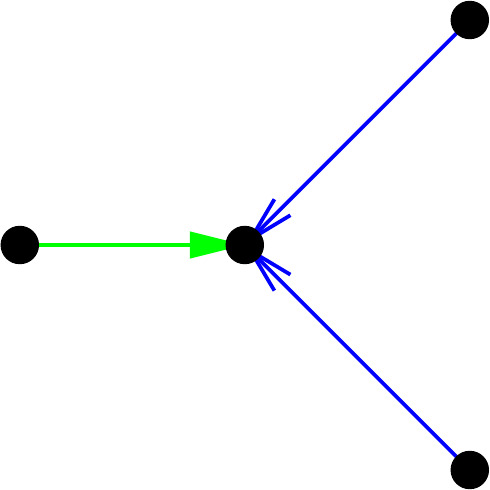} & \hspace{1em}
    & c.1 & \includegraphics[scale=0.2]{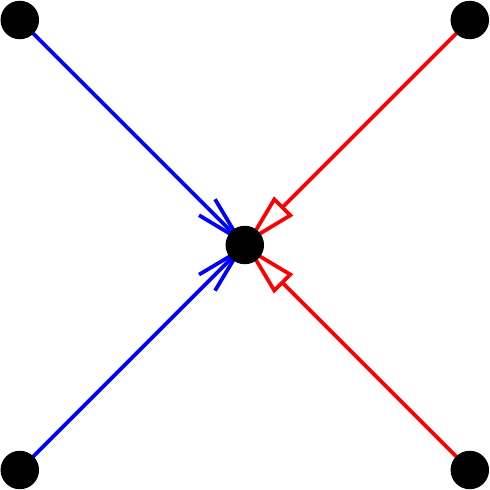} \\
    a.2 & \includegraphics[scale=0.4]{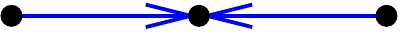} & & b.2
          &\includegraphics[scale=0.2]{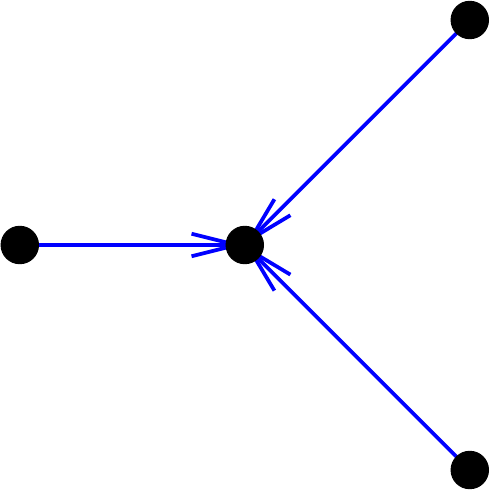} & & c.2
          & \includegraphics[scale=0.2]{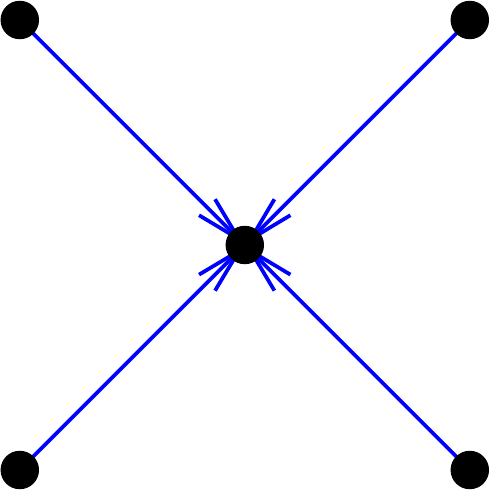} \\
    a.3 &\includegraphics[scale=0.4]{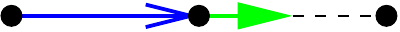} & & b.3
          & \includegraphics[scale=0.2]{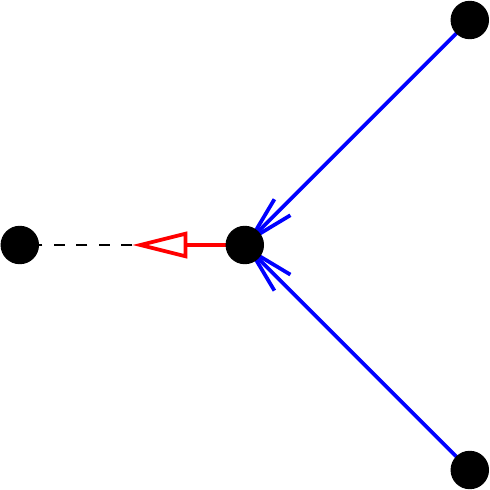} & & c.3
          & \includegraphics[scale=0.2]{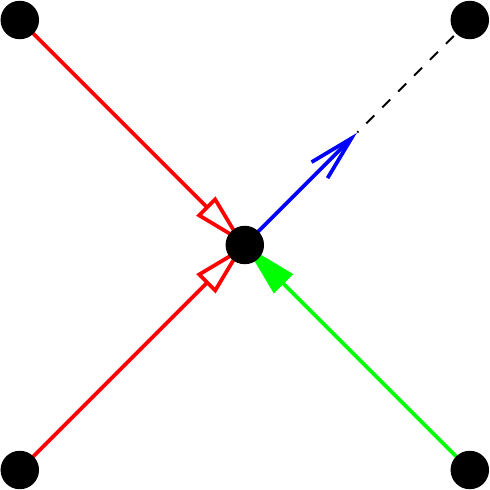} \\
    a.4 & \includegraphics[scale=0.4]{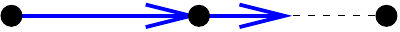} & & b.4
          & \includegraphics[scale=0.2]{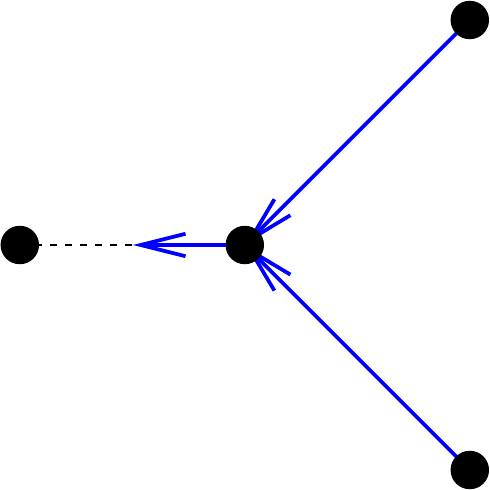} & & c.4
          & \includegraphics[scale=0.2]{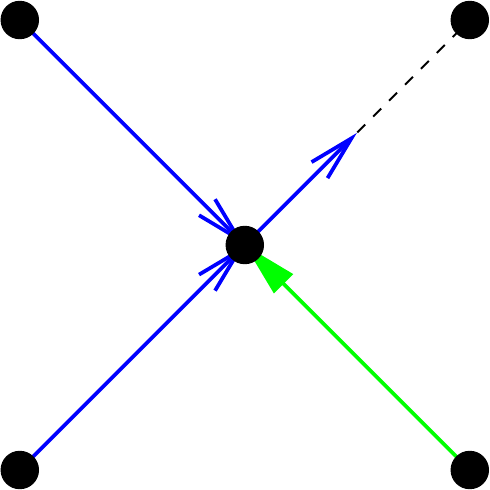}\\
    a.5 & \includegraphics[scale=0.4]{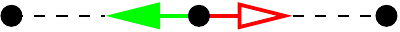} &  & b.5
          & \includegraphics[scale=0.2]{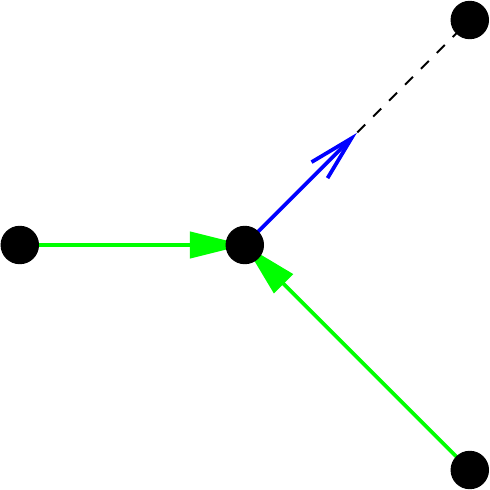} & & c.5
          & \includegraphics[scale=0.2]{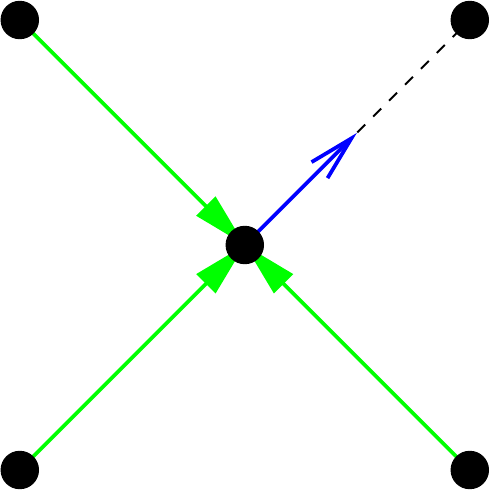}\\
       & & &  b.6  & \includegraphics[scale=0.2]{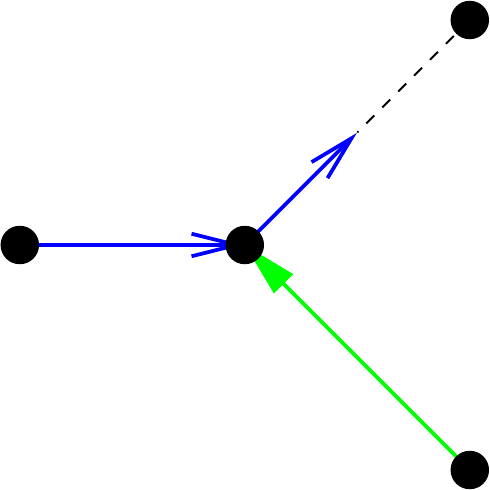} & &
                                                                    c.6 & \includegraphics[scale=0.2]{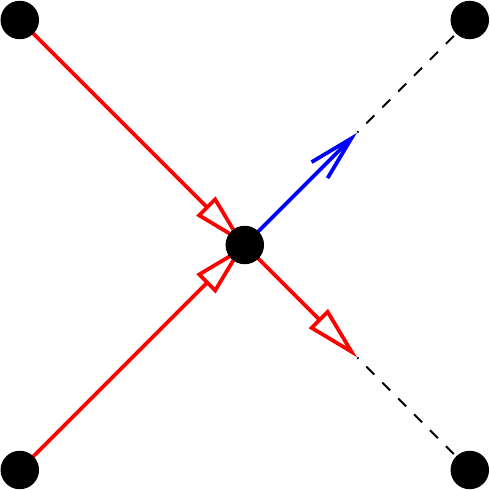}\\
       & & & b.7  &\includegraphics[scale=0.2]{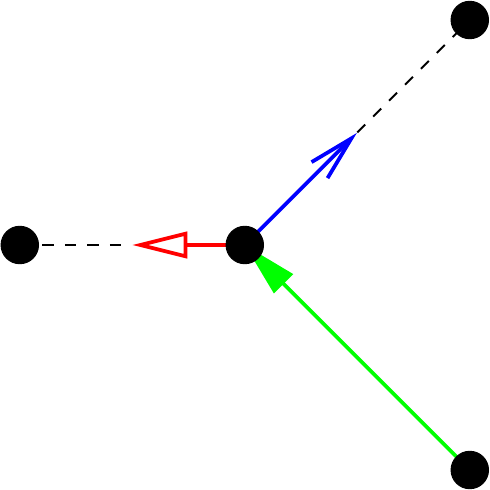} & & c.7
          & \includegraphics[scale=0.2]{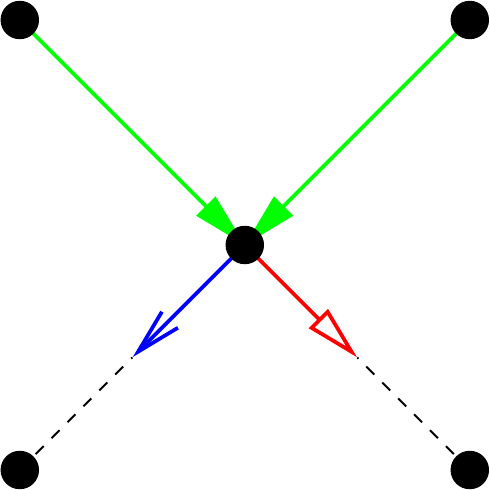}\\
       & & & b.8  &\includegraphics[scale=0.2]{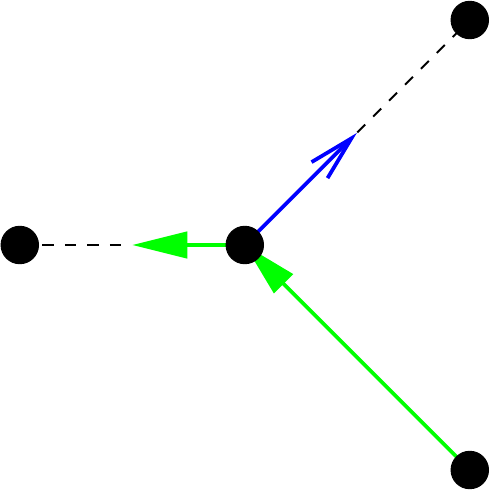} & & c.8
          & \includegraphics[scale=0.2]{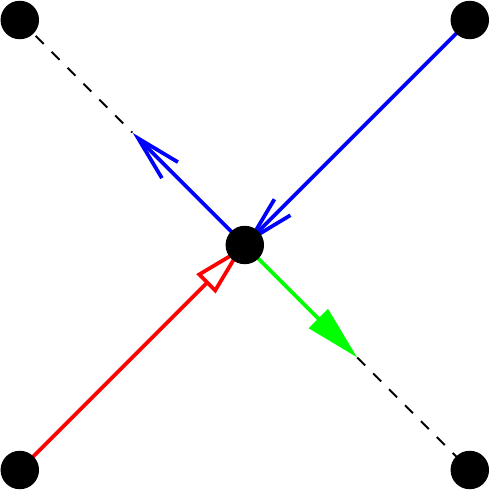}\\
       & & & b.9 &\includegraphics[scale=0.2]{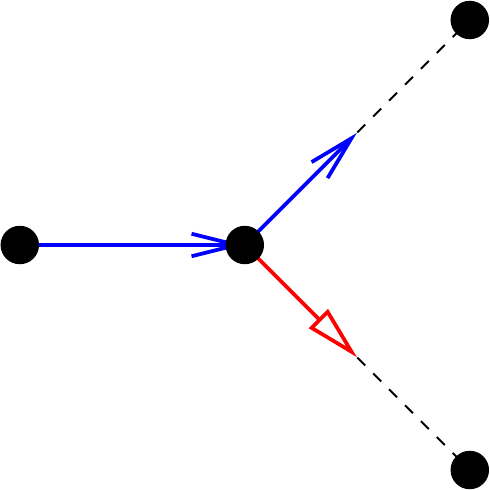} & & c.9
          & \includegraphics[scale=0.2]{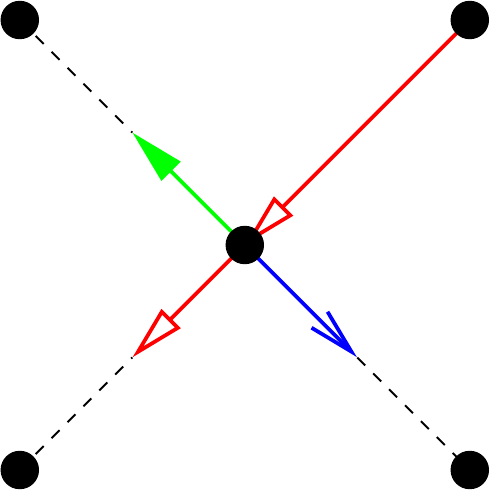} \\
       & & & b.10  &\includegraphics[scale=0.2]{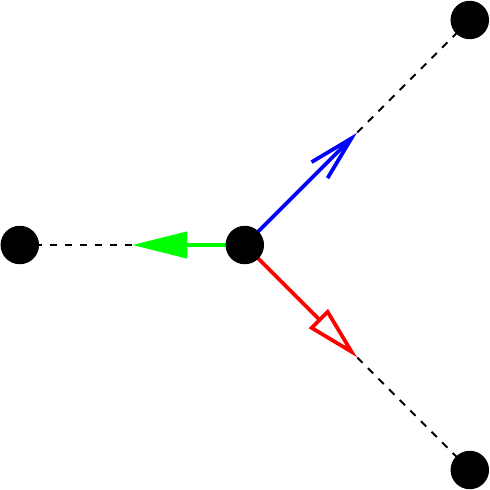} \\
  \end{tabular}
\caption{Orientation and coloring possibilities  for cases (a) (b) and (c).}
\label{fig:case-b}
\end{figure}

\begin{figure}[!h]
\center
\includegraphics[scale=0.4]{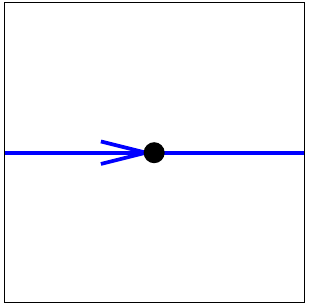} %& \hspace{1em}
  \caption{Orientation and coloring for cases (d) (e)
    and (f).}
\label{fig:case-e}
\end{figure}

We first need to relax the
definition of crossing Schnyder wood:

\begin{lemma}
\label{lem:tri-relax}
Let $G$ be a toroidal map given with an orientation and coloring of
the edges of $G$ with the colors $0$, $1$, $2$, where every edge $e$
is oriented in one direction or in two opposite directions (each
direction having a distinct color and being outgoing).  The
orientation and coloring is a crossing Schnyder wood if and only if it
satisfies the following:

\begin{itemize}
\item[(T1)] Every vertex $v$ satisfies the Schnyder property.

\item[(T2)] For each pair $i,j$ of different colors, there exists a
  $i$-cycle intersecting a $j$-cycle.

\item[(T3)] There is no monochromatic cycles $C_i,C_j$ of different
  colors $i,j$ such that $C_i=C_j^{-1}$.
\end{itemize}

\end{lemma}

\begin{proof} 
  ($\Longrightarrow$) If we have a crossing Schnyder wood, then
  Property (T1) is clearly satisfied by definition.  Property (T1)
  implies that there exists a monochromatic cycles of each color, thus
  Property (T2) is a consequence of the crossing property. Property
  (T3) is implied by Theorem~\ref{lem:type-cross}.

  ($\Longleftarrow$) Conversely, suppose we have an orientation and
  coloring satisfying (T1), (T2), (T3).  By (T3), there is no face the
  boundary is a monochromatic cycle since, by (T1), such a cycle would
  be colored one way in color $i$ and the other way in color $j$, and
  it contradicts (T3). Thus the orientation and coloring is a Schnyder
  wood. 

  Let us now prove that the Schnyder wood is crossing.  Let $C_i$ be
  any $i$-cycle of color $i$. We have to prove that $C_i$ crosses at
  least one $(i-1)$-cycle and at least one $(i+1)$-cycle.  Let $j$ be
  either $i-1$ or $i+1$.  By (T2), there exists a $i$-cycle $C'_i$
  intersecting a $j$-cycle $C'_j$ of color $j$. The two cycles
  $C'_i, C'_j$ are not reversal by (T3), thus they are crossing. So
  the Schnyder wood is half-crossing and more precisely $i$-crossing
  or $j$-crossing and the two cycles $C_i$ and $C'_j$ are crossing.
\end{proof}

Note that for toroidal triangulations, there is no edges oriented in
two directions in an orientation and coloring of the edges satisfying
(T1). So (T3) is automatically satisfied. Thus in the case of
toroidal triangulations it is sufficient to have properties (T1) and
(T2) to have a crossing Schnyder wood.  This is not true in general as shown
by the example of Figure~\ref{fig:relax} that satisfies (T1) and
(T2) but that is not a crossing Schnyder wood. There is a monochromatic cycle
of color $1$ (blue) that is not intersecting any monochromatic cycle of color
$2$ (green).

\begin{figure}[!h]
\center
\includegraphics[scale=0.5]{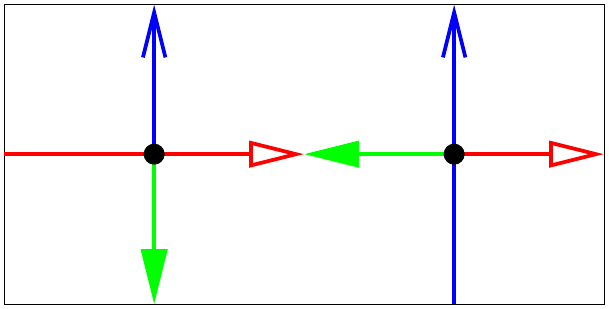}
\caption{An orientation and coloring of the edges of a toroidal map
  satisfying (T1) and (T2) but that is not a crossing Schnyder wood.}
\label{fig:relax}
\end{figure}

By Lemma~\ref{lem:tri-relax}, instead of proving that the crossing
property is preserved during the decontraction process, by considering
intersections between every pair of monochromatic cycles, we prove
(T2) and (T3). Property (T2) is simpler than the definition of
crossing as it is considering just one intersection for each pair of
colors instead of all the intersections. Even with this simplification
proving (T2) is the main difficulty of the proof.  For each considered
cases, one has to analyze the different ways that the monochromatic
cycles go through the contracted vertex or not and to show that there
always exists a coloring and orientation of $G$ where (T2) is
satisfied.  In fact, among all the cases of Figures~\ref{fig:case-b}
and~\ref{fig:case-e}, we detail only one case completely to show how
the proof goes. The numerous other cases can be proved similarly.

\

\emph{Proof of Lemma~\ref{lem:contractype1}.}\
\setcounter{ssclaim}{0}

Let $u,v$ be the two extremities of $e$. Vertices $u$ and $v$ are of
degree at least three.
Depending on whether the faces incident to $e$ have size three or
not, and that some vertices are identified or not, we are, by
symmetry, in one of the  cases (a), (b), (c), (d), (e), (f) of
Figure~\ref{fig:contraction} and~\ref{fig:contraction-loop}.
 Let $G'=G/e$ and consider
  a crossing Schnyder wood  of $G'$. Let $w$ be the vertex of $G'$
  resulting from the contraction of $e$.

  Like in the proof of Theorem~\ref{th:proofbycontractiontri}, for
  each case (a), (b), (c), (d), (e), (f), there are different cases to
  consider corresponding to the different possibilities of
  orientations and colorings of the labeled edges of $G'$ in a
  crossing Schnyder wood of $G'$. These cases are represented (up to
  symmetry) on Figure~\ref{fig:case-b} for cases (a), (b), (c) and
  Figure~\ref{fig:case-e} for cases (d), (e), (f). A dotted half-edge
  represent the possibility for an edge to be uni- or bidirected.

  In each case, one can easily color and orient the edges of $G$ to
  satisfy (T1). Just the edges of $G$ that are labeled on
  Figures~\ref{fig:contraction} and~\ref{fig:contraction-loop} have to
  be specified, all the other edges of $G$ keep the orientation and
  coloring that they have in the crossing Schnyder wood of $G'$.
  There might be several possibilities to color and orient the edges
  in order to satisfy (T1).  These different possibilities can easily
  be found by considering the angle labeling of $G$ and $G'$ around
  vertices $u,v,w$ (as in the proof of
  Theorem~\ref{th:proofbycontractiontri}).

For example, for case a.1, the labeling of the angles in $G'$ are
depicted on the left of Figure~\ref{fig:angle-vierge-2}.  There is
exactly three possible orientations and colorings of the edges of $G$
satisfying this labeling and represented by cases a.1.1, a.1.2, a.1.3 on
Figure~\ref{fig:angle-vierge}.

\begin{figure}[!h]
\center 
  \begin{tabular}[]{ccccccccc}
\includegraphics[scale=0.4]{1-1-1-angle} & \hspace{0em} &
\includegraphics[scale=0.4]{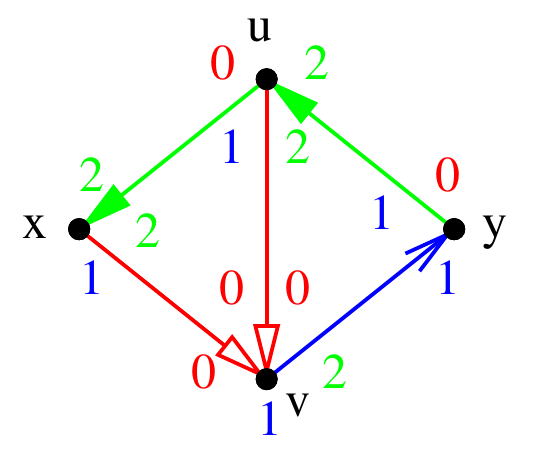} & \hspace{0em} &
\includegraphics[scale=0.4]{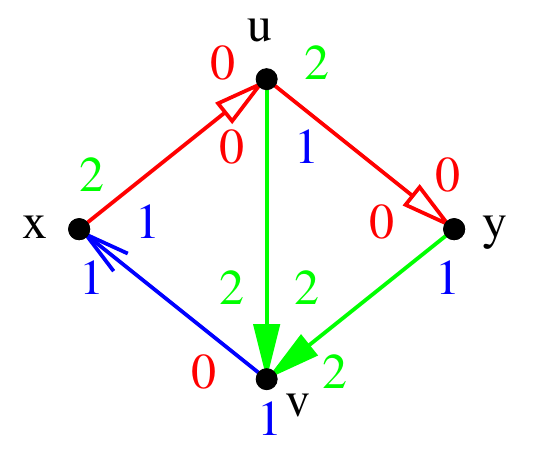} \\
                                a.1.1 & &
                                a.1.2 & & 
                                a.1.3 \\ \\ 
  \end{tabular}

  \caption{The three possible orientation and coloring of $G$ for case
    a.1.}
\label{fig:angle-vierge}
\end{figure}

One can check that for all cases of Figure~\ref{fig:angle-vierge},
property (T3) is preserved.  Indeed, if (T3) is not satisfied in $G$
after applying one of the possible orientation and coloring satisfying
(T1), then there exists two monochromatic cycles $C,C'$ of different
colors that are reversal.  By Lemma~\ref{lem:tri-relax}, there is no
reversal cycles in the crossing Schnyder wood of $G'$. Thus $C,C'$
have to use a bidirected edge that is newly created. There is no such
edges in the cases of Figure~\ref{fig:angle-vierge}, so (T3) is always
preserved in this case.

For the other cases (a), (b), (c), (d), (e), (f), it is just a little
case analysis to identify among the orientation and coloring
preserving (T1), which one also preserves (T3).  We illustrate this
for the case c.4. Figure~\ref{fig:exampleT3} gives two possible
orientation and coloring of $G$ preserving (T1) for case c.4.  Only
for the second case we are sure to preserve (T3). Indeed, in the first
case the decontraction process might create a green cycle and a blue
cycle that are reversal to each other. This is not possible in the
second case as two reversal cycles will have to use the red-blue or
green-blue edge. If the green-blue edge is used, then the two reversal
cycles are green and blue, but the blue cycle continues with the
red-blue edge, a contradiction. If the red-blue edge is used, then the
two reversal cycles are red and blue and already present in $G'$, a
contradiction.  We do not detail more this part that can easily be
done for all cases.

\begin{figure}[!h]
\center
\includegraphics[scale=0.4]{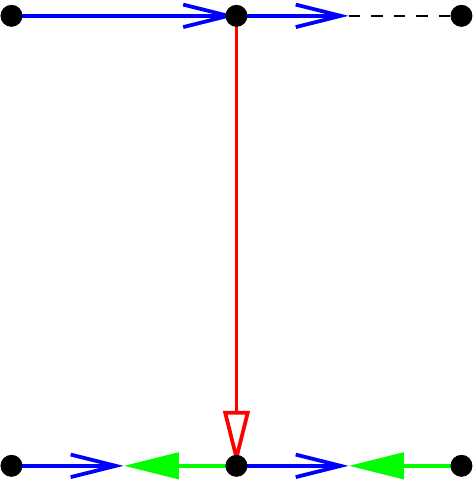} \hspace{2em}
\includegraphics[scale=0.4]{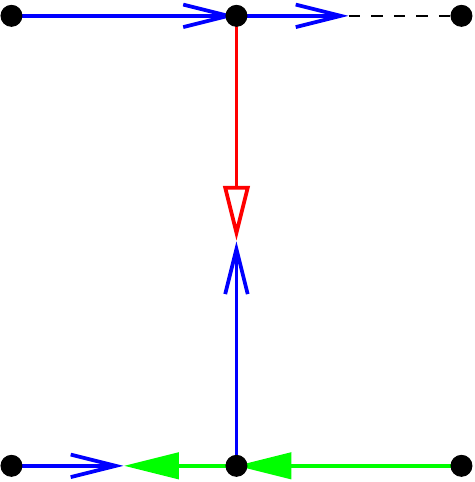} 
\caption{Discussion on (T3) preservation during decontraction of case c.4.}
\label{fig:exampleT3}
\end{figure}

Among all the possible orientation and coloring of $G$ satisfying (T1)
and (T3), we claim that there is always at least one that satisfies
(T2) and thus gives a crossing Schnyder wood of $G$.  Identifying
which one satisfies (T2) depends on a huge case analysis concerning
the behavior of the monochromatic cycles in $G'$.  and is the main
difficulty of the proof. We illustrate this part by doing only case
a.1 completely. This case is the most difficult cases for proving that
(T2) can be preserved.  We omit the proof of the numerous other cases
that can be obtained with similar arguments.

The sector $[e_1,e_2]$ of a vertex $w$, for $e_1$ and $e_2$ two edges
incident to $w$, is the counterclockwise sector of $w$ between $e_1$
and $e_2$, including the edges $e_1$ and $e_2$. The sector $]e_1,e_2]$,
$[e_1,e_2[$, and $]e_1,e_2[$ are defined analogously by excluding the
corresponding edges from the sectors.

We assume that we are in case a.1, where $e_{wx}=e_0(x)$ and
$e_{wy}=e_2(y)$.  We say that a monochromatic cycle $C$ of $G'$ is
\emph{safe} if $C$ does not contain $w$. Depending on whether there
are safe monochromatic cycles or not of each colors, it is a different
argument that is used to prove that property (T2) can be preserved by
one of the coloring of Figure~\ref{fig:angle-vierge}. All the cases
are considered below.

\begin{itemize}
\item \emph{There are safe monochromatic cycles of
  colors $\{0,1,2\}$.}
\end{itemize}

Let $C'_0,C'_1,C'_2$ be safe monochromatic cycles of color $0,1,2$ in
$G'$. As the Schnyder wood of $G'$ is crossing, they pairwise
intersects in $G'$.  Apply the coloring a.1.1 on $G$.  As
$C'_0,C'_1,C'_2$ do not contain vertex $w$, they are not modified in
$G$. Thus they still pairwise intersect in $G$. So (T2) is satisfied.

\begin{itemize}
\item \emph{There are safe monochromatic cycles of
  colors exactly $\{0,2\}$.}
\end{itemize}
Let $C'_0,C'_2$ be safe monochromatic cycles of color $0,2$ in
$G'$. Let $C'_1$ be a $1$-cycle in $G'$.  As the Schnyder wood of $G'$
is crossing, $C'_0,C'_1,C'_2$ pairwise intersects in $G'$. None of
those intersections contain $w$ as $C'_0$ and $C'_2$ do not contain
$w$.  By (T1), the cycle $C'_1$ enter $w$ in the sector
$]e_{wx},e_{wy}[$ and leaves in the sector $]e_{wy},e_{wx}[$.  Apply
the coloring a.1.1 on $G$. The cycle $C'_1$ is replaced by a new cycle
$C_1=C'_1\setminus\{w\}\cup\{u,v\}$.  The cycles $C'_0,C'_1,C'_2$ were
intersecting outside $w$ in $G'$ so $C'_0,C_1,C'_2$ are intersecting
in $G$. So (T2) is satisfied.

\begin{itemize}
\item \emph{There are safe monochromatic cycles of
  colors exactly $\{1,2\}$.}
\end{itemize}

Let $C'_1,C'_2$ be safe monochromatic cycles of color $1,2$ in
$G'$. Let $C'_0$ be a $0$-cycle in $G'$.  The cycles $C'_0,C'_1,C'_2$
pairwise intersects outside $w$.  The cycle $C'_0$ enters $w$ in the
sector $[e_1(w),e_{wx}[$, $[e_{wx},e_{wx}]$ or
$]e_{wx},e_2(w)]$. Apply the coloring a.1.2 
on $G$.  Depending on which of the three sectors $C'_0$ enters, it is
is replaced by one of the three following cycle
$C_0=C'_0\setminus\{w\}\cup\{u,v\}$,
$C_0=C'_0\setminus\{w\}\cup\{x,v\}$,
$C_0=C'_0\setminus\{w\}\cup\{v\}$. In any of the three possibilities,
$C_0,C'_1,C'_2$ are intersecting in $G$. So (T2) is satisfied.

\begin{itemize}
\item \emph{There are safe
  monochromatic cycles of colors exactly $\{0,1\}$.}
\end{itemize}

This case is completely symmetric to the previous case where there are safe monochromatic cycles of
  colors exactly $\{1,2\}$.

\begin{itemize}
\item  \emph{There are safe
  monochromatic cycles of color $2$ only.}
\end{itemize}

Let $C'_2$ be a safe $2$-cycle in $G'$.  Let $C'_0,C'_1$ be
monochromatic cycles of color $0,1$ in $G'$.

Suppose that there exists a path $Q'_0$ of color $0$, from $y$ to $w$
such that this path does not intersect $C'_2$. Suppose also that there
exists a path $Q'_1$ of color $1$, from $y$ to $w$ such that this path
does not intersect $C'_2$.  Let $C''_{0}=Q'_{0}\cup \{e_{wy}\}$ and
$C''_{1}=Q'_{1}\cup \{e_{wy}\}$. By Lemma~\ref{lem:nodirectedcycle},
cycles $C''_{0}, C''_1, C'_{2}$ are non-contractible. Both of
$C''_{0}, C''_{1}$ do not intersect $C'_2$ so by
Lemma~\ref{lem:intersect2}, they are both weakly homologous to
$C'_2$. Thus cycles $C''_{0}, C''_{1}$ are weakly homologous to each
other.  The path $Q'_0$ is leaving $C''_1$ at $y$ on the right side of
$C''_1$.  Since $C''_1$ and $C''_0$ are weakly homologous, the path
$Q'_0$ is entering $C''_1$ at least once from its right side. This is
in contradiction with the Schnyder property.  So we can assume that
one of $Q'_0$ or $Q'_1$ as above does not exist.

Suppose that in $G'$, there does not exist a path of color $0$, from
$y$ to $w$ such that this path does not intersect $C'_2$.  Apply the
coloring a.1.1 on $G$. Cycle $C'_1$ is replaced by
$C_1=C'_1\setminus\{w\}\cup\{u,v\}$, and intersect $C'_2$. Let $C_0$
be a $0$-cycle of $G$.  Cycle $C_0$ has to contain $u$ or $v$ or both,
otherwise it is a safe cycle of $G'$ of color $0$. In any case it
intersects $C_1$.  If $C_0$ contains $v$, then
$C'_0=C_0\setminus\{v\}\cup\{w\}$ and so $C_0$ is intersecting $C'_2$
and (T2) is satisfied.  Suppose now that $C_0$ does not contain $v$.
Then $C_0$ contains $u$ and $y$, the extremity of the edge leaving $u$
in color $0$. Let $Q_0$ be the part of $C_0$ consisting of the path
from $y$ to $u$. The path $Q'_0=Q_0\setminus\{u\}\cup\{w\}$ is from
$y$ to $w$. Thus by assumption $Q'_0$ intersects $C'_2$. So $C_0$
intersects $C'_2$ and (T2) is satisfied.

Suppose now that in $G'$, there does not exist a path of color $1$,
from $y$ to $w$ such
that this path does not intersect $C'_2$. Apply the coloring a.1.2 on
$G$.  Depending on which of the three sectors $C'_0$ enters,
$[e_1(w),e_{wx}[$, $[e_{wx},e_{wx}]$ or $]e_{wx},e_2(w)]$, it is is
replaced by one of the three following cycle
$C_0=C'_0\setminus\{w\}\cup\{u,v\}$,
$C_0=C'_0\setminus\{w\}\cup\{x,v\}$,
$C_0=C'_0\setminus\{w\}\cup\{v\}$. In any of the three possibilities,
$C_0$ contains $v$ and intersect $C'_2$.  Let $C_1$ be a $1$-cycle of
$G$.  Cycle $C_1$ has to contain $u$ or $v$ or both, otherwise it is a
safe cycle of $G'$ of color $1$.  Vertex $u$ has no edge entering it
in color $1$ so $C_1$ does not contain $u$ and thus it contains $v$
and intersects $C_0$.  Then $C_1$ contains $y$, the extremity of the
edge leaving $v$ in color $1$. Let $Q_1$ be the part of $C_1$
consisting of the path from $y$ to $v$. The path
$Q'_1=Q_1\setminus\{v\}\cup\{w\}$ is from $y$ to $w$. Thus by assumption $Q'_1$ intersects
$C'_2$. So $C_1$ intersects $C'_2$ and (T2) is satisfied.

\begin{itemize}
\item \emph{There are safe
  monochromatic cycles of color $0$ only.}
\end{itemize}

This case is completely symmetric to the case $\{2\}$.

\begin{itemize}
\item 
\emph{There are safe
  monochromatic cycles of color  $1$ only.}
\end{itemize}

Let $C'_1$ be a safe $1$-cycle in $G'$.  Let $C'_0$ and $C'_2$ be
monochromatic cycles of color $0$ and $2$ in $G'$.

Suppose $C'_0$ is entering $w$ in the sector $]e_{wx},e_2(w)]$.  Apply
the coloring a.1.3 on $G$.  The $0$-cycle $C'_0$ is replaced by
$C_0=C'_0\setminus\{w\}\cup\{v\}$ and thus contains $v$ and still
intersect $C'_1$.  Depending on which of the three sectors $C'_2$
enters, $[e_0(w),e_{wy}[$, $[e_{wy},e_{wy}]$ or $]e_{wy},e_1(w)]$, it
is replaced by one of the three following cycles
$C_2=C'_2\setminus\{w\}\cup\{v\}$,
$C_2=C'_2\setminus\{w\}\cup\{y,v\}$,
$C_2=C'_2\setminus\{w\}\cup\{u,v\}$.  In any case, $C_2$ contains $v$
and still intersect $C'_1$. Cycle $C_0$ and $C_2$ intersect on $v$.
So (T2) is satisfied.

The case where $C'_2$ is entering $w$ in the sector $[e_0(w),e_{wy}[$
is completely symmetric and we apply the coloring a.1.2 on $G$.

It remains to deal with the case where $C'_0$ is entering $w$ in the
sector $[e_1(w),e_{wx}]$ and $C'_2$ is entering $w$ in the sector
$[e_{wy},e_1(w)]$.  Suppose that there exists a path $Q'_0$ of color
$0$, from $y$ to $w$, entering $w$ in the sector $[e_1(w),e_{wx}]$,
such that this path does not intersect $C'_1$. Suppose also that there
exists a path $Q'_2$ of color $2$, from $x$ to $w$, entering $w$ in
the sector $[e_{wy},e_1(w)]$, such that this path does not intersect
$C'_1$.  Let $C''_{0}=Q'_{0}\cup \{e_{wy}\}$ and
$C''_{2}=Q'_{2}\cup \{e_{wx}\}$. By Lemma~\ref{lem:nodirectedcycle},
cycles $C''_{0}, C'_1, C''_{2}$ are non-contractible. Cycles
$C''_{0}, C''_{2}$ do not intersect $C'_1$ so by
Lemma~\ref{lem:intersect2}, they are weakly homologous to $C'_1$. Thus
cycles $C''_{0}, C''_{2}$ are weakly homologous to each other.  By
(T3), we have $C''_0\neq (C''_2)^{-1}$. So cycle $C''_2$ is leaving
$C''_0$. By (T1) and the assumption on the sectors, $C''_2$ is leaving
$C''_0$ on its right side. Since $C''_0$ and $C''_2$ are weakly
homologous, the path $Q'_2$ is entering $C''_0$ at least once from its
right side. This is in contradiction with the Schnyder property.  So
we can assume that one of $Q'_0$ or $Q'_2$ as above does not exist.

By symmetry, suppose that in $G'$, there does not exist a path of
color $0$, from $y$ to $w$, entering $w$ in the sector
$[e_1(w),e_{wx}]$, such that this path does not intersect $C'_1$.
Apply the coloring a.1.3 on $G$.  Depending on which of the two
sectors $C'_2$ enters, $[e_{wy},e_{wy}]$ or $]e_{wy},e_1(w)]$, it is
is replaced by one of the two following cycle
$C_2=C'_2\setminus\{w\}\cup\{y,v\}$,
$C_2=C'_2\setminus\{w\}\cup\{u,v\}$.  In any case, $C_2$ still
intersects $C'_1$. Let $C_0$ be a $0$-cycle of $G$.  Cycle $C_0$ has
to contain $u$ or $v$ or both, otherwise it is a safe cycle of $G'$ of
color $0$.  Suppose $C_0$ does not contain $u$, then
$C'_0=C_0\setminus\{v\}\cup\{w\}$ and $C'_0$ is not entering $w$ in
the sector $[e_1(w),e_{wx}]$, a contradiction.  So $C_0$ contains
$u$. Thus $C_0$ contains $y$, the extremity of the edge leaving $u$ in
color $0$, and it intersects $C_2$.  Let $Q_0$ be the part of $C_0$
consisting of the path from $y$ to $u$. The path
$Q'_0=Q_0\setminus\{u\}\cup\{w\}$ is from $y$ to $w$ and entering $w$
in the sector $[e_1(w),e_{wx}]$. Thus by assumption $Q'_0$ intersects
$C'_1$. So $C_0$ intersects $C'_1$ and (T2) is satisfied.

\begin{itemize}
\item \emph{There are no safe
  monochromatic cycle.}
\end{itemize}

Let $C'_0,C'_1,C'_2$ be monochromatic cycles of color $0,1,2$ in
$G'$. They all pairwise intersect on $w$.  

Suppose first that $C'_0$ is entering $w$ in the sector
$]e_{wx},e_2(w)]$.  Apply the coloring a.1.3 on $G$.  The $0$-cycle
$C'_0$ is replaced by $C_0=C'_0\setminus\{w\}\cup\{v\}$ and thus
contains $v$.  Depending on which of the three sectors $C'_2$ enters,
$[e_0(w),e_{wy}[$, $[e_{wy},e_{wy}]$ or $]e_{wy},e_1(w)]$, it is is
replaced by one of the three following cycle
$C_2=C'_2\setminus\{w\}\cup\{v\}$,
$C_2=C'_2\setminus\{w\}\cup\{y,v\}$,
$C_2=C'_2\setminus\{w\}\cup\{u,v\}$. In any case, $C_2$ contains $v$.
Let $C_1$ be a $1$-cycle in $G$. Cycle $C_1$ has to contain $u$ or $v$
or both, otherwise it is a safe cycle of $G'$ of color $1$. Vertex $u$
has no edge entering it in color $1$ so $C_1$ does not contain $u$ and
thus it contains $v$. So $C_0,C_1,C_2$ all intersect on $v$ and (T2)
is satisfied.

The case where $C'_2$ is entering $w$ in the sector $[e_0(w),e_{wy}[$
is completely symmetric and we apply the coloring a.1.2 on $G$.

It remains to deal with the case where $C'_0$ is entering $w$ in the
sector $[e_1(w),e_{wx}]$ and $C'_2$ is entering $w$ in the sector
$[e_{wy},e_1(w)]$. Apply the coloring a.1.1 on $G$.  Cycle $C'_1$ is
replaced by $C_1=C'_1\setminus\{w\}\cup\{u,v\}$.  Let $C_0$ be a
$0$-cycle in $G$.  Cycle $C_0$ has to contain $u$ or $v$ or both,
otherwise it is a safe cycle of $G'$ of color $0$. Suppose
$C_0\cap\{v,x\}=\{v\}$, then $C_0\setminus\{v\}\cup\{w\}$ is a
$0$-cycle of $G'$ entering $w$ in the sector $]e_{wx},e_2(w)]$,
contradicting the assumption on $C'_0$. Suppose $C_0$ contains $u$,
then $C_0$ contains $y$, the extremity of the edge leaving $u$ in
color $0$.  So $C_0$ contains $\{v,x\}$ or $\{u,y\}$.  Similarly $C_2$
contains $\{v,y\}$ or $\{u,x\}$.  In any case $C_0,C_1,C_2$ pairwise
intersect. So (T2) is satisfied.

\hfill $\Box$\vspace{1em}

\section{Proof of existence}
\label{sec:existencemainproof}
We are now able to prove Theorem~\ref{th:existencebasic}:

\

\noindent \emph{Proof of Theorem~\ref{th:existencebasic}.}\ 

The second
part of the theorem is clear by Lemma~\ref{lem:basic}.

($\Longrightarrow$) If $G$ is a toroidal map given with a crossing
Schnyder wood. Then, by Lemma~\ref{lem:conjessentially}, $G$ is
essentially 3-connected and by Lemma~\ref{lem:basic}, $G$ is not
basic.

($\Longleftarrow$) Let $G$ be a non-basic essentially 3-connected
toroidal map.  By Lemma~\ref{lem:contractionbrique}, $G$ can be
contracted to the 3-loops or to the brick by keeping the map
essentially 3-connected (so all the vertices have degree at least
$3$).  Both of these maps admit crossing Schnyder woods (see
Figure~\ref{fig:essential}.(a) and~\ref{fig:brique}). So by
Lemma~\ref{lem:contractype1} applied successively, $G$ admits a
crossing Schnyder wood.  \hfill $\square$\vspace{1em}

Here is a remark about how to compute a crossing Schnyder wood for an
essentially 3-connected toroidal triangulation. Instead of completing
the technical proof of Lemma~\ref{lem:contractype1} to know which
coloring of the decontracted map has to be chosen among the possible
choice to preserve the crossing property, one can proceed as
follows. At each decontraction step, one can try all the possible
decontraction that preserves the Schnyder property and check afterward
which one gives a crossing Schnyder wood.  This method gives a
polynomial algorithm that is not linear since at each decontraction
step there is already a test of the crossing property that takes
linear time.

\part{Applications in the toroidal case}
\label{part:5}

\chapter{Optimal encoding and bijection}
\label{chap:encoding}

\section{\npss method}
\label{sec:introduction}

\ps introduced in~\cite{PS06} an elegant method (called here PS method
for short) to linearly encode a planar triangulation $G$ with a binary
word of length
$\log_2\binom{4n}{n}\sim n\, \log_2(\frac{256}{27})\approx 3.2451\,n$
bits. This is asymptotically optimal since it matches the information
theory lower bound. 
The method is the following. Consider the set of
orientations of $G$ where inner vertices have outdegree $3$ and outer
vertices have outdegree $1$. These set carries a structure of
distributive lattice and one can consider its minimal element
w.r.t.~the outer face.  Then, starting from the outer face, a special
depth-first search algorithm is applied by ``following'' ingoing edges
and ``cutting'' outgoing ones.
The algorithm outputs a rooted
spanning tree with exactly two leaves (also called stems) on each
inner vertex from which the original  triangulation can be recovered in
a straightforward way. This tree can be encoded very efficiently.  A
nice aspect of this work, besides its interesting encoding properties,
is that the method gives a bijection between planar triangulations and
a particular type of plane trees with some consequences for counting
and sampling planar triangulation.

Castelli Aleardi, Fusy and Lewiner~\cite{CFL10} adapt \pss method to
encode planar triangulations with boundaries. A consequence is that a
triangulation of any orientable surface can be encoded by cutting the
surface along non-contractible cycles and see the surface as a planar
map with boundaries.  The obtained algorithm is asymptotically optimal
(in terms of number of bits) but it is not linear, nor bijective. 

We propose an alternative generalization of \pss algorithm to higher
genus  whose execution on a minimal balanced toroidal Schnyder wood
of a toroidal triangulation leads to a rooted unicellular map (which
corresponds to the natural generalization of trees when going to
higher genus). This map can be encoded optimally using $3.2451\,n$
bits. The algorithm can be performed in linear time and leads to a new
bijection between toroidal triangulations and a particular type of
unicellular maps.

The two main ingredients that make \pss algorithm work in an
orientation of a planar map are minimality and accessibility of the
orientation. \emph{Minimality} means  that there is no \cw cycle
and can be translated with the terminology of
section~\ref{sec:lattice} by saying that it is minimal w.r.t.~the
outer face.  \emph{Accessibility} means that there exists a root
vertex such that all the vertices have an oriented path directed
toward the root vertex.  Felsner's result~\cite{Fel04} concerning
existence and uniqueness of a minimal $\alpha$-orientation as soon as
an $\alpha$-orientation exists, enables several analogues of \pss
method to different kinds of planar maps, see~\cite{AP13,Ber07,DPS13}. In
all these cases the accessibility of the considered
$\alpha$-orientations is a consequence of the natural choice of
$\alpha$, like in \npss original work~\cite{PS06} where any
$3$-orientation of the inner edges of a planar triangulation is
accessible for any choice of root vertex on the outer face. (Note that
the conventions may differs in the literature: the role of outgoing
and incoming edges are sometimes exchanged and/or the role of \cw and
\ccww.)

For higher genus, the minimality can be obtained by results of
Section~\ref{sec:lattice} showing that on any orientable surface the
set of orientations of a given map having the same homology carries a
structure of distributive lattice.  But a given map on an orientable
surface can have several $\alpha$-orientations (for the same given
$\alpha$) that are not homologous (see for instance the example of
Chapter~\ref{sec:example}). So the set of $\alpha$-orientations of a
given map is now partitioned into distributive lattices contrarily to
the planar case where there is only one lattice (and thus only one
minimal element). For toroidal triangulations we manage to face this
problem and maintain a bijection by using the balanced property of
Chapter~\ref{sec:balanced} that defines a canonical lattice.

The main issue while trying to extend \pss algorithm to higher genus
is the accessibility.  Accessibility toward the outer face is given
almost for free in the planar case because of Euler's formula that
sums to a strictly positive value. For an orientable surface of genus
$g\geq 1$ new difficulties occur. Already in genus $1$ (the torus),
even if the orientation is minimal and accessible \pss algorithm can
visit all the vertices but not all the angles of the map because of
the existence of non-contractible cycles.
Lemma~\ref{lem:incomingedges} can be used to show that this problem
never occurs  for minimal balanced Schnyder woods. In genus
$g\geq 2$ things get even more difficult with separating
non-contractible cycles that make having accessibility of the vertices
difficult to obtain.

Another problem is to recover the original map after the execution of
the algorithm. If what remains after the execution of PS method is a
spanning unicellular map then the map can be recovered with the same
simple rules as in the plane. Unfortunately for many minimal
orientations the algorithm leads to a spanning unicellular embedded
graph that is not a map (the only face is not a disk) and it is not
possible to directly recover the original map.  Here again, the choice
of the minimal balanced Schnyder wood ensures that this never happens.

\section{General properties of {\sc Algorithm PS}}
\label{sec:algo}
 
We introduce the following reformulation of \npss original
algorithm. This version is more general in order to be applicable to
 any orientation of any map on an orientable surface.

\

\noindent {\bf \sc Algorithm PS } 

{\sc Input}: An oriented map $G$ on an oriented surface $S$, a root vertex $v_0$ and a root edge
$e_0$ incident to $v_0$.

{\sc Output}: A graph $U$ with stems, embedded on the oriented surface $S$.
\begin{enumerate}
\item Let $v:=v_0$, $e:=e_0$, $U:=\emptyset$.
  \item 
Let $v'$ be the extremity of $e$ different from $v$.
\begin{itemize}
\item[\underline{Case 1}:] \emph{$e$ is non-marked and entering $v$.}
Add $e$ to $U$ and let $v:=v'$.
\item[\underline{Case 2}:] \emph{$e$ is non-marked and leaving $v$.}  Add a stem
  to $U$ incident to $v$ and corresponding to $e$.
\item[\underline{Case 3}:] \emph{$e$ is already marked and entering $v$.}
Do nothing.
\item[\underline{Case 4}:] \emph{$e$ is already marked and leaving $v$.}
Let $v:=v'$.
\end{itemize}
\item 
Mark $e$. 
\item Let  $e$ be the  next edge around $v$ in \ccw order after the current $e$.
\item While $(v,e)\neq (v_0,e_0)$ go back to 2.
\item Return $U$.
\end{enumerate}

 We
show general properties of {\sc Algorithm PS} in this section before
considering toroidal triangulations in the forthcoming sections.

We insist on the fact that the output of \aps is a graph embedded on
the same surface as the input map but that this embedded graph is not
necessarily a map (i.e some faces may not be homeomorphic to open
disks). In the following section we show that in our specific case the
output $U$ is an unicellular map.

Consider any oriented map $G$ on an orientable surface given with a root
vertex $v_0$ and a root edge $e_0$ incident to $v_0$. When \aps is
considering a couple $(v,e)$ we see this like it is considering the
angle at $v$ that is just before $e$ in counterclockwise order. The
particular choice of $v_0$ and $e_0$ is thus in fact a particular
choice of a root angle $a_0$ that automatically defines a root vertex
$v_0$, a root edge $e_0$, as well as a root face $f_0$. From now on we
 consider that the input of \aps is an oriented map plus a root
angle (without specifying the root vertex, face and edge).

The \emph{angle graph} of $G$, is the graph defined on the angles of
$G$ and where two angles are adjacent if and only if they are
consecutive around a vertex or around a face.  An execution of \aps
can be seen as a walk in the angle graph. Figure~\ref{fig:anglerule}
illustrates the behavior of the algorithm corresponding to Case~1
to~4. In each case, the algorithm is considering the angle in top left
position and depending on the marking of the edge and its orientation
the next angle that is considered is the one that is the end of the
magenta arc of the angle graph. The cyan edge of Case 1 represents the
edge that is added to $U$ by the algorithm. The stems of $U$ added in
Case~2 are not represented in cyan, in fact we will represent them
later by an edge in the dual. Indeed seeing the execution of \aps as a
walk in the angle graph enables us to show that \aps behaves exactly
the same in the primal or in the dual map (as explained later).

\begin{figure}[h!]
\center
\begin{tabular}[c]{ll}
\includegraphics[scale=0.5]{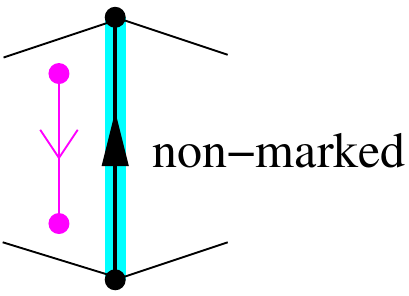} \ \ \ \ &
\includegraphics[scale=0.5]{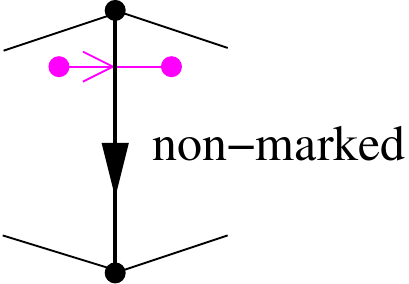} \\
\hspace{1em} Case 1 & \hspace{1em} Case 2 \\
\ &\  \\
\includegraphics[scale=0.5]{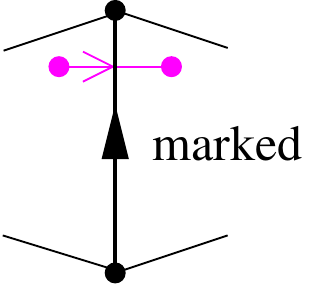} \ \ \ \ &
\includegraphics[scale=0.5]{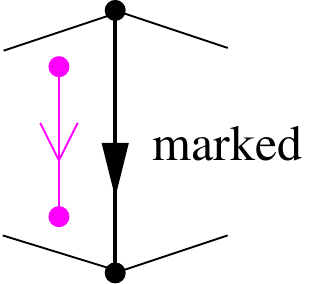} \\
\hspace{1em} Case 3 & \hspace{1em} Case 4
\end{tabular}
\caption{The four cases of \aps.}
\label{fig:anglerule}
\end{figure}

The crossing Schnyder wood of Figure~\ref{fig:tore-primal-min} is the
minimal balanced Schnyder wood of $K_7$  for the choice of $f_0$ corresponding to the
shaded face. This example is used in all this chapter to illustrate 
\npss method.

\begin{figure}[h!]
\center
\includegraphics[scale=0.5]{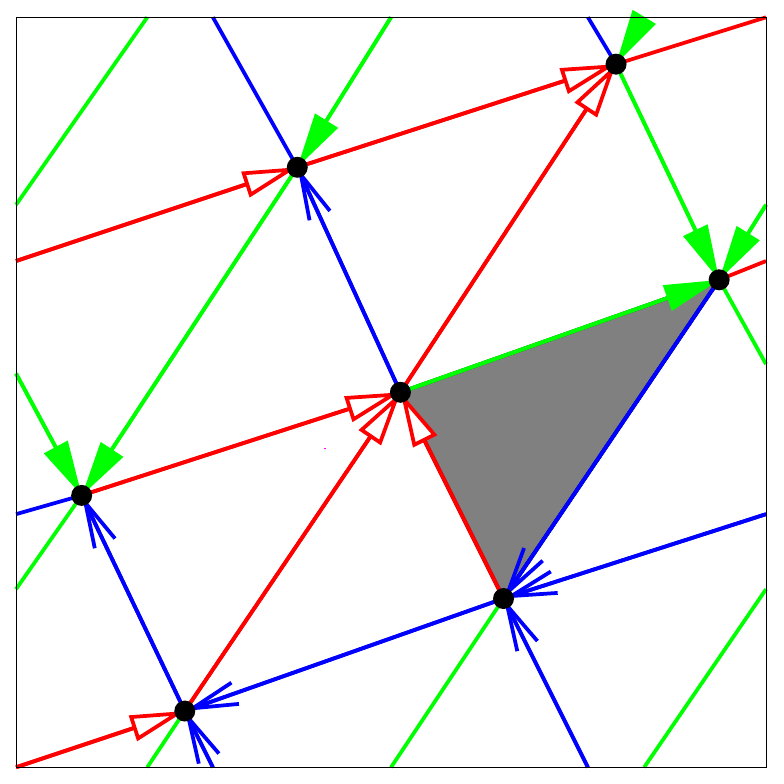}
\caption{The minimal balanced Schnyder wood of $K_7$ w.r.t.~the shaded
  face.}
\label{fig:tore-primal-min}
\end{figure}

On Figure~\ref{fig:tore-example}, we give an example of an execution
of \aps on the orientation of Figure~\ref{fig:tore-primal-min}.

\begin{figure}[h!]
\center
\begin{tabular}{cc}
\includegraphics[scale=0.5]{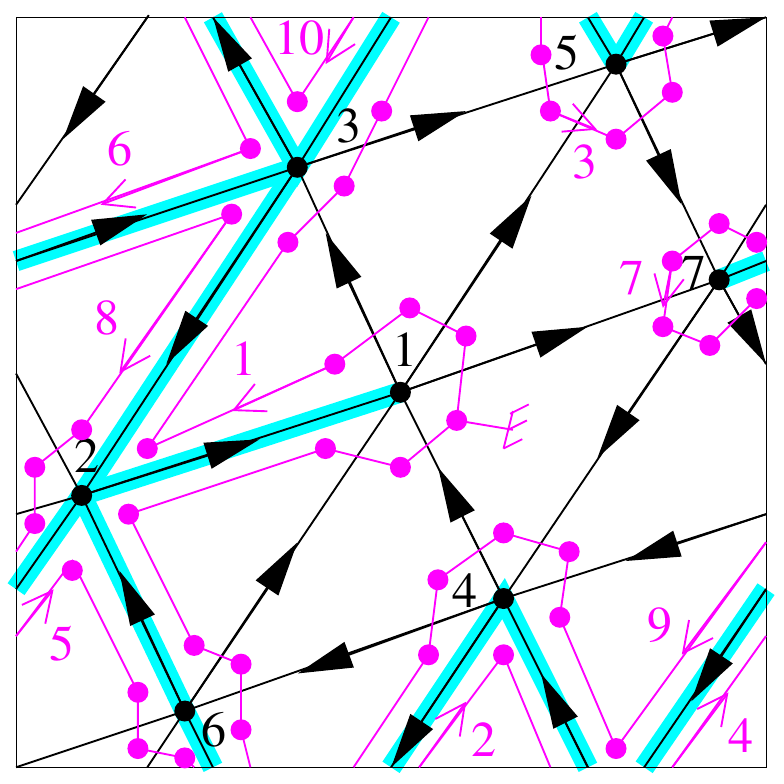} \ & \
\includegraphics[scale=0.5]{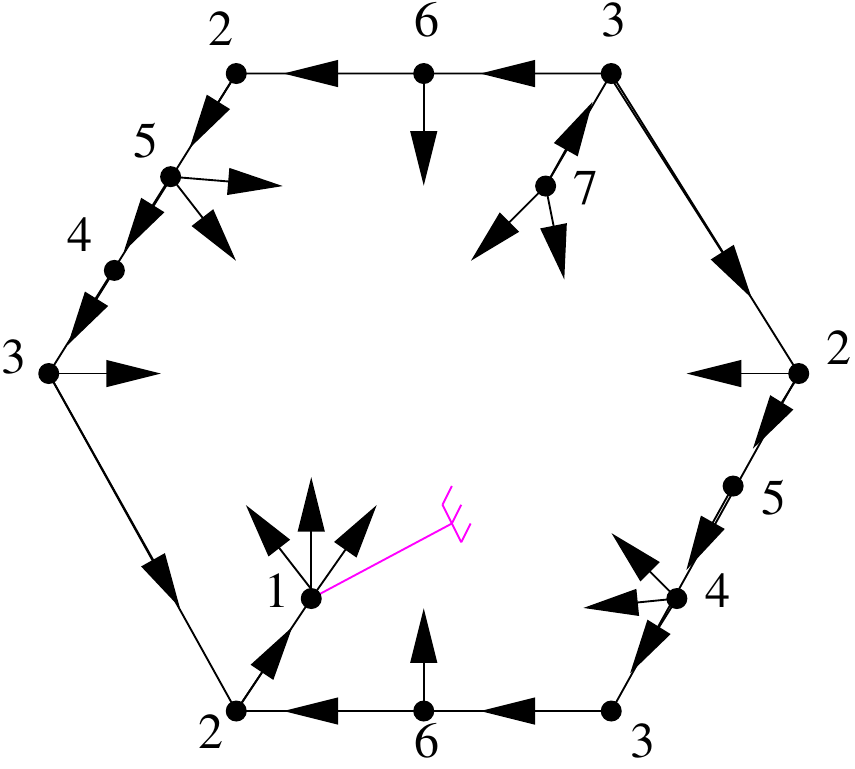}\\
Execution \ &\ Output \\
\end{tabular}
\caption{An execution of \aps on $K_7$ given with the orientation
  corresponding to the minimal balanced Schnyder wood of
  Figure~\ref{fig:tore-primal-min}. Vertices are numbered in
  black. The root angle is identified by a root symbol and chosen in
  the face for which the orientation is minimal (i.e. the shaded face
  of Figure~\ref{fig:tore-primal-min}). The magenta arrows and numbers
  are here to help the reader to follow the cycle in the angle graph.
  The output $U$ is a toroidal unicellular map, represented here as an
  hexagon where the opposite sides are identified.}
\label{fig:tore-example}
\end{figure}

Let $a$ be a particular angle of the map $G$. It is adjacent to four other angles
in the \emph{angle graph} (see Figure~\ref{fig:angledirected}). Let
$v,f$ be such that $a$ is an angle of vertex $v$ and face $f$. The
\emph{next-vertex} (resp. \emph{previous-vertex}) angle of $a$ is the
angle appearing just after (resp. before) $a$ in \ccw order around
$v$. Similarly, the \emph{next-face} (resp. \emph{previous-face})
angle of $a$ is the angle appearing just after (resp. before) $a$ in
\cw order around $f$. These definitions enable one to orient
consistently the edges of the angle graph like in
Figure~\ref{fig:angledirected} so that for every oriented edge
$(a,a')$, $a'$ is a next-vertex or next-face angle of $a$.

\begin{figure}[h!]
\center
\includegraphics[scale=0.5]{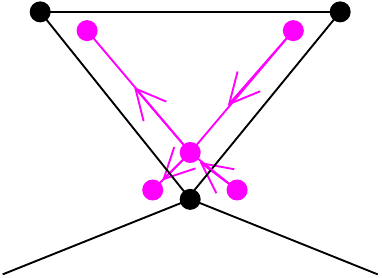} 
\caption{Orientation of the edges of the angle graph.}
\label{fig:angledirected}
\end{figure}

The different cases depicted in Figure~\ref{fig:anglerule} show that
an execution of \aps is just an oriented walk in the angle graph
(i.e. a walk that is following the orientation of the edges described in
Figure~\ref{fig:angledirected}).  The condition in the while loop
ensures that when the algorithm terminates, this walk is back to the
root angle. The following proposition shows that the algorithm actually
terminates:

\begin{proposition}
  \label{prop:terminates}
  Consider an oriented map $G$ on an orientable surface and a root angle
  $a_0$. The execution of \aps on $(G,a_0)$ terminates and corresponds
  to a cycle in the angle graph.
\end{proposition}
\begin{proof}
  We consider the oriented walk $W$ in the angle graph corresponding
  to the execution of \aps. Note that $W$ may be infinite. The walk
  $W$ starts with $a_0$, and if it is finite it ends with $a_0$ and
  contains no other occurrence of $a_0$ (otherwise the algorithm
  should have stopped earlier). Toward a contradiction, suppose that
  $W$ is not simple (i.e. some angles different from the root angle
  $a_0$ are repeated). Let $a\neq a_0$ be the first angle along $W$
  that is met for the second time.  Let $a_1,a_2$ be the angles
  appearing before the first and second occurrence of $a$ in $W$,
  respectively. Note that $a_1\neq a_2$ by the choice of $a$.

  If $a_1$ is the previous-vertex angle of $a$, then $a_2$ is the
  previous-face angle of $a$. When the algorithm considers $a_1$, none
  of $a$ and $a_2$ are already visited, thus the edge $e$ between $a$
  and $a_1$ is not marked. Since the execution then goes to $a$ after
  $a_1$, we are in Case 2 and edge $e$ is oriented from $v$, where $v$
  is the vertex incident to $a$. Afterward, when the algorithm reaches
  $a_2$, Case 3 applies and the algorithm cannot go to $a$, a
  contradiction. The case where $a_1$ is the previous-face angle of
  $a$ is similar.

  So $W$ is simple. Since the angle graph is finite, $W$ is finite. So
  the algorithm terminates, thus $W$ ends on the root angle and $W$ is
  a cycle.
\end{proof}

In the next section we see that in some particular cases the
cycle in the angle graph corresponding to the execution of \pss
algorithm (Proposition~\ref{prop:terminates}) can be shown to be
Hamiltonian like on Figure~\ref{fig:tore-example}.

By Proposition~\ref{prop:terminates}, an angle is considered at most
once by \apss. This implies that the angles around an edge can be
visited in different ways depicted on Figure~\ref{fig:anglerulesum}.
Consider an execution of \aps on $G$.  Let $C$ be the cycle formed in
the angle graph by Proposition~\ref{prop:terminates}.  Let $P$ be the
set of edges of the output $U$ (without the stems) and $Q$ be the set
of dual edges of edges of $G$ corresponding to stems of $U$. These
edges are represented on Figure~\ref{fig:anglerulesum} in cyan for $P$
and in yellow for $Q$.  They are considered with the convention
 that the dual edge $e^*$ of an edge $e$ goes from the face on
the left of $e$ to the face on the right of $e$.  Note that $C$ does
not cross an edge of $P$ or $Q$, and moreover $P$ and $Q$ do not
intersect (i.e. an edge can be in $P$ or its dual in $Q$ but both
cases cannot happen).

\begin{figure}[h!]
\center
\begin{tabular}[c]{cc}
  \includegraphics[scale=0.5]{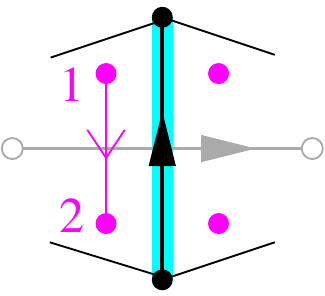} \ \ & \ \
\includegraphics[scale=0.5]{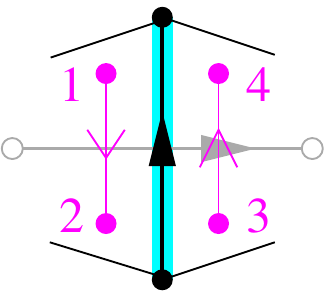} \\
 Case 1 alone \ \ & \ \ Case 1 and then 4 \\
\ \\
\includegraphics[scale=0.5]{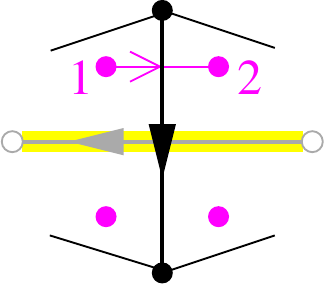} \ \ & \ \
\includegraphics[scale=0.5]{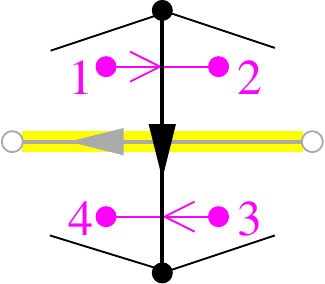} \\
 Case 2 alone \ \ & \ \ Case 2 and then 3 \\
\end{tabular} \ \ \ \
 \begin{tabular}[c]{cc}
   \includegraphics[scale=0.5]{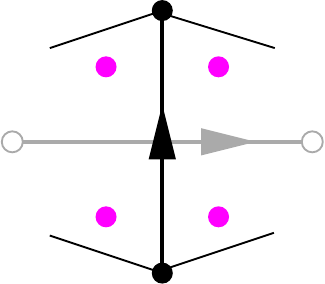} \\
 Not visited
 \end{tabular}
 \caption{The different cases of \aps seen in a dual way. The number
   of the angles gives the order in which the algorithm visits them
   (unvisited angles are not numbered). The edges of $P$ and $Q$
    are respectively cyan and yellow.}
\label{fig:anglerulesum}
\end{figure}

One can remark that the cases of Figure~\ref{fig:anglerulesum} are
dual of each other. One can see that \aps behaves exactly the same if
applied on the primal map or on the dual map. The only modifications
to make is to start the algorithm with the face $f_0$ as the root
vertex, the dual of edge $e_0$ as the root edge and to replace \ccw by
\cw at Line 4. Then the cycle $C$ formed in the angle graph is exactly
the same and the output is $Q$ with stems corresponding to $P$
(instead of $P$ with stems corresponding to $Q$). Note that this
duality is also illustrated by the fact that the minimality of the
orientation of $G$ w.r.t.~the root face is nothing else than the
accessibility of the dual orientation toward the root face. Indeed, a
clockwise $0$-homologous oriented subgraph of $G$ w.r.t $f_0$
corresponds to a directed cut of the dual where all the edges are
oriented from the part containing $f_0$. The following lemma shows the
connectivity of $P$ and $Q$:

\begin{lemma}
\label{lem:PDconnected}
At each step of the algorithm, for every vertex $v$ appearing in an edge
of $P$ (resp. $Q$), there is an oriented path from $v$ to $v_0$ (resp.
$f_0$) consisting only of edges of $P$ (resp. $Q$). In particular $P$
and $Q$ are connected.
\end{lemma}
\begin{proof}
  If at a step a new vertex is reached then it correspond to Case 1
  and the corresponding edge is added in $P$ and oriented from the new
  vertex, so the property is satisfied by induction. As observed
  earlier the algorithm behaves similarly in the dual map.
\end{proof}

Let $\overline{C}$ be the set of angles of $G$ that are not in $C$.
Any edge of $G$ is bounded by exactly 4 angles. Since $C$ is a cycle,
the 4 angles around an edge are either all in $C$, all in
$\overline{C}$ or 2 in each set (see
Figure~\ref{fig:anglerulesum}). Moreover, if they are 2 in each set,
these sets are separated by an edge of $P$ or an edge of $Q$. Hence
the frontier between $C$ and $\overline{C}$ is a set of edges of $P$
and $Q$. Moreover this frontier is an union of oriented closed walks
of $P$ and of oriented closed walks of $Q$.  In the next section we
study this frontier in more details to show that $\overline{C}$ is
empty in the case considered there.

\section{From toroidal triangulations to unicellular maps}
\label{sec:open}

Let $G$ be a toroidal triangulation.  In order to choose appropriately
the root angle $a_0$, we have to consider separating triangles. 
We say that an angle is \emph{in the
  strict interior of a separating triangle} $T$ if it is in the
interior of $T$ and not incident to a vertex of $T$.  We choose as
root angle $a_0$ any angle that is not in the strict interior of a
separating triangle.  One can easily see that such an angle $a_0$
always exists. Indeed the interiors of two triangles are either
disjoint or one is included in the other. So, the angles that are
incident to a triangle whose interior is maximal by inclusion satisfy
the property.

A subgraph of a graph is \emph{spanning} if it is covering all the
vertices.  The main result of this section is the following theorem
(see Figure~\ref{fig:tore-example} for an example):

\begin{theorem}
\label{th:uni}
Consider a toroidal triangulation $G$, a root angle $a_0$ that is not
in the strict interior of a separating triangle and the orientation of
the edges of $G$ corresponding to the minimal balanced Schnyder wood
w.r.t.~the root face $f_0$ containing $a_0$. Then the output $U$ of
\aps applied on $(G,a_0)$ is a toroidal spanning unicellular map.
\end{theorem}

The choice of a root angle that is not in the interior of a separating
triangle is necessary to be able to use \ps method.  Indeed, in a
$3$-orientation of a toroidal triangulation, by Euler's formula, all
the edges that are incident to a separating triangle and in its
interior are oriented toward the triangle.  Thus if one applies \aps
from an angle in the strict interior of a triangle, the algorithm will
remain stuck in the interior of the triangle and will not visit all
the vertices.

Consider a toroidal triangulation $G$, a root angle $a_0$ that is not
in the strict interior of a separating triangle and the orientation of
the edges of $G$ corresponding to the minimal balanced Schnyder wood
w.r.t.~the root face $f_0$ containing $a_0$. Let $U$ be the output of
\aps applied on $(G,a_0)$.  We use the same notation as in the previous
section: the cycle in the angle graph is $C$, the set of angles that
are not in $C$ is $\overline{C}$, the set of edges of $U$ is $P$, the
dual edges of stems of $U$ is $Q$.

\begin{lemma}
\label{lem:noD}
The frontier between $C$ and $\overline{C}$ contains no oriented
closed walk of $Q$.
\end{lemma}

\begin{proof}
  Suppose by contradiction that there exists such a walk $W$. Then
  along this walk, all the dual edges of $W$ are edges of $G$ oriented
  from the region containing $C$ toward $\overline{C}$ as one can see
  in Figure~\ref{fig:anglerulesum}. By Lemma~\ref{lem:incomingedges},
  walk $W$ does not contain any non-contractible cycle. So $W$
  contains an oriented contractible cycle $W'$, and then either $C$ is
  in the contractible region delimited by $W'$, or not. The two case
  are considered below:

\begin{itemize}
\item \emph{$C$ lies in the non-contractible region of $W'$:}

  Then consider the plane map $G'$ obtained from $G$ by keeping only
  the vertices and edges that lie (strictly) in the contractible
  region delimited by $W'$. Let $n'$ be the number of vertices of
  $G'$. All the edges incident to $G'$ that are not in $G'$ are
  entering $G'$. So in $G'$ all the vertices have outdegree $3$ as we
  are considering $3$-orientations of $G$. Thus the number of edges of
  $G'$ is exactly $3n'$, contradicting the fact that the maximal
  number of edges of planar map on $n$ vertices is $3n-6$ by Euler's
  formula.

\item \emph{$C$ lies in the contractible region of $W'$:}

  All
  the dual edges of $W'$ are edges of $G$ oriented from its
  contractible region toward its exterior. Consider the graph
  $G_{out}$ obtained from $G$ by removing all the edges that are cut
  by $W'$ and all the vertices and edges that lie in the contractible
  region of $W'$. As $G$ is a map, the face of $G_{out}$ containing
  $W'$ is homeomorphic to an open disk. Let $F$ be its facial walk (in
  $G_{out}$) and let $k$ be the length of $F$.  We consider the map
  obtained from the facial walk $F$ by putting back the vertices and
  edges that lied inside. We transform this map into a plane map $G'$
  by duplicating the vertices and edges appearing several times in
  $F$, in order to obtain a triangulation of a cycle of length
  $k$. Let $n',m',f'$ be the number of vertices, edges and faces of
  $G'$. Every inner vertex of $G'$ has outdegree $3$, there are no
  other inner edges, so the total number of edges of $G'$ is
  $m'=3(n'-k)+k$. All the inner faces have size $3$ and the outer face
  has size $k$, so $2m'=3(f'-1)+k$.  By Euler's formula
  $n'-m'+f'=2$. Combining the three equalities gives $k=3$ and $F$ is
  hence a separating triangle of $G$. This contradicts the choice of
  the root angle, as it should not lie in the strict interior of a
  separating triangle.
\end{itemize}
\end{proof}

A \emph{Hamiltonian cycle} of a graph is a cycle visiting every
vertex once.

\begin{lemma}
\label{lem:ham}
The cycle $C$ is a Hamiltonian cycle of the angle graph, all the
edges of $G$ are marked exactly twice, the subgraph $Q$ of $G^*$ is
spanning, and, if $n\geq 2$, the subgraph $P$ of $G$ is spanning.
\end{lemma}

\begin{proof}
  Suppose for a contradiction that $\overline{C}$ is non empty. By
  Lemma~\ref{lem:noD}, the frontier $T$ between $C$ and $\overline{C}$
  is an union of oriented closed walks of $P$.  Hence a face of $G$
  has either all its angles in $C$ or all its angles in
  $\overline{C}$. Moreover $T$ is a non-empty union of oriented closed
  walk of $P$ that are oriented \cw according to the set of faces
  containing $\overline{C}$ (see the first case of
  Figure~\ref{fig:anglerulesum}). This set does not contain $f_0$
  since $a_0$ is in $f_0$ and $C$.  As in Section~\ref{sec:lattice},
  let $\mc{F}$ be the set of counterclockwise facial walks of $G$ and
  $F_0$ be the counterclockwise facial walk of $f_0$. Let
  $\mc{F}'=\mc{F}\setminus F_0$, and
  $\mc{F}_{\overline{C}}\subseteq\mc{F}'$ be the set of
  counterclockwise facial walks of the faces containing
  $\overline{C}$. We have
  $\phi(T)=-\sum_{F\in\mc{F}_{\overline{C}}}\phi(F)$.  So $T$ is a
  clockwise non-empty $0$-homologous oriented subgraph
  w.r.t.~$f_0$. This contradicts Lemma~\ref{prop:maximal} and the
  minimality of the orientation w.r.t.~$f_0$. So $\overline{C}$ is
  empty, thus $C$ is Hamiltonian and all the edges of $G$ are marked
  twice.

  Suppose for a contradiction that $n\geq 2$ and $P$ is not
  spanning. Since the algorithm starts at $v_0$, $P$ is not covering a
  vertex $v$ of $G$ different from $v_0$. Then the angles around $v$
  cannot be visited since by Figure~\ref{fig:anglerulesum} the only
  way to move from an angle of one vertex to an angle of another
  vertex is through an edge of $P$ incident to them. So $P$ is
  spanning.  The proof is similar for $Q$ (note that  we always
  have $f\geq 2$).
\end{proof}

\begin{lemma}
\label{lem:first}
The first cycle created in $P$ (resp. in $Q$) by the algorithm is
oriented.
\end{lemma}

\begin{proof}
  Let $e$ be the first edge creating a cycle in $P$ while executing
  \aps and consider the steps of \aps before $e$ is added to $P$. So
  $P$ is a tree during all these steps.  For every vertex of $P$ we
  define $P(v)$ the unique path from $v$ to $v_0$ in $P$ (while $P$ is
  empty at the beginning of the execution, we define $P(v_0)=\{v_0\}$).
  By Lemma~\ref{lem:PDconnected}, this path $P(v)$ is an oriented
  path.  We prove the following

  \begin{claim}
\label{cl:leftangles}
Consider a step of the algorithm before $e$ is added to $P$ and where
the algorithm is considering a vertex $v$. Then all the angles around
the vertices of $P$ different from the vertices of $P(v)$ are already
visited.
  \end{claim}

  \begin{proof}
    Suppose by contradiction that there is such a step of the
    algorithm where some angles around the vertices of $P$ different
    from the vertices of $P(v)$ have not been visited. Consider the
    first such step. Then clearly we are not at the beginning of the
    algorithm since $P=P(v)=\{v_0\}$. So at the step just before, the
    conclusion holds and now it does not hold anymore. Clearly, at the
    step before, we were considering a vertex $v'$ distinct from $v$,
    otherwise $P(v)$ and $P$ have not changed and we have the
    conclusion. So from $v'$ to $v$ we are either in Case~1 or Case~4
    of \apss. If $v$ has been considered by Case~1, then $P(v)$
    contains $P(v')$ and the conclusion holds. If $v$ has been
    considered by Case~4, then since $P$ is a tree, all the angles
    around $v'$ have been considered and $v'$ is the only element of
    $P\setminus P(v)$ that is not in $P\setminus P(v')$. Thus the
    conclusion also holds.
  \end{proof}

  Consider the iteration of \aps where $e$ is added to $P$. The edge
  $e$ is added to $P$ by Case~1, so $e$ is oriented from a vertex $u$
  to a vertex $v$ such that $v$ is already in $P$ or $v$ is the root
  vertex $v_0$. Consider the step of the algorithm just before $u$ is
  added to $P$.  By Claim~\ref{cl:leftangles}, vertex $u$ is not in
  $P\setminus P(v)$ (otherwise $e$ would have been considered before
  and it would be a stem). So $u\in P(v)$ and $P(v)\cup\{e\}$ induces
  an oriented cycle of $G$. The proof is similar for $Q$.
\end{proof}

\begin{lemma}
\label{lem:unicellular}
  $P$ is a spanning unicellular map of $G$ and $Q$ is a spanning tree
  of $G^*$. Moreover one is the dual of the complement of the other.
\end{lemma}
\begin{proof}
  Suppose that $Q$ contains a cycle, then by Lemma~\ref{lem:first} it
  contains an oriented cycle of $G^*$. This cycle is contractible by
  Lemma~\ref{lem:incomingedges}. Recall that by Lemma~\ref{lem:ham},
  $C$ is a Hamiltonian cycle, moreover it does not cross $Q$, a
  contradiction.  So $Q$ contains no cycle and is a tree.

  By Lemma~\ref{lem:ham}, all the edges of $G$ are marked at the
  end. So every edge of $G$ is either in $P$ or its dual in $Q$ (and
  not both). Thus $P$ and $Q$ are the dual of the complement of each
  other. So $P$ is the dual of the complement of a spanning tree of
  $G^*$. Thus $P$ is a spanning unicellular map of $G$.
\end{proof}

Theorem~\ref{th:uni} is then a direct reformulation of
Lemma~\ref{lem:unicellular} by the definition of $P$ and $Q$.

A toroidal unicellular map on $n$ vertices has exactly $n+1$ edges:
$n-1$ edges of a tree plus $2$ edges corresponding to the size of a
homology-basis (i.e. plus $2g$ in general for an oriented
surface of genus $g$).  Thus a consequence of Theorem~\ref{th:uni} is
that the obtained unicellular map $U$ has exactly $n$ vertices, $n+1$
edges and $2n-1$ stems since the total number of edges is $3n$.  The
orientation of $G$ induces an orientation of $U$ such that the stems
are all outgoing, and such that while walking \cw around the unique
face of $U$ from $a_0$, the first time an edge is met, it is oriented
\ccw according to this face, see Figure~\ref{fig:hexasquare} where all the
tree-like parts and stems are not represented.  There are two types
of toroidal unicellular maps depicted on
Figure~\ref{fig:hexasquare}. Two cycles of $U$ may intersect either
on a single vertex (square case) or on a path (hexagonal case).  The
square can be seen as a particular case of the hexagon where one side
has length zero and thus the two corners of the hexagon are
identified.

\begin{figure}[!h]
\center
\begin{tabular}{cc}
\includegraphics[scale=0.5]{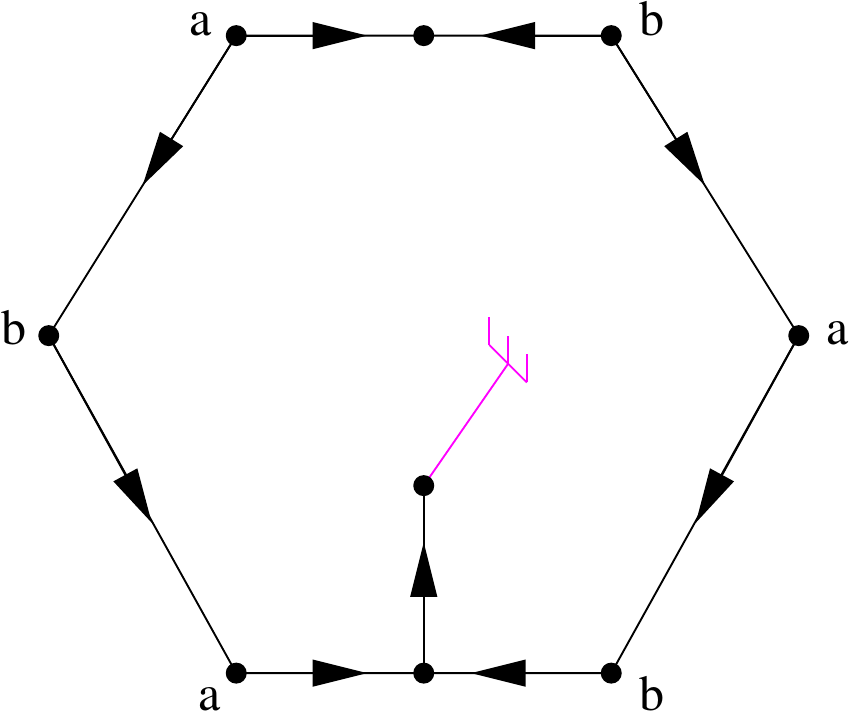} \ \ & \ \
\includegraphics[scale=0.5]{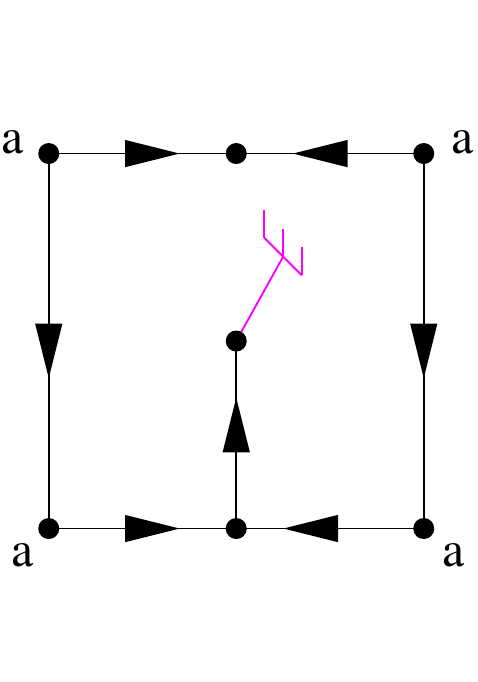}\\
Hexagon \ \ & \ \ Square
\end{tabular}

\caption{The two types of rooted toroidal unicellular maps.}
\label{fig:hexasquare}
\end{figure}

In Figure~\ref{fig:contre-exemple}, we give several examples of
executions of \aps on minimal $3$-orientations.  These examples show
how important is the choice of the minimal balanced Schnyder wood in order
to obtain Theorem~\ref{th:uni}. In particular, the third example shows
that \aps can visit all the angles of the triangulation (i.e. the
cycle in the angle graph is Hamiltonian) without outputting an
unicellular map.

\begin{figure}[h!]
\center
\begin{tabular}[c]{cc}
\includegraphics[scale=0.5]{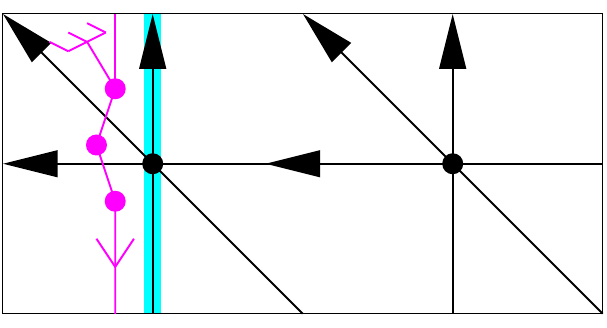} \ \ \ \ & \ \ \ \
\includegraphics[scale=0.5]{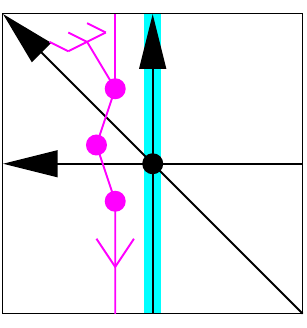}\\
(1) \ \ \ \ & \ \ \ \ (2)
\end{tabular}

\ \\

\ \\

 \begin{tabular}[c]{c}
 \includegraphics[scale=0.5]{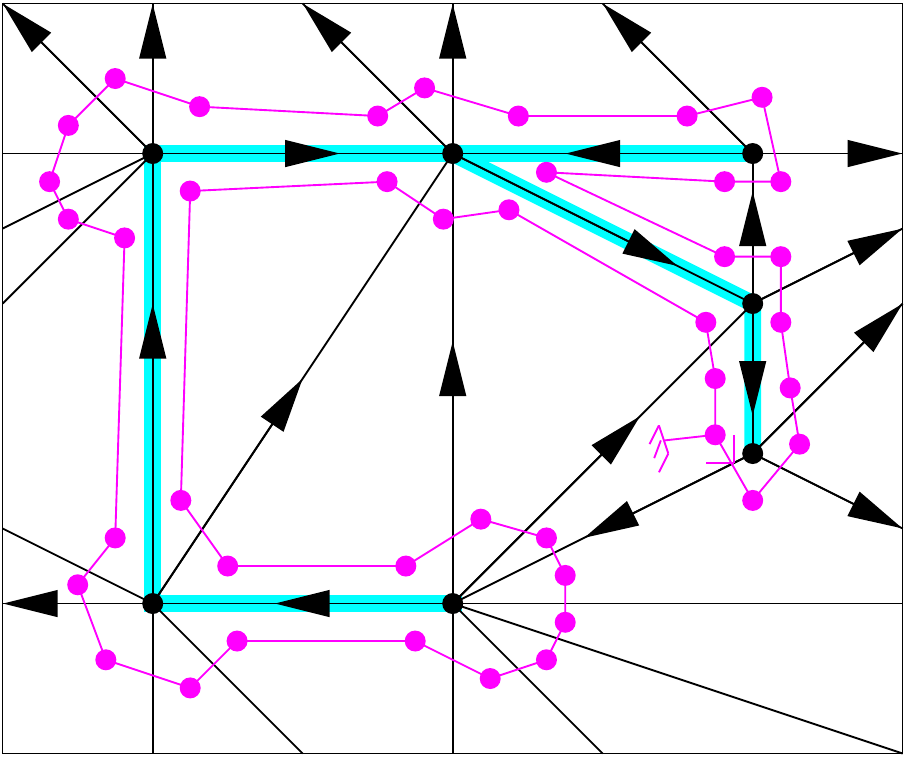} \\
 (3)
 \end{tabular}

\caption{Examples of minimal $3$-orientations that are not balanced
  Schnyder woods and where \aps respectively: $(1)$ does not visit all
  the vertices, $(2)$ visits all the vertices but not all the angles, and
  $(3)$ visits all the angles but does not output an unicellular map.}
\label{fig:contre-exemple}
\end{figure}

Note that the orientations of Figure~\ref{fig:contre-exemple} are not
Schnyder woods. One may wonder if the fact of being a Schnyder wood is
of any help for our method. This is not the case since there are
examples of minimal Schnyder woods that are not balanced and where \aps
does not visit all the vertices. One can obtain such an example by
replicating $3$ times horizontally and then $3$ times vertically the
second example of Figure~\ref{fig:contre-exemple} to form a
$3\times 3$ tiling and starts \aps from the same root angle.
Conversely, there are minimal Schnyder woods that are not balanced where
\aps does output a toroidal spanning unicellular map (the Schnyder
wood of Figure~\ref{fig:gamma0glue} can serve as an example while
starting from an angle of the only face oriented \cww).

\section{Recovering the original triangulation}
\label{sec:close}

This section is dedicated to showing how to recover the original
triangulation from the output of \apss. The method is very similar
to~\cite{PS06} since like in the plane the output has only one face
that is homeomorphic to an open disk (i.e. a tree in the plane and an
unicellular map in general). 

\begin{theorem}
\label{th:recover}
Consider a toroidal triangulation $G$, a root angle $a_0$ that is not
in the strict interior of a separating triangle and the orientation of
the edges of $G$ corresponding to the minimal balanced Schnyder wood
w.r.t.~the root face $f_0$ containing $a_0$. From the output $U$ of
\aps applied on $(G,a_0)$ one can reattach all the stems to obtain $G$
by starting from the root angle $a_0$ and walking along the face of
$U$ in \ccw order (according to this face): each time a stem is met,
it is reattached in order to create a triangular face on its left
side.
\end{theorem}

Theorem~\ref{th:recover} is illustrated on
Figure~\ref{fig:reconstruction} where one can check that the obtained
toroidal triangulation is $K_7$ (like on the input of
Figure~\ref{fig:tore-example}).

\begin{figure}[h!]
\center
\begin{tabular}{cc}
\includegraphics[scale=0.5]{tore-tri-exe-4---}
 & 
\includegraphics[scale=0.5]{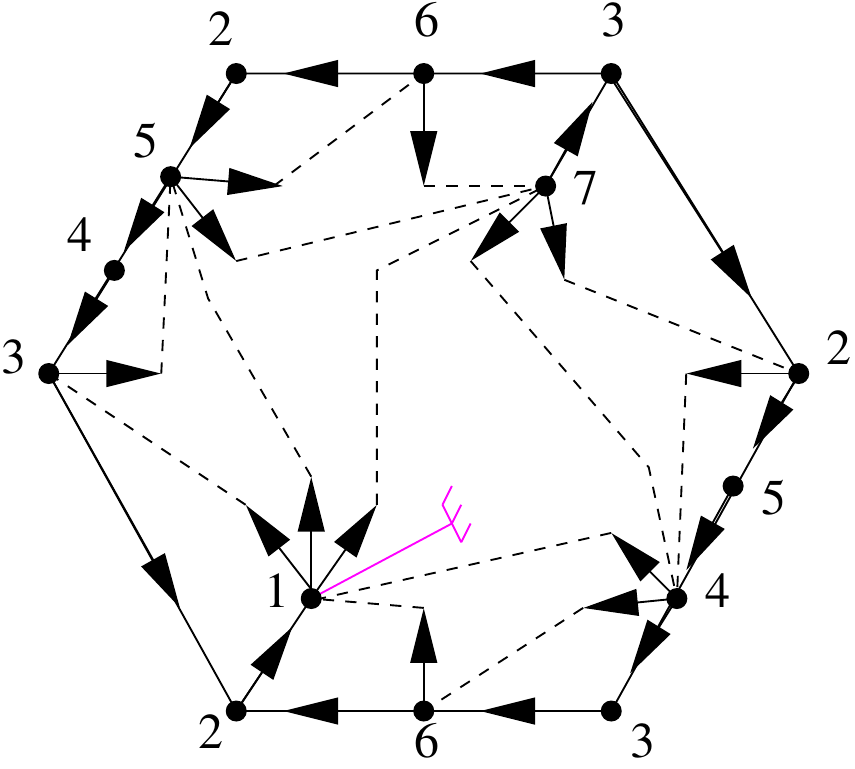}\\
Output & Recovering \\
\end{tabular}
\caption{Example of how to recover the original toroidal
  triangulation $K_7$ from the output of \aps.}
\label{fig:reconstruction}
\end{figure}

In fact in this section we define a method, more general than the one
described in Theorem~\ref{th:recover}, that will be useful in next
sections.

We consider the two type of unicellular toroidal map depicted on
Figure~\ref{fig:hexasquare} and call \emph{corner} a vertex that is a
corner of an hexagon or the square

Let $\mathcal U_r(n)$ denote the set of toroidal unicellular maps $U$
rooted on a particular angle, with exactly $n$ vertices, $n+1$ edges
and $2n-1$ stems satisfying the following property. 
A vertex that is
not the root, has exactly $2$ stems if it is not a corner, $1$ stem if
it is the corner of a hexagon and $0$ stem if it is the corner of a
square. The root vertex has $1$ additional stem, i.e.  it has $3$
stems if it is not a corner, $2$ stems if it is the corner of a hexagon
and $1$ stem if it is the corner of a square.  Note that the output
$U$ of \aps given by Theorem~\ref{th:uni} is an element of
$\mathcal U_r(n)$.

Similarly to the planar case~\cite{PS06}, we define a general way to
reattach step by step all the stems of an element $U$ of
$\mathcal U_r(n)$.  Let $U_0=U$, and, for $1\leq k \leq 2n-1$, let
$U_{k}$ be the map obtained from $U_{k-1}$ by reattaching one of its
stem (we explicit below which stem is reattached and how). The
\emph{special face of $U_0$} is its only face. For
$1\leq k \leq 2n-1$, the \emph{special face of $U_{k}$} is the face on
the right of the stem of $U_{k-1}$ that is reattached to obtain
$U_{k}$.  For $0\leq k\leq 2n-1$, the border of the special face of
$U_k$ consists of a sequence of edges and stems. We define an
\emph{admissible triple} as a sequence $(e_1,e_2,s)$, appearing in
\ccw order along the border of the special face of $U_k$, such that
$e_1=(u,v)$ and $e_2=(v,w)$ are edges of $U_k$ and $s$ is a stem
attached to $w$. The \emph{closure} of the admissible triple consists
in attaching $s$ to $u$, so that it creates an edge $(w,u)$ oriented
from $w$ to $u$ and so that it creates a triangular face $(u,v,w)$ on
its left side.  The \emph{complete closure} of $U$ consists in closing
a sequence of admissible triple, i.e.  for $1\leq k \leq 2n-1$, the
map $U_{k}$ is obtained from $U_{k-1}$ by closing any admissible
triple.

Note that, for $0\leq k\leq 2n-1$, the special face of $U_k$ contains
all the stems of $U_k$. The closure of a stem reduces the number of
edges on the border of the special face and the number of stems by
$1$. At the beginning, the unicellular map $U_0$ has $n+1$ edges and
$2n-1$ stems. So along the border of its special face, there are
$2n+2$ edges and $2n-1$ stems. Thus there is exactly three more edges
than stems on the border of the special face of $U_0$ and this is
preserved while closing stems. So at each step there is necessarily at
least one admissible triple and the sequence $U_k$ is well defined.
Since the difference of three is preserved, the special face of
$U_{2n-2}$ is a quadrangle with exactly one stem. So the reattachment
of the last stem creates two faces that have size three and at the end
$U_{2n-1}$ is a toroidal triangulation.  Note that at a given step
there might be several admissible triples but their closure are
independent and the order in which they are performed does not modify
the obtained triangulation $U_{2n-1}$.

We now apply the closure method to our particular case.  Consider a
toroidal triangulation $G$, a root angle $a_0$ that is not in the
strict interior of a separating triangle and the orientation of the
edges of $G$ corresponding to the minimal balanced Schnyder wood w.r.t.~the
root face $f_0$. Let $U$ be the output of \aps
applied on $(G,a_0)$.

\begin{lemma}
\label{lem:stemright}
When a stem of $U$ is reattached to form the corresponding edge of
$G$, it splits the (only) face of $U$ into two faces. The root angle
of $U$ is in the face that is on the right side of the stem.
\end{lemma}

\begin{proof}
  By Lemma~\ref{lem:ham}, the execution of \aps corresponds to a
  Hamiltonian cycle $C=(a_0, \ldots, a_{2m},a_0)$ in the angle graph
  of $G$. Thus $C$ defines a total order $<$ on the angles of $G$
  where $a_i< a_j$ if and only if $i< j$.  Let us consider now the
  angles on the face of $U$.  Note that such an angle corresponds to
  several angles of $G$, that are consecutive in $C$ and that are
  separated by a set of incoming edges of $G$ (those incoming edges
  corresponding to stems of $U$).  Thus the order on the angles of $G$
  defines automatically an order on the angles of $U$.  The angles of
  $U$ considered in \cw order along the border of its face, starting
  from the root angle, correspond to a sequence of strictly
  increasing angles for $<$.

  Consider a stem $s$ of $U$ that is reattached to form an edge $e$ of
  $G$.  Let $a_s$ be the angle of $U$ that is situated just before $s$
  (in \cw order along the border of the face of $U$) and $a'_s$ be the
  angle of $U$ where $s$ should be reattached.  If $a'_s< a_s$, then
  when \aps consider the angle $a_s$, the edge corresponding to $s$ is
  already marked and we are not in Case~2 of \aps. So $a_s< a'_s$ and
  $a_0$ is on the right side of $s$.
\end{proof}

Recall that $U$ is an element of $\mathcal U_r(n)$ so we can apply on
$U$ the complete closure procedure described above.  We use the same
notation as before, i.e. let $U_0=U$ and for $1\leq k \leq 2n-1$, the
map $U_{k}$ is obtained from $U_{k-1}$ by closing any admissible
triple.  The following lemma shows that the triangulation obtained by
this method is $G$:

\begin{lemma}
\label{lem:stemstep}
  The complete closure of $U$ is $G$,  i.e. $U_{2n-1}=G$.
\end{lemma}

\begin{proof}
  We prove by induction on $k$ that every face of $U_k$ is a face of
  $G$, except for the special face.  This is true for $k=0$ since
  $U_0=U$ has only one face, the special face.  Let
  $0\leq k\leq 2n-2$, and suppose by induction that every non-special
  face of $U_k$ is a face of $G$.  Let $(e_1,e_2,s)$ be the admissible
  triple of $U_k$ such that its closure leads to $U_{k+1}$, with
  $e_1=(u,v)$ and $e_2=(v,w)$. The closure of this triple leads to a
  triangular face $(u,v,w)$ of $U_{k+1}$. This face is the only
  ``new'' non-special face while going from $U_k$ to $U_{k+1}$.

  Suppose, by contradiction, that this face $(u,v,w)$ is not a face of
  $G$.  Let $a_v$ (resp. $a_w$) be the angle of $U_k$ at the special
  face, between $e_1$ and $e_2$ (resp. $e_2$ and $s$).  Since $G$ is a
  triangulation, and $(u,v,w)$ is not a face of $G$, there exists at
  least one stem of $U_k$ that should be attached to $a_v$ or $a_w$ to
  form a proper edge of $G$. Let $s'$ be such a stem that is the
  nearest from $s$. In $G$ the edges corresponding so $s$ and $s'$
  should be incident to the same triangular face. Let $x$ be the
  origin of the stem $s'$.  Let $z\in \{v,w\}$ such that $s'$ should
  be reattached to $z$.  If $z=v$, then $s$ should be reattached to
  $x$ to form a triangular face of $G$. If $z=w$, then $s$ should be
  reattached to a common neighbor of $w$ and $x$ located on the border
  of the special face of $U_k$ in \ccw order between $w$ and $x$. So
  in both cases $s$ should be reattached to a vertex $y$ located on
  the border of the special face of $U_k$ in \ccw order between $w$
  and $x$ (with possibly $y=x$). To summarize $s$ goes from $w$ to $y$
  and $s'$ from $x$ to $z$, and $z,x,y,w$ appear in \cw order along
  the special face of $U_k$. By Lemma~\ref{lem:stemright}, the root
  angle is on the right side of both $s$ and $s'$, this is not
  possible since their right sides are disjoint, a contradiction.

  So for $0\leq k\leq 2n-2$, all the non-special faces of $U_k$ are
  faces of $G$. In particular every face of $U_{2n-1}$ except one is a
  face of $G$. Then clearly the (triangular) special face of
  $U_{2n-1}$ is also a face of $G$, hence $U_{2n-1}=G$.
\end{proof}

Lemma~\ref{lem:stemstep} shows that one can recover the original
triangulation from $U$ with any sequence of admissible triples
that are closed successively. This does not explain how to find the
admissible triples efficiently. In fact the root angle can be used to
find a particular admissible triple of $U_k$:

\begin{lemma}
\label{lem:firststem}
For $0\leq k\leq 2n-2$, let $s$ be the first stem met while walking
\ccw from $a_0$ in the special face of $U_k$. Then before $s$, at
least two edges are met and  the last two of these edges form an
admissible triple with $s$.
\end{lemma}

\begin{proof}
  Since $s$ is the first stem met, there are only edges that are met
  before $s$. Suppose by contradiction that there is only zero or one
  edge met before $s$. Then the reattachment of $s$ to form the
  corresponding edge of $G$ is necessarily such that the root angle is
  on the left side of $s$, a contradiction to
  Lemma~\ref{lem:stemright}. So at least two edges are met before $s$
  and the last two of these edges form an admissible triple with $s$.
\end{proof}

Lemma~\ref{lem:firststem} shows that one can reattach all the stems
by walking once along the face of $U$ in \ccw order. Thus we obtain
Theorem~\ref{th:recover}.

Note that $U$ is such that the complete closure procedure described
here never \emph{wraps over the root angle}, i.e. when a stem is
reattached, the root angle is always on its right side (see
Lemma~\ref{lem:stemright}). The property of never wrapping over the
root angle is called \emph{safe}. Usually, this property is called
balanced in the literature but the word balanced is already used in
this manuscript and in this chapter also.  Let $\mathcal U_{r,s}(n)$
denote the set of elements of $\mathcal U_r(n)$ that are safe. So the
output $U$ of \aps given by Theorem~\ref{th:uni} is an element of
$\mathcal U_{r,s}(n)$.  We exhibit in Section~\ref{sec:bijection} a
bijection between appropriately rooted toroidal triangulations and a
particular subset of $\mathcal U_{r,s}(n)$.

The possibility to close admissible triples in any order to recover
the original triangulation is interesting compared to the simpler
method of Theorem~\ref{th:recover} since it enables to recover the
triangulation even if the root angle is not given.  This property is
used in Section~\ref{sec:bij2} to obtain a bijection between toroidal
triangulations and some unrooted unicellular maps.

Moreover if the root angle is not given, then one can simply start
from any angle of $U$, walk twice around the face of $U$ in \ccw order
and reattach all the admissible triples that are encountered along
this walk. Walking twice ensures that at least one complete round is
done from the root angle. Since only admissible triples are
considered, we are sure that no unwanted reattachment is done during
the process and that the final map is $G$. This enables to reconstruct
$G$ in linear time even if the root angle is not known. This property
is used in Section~\ref{sec:coding}.

\section{Optimal encoding}
\label{sec:coding}
The results presented in the previous sections allow us to generalize
the encoding of planar triangulations, defined by Poulalhon and
Schaeffer~\cite{PS06}, to triangulations of the torus. The
construction is direct and it is hence really different from the one
of~\cite{CFL10} where triangulations of surfaces are cut in order to
deal with planar triangulations with boundaries. Here we encode the
unicellular map outputted by \aps by a plane rooted tree with
$n$ vertices and with exactly two stems attached to each vertex, plus
$O(\log(n))$ bits.  As in~\cite{CFL10}, this encoding is
asymptotically optimal and uses approximately $3.2451 n$ bits. The
advantage of our method is that it can be implemented in linear
time. Moreover we believe that our encoding gives a better
understanding of the structure of triangulations of the torus. It is
illustrated with new bijections that are obtained in
Sections~\ref{sec:bijection} and~\ref{sec:bij2}.

Consider a toroidal triangulation $G$, a root angle $a_0$ that is not
in the strict interior of a separating triangle and the orientation of
the edges of $G$ corresponding to the minimal balanced Schnyder wood
w.r.t.~the root face $f_0$. Let $U$ be the output of \aps applied on
$(G,a_0)$. As already mentioned at the end of Section~\ref{sec:close},
to retrieve the triangulation $G$ one just needs to know $U$ without
the information of its root angle (by walking twice around the face of
$U$ in \ccw order and reattaching all the admissible triples that are
encountered along this walk, one can recover $G$). Hence to encode
$G$, one just has to encode $U$ without the position of the root angle
around the root vertex (see Figure~\ref{fig:Coding}.(a)).

By Lemma~\ref{lem:PDconnected}, the unicellular map $U$ contains a
spanning tree $T$ which is oriented from the leaves to the root
vertex. The tree $T$ contains exactly $n-1$ edges, so there is exactly
$2$ edges of $U$ that are not in $T$. We call these edges the
\emph{special edges} of $U$. We cut these two special edges to
transform them into stems of $T$ (see Figures~\ref{fig:Coding}.(a)
and~(b)). We keep the information of where are the special stems in
$T$ and on which angle of $T$ they should be reattached. This
information can be stored with $O(\log(n))$ bits.  One can recover $U$
from $T$ by reattaching the special stems in order to form
non-contractible cycles with $T$ (see Figure~\ref{fig:Coding}.(c)).

\begin{figure}[h!]
\center
\begin{tabular}{ccc}
\includegraphics[scale=0.3]{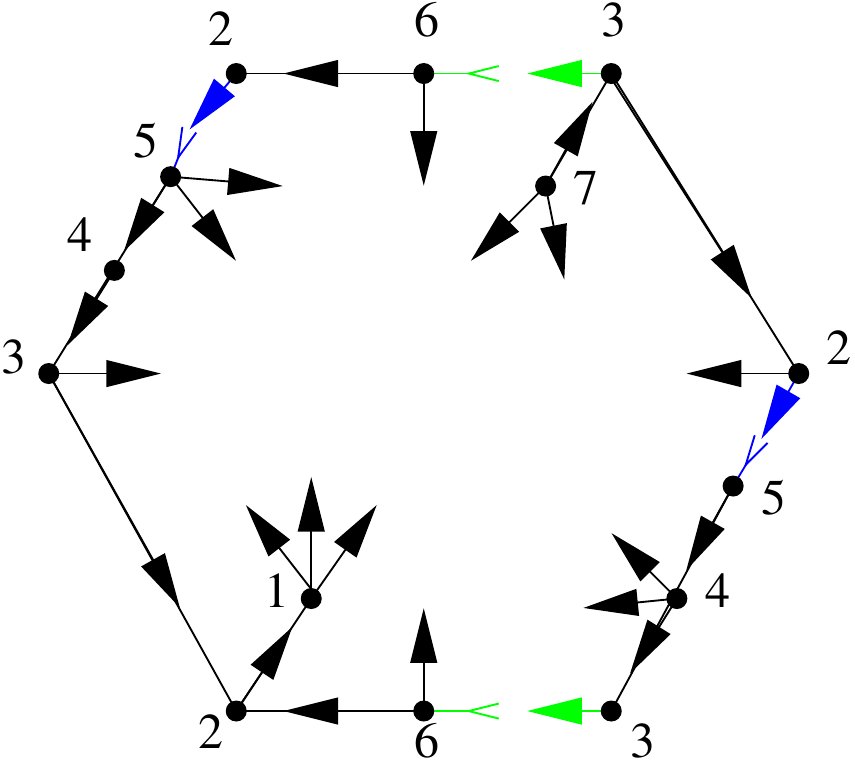} \ & \
\includegraphics[scale=0.3]{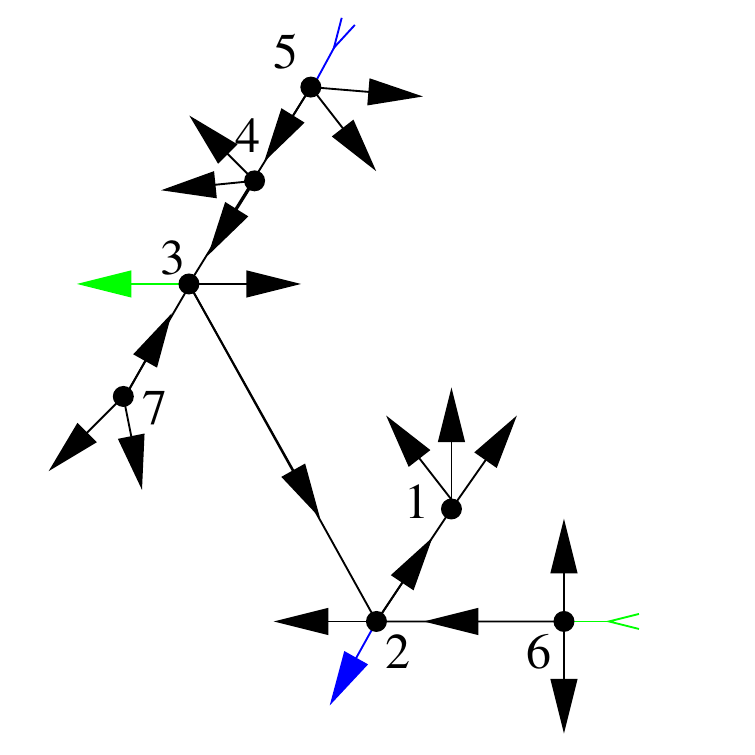} \ & \ 
\includegraphics[scale=0.3]{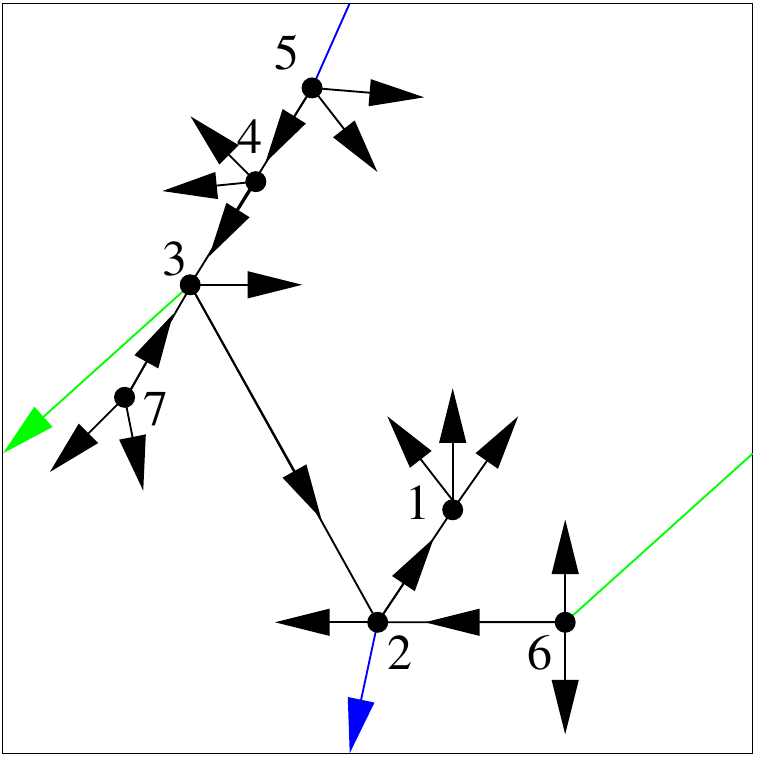} \\
(a) \ &\ (b) \ &\ (c) \\
\end{tabular}
\caption{From unicellular maps to trees with special stems and back.}
\label{fig:Coding}
\end{figure}

So $T$ is a plane tree on $n$ vertices, each vertex having $2$ stems
except the root vertex $v_0$ having three stems. Choose any stem $s_0$
of the root vertex, remove it and consider that $T$ is rooted at the
angle where $s_0$ should be attached. The information of the root
enables to put back $s_0$ at its place. So now we are left with a
rooted plane tree $T$ on $n$ vertices where each vertex has exactly
$2$ stems (see Figure~\ref{fig:CodingTree}.(a)).

This tree $T$ can easily be encoded by a binary word on $6n-2$ bits:
that is, walking in \ccw order around $T$ from the root angle, writing
a ``1'' when going down along $T$, and a ``0'' when going up along $T$
(see Figure~\ref{fig:CodingTree}.(a)). As in~\cite{PS06}, one can
encode $T$ more compactly by using the fact that each vertex has
exactly two stems.  Thus $T$ is encoded by a binary word on $4n-2$
bits: that is, walking in \ccw order around $T$ from the root angle,
writing a ``1'' when going down along an edge of $T$, and a ``0'' when
going up along an edge or along a stem of $T$ (see
Figure~\ref{fig:CodingTree}.(b) where the ``red 1's'' of
Figure~\ref{fig:CodingTree}.(a) have been removed).  Indeed there is
no need to encode when going down along stems, this information can be
retrieved afterward. While reading the binary word to recover $T$,
when a ``0'' is met, we should go up in the tree, except if the
vertex that we are considering does not have already its two stems,
then in that case we should create a stem (i.e. add a ``red 1'' before
the ``0''). So we are left with a binary word on $4n-2$ bits with
exactly $n-1$ bits ``1'' and $3n-1$ bits ``0''.

\begin{figure}[h!]
\center
\begin{tabular}{cc}
\includegraphics[scale=0.4]{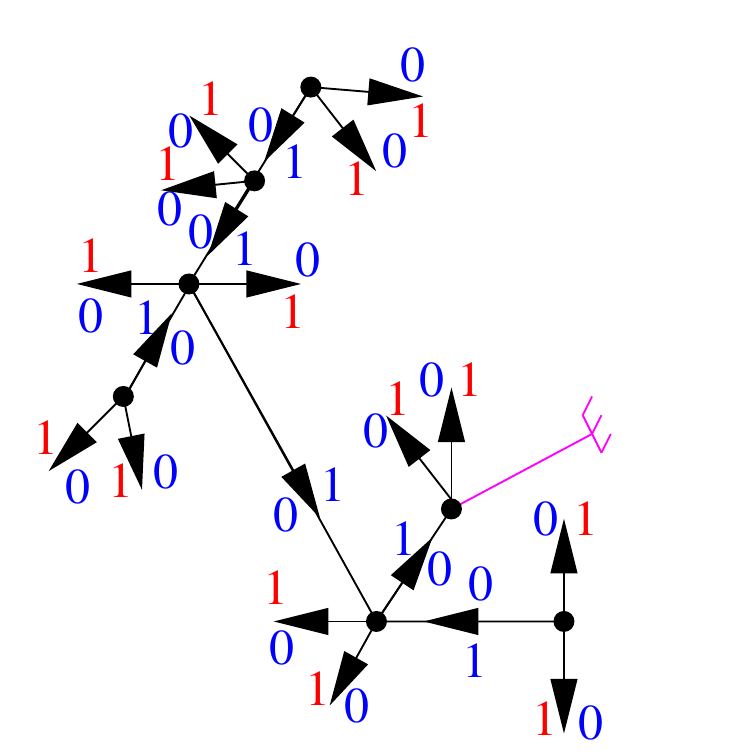}\ & \
\includegraphics[scale=0.4]{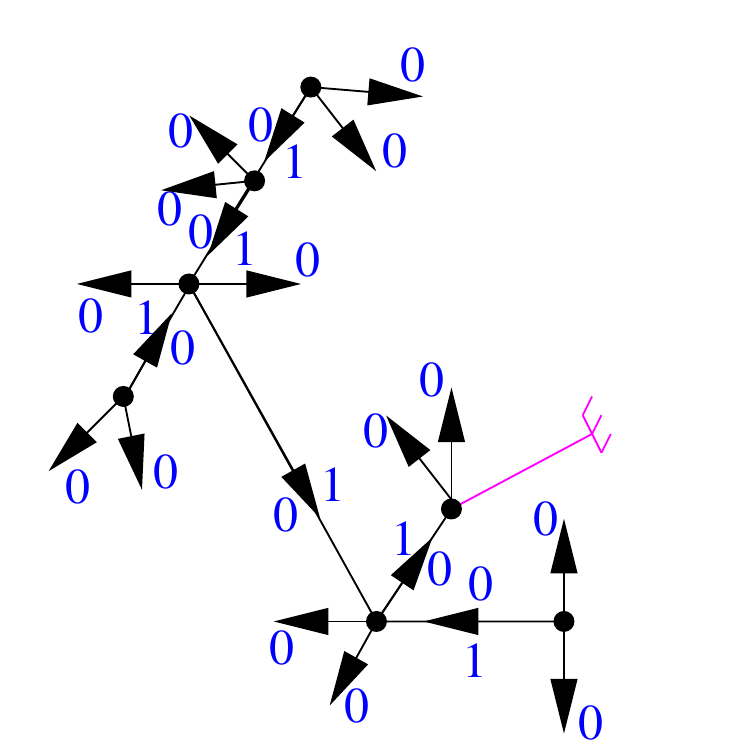}\\

{\color{blue}{\color{red}1}0{\color{red}1}011{\color{red}1}011{\color{red}1}0{\color{red}1}00{\color{red}1}0{\color{red}1}00{\color{red}1}01{\color{red}1}0{\color{red}1}000{\color{red}1}0{\color{red}1}01{\color{red}1}0{\color{red}1}000}  \ & \
{\color{blue}00110110000000100000010000} \\
(a) \ &\ (b) \\
\end{tabular}
\caption{Encoding a rooted tree with two stems at each vertex.}
\label{fig:CodingTree}
\end{figure}

Similarly to~\cite{PS06}, using \cite[Lemma~7]{BGH03}, this word can
then be encoded with a binary word of length
$\log_2\binom{4n-2}{n-1}+o(n)\sim n\, \log_2(\frac{256}{27})\approx
3.2451\,n$
bits. Thus we have the following theorem whose linearity is discussed
below:

\begin{theorem}
\label{th:encoding}
  Any toroidal triangulation on $n$ vertices, can be encoded with a
  binary word of length $3.2451 n +o(n)$ bits, the encoding and
  decoding being linear in $n$.
\end{theorem}

The encoding method described here, that is encoding a toroidal
triangulation via an unicellular map and recovering the original
triangulation, can be performed in linear time.  The only difficulty
lies in providing \aps with the appropriate input it needs in order to
apply Theorem~\ref{th:uni}. Then clearly the execution of \apss, the
encoding phase and the recovering of the triangulation are linear like
in the planar case.  Thus we have to show how one can find in linear
time a root angle $a_0$ that is not in the strict interior of a
separating triangle, as well as the minimal balanced Schnyder wood
w.r.t.~the root face $f_0$.

Consider a balanced Schnyder wood of $G$ that can be computed in
linear time as explained in Section~\ref{sec:contractionproof}.  From
this balanced Schnyder wood, one can compute in linear time a root
angle $a_0$ not in the strict interior of a separating triangle. First
note that in a $3$-orientation of a toroidal triangulation, the edges
that are inside a separating triangle and that are incident to the
three vertices on the border are all oriented toward these three
vertices by Euler's formula. Thus an oriented non-contractible cycle
cannot enter in the interior of a separating triangle. Now follow any
oriented monochromatic path of the Schnyder wood and stop the first
time this path is back to a previously met vertex $v_0$. The end of
this path forms an oriented monochromatic cycle $C$ containing $v_0$.
 Thus $C$ is an oriented non-contractible cycle
and cannot contain some vertices that are in the interior of a
separating triangle. So $v_0$ is not in the interior of a separating
triangle and we can choose as root angle $a_0$ any angle incident to
$v_0$.  Then using the method described in Section~\ref{sec:lattice}.
one can transform the balanced Schnyder wood of $G$ into a minimal
balanced Schnyder wood w.r.t.~the root face $f_0$.

\section{Bijection with rooted unicellular maps}
\label{sec:bijection}
 
Given a toroidal triangulation $G$ with a root angle $a_0$, we have
defined a unique associated orientation: the minimal balanced Schnyder wood
w.r.t.~the root face $f_0$. Suppose that $G$ is
oriented according to the minimal balanced Schnyder wood.  If $a_0$ is not
in the strict interior of a separating triangle then
Theorems~\ref{th:uni} and~\ref{th:recover} show that the execution of
\aps on $(G,a_0)$ gives a toroidal unicellular map with stems from
which one can recover the original triangulation. Thus there is a
bijection between toroidal triangulations rooted from an appropriate
angle and their image by \apss. The goal of this section is to
describe this image.

Recall from Section~\ref{sec:close} that the output of \aps on
$(G,a_0)$ is an element of $\mathcal U_{r,s}(n)$.  One may hope that
there is a bijection between toroidal triangulations rooted from an
appropriate angle and $\mathcal U_{r,s}(n)$ since this is how it works
in the planar case. Indeed, given a planar triangulation $G$, there is
a unique orientation of $G$ (the minimal Schnyder wood) on which
\apss, performed from an outer angle, outputs a spanning tree. In the
toroidal case, things are more complicated since the behavior of \aps
on minimal balanced Schnyder woods does not characterize such orientations.

Figure~\ref{fig:twounicellular} gives an example of two
(non homologous) orientations of the same triangulation that are both
minimal w.r.t.~the same root face. For these two orientations, the
execution of \aps from the same root angle gives two different
elements of $\mathcal U_{r,s}(2)$ (from which the original
triangulation can be recovered by the method of
Theorem~\ref{th:recover}). 

\begin{figure}[h!]
\center
\begin{tabular}[c]{cc}
\includegraphics[scale=0.5]{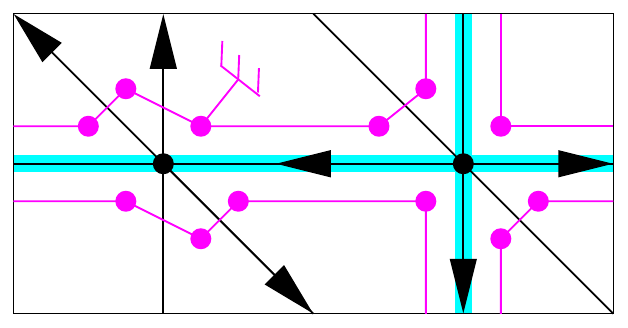}  \ \ & \ \
\includegraphics[scale=0.5]{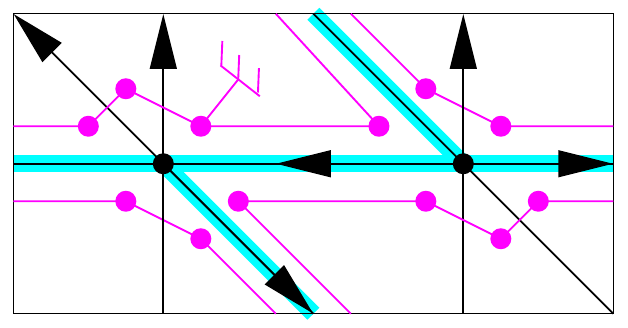} \\
balanced orientation \ \ &\ \ non-balanced orientation \\
\end{tabular}
\caption{A graph that can be represented by two different unicellular
  maps.}
\label{fig:twounicellular}
\end{figure}

Let us translate the balanced property on $\mathcal U_{r}(n)$.
Consider an element $U$ of $\mathcal U_{r}(n)$ whose edges and stems
are oriented w.r.t.~the root angle as follows: the stems are all
outgoing, and while walking \cw around the unique face of $U$ from
$a_0$, the first time an edge is met, it is oriented \ccw w.r.t.~the
face of $U$. Then one can compute $\gamma$ (define in
Section~\ref{sec:gencharacterization}) on the cycles of $U$ (edges and
stems count).  We say that an unicellular map of $\mathcal U_{r}(n)$
satisfies the $\gamma_0$ property if $\gamma$ equals zero on its
(non-contractible) cycles. Let us call $\mathcal U_{r,s,\gamma_0}(n)$
the set of elements of $\mathcal U_{r,s}(n)$ satisfying the $\gamma_0$
property. So the output of \aps given by Theorem~\ref{th:uni} is an
element of $\mathcal U_{r,s,\gamma_0}(n)$.

Let $\mc T_r(n)$ be the set of toroidal triangulations on $n$ vertices
rooted at an angle that is not in the strict interior of a separating
triangle. Then we have the following bijection:

 \begin{theorem}
\label{th:bij1}
There is a bijection between $\mc T_r(n)$ and
$\mathcal U_{r,s,\gamma_0}(n)$.
\end{theorem}

\begin{proof}
  Consider the mapping $g$ that associates to an element of
  $\mc T_r(n)$, the output of \aps executed on the minimal balanced
  Schnyder wood w.r.t.~the root face. By the above discussion the
  image of $g$ is in $\mathcal U_{r,s,\gamma_0}(n)$ and $g$ is
  injective since one can recover the original triangulation from its
  image by Theorem~\ref{th:recover}.

  Conversely, given an element $U$ of $\mathcal U_{r,s,\gamma_0}(n)$
  with root angle $a_0$, one can build a toroidal map $G$ by the
  complete closure procedure described in Section~\ref{sec:close}. The
  number of stems and edges of $U$ implies that all faces of $G$ are
  triangles.  We explain below why $G$ has no contractible loop nor
  multiple edges.  Recall that $a_0$ defines an orientation on the
  edges and stems of $U$.  Consider the orientation $D$ of $G$ induced
  by this orientation. Since $U$ is safe, the execution of \aps on
  $(G,a_0)$ corresponds to the cycle in the angle graph of $U$
  obtained by starting from the root angle and walking \cw in the face
  of $U$. Thus the output of \aps executed on $(G,a_0)$ is $U$.  First
  note that by definition of $\mathcal U_{r}(n)$, the orientation $D$
  is a $3$-orientation. Thus a simple counting argument shows that $G$
  has no contractible loop nor multiple edges.  It remains to show
  that $G$ is appropriately rooted and that $D$ corresponds to the
  minimal balanced Schnyder wood w.r.t.~this root.

  First note that by definition of $\mathcal U_{r}(n)$, the
  orientation $D$ is a $3$-orientation.

  Suppose by contradiction that $a_0$ is in the strict interior of a
  separating triangle. Then, since we are considering a
  $3$-orientation, by Euler's formula, the edges in the interior of
  this triangle and incident to its border are all entering the
  border. So \aps starting  from the strict interior cannot visit the
  vertices on the border of the triangle and outside. Thus the output
  of \aps is not a toroidal unicellular map, a contradiction. So $a_0$
  is not in the strict interior of a separating triangle.

  The $\gamma_0$ property of $U$ implies that $\gamma$ equals zero on
  two cycles of $U$. Hence these two cycles considered in $G$ also
  satisfy $\gamma$ equals $0$ so $D$ is a balanced Schnyder wood by
  Lemma~\ref{lem:typegamma0foranybasis}.
  
  Suppose by contradiction that $D$ is not minimal. We use the
  terminology and notations of Section~\ref{sec:lattice}. Then $D$
  contains a clockwise (non-empty) $0$-homologous oriented subgraph
  w.r.t.~$f_0$. Let $T$ be such a subgraph with
  $\phi(T)=-\sum_{F\in\mc{F}'}\lambda_F\phi(F)$, with
  $\lambda\in\mathbb{N}^{|\mc{F}'|}$.  Let $\lambda_{f_0}=0$, and
  $\lambda_{\max}=\max_{F\in\mathcal{F}}\lambda_F$.  For
  $0\leq i\leq \lambda_{\max}$, let
  $X_{i}=\{F\in\mathcal{F}\,|\,\lambda_F\geq i\}$.  For
  $1\leq i \leq \lambda_{\max}$, let $T_i$ be the oriented subgraph
  such that $\phi(T_i)=-\sum_{F\in X_i}\phi(F)$.  Then we have
  $\phi(T)=\sum_{1\leq i \leq \lambda_{\max}} \phi(T_i)$.  Since $T$
  is an oriented subgraph, we have
  $\phi(T)\in\{-1,0,1\}^{|E(G)|}$. Thus for any edge of ${G}$,
  incident to faces $F_1$ and $F_2$, we have
  $(\lambda_{F_1}-\lambda_{F_2})\in\{-1,0,1\}$.  So, for
  $1\leq i\leq \lambda_{\max}$, the oriented graph $T_i$ is the
  frontier between the faces with $\lambda$ value equal to $i$ and
  $i-1$.  So all the $T_i$ are edge disjoint and are oriented
  subgraphs of $D$.  Since $T$ is non-empty, we have
  $\lambda_{\max}\geq 1$, and $T_1$ is non-empty.  All the edges of
  $T_1$ have a face of $X_1$ on their right and a face of $X_0$ on
  their left.  Since $U$ is an unicellular map, and $T_1$ is a
  (non-empty) $0$-homologous oriented subgraph, at least one edge of
  $T_1$ corresponds to a stem of $U$. Let $s$ be the last stem of $U$
  corresponding to an edge of $T_1$ that is reattached by the complete
  closure procedure. Consider the step where $s$ is reattached.  As
  the root angle (and thus $f_0$) is in the special face (see the
  terminology of Section~\ref{sec:close}), the special face is in the
  region defined by $X_0$. Thus it is on the left of $s$ when it is
  reattached. This contradicts the fact that $U$ is safe.  Thus $D$ is
  the minimal balanced Schnyder wood w.r.t.~$f_0$.
\end{proof}

\section{Bijection with unrooted unicellular maps}
\label{sec:bij2}

The goal of this section is to remove the root and the safe property
of the unicellular maps considered in Theorem~\ref{th:bij1}. For this
purpose, we have to root the toroidal triangulation more precisely
than before. We say that an angle is not \emph{in the \cw interior of
  a separating triangle} if it is not in its interior, or if it is
incident to a vertex $v$ of the triangle and situated just before an
edge of the triangle in \ccw order around $v$ (see
Figure~\ref{fig:rootangles}).

\begin{figure}[h!]
\center
\includegraphics[scale=0.5]{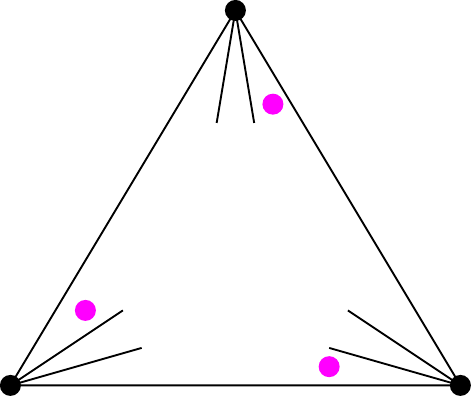}
\caption{Angles that are in a separating triangle but not in its \cw
  interior.}
\label{fig:rootangles}
\end{figure}

Consider a toroidal triangulation $G$. Consider a root angle $a_0$
that is not in the \cw interior of a separating triangle.  Note that
the choice of $a_0$ is equivalent to the choice of a root vertex $v_0$
and a root edge $e_0$ incident to $v_0$ such that none is in the
interior of a separating triangle.  Consider the orientation of the
edges of $G$ corresponding to the minimal balanced Schnyder wood w.r.t.~the
root face $f_0$.  By Lemma~\ref{lem:trianglef0}, there is a \cw
triangle containing $f_0$. Thus by the choice of $a_0$, the edge $e_0$
is leaving the root vertex $v_0$. This is the essential property used
in this section.  Consider the output $U$ of \aps on $(G,a_0)$.  Since
$e_0$ is leaving $v_0$ and $a_0$ is just before $e_0$ in \ccw order
around $v_0$, the execution of \aps starts by Case 2 and $e_0$
corresponds in $U$ to a stem $s_0$ attached to $v_0$. We call this
stem $s_0$ the \emph{root stem}.

The recovering method defined in Theorem~\ref{th:recover} says that
$s_0$ is the last stem reattached by the procedure. So there exists a
sequence of admissible triples of $U$ (see the terminology and
notations of Section~\ref{sec:close}) such that $s_0$ belongs to the
last admissible triple. Let $U_0=U$ and for $1\leq k \leq 2n-2$, the
map $U_{k}$ is obtained from $U_{k-1}$ by closing any admissible
triple that does no contain $s_0$. As noted in
Section~\ref{sec:close}, the special face of $U_{2n-2}$ is a
quadrangle with exactly one stem. This stem being $s_0$, we are in the
situation of Figure~\ref{fig:rootstem}.

\begin{figure}[h!]
\center
\includegraphics[scale=0.4]{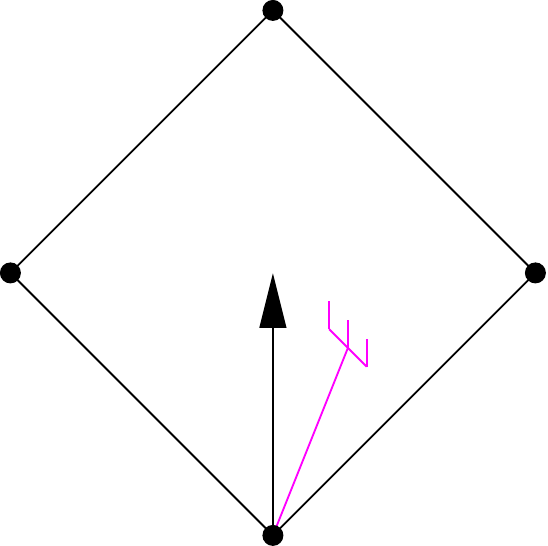} 
\caption{The situation just before the last stem (i.e. the root stem)
  is reattached}
\label{fig:rootstem}
\end{figure}

Consequently, if one removes the root stem $s_0$ from $U$ to obtain an
unicellular map $U'$ with $n$ vertices, $n+1$ edges and $2n-2$ stems,
one can recover the graph $U_{2n-2}$ by applying a complete closure
procedure on $U'$ (see example of Figure~\ref{fig:uprime}). Note that
then, there are four different ways to finish the closure of
$U_{2n-2}$ to obtain an oriented toroidal triangulation. This four
cases correspond to the four ways to place the (removed) root stem in
a quadrangle, they are obtained by pivoting Figure~\ref{fig:rootstem}
by 0°, 90°, 180° and 270°. Note that only one of this four cases leads
to the original rooted triangulation $G$, except if there are some
symmetries (like in the example of Figure~\ref{fig:uprime}).

\begin{figure}[h!]
\center
\begin{tabular}{cc}
\includegraphics[scale=0.5]{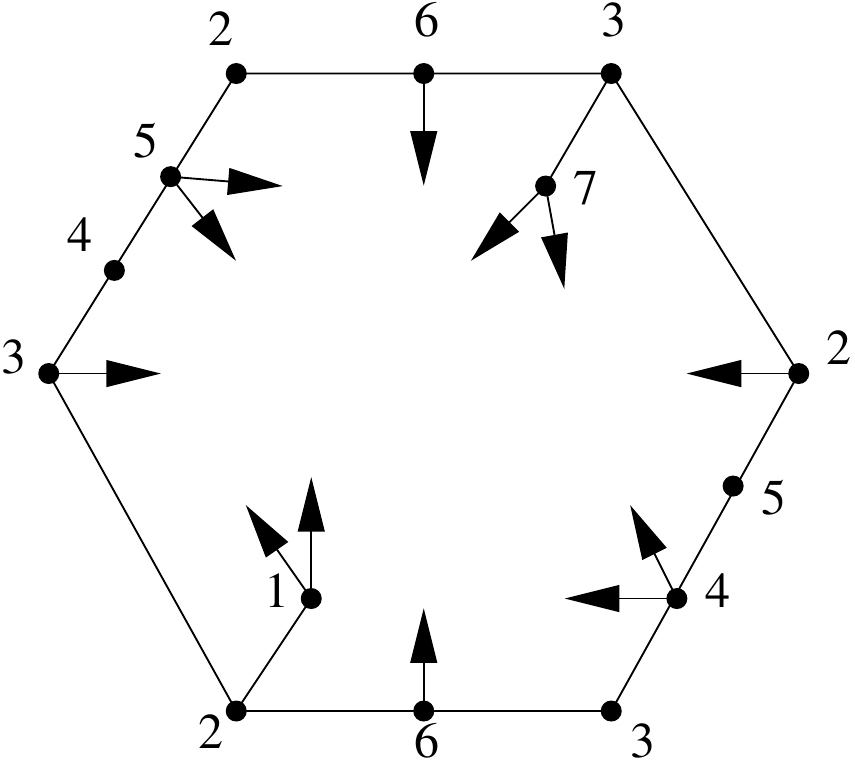}
 & 
\includegraphics[scale=0.5]{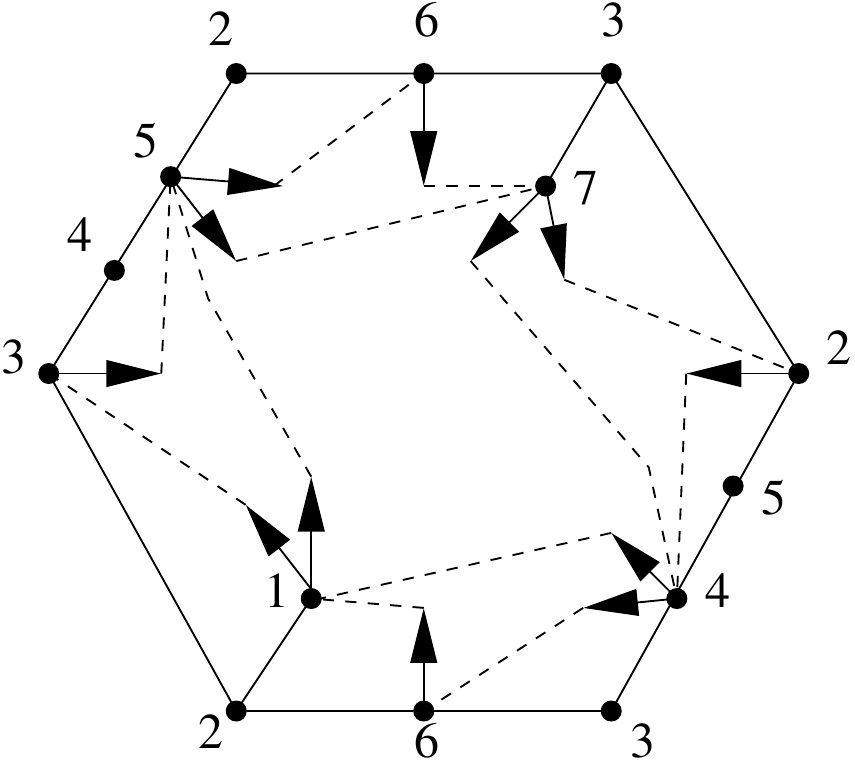}\\
Non-rooted output & Partial recovering \\
\end{tabular}
\caption{Example of $K_7$ where the root angle, the root stem and the
  orientation w.r.t.~the root angle have been removed from the output
  of Figure~\ref{fig:tore-example}. The complete closure procedure
  leads to a quadrangular face.}
\label{fig:uprime}
\end{figure}

Let $\mathcal U(n)$ denote the set of (non-rooted) toroidal
unicellular maps, with exactly $n$ vertices, $n+1$ edges and $2n-2$
stems satisfying the following: a vertex has exactly $2$ stems if it
is not a corner, $1$ stem if it is the corner of a hexagon and $0$
stem if it is the corner of a square. Note that the output of
Theorem~\ref{th:uni} on an appropriately rooted toroidal triangulation
is an element of $\mathcal U(n)$ when the root stem is removed.

Note that an element $U'$ of $\mathcal U(n)$ is non-rooted so we
cannot orient automatically its edges w.r.t.~the root angle like in
Section~\ref{sec:bijection}.  Nevertheless one can still orient all
the stems as outgoing and compute $\gamma$ on the cycles of $U'$ by
considering only its stems in the counting (and not the edges nor the
root stem anymore).  We say that an unicellular map of $\mathcal U(n)$
satisfies the $\gamma_0$ property if $\gamma$ equals zero on its
(non-contractible) cycles.  Let us call $\mathcal U_{\gamma_0}(n)$ the
set of elements of $\mathcal U(n)$ satisfying the $\gamma_0$ property.

A surprising property is that an element $U'$ of $\mathcal U(n)$
satisfies the $\gamma_0$ property if and only if any element $U$ of
$\mathcal U_{r}(n)$ obtained from $U'$ by adding a root stem anywhere
in $U'$ satisfies the $\gamma_0$ property (note that in $U$ we count
the edges and the root stem to compute $\gamma$).  One can see this by
considering the unicellular map of
Figure~\ref{fig:hexasquaregamma}. It represents the general case of
the underlying rooted hexagon of $U$. The edges represent in fact
paths (some of which can be of length zero).  One can check that it
satisfies $\gamma$ equals zero on its (non-contractible) cycles.  It
corresponds exactly to the set of edges that are taken into
consideration when computing $\gamma$ on $U$ but not when computing
$\gamma$ on $U'$. Thus it does not affect the counting (the tree-like
parts are not represented since they do not affect the value
$\gamma$). So the output of Theorem~\ref{th:uni} on an appropriately
rooted toroidal triangulation is an element of
$\mathcal U_{\gamma_0}(n)$ when the root stem is removed.

\begin{figure}[!h]
\center
\includegraphics[scale=0.5]{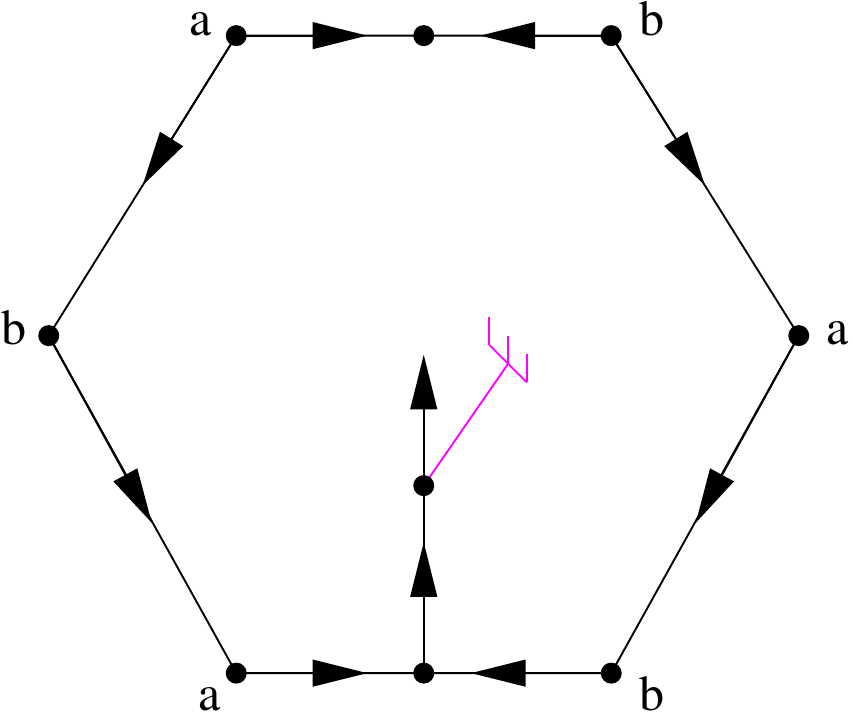}
\caption{The parts of the unicellular map showing the correspondence
  while computing $\gamma$ with or without the orientation w.r.t.~the root plus the root stem.}
\label{fig:hexasquaregamma}
\end{figure}

For the particular case of $K_7$, the difference between the rooted
output of Figure~\ref{fig:tore-example} and the non-rooted output of
Figure~\ref{fig:uprime} is represented on Figure~\ref{fig:uprimediff}
(one can superimpose the last two to obtain the first). One can check
that these three unicellular maps (rooted, non-rooted and the
difference) all satisfy $\gamma$ equals zero on their cycles.

\begin{figure}[h!]
\center
\includegraphics[scale=0.5]{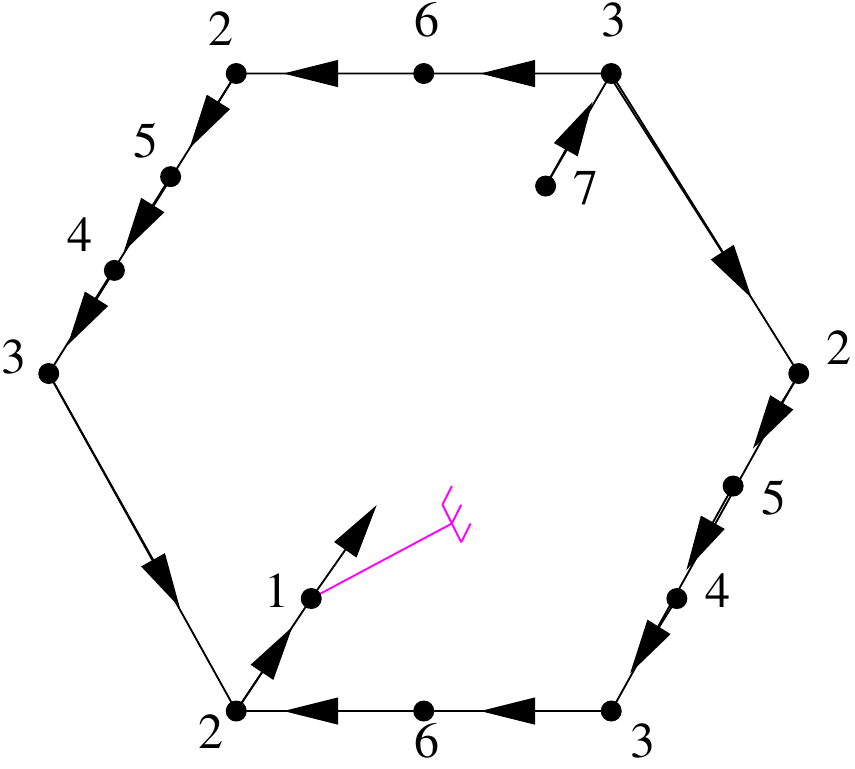}\\
\caption{The difference between the rooted output of
  Figure~\ref{fig:tore-example} and the non-rooted output of
  Figure~\ref{fig:uprime}.}
\label{fig:uprimediff}
\end{figure}

There is an ``almost'' four-to-one correspondence between toroidal
triangulations on $n$ vertices, given with a root angle that is not in
the \cw interior of a separating triangle, and elements of
$\mathcal U_{\gamma_0}(n)$. The ``almost'' means that if the
automorphism group of an element $U$ of $\mathcal U_{\gamma_0}(n)$ is
not trivial, some of the four ways to add a root stem in $U$ are
isomorphic and lead to the same rooted triangulation. In the example
of Figure~\ref{fig:uprime}, one can root in four ways the quadrangle
but this gives only two different rooted triangulations (because of
the symmetries of $K_7$). We face this problem by defining another
class for which we can formulate a bijection.

Let $\mc T(n)$ be the set of toroidal maps on $n$ vertices, where all
the faces are triangles, except one that is a quadrangle and which is
not in a separating triangle. Then we have the following bijection:

 \begin{theorem}
\label{th:bij2}
There is a bijection between $\mc T(n)$ and
$\mathcal U_{\gamma_0}(n)$.
\end{theorem}

\begin{proof}
  Let $a$ (for ``add'') be an arbitrarily chosen mapping defined on
  the maps $G'$ of $\mc T(n)$ that adds a diagonal $e_0$ in the
  quadrangle of $G'$ and roots the obtained toroidal triangulation $G$
  at a vertex $v_0$ incident to $e_0$ (this defines the root angle
  $a_0$ situated just before $e_0$ in \ccw order around $v_0$).  Note
  that the added edge cannot create homotopic multiple edges, since
  otherwise the quadrangle would be in a separating triangle. Moreover
  the root angle of $G$ is not in the \cw interior of a separating
  triangle. Thus the image of $a$ is in $\mc T'_r(n)$, the subset of
  $\mc T_r(n)$ corresponding to toroidal triangulations rooted at an
  angle that is not in the \cw interior of a separating triangle.

  Let $\mathcal U'_{r,s,\gamma_0}(n)$ be the elements of
  $\mathcal U_{r,s,\gamma_0}(n)$ that have their root angle just
  before a stem in \ccw order around the root vertex.  Consider the
  mapping $g$, defined in the proof of Theorem~\ref{th:bij1}.
  By Lemma~\ref{lem:trianglef0} and Theorem~\ref{th:bij1}, the image of $g$
  restricted to $\mc T'_r(n)$ is in $\mathcal U'_{r,s,\gamma_0}(n)$.
  Let $r$ (for ``remove'') be the mapping that associates to an
  element of $\mathcal U'_{r,s,\gamma_0}(n)$ an element of
  $\mathcal U_{\gamma_0}(n)$ obtained by removing the root angle and
  its corresponding stem.  Finally, let $h=r\circ g\circ a$ which
  associates to an element of $\mc T(n)$ an element of
  $\mathcal U_{\gamma_0}(n)$.  Let us show that $h$ is a bijection.

  Consider an element $G'$ of $\mc T(n)$ and its image $U'$ by
  $h$. The complete closure procedure on $U'$ gives $G'$ thus the
  mapping $h$ is injective.

  Conversely, consider an element $U'$ of $\mathcal U_{\gamma_0}(n)$.
  Apply the complete closure procedure on $U'$. At the end of this
  procedure, the special face is a quadrangle whose angles are denoted
  $\alpha^1, \ldots, \alpha^4$. We denote also by
  $\alpha^1,\ldots, \alpha^4$ the corresponding angles of $U'$.  For
  $i\in\{1,\ldots,4\}$, let $U^i$ be the element of
  $\mathcal U_{r}(n)$ obtained by adding a root stem and a root angle
  in the angle $\alpha^i$ of $U'$, with the root angle just before the
  stem in \ccw order around the root vertex.  Note that by the choice
  of $\alpha^i$, the $U^i$ are all safe.  By above remarks they
  also satisfy the $\gamma_0$ property and thus they are in
  $\mathcal U'_{r,s,\gamma_0}(n)$.

  By the proof of Theorem~\ref{th:bij1}, the complete closure
  procedure on $U^i$ gives a triangulation $G^i$ of $\mc T_r(n)$ that
  is rooted from an angle $a_0^i$ not in the strict interior of a
  separating triangle and oriented according to the minimal balanced
  Schnyder wood w.r.t.~the root face. Moreover the output of \aps
  applied on $(G^i,a_0^i)$ is $U^i$. Since in $U^i$, the root stem is
  present just after the root angle, the first edge seen by the
  execution of \aps on $(G^i,a_0^i)$ is outgoing. So $a_0$ is not in
  the \cw interior of a separating triangle (in a $3$-orientation, all
  the edges that are in the interior of a separating triangle and
  incident to the triangle are entering the triangle). Thus the $G^i$
  are appropriately rooted and are elements of $\mc T'_r(n)$. Removing
  the root edge of any $G^i$, gives the same map $G'$ of $\mc T(n)$.
  Exactly one of the $G^i$ is the image of $G'$ by the mapping
  $a$. Thus the image of $G'$ by $h$ is $U'$ and the mapping $h$ is
  surjective.
\end{proof}

A nice aspect of Theorem~\ref{th:bij2} comparing to
Theorem~\ref{th:bij1} is that the unicellular maps that are considered
are much simpler. They have no root nor safe property anymore.

\section{Conjecture for higher genus}
\label{sec:highergenus}

Note that the work presented here is related to a work of Bernardi and
Chapuy~\cite{BC11} (their convention for the orientation of the edges
is the reverse of ours).  Consider a map $G$ (not necessarily a
triangulation) on an orientable surface of genus $g$, rooted at a
particular angle $a_0$.  An orientation of $G$ is \emph{right} if for
each edge $e$, the \emph{right-walk} starting from $e$ (when entering
a vertex, the next chosen edge is the one leaving on the right)
reaches the root edge $e_0$ via the root vertex $v_0$.
A consequence of~\cite{BC11} is that \aps applied on an orientation of
$(G,a_0)$ outputs a spanning unicellular submap $U$ if and only if the
considered orientation is right. Note that in this characterization,
the submap $U$ is not necessarily a map of genus $g$, its genus can be
any value in $\{0,\ldots,g\}$.  In the particular case of toroidal
triangulations we show that by considering minimal balanced Schnyder woods
the output $U$ is a toroidal spanning unicellular map. Hence by the
above characterization, minimal balanced Schnyder woods are right.  But
here, the fact that $U$ and $G$ have the same genus is of particular
interest as it yields a simple bijection.

The key property that makes $U$ and $G$ have same genus is the
conclusion of Lemma~\ref{lem:incomingedges} (no oriented
non-contractible cycle in the dual orientation).  Thus we we hope for
a possible generalization of Theorem~\ref{th:AGK} satisfying the
conclusion of Lemma~\ref{lem:incomingedges}:

\begin{conjecture}
\label{conj:accessibility}
A triangulation on a genus $g\geq 1$ orientable surface admits an
orientation of its edges such that every vertex has outdegree at least
$3$, divisible by $3$, and such that there is no oriented
non-contractible cycle in the dual orientation.
\end{conjecture}
 
If Conjecture~\ref{conj:accessibility} is true, one can consider a
minimal orientation satisfying its conclusion and apply \aps to obtain
a unicellular map of the same genus as $G$.  Note that more efforts
should be made to obtain a bijection since there might be several
minimal elements satisfying the conjecture and a particular one has to be
identified (as the minimal balanced Schnyder wood in our case).

\chapter{Orthogonal surfaces and straight-line drawing}
\label{sec:ortho}

\section{Periodic geodesic embedding}

 Generalizing results for planar maps \cite{BFM07,Fel03,Mil02}, we
 show that Schnyder woods can be used to embed the universal cover of
 an essentially 3-connected toroidal map on an infinite and periodic
 orthogonal surface. This embedding is then used to obtain a
 straight-line flat torus representation of any toroidal map in a
 polynomial size grid.

Given two points $u=(u_0,u_1,u_2)$ and $v=(v_0,v_1,v_2)$ in $\mathbb
R^3$, we note $u\vee v=(\max(u_i,v_i))_{i=0,1,2}$ and $u\wedge
v=(\min(u_i,v_i))_{i=0,1,2}$. We define an order $\geq$ among the
points in $\mathbb R^3$, in such a way that $u\geq v$ if $u_i\geq v_i$
for $i=0,1,2$.

Given a set $\mv$ of pairwise incomparable elements in $\mr^3$, we
define the set of vertices that dominates $\mv$ as
$\md_\mv=\{u\in\mathbb R^3\ |\ \exists\, v\in \mathcal V$ such that
$u\geq v\}$. The \emph{orthogonal surface $\ms$} generated by $\mv$ is
the boundary of $\md_{\mv}$. (Note that orthogonal surfaces are well
defined even when $\mv$ is an infinite set.)  If $u,v\in \mv$ and
$u\vee v\in \ms$, then $\ms$ contains the union of the two line
segments joining $u$ and $v$ to $u\vee v$. Such arcs are called
\emph{elbow geodesic}. The \emph{orthogonal arc} of $v\in \mv$ in the
direction of the standard basis vector $e_i$ is the intersection of
the ray $v+\lambda e_i$ with $\ms$.

Let $G$ be a planar map. A \emph{geodesic embedding} of $G$ on the
orthogonal surface $\ms$ is a drawing of $G$ on $\ms$ satisfying the
following:

\begin{itemize}
\item[(D1)] There is a bijection between the vertices of $G$ and $\mv$.

\item[(D2)] Every edge of $G$ is an elbow geodesic.

\item[(D3)] Every 
orthogonal arc in $\ms$ is part of an edge of $G$.

\item[(D4)] There are no crossing edges in the embedding of $G$ on $\ms$.
\end{itemize}

Miller~\cite{Mil02} (see also \cite{Fel03, FZ08}) proved that a
geodesic embedding of a planar map $G$ on an orthogonal surface $\ms$
induces a planar Schnyder wood of $G$. The edges of $G$ are colored with the
direction of the orthogonal arc contained in the edge. An orthogonal
arc intersecting the ray $v+\lambda e_i$ corresponds to the edge
leaving $v$ in color $i$. Edges represented by two orthogonal arcs
correspond to edges oriented in two directions.

Conversely, it has been proved that a Schnyder wood of a planar map
$G$ can be used to obtain a geodesic embedding of $G$.  Let $G$ be a
planar map given with a Schnyder wood.  The method is the following
(see \cite{Fel03} for more details): For every vertex $v$, one can
divide $G$ into the three regions bounded by the three monochromatic
paths going out from $v$ (see Figure~\ref{fig:regions-planar}). The \emph{region vector} associated to $v$ is
the vector obtained by counting the number of faces in each of these
three regions.  The mapping of each vertex on its region vector gives
the geodesic embedding. Note that in this approach, the vertices are
all mapped on the same plane as the sum of the coordinates of each
region vector is equal to the total number of inner faces of the map.

Our goal is to generalize geodesic embeddings to the torus.  More
precisely, we want to represent the universal cover of a toroidal map
on an infinite and periodic orthogonal surface.

Let $G$ be a toroidal map.  Consider a representation of $G$ in a
parallelogram $P$ whose opposite sides are identified.  The graph
$G^\infty$ is obtained by replicating $P$ to tile the plane.  Given
any parallelogram $Q$ representing a copy of $G$ in $G^\infty$, we
choose arbitrarily two consecutive sides of $Q$ in \cw order and say
that the first one is the top side and the second one is the right
side.  Let $Q^{top}$ (resp. $Q^{right}$) be the copy of $Q$ just above
(resp. on the right of) $Q$, w.r.t.~the top and right side of $Q$.
Given a vertex $v$ in $Q$, we note $v^{top}$ (resp. $v^{right}$) its
copies in $Q^{top}$ (resp. $Q^{right}$).

A mapping of the vertices of $G^\infty$ in $\mr^d$, $d\in\{2,3\}$, is
\emph{periodic} with respect to vectors $S$ and $S'$ of $\mr^d$, if
there exists a parallelogram $Q$ representing a copy of $G$ in
$G^\infty$ such that for any vertex $v$ of $G^\infty$, vertex
$v^{top}$ is mapped on $v+S$ and $v^{right}$ is mapped on $v+S'$.  A
\emph{geodesic embedding} of a toroidal map $G$ is a geodesic
embedding of $G^\infty$ on $\mathcal S_{\mathcal V^\infty}$, where
$\mathcal V^\infty$ is a periodic mapping of $G^\infty$ w.r.t.~two non
collinear vectors.

The goal of the next two sections is to prove the following theorem
 (see example of Figure~\ref{fig:example-primal},
where the black parallelogram represents a copy of the map that is the
basic tile used to fill the plane):

\begin{theorem}
\label{th:triortho}
The universal cover of an essentially 3-connected toroidal map admits
a geodesic embedding on an infinite and periodic orthogonal surface.
\end{theorem}

\begin{figure}[h!]
\center
\includegraphics[scale=0.2]{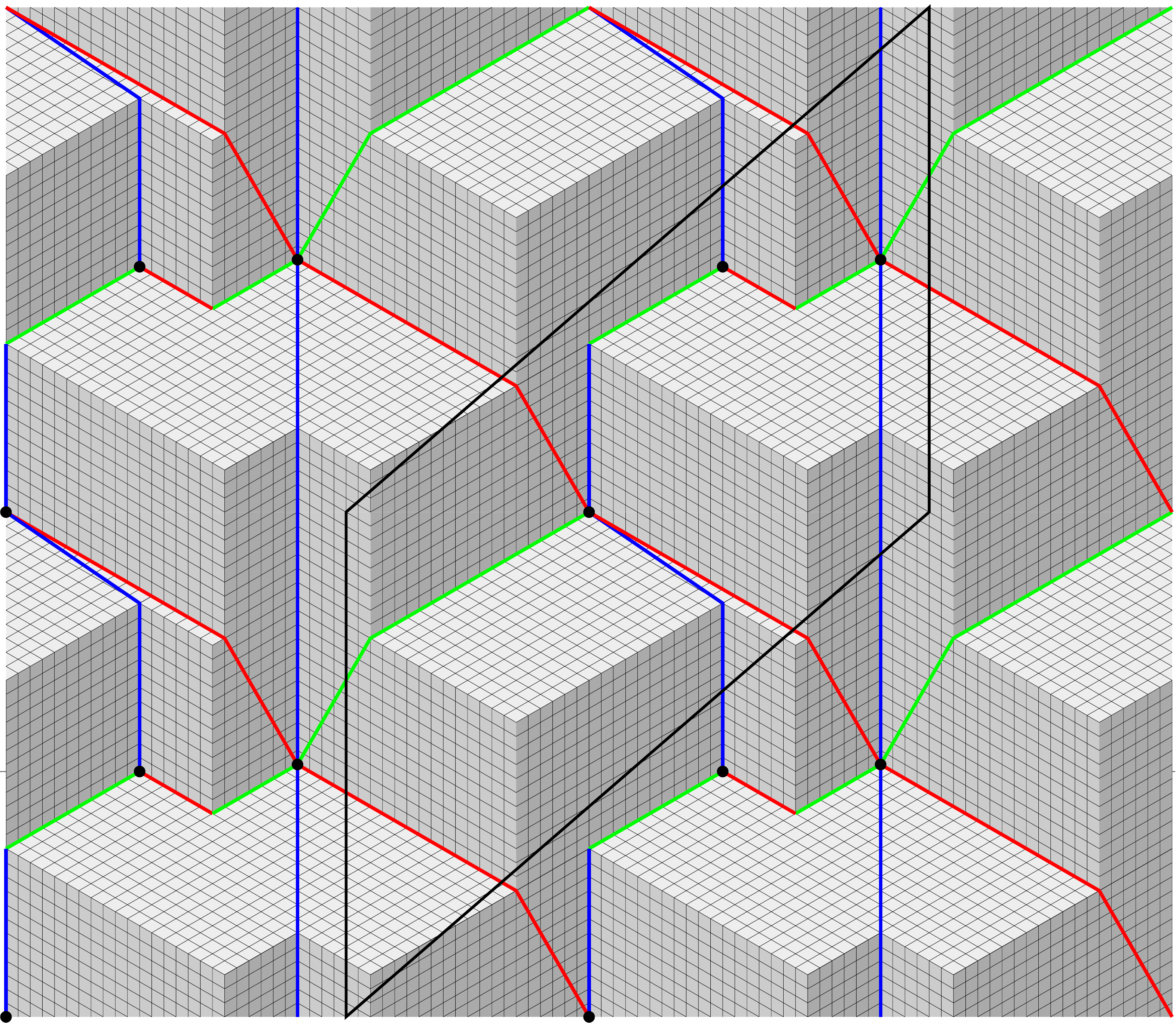}
\caption{Geodesic embedding of the toroidal map of
  Figure~\ref{fig:example-dual-tore}.}
\label{fig:example-primal}
\end{figure}

\section{Generalization of the region vector}
\label{sec:regionvector}

Like in the plane, Schnyder woods can be used to obtain geodesic
embeddings of toroidal maps. For that purpose, we need to generalize
the region vector method.  The idea is to use the regions to compute
the coordinates of the vertices of the universal cover. The problem is
that contrarily to the planar case, these regions are unbounded and
contains an infinite number of faces. The method is thus generalized
by the following.

Consider a  toroidal map $G$ given with an
intersecting Schnyder wood.  We say that the Schnyder wood is of
\emph{type 1} if it is crossing and \emph{type 2} if it is not
crossing. When it is not crossing, we may specify the only color $i$
for which it is $i$-crossing (see Theorem~\ref{lem:type-cross}) by
saying that the Schnyder wood is of \emph{type 2.i}.

Recall that $\mathcal C_i$ denotes the set of $i$-cycles of $G$.  The
Schnyder wood is intersecting so, by Theorem~\ref{lem:type-cross},
the cycles of $\mathcal C_i$ are non-contractible, non intersecting
and homologous.  So we can order them as follows
$\mathcal C_i=\{C_i^0,\ldots,C_i^{k_i-1}\}$, $k_i\geq 1$, such that,
for $0\leq j\leq k_i-1$, there is no $i$-cycle in the region
$R(C_i^j,C_i^{j+1})$ between $C_i^j$ and $C_i^{j+1}$ containing the
right side of $C_i^j$ (superscript understood modulo $k_i$).  Recall
that if the Schnyder wood is of type 2.i, then
$\mathcal C_{i-1}=(\mathcal C_{i+1})^{-1}$ by Theorem~\ref{lem:type-cross}.

A directed monochromatic cycle $C_i^j$ corresponds to a family of
$i$-lines of $G^\infty$, denoted $\mathcal L_i^j$.  Given any two
$i$-lines $L$, $L'$, the unbounded region between $L$ and $L'$ is
noted $R(L,L')$.  We say that two $i$-lines $L,L'$ are
\emph{consecutive} if there is no $i$-line contained in $R(L,L')$.
The \emph{positive side} of a $i$-line is defined as the right side
while ``walking'' along the directed path by following the orientation
of the edges colored $i$.

\begin{lemma}
\label{lem:linesintersection}
For any vertex $v$, the two monochromatic lines $L_{i-1}(v)$ and
$L_{i+1}(v)$ intersect. Moreover, if the Schnyder wood is of type
2.i, then $L_{i+1}(v)=(L_{i-1}(v))^{-1}$ and $v$ is situated on the
right of $L_{i+1}(v)$.
\end{lemma}

\begin{proof}
  Let $j,j'$ be such that $L_{i-1}(v)\in \mathcal L_{i-1}^{j}$ and
  $L_{i+1}(v)\in \mathcal L_{i+1}^{j'}$.  If the Schnyder wood is of
  type 1 or type 2.j with $j\neq i$, then the two cycles $C_{i-1}^{j}$
  and $C_{i+1}^{j'}$ are crossing, and so the two lines
  $L_{i-1}(v)$ and $L_{i+1}(v)$ intersect.  

  Suppose now that the Schnyder wood is of type 2.i.  If
  $v\in L_{i-1}(v)$, then clearly $L_{i+1}(v)=(L_{i-1}(v))^{-1}$ and
  $v\in L_{i+1}(v)$. The case where $v\in L_{i+1}(v)$ is similar and
  we can assume now that $v$ does not belong to $L_{i-1}(v)$ nor
  $L_{i+1}(v)$.  Then $v$ lies between two consecutive $(i+1)$-lines
  (which are also $(i-1)$-lines). Let us denote those two lines
  $L_{i+1}$ and $L'_{i+1}$, such that $L'_{i+1}$ is situated on the
  right of $L_{i+1}$. By the Schnyder property, paths $P_{i+1}(v)$ and
  $P_{i-1}(v)$ cannot reach $L'_{i+1}$. Thus
  $L_{i+1}=L_{i+1}(v)=(L_{i-1}(v))^{-1}$.
\end{proof}

The \emph{size} of the region $R(C_i^j,C_i^{j+1})$ of $G$, denoted
$f_i^j=|R(C_i^j,C_i^{j+1})|$, is equal to the number of faces in
$R(C_i^j,C_i^{j+1})$.  Remark that for each color, we have
$\sum_{j=0}^{k_i-1} f_i^j$ equals the total number of faces $f$ of
$G$.  If $L$ and $L'$ are consecutive $i$-lines of $G^\infty$ with
$L\in \mathcal L_i^j$ and $L'\in \mathcal L_i^{j+1}$, then the
\emph{size} of the (unbounded) region $R(L,L')$, denoted $|R(L,L')|$,
is equal to $f_i^j$. If $L$ and $L'$ are any $i$-lines, the
\emph{size} of the (unbounded) region $R(L,L')$, denoted $|R(L,L')|$,
is equal to the sum of the size of all the regions delimited by
consecutive $i$-lines inside $R(L,L')$.  For each color $i$, choose
arbitrarily a $i$-line $L_i^*$ in $\mathcal L_i^0$ that is used as an
origin for $i$-lines.  Given a $i$-line $L$, we define the value
$f_i(L)$ of $L$ as follows: $f_i(L)=|R(L,L_i^*)|$ if $L$ is on the
positive side (right side) of $L_i^*$ and $f_i(L)=-|R(L,L_i^*)|$ otherwise.

Consider two vertices $u,v$ such that $L_{i-1}(u)=L_{i-1}(v)$ and
$L_{i+1}(u)=L_{i+1}(v)$. Even if the two regions $R_{i}(u)$ and
$R_{i}(v)$ are unbounded, their \emph{difference} is bounded.  Let
$d_i(u,v)$ be the number of faces in $R_{i}(u)\setminus R_{i}(v)$
minus the number of faces in $R_{i}(v)\setminus R_{i}(u)$.  For any
vertex, by Lemma~\ref{lem:linesintersection}, there exists $z_i(v)$ a
vertex on the intersection of the two lines $L_{i-1}(v)$ and
$L_{i+1}(v)$.  Let $N$ be a constant $\geq n$ (in this chapter we can
have $N=n$ but in chapter~\ref{sec:straight} we need to choose $N$
bigger).  We are now able to define the region vector of a vertex of
$G^\infty$, that is a mapping of this vertex in $\mathbb R^3$ (see
Figure~\ref{fig:region}).

\begin{definition}[Region vector]
\label{def:regionvector}
  The $i$-th coordinate of the \emph{region vector} of a vertex $v$ of
  $G^\infty$ is equal to $v_i\ =\ d_i(v,z_i(v))\ +\
  N\times\big(f_{i+1}(L_{i+1}(v))-f_{i-1}(L_{i-1}(v))\big)$.
\end{definition}

\begin{figure}[!h]
\center
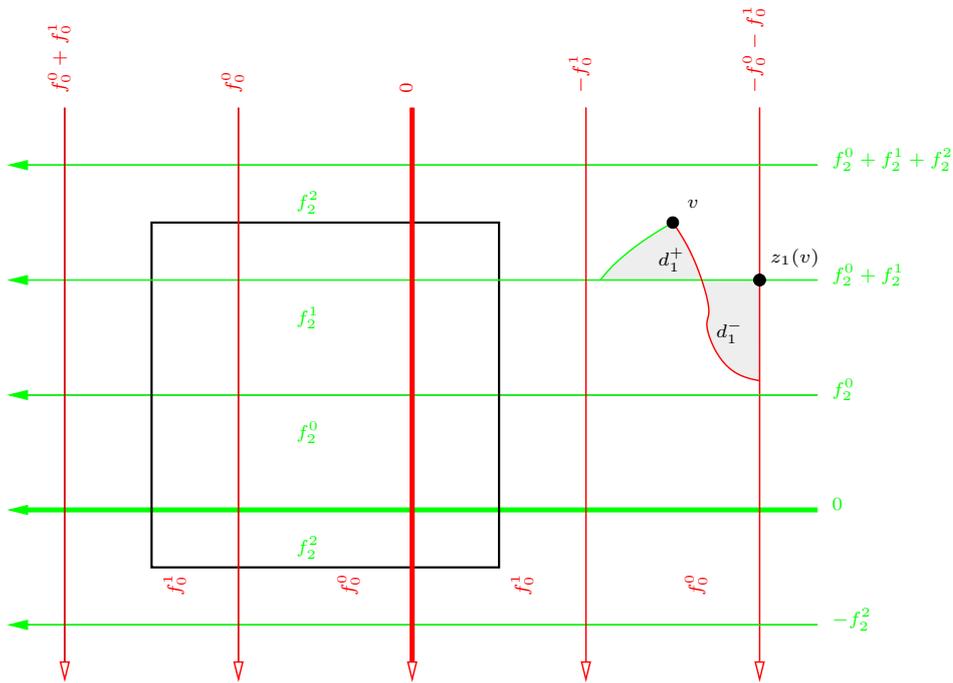
\caption{The toroidal map is represented in a black square that is
  used to tile the plane. The bold red and green lines are used as an
  origin for 0- and 2-lines.  Coordinate $1$ of vertex $v$, is equal
  to the number of faces in the region $d_1^+$, minus the number of
  faces in the region $d_1^-$, plus $N$ times
  $(f_2^0+f_2^1)-(-f_0^0-f_0^1)$.}
\label{fig:region}
\end{figure}

Note that Lemma~\ref{lem:linesintersection} is not necessarily true
for half-crossing Schnyder woods. For example, one can consider the
half-crossing Schnyder wood of Figure~\ref{fig:cross-0}, for which
there is no red lines intersecting green lines in the universal
cover. Thus our definition of region vector
(definition~\ref{def:regionvector}) is valid only in the case of
intersecting Schnyder woods. This explains why we are only
considering intersecting Schnyder woods in this chapter.

\section{Region vector and geodesic embedding}

Consider a toroidal map $G$ given with an intersecting Schnyder wood,
a representation of $G$ in a parallelogram $P$ and the mapping of each
vertex of $G^\infty$ on its region vector (see
definition~\ref{def:regionvector}). We show in this section that this
gives a geodesic embedding of $G$.

\begin{lemma}
\label{lem:sum}
The sum of the coordinates of a vertex $v$ equals the number of faces
in the bounded region delimited by the lines $L_0(v)$, $L_1(v)$ and
$L_2(v)$ if the Schnyder wood is of type 1 and this sum equals zero if
the Schnyder wood is of type 2.
\end{lemma}

\begin{proof}
  We have $v_0+v_1+v_2 = d_0(v,z_0(v))+ d_1(v,z_1(v))+ d_2(v,z_2(v))=
  \sum_i(|R_i(v) \setminus R_i(z_i(v))|-|R_i(z_i(v)) \setminus
  R_i(v)|) $.  We use the characteristic function $\bf 1$ to deal with
  infinite regions. We note ${\bf 1}(R)$, the function defined on the
  faces of $G^\infty$ that has value $1$ on each face of region $R$
  and $0$ elsewhere. Given a function $g:F(G^\infty)\longrightarrow
  \mathbb Z$, we note $|g|=\sum_{F\in F(G^\infty)} g(F)$ (when the sum
  is finite).  Thus 
$\sum_i v_i
=\sum_i(|\mathbf 1(R_i(v) \setminus
  R_i(z_i(v)))|-|\mathbf 1(R_i(z_i(v)) \setminus R_i(v))|)
=
  |\sum_i(\mathbf 1(R_i(v) \setminus R_i(z_i(v)))-\mathbf
  1(R_i(z_i(v)) \setminus R_i(v)))|$. Now we compute $g=\sum_i(\mathbf
  1(R_i(v) \setminus R_i(z_i(v)))-\mathbf 1(R_i(z_i(v)) \setminus
  R_i(v)))$.  We have:
$$ g =\sum_i(\mathbf 1(R_i(v) \setminus
R_i(z_i(v)))+ \mathbf 1(R_i(v) \cap R_i(z_i(v))) -\mathbf
1(R_i(z_i(v)) \setminus R_i(v)) -\mathbf 1(R_i(v) \cap R_i(z_i(v))))
$$
As $R_i(v) \setminus
R_i(z_i(v))$ and $R_i(z_i(v)) \setminus R_i(v)$
are disjoint from $R_i(v) \cap R_i(z_i(v))$, we have
$$
g= \sum_i(\mathbf 1(R_i(v)) -\mathbf
1(R_i(z_i(v))))=\sum_i\mathbf 1(R_i(v)) -
\sum_i\mathbf 1(R_i(z_i(v)))
$$
The interior of the three regions $R_i(v)$, for $i=0,1,2$, being
disjoint and spanning the whole plane $\mathbb P$ (see
Section~\ref{sec:crossinguniversalcover}), we have
$\sum_i\mathbf 1(R_i(v))= \mathbf 1(\cup_i(R_i(v)))=\mathbf 1(\mathbb
P)$.
Moreover the regions $R_i(z_i(v))$, for $i=0,1,2$, are also disjoint
and
$\sum_i\mathbf 1(R_i(z_i(v)))= \mathbf 1(\cup_i(R_i(z_i(v))))=\mathbf
1(\mathbb P\setminus T)$
where $T$ is the bounded region delimited by the lines $L_0(v)$,
$L_1(v)$ and $L_2(v)$.  So
$ g= \mathbf 1(\mathbb P) - \mathbf 1(\mathbb P\setminus T)= \mathbf
1( T) $.  And thus $\sum_i v_i=|g|=| \mathbf 1( T)|$.
\end{proof}

Lemma~\ref{lem:sum} shows that if the Schnyder wood is of type 1,
then the set of points are not necessarily coplanar like in the planar
case~\cite{Fel-book}, but all the copies of a vertex lies on the
same plane (the bounded region delimited by the lines $L_0(v)$,
$L_1(v)$ and $L_2(v)$ has the same number of faces for any copies of a
vertex $v$). Surprisingly, for Schnyder woods of type 2, all the
points are coplanar.

For each color $i$, let $c_i$ (resp. $c'_i$), be the algebraic number
of times a $i$-cycle is traversing the ``left'' (resp. ``top'')
side of the parallelogram $P$  from right to left (resp. from
bottom to top).  This number increases by one each time a
monochromatic cycle traverses the side in the given direction and
decreases by one when it traverses in the other way.  Let $S$ and $S'$
be the two vectors of $\mr^3$ with coordinates
$S_i=N(c_{i+1}-c_{i-1})f$ and $S'_i=N(c'_{i+1}-c'_{i-1})f$.  Note that
$S_0+S_1+S_2=0$ and $S'_0+S'_1+S'_2=0$

\begin{lemma}
\label{lem:periodic}
The mapping is periodic with respect to $S$ and $S'$.
\end{lemma}

\begin{proof}
  Let $v$ be any vertex of $G^\infty$.  Then
  $v^{top}_i-v_i
  =N(f_{i+1}(L_{i+1}(v^{top}))-f_{i+1}(L_{i+1}(v)))-N(f_{i-1}(L_{i-1}(v^{top}))-f_{i-1}(L_{i-1}(v)))=N(c_{i+1}-c_{i-1})f$.
  So $v^{top}=v+S$. Similarly $v^{right}=v+S'$.
\end{proof}

Recall that, for each color $i$, the $i$-winding number $\omega_i$ is
the integer such that two monochromatic cycles of $G$ of respective
colors $i-1$ and $i+1$ cross exactly $\omega_i$ times.  If the
Schnyder wood is of type 1, then $\omega_{i-1}\neq 0$,
$\omega_i\neq 0 $ and $\omega_{i+1}\neq 0$.  If the Schnyder wood is
of type 2.i, then $\omega_i= 0$ and $\omega_{i-1}=\omega_{i+1}$. Let
$\omega = \max(\omega_0,\omega_1,\omega_2)$. Let
$Z_0=( (\omega_1+\omega_2) Nf, - \omega_1 Nf, - \omega_2 Nf)$ and
$Z_1=(- \omega_0Nf, (\omega_0+\omega_2) Nf, - \omega_2 Nf)$ and
$Z_2=(- \omega_0Nf,- \omega_1 Nf,(\omega_0+\omega_1)Nf)$.

\begin{lemma}
\label{lem:copies}
For any vertex $u$, we have $\{u+k_0Z_0+k_1Z_1+k_2Z_2\, |\,
k_0,k_1,k_2\in \mz\} \subseteq \{u+kS+k'S'\, |\, k,k'\in \mz\}$.
\end{lemma}

\begin{proof}
  Let $u,v$ be two copies of the same vertex, such that $v$ is the
  first copy of $u$ in the direction of $L_0(u)$.  (That is
  $L_0(u)=L_0(v)$ and on the path $P_0(u)\setminus P_0(v)$ there is no
  two copies of the same vertex.)  Then $v_i-u_i
  =N(f_{i+1}(L_{i+1}(v))-f_{i+1}(L_{i+1}(u)))-N(f_{i-1}(L_{i-1}(v))-f_{i-1}(L_{i-1}(u)))$.
  We have $|R(L_{0}(v),L_{0}(u))|=0$, $|R(L_{1}(v),L_{1}(u))|=\omega_2
  f$ and $|R(L_{2}(v),L_{2}(u))|=\omega_1 f$. So $v_0-u_0= N
  (\omega_1+\omega_2) f$ and $v_1-u_1= -N \omega_1 f$ and $v_2-u_2= -
  N \omega_2 f$. So $v=u+Z_0$.  Similarly for the other colors. So the
  first copy of $u$ in the direction of $L_i(u)$ is equal to $u+Z_i$.
  By Lemma~\ref{lem:periodic}, all the copies of $u$ are mapped on
  $\{u+kS+k'S'\, |\, k,k'\in \mz\}$, so we have the result.
\end{proof}

\begin{lemma}
\label{lem:dim}
We have $dim(Z_0,Z_1,Z_2)=2$ and if the Schnyder wood is not of type
2.i, then $dim(Z_{i-1},Z_{i+1})=2$.
\end{lemma}

\begin{proof}
  We have $\omega_0Z_0+\omega_1Z_1+\omega_2Z_2=0$ and so
  $dim(Z_0,Z_1,Z_2)\leq 2$.  We can assume by symmetry that the
  Schnyder wood is not of type 2.1 and so $\omega_1\neq 0$. Thus
  $Z_0\neq 0$ and $Z_2\neq 0$.  Suppose by contradiction that
  $dim(Z_0,Z_2)=1$. Then there exist $\alpha\neq 0$, $\beta\neq 0$,
  such that $\alpha Z_0 + \beta Z_2=0$.  The sum of this equation for
  the coordinates $0$ and $2$ gives $(\alpha+\beta)\omega_1=0$ and
  thus $\alpha = -\beta$. Then the equation for coordinate $0$ gives
  $\omega_0+\omega_1+\omega_2=0$ contradicting the fact that
  $\omega_1>1$ and $\omega_0,\omega_2\geq 0$.
\end{proof}

\begin{lemma}
\label{lem:notcolinear}
The vectors $S,S'$ are not collinear.
\end{lemma}

\begin{proof}
  By Lemma~\ref{lem:copies}, the set $\{u+k_0Z_0+k_1Z_1+k_2Z_2\, |\,
  k_0,k_1,k_2\in \mz\}$ is a subset of $\{u+kS+k'S'\, |\, k,k'\in
  \mz\}$.  By Lemma~\ref{lem:dim}, we have $dim(Z_0,Z_1,Z_2)=2$, thus
  $dim(S,S')=2$.
\end{proof}

\begin{lemma}
  \label{lem:regionslongue}
  If $u,v$ are two distinct vertices such that $v$ is in $L_{i-1}(v)$,
  $u$ is in $P_{i-1}(v)$, both $u$ and $v$ are in the region
  $R(L_{i+1}(u),L_{i+1}(v))$ and $L_{i+1}(u)$ and $L_{i+1}(v)$ are two
  consecutive $(i+1)$-lines with $L_{i+1}(u)\in \mathcal L_{i+1}^{j}$
  (see Figure~\ref{fig:longregion}).  Then
  $d_i(z_i(v),v)+d_i(u,z_i(u))< (n-1)\times f_{i+1}^{j}$.
\end{lemma}

\begin{figure}[!h]
\center
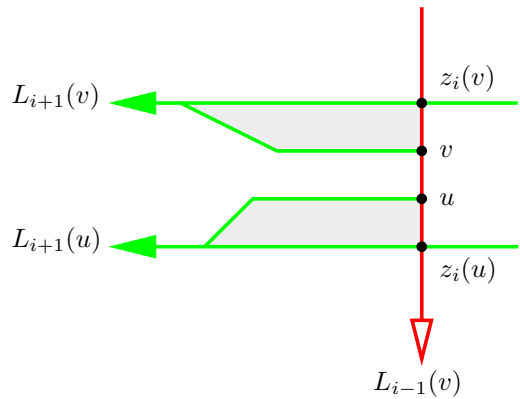
\caption{The gray area, corresponding to the quantity
  $d_i(z_i(v),v)+d_i(u,z_i(u))$, has size bounded by $(n-1)\times
  f_{i+1}^{j}$.
}
\label{fig:longregion}
\end{figure}

\begin{proof}
  Let $Q_{i+1}(v)$ the subpath of $P_{i+1}(v)$ between $v$ and
  $L_{i+1}(v)$ (maybe $Q_{i+1}(v)$ has length $0$ if $v=z_i(v)$).  Let
  $Q_{i+1}(u)$ the subpath of $P_{i+1}(u)$ between $u$ and $L_{i+1}(u)$
  (maybe $Q_{i+1}(u)$ has length $0$ if $u=z_i(u)$).  The path
  $Q_{i+1}(v)$ cannot contain two different copies of a vertex of $G$,
  otherwise $Q_{i+1}(v)$ will corresponds to a non-contractible cycle
  of $G$ and thus will contain an edge of $L_{i+1}(v)$. So the length
  of $Q_{i+1}(v)$ is $\leq n-1$.
  
  The total number of times a copy of a given face of $G$ can appear
  in the region $R=R_i(z_i(v))\setminus R_i(v)$, corresponding to
  $d_i(z_i(v),v)$, can be bounded as follow. Region $R$ is between two
  consecutive copies of $L_{i+1}(u)$. So in $R$, all the copies of a
  given face are separated by a copy of $L_{i-1}(v)$. Each copy of
  $L_{i-1}(v)$ intersecting $R$ has to intersect $Q_{i+1}(v)$ on a
  specific vertex. As $Q_{i+1}(v)$ has at most $n$ vertices. A given
  face can appear at most $n-1$ times in $R$.  Similarly, the total
  number of times that a copy of a given face of $G$ can appear in the
  region $R_i(u)\setminus R_i(z_i(u))$, corresponding to
  $d_i(u,z_i(u))$, is $\leq (n-1)$.

  A given face of $G$ can appear in only one of the two gray regions
  of Figure~\ref{fig:longregion}. So a face is counted $\leq n-1$
  times in the quantity $d_i(z_i(v),v)+d_i(u,z_i(u))$.
  Only the faces of the region $R(C_{i+1}^{j},C_{i+1}^{j+1})$ can be
  counted. And there is at least one face of
  $R(C_{i+1}^{j},C_{i+1}^{j+1})$ (for example one incident to $v$)
  that is not counted. So in total $d_i(z_i(v),v)+d_i(u,z_i(u))\leq
  (n-1)\times (f_{i+1}^{j}-1)< (n-1)\times f_{i+1}^{j}$.
\end{proof}

Clearly, the symmetric of Lemma~\ref{lem:regionslongue},  where
the role of $i+1$ and $i-1$ are exchanged, is also true.

The bound of Lemma~\ref{lem:regionslongue} is somehow sharp. In the
example of Figure~\ref{fig:examplelongue}, the rectangle represents a
toroidal map $G$ and the universal cover is partially represented.  If
the map $G$ has $n$ vertices and $f$ faces ($n=5$ and $f=5$ in the
example), then the gray region, representing the quantity
$d_1(z_1(v),v)+d_1(u,z_1(u))$, has size
$\frac{n(n-1)}{2}=\Omega(n\times f)$.

 \begin{figure}[!h]
 \center
 \includegraphics[scale=0.3]{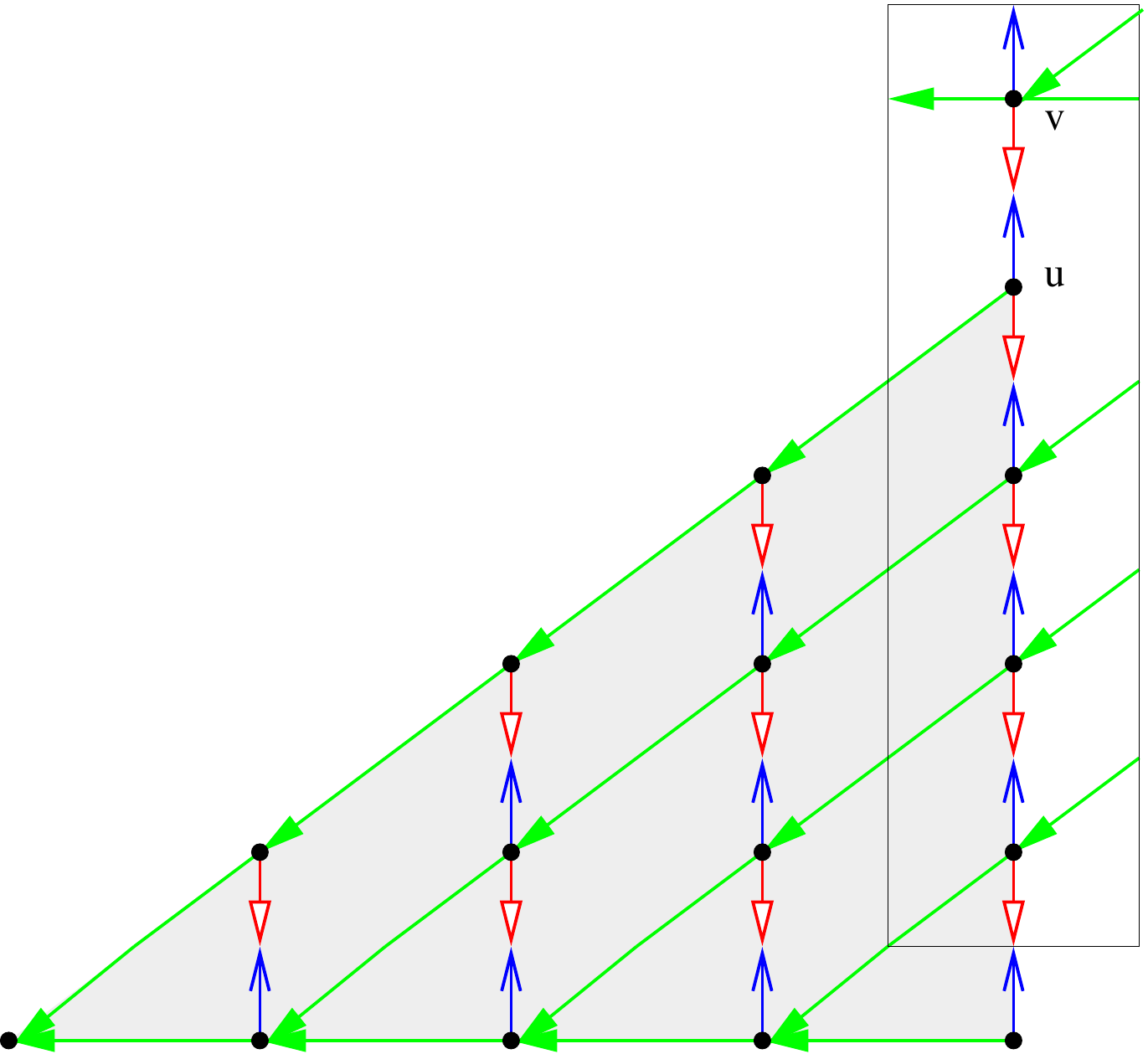}
 \caption{Example of a toroidal map where $d_1(u,z_1(u))$ has size
   $\Omega(n\times f)$.}
 \label{fig:examplelongue}
 \end{figure} 

\begin{lemma}
  \label{lem:regionorder}
  Let $u,v$ be vertices of $G^\infty$ such that $R_i(u)\subseteq
  R_i(v)$, then $u_i\leq v_i$. Moreover if $R_i(u)\subsetneq R_i(v)$,
  then
  $v_i-u_i>(N-n)(|R(L_{i-1}(u),L_{i-1}(v))|+|R(L_{i+1}(u),L_{i+1}(v))|)\geq
  0$.
\end{lemma}

\begin{proof}
  We distinguish two cases depending of the fact that the Schnyder
  wood is of type 2.i or not.

  \noindent
\emph{$\bullet$ Case 1: The Schnyder wood is not of type 2.i.}

  Suppose first that $u$ and $v$ are both in a region delimited by two
  consecutive lines of color $i-1$ and two consecutive lines of color
  $i+1$.  Let $L_{i-1}^j,L_{i-1}^{j+1}$, $L_{i+1}^{j'},L_{i+1}^{j'+1}$
  be these lines such that $L_{i-1}^{j+1}$ is on the positive side of
  $L_{i-1}^j$, $L_{i+1}^{j'+1}$ is on the positive side of
  $L_{i+1}^{j'}$, and $L_k^\ell \in \mathcal L_k^\ell$ (see
  Figure~\ref{fig:regionrectangle}).
  We distinguish cases corresponding to equality or not between lines
  $L_{i-1}(u)$, $L_{i-1}(v)$ and $L_{i+1}(u)$, $L_{i+1}(v)$.

\begin{figure}[!h]
\center
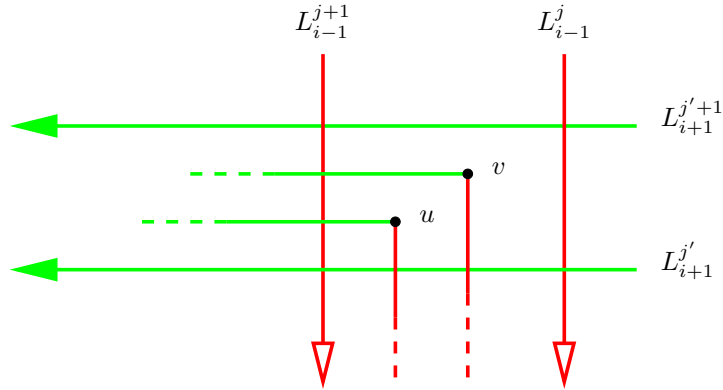
\caption{Position of $u$ and $v$ in the proof of
  Lemma~\ref{lem:regionorder}}
\label{fig:regionrectangle}
\end{figure}

\noindent
\emph{$\star$ Case 1.1: $L_{i-1}(u)=L_{i-1}(v)$ and
  $L_{i+1}(u)=L_{i+1}(v)$.}  Then
$v_i-u_i=d_i(v,z_i(v))-d_i(u,z_i(u))=d_i(v,u)$. Thus clearly, if
$R_i(u)\subseteq R_i(v)$, then $u_i\leq v_i$ and if $R_i(u)\subsetneq
R_i(v)$, 
$v_i-u_i>0=(N-n)(|R(L_{i-1}(u),L_{i-1}(v))|+|R(L_{i+1}(u),L_{i+1}(v))|)$.

\noindent
\emph{$\star$ Case 1.2: $L_{i-1}(u)=L_{i-1}(v)$ and $L_{i+1}(u)\neq
  L_{i+1}(v)$.}  As $u\in R_i(v)$, we have $L_{i+1}(u)=L_{i+1}^{j'}$
and $L_{i+1}(v)=L_{i+1}^{j'+1}$.  Then
$v_i-u_i=d_i(v,z_i(v))-d_i(u,z_i(u)) +
N(f_{i+1}(L_{i+1}(v))-f_{i+1}(L_{i+1}(u)))=d_i(v,z_i(v))-d_i(u,z_i(u))
+ N f_{i+1}^{j'}$. Let $u'$ be the intersection of $P_{i+1}(u)$ with
$L_{i-1}^{j+1}$ (maybe $u=u'$). Let $v'$ be the intersection of
$P_{i+1}(v)$ with $L_{i-1}^{j+1}$ (maybe $v=v'$). Since
$L_{i+1}(u)\neq L_{i+1}(v)$, we have $u'\neq v'$. Since $u\in R_i(v)$,
we have $u'\in R_i(v')$ and so $u'\in P_{i-1}(v')$. Then, by
Lemma~\ref{lem:regionslongue}, $d_i(z_i(v'),v')+d_i(u',z_i(u'))<(n-1)
f_{i+1}^{j'}$.  If $L_{i-1}(u)=L_{i-1}^{j+1}$ then, one can see that
$d_i(v,z_i(v))-d_i(u,z_i(u))\geq d_i(v',z_i(v'))-d_i(u',z_i(u'))$.  If
$L_{i-1}(u)=L_{i-1}^j$, one can see that
$d_i(v,z_i(v))-d_i(u,z_i(u))\geq
d_i(v',z_i(v'))-d_i(u',z_i(u'))-f_{i+1}^{j'}$.  So finally,
$v_i-u_i=d_i(v,z_i(v))-d_i(u,z_i(u)) + N f_{i+1}^{j'}\geq
d_i(v',z_i(v'))-d_i(u',z_i(u')) + (N-1) f_{i+1}^{j'} >(N-n)
f_{i+1}^{j'}=(N-n)(|R(L_{i-1}(u),L_{i-1}(v))|+|R(L_{i+1}(u),L_{i+1}(v))|)\geq 0$.

\noindent
\emph{$\star$ Case 1.3: $L_{i-1}(u)\neq L_{i-1}(v)$ and $L_{i+1}(u)=
  L_{i+1}(v)$.}  
This case is completely symmetric to the previous
case.

\noindent
\emph{$\star$ Case 1.4: $L_{i-1}(u)\neq L_{i-1}(v)$ and
  $L_{i+1}(u)\neq L_{i+1}(v)$.}  As $u\in R_i(v)$, we have
$L_{i+1}(u)=L_{i+1}^{j'}$, $L_{i+1}(v)=L_{i+1}^{j'+1}$,
$L_{i-1}(u)=L_{i-1}^{j+1}$, and $L_{i-1}(v)=L_{i-1}^{j}$.  Then
$v_i-u_i=d_i(v,z_i(v))-d_i(u,z_i(u)) +
N(f_{i+1}(L_{i+1}(v))-f_{i+1}(L_{i+1}(u)))-
N(f_{i-1}(L_{i-1}(v))-f_{i-1}(L_{i-1}(u)))
=d_i(v,z_i(v))-d_i(u,z_i(u)) + N f_{i+1}^{j'}+Nf_{i-1}^{j}$.  Let $u'$
be the intersection of $P_{i+1}(u)$ with $L_{i-1}^{j+1}$ (maybe
$u=u'$).  Let $u''$ be the intersection of $P_{i-1}(u)$ with
$L_{i+1}^{j'}$ (maybe $u=u''$).  Let $v'$ be the intersection of
$P_{i+1}(v)$ with $L_{i-1}^{j+1}$ (maybe $v=v'$).  Let $v''$ be the
intersection of $P_{i-1}(v)$ with $L_{i+1}^{j'}$ (maybe $v=v''$).
Since $L_{i+1}(u)\neq L_{i+1}(v)$, we have $u'\neq v'$. Since $u\in
R_i(v)$, we have $u'\in R_i(v')$ and so $u'\in P_{i-1}(v')$. Then, by
Lemma~\ref{lem:regionslongue}, $d_i(z_i(v'),v')+d_i(u',z_i(u'))<(n-1)
f_{i+1}^{j'}$.  Symmetrically,
$d_i(z_i(v''),v'')+d_i(u'',z_i(u''))<(n-1) f_{i-1}^{j}$.  Moreover, we
have $d_i(v,z_i(v))-d_i(u,z_i(u))\geq
d_i(v',z_i(v'))-d_i(u',z_i(u'))+d_i(v'',z_i(v''))-d_i(u'',z_i(u''))-f_{i+1}^{j'}-f_{i-1}^j$.
So finally, $v_i-u_i=d_i(v,z_i(v))-d_i(u,z_i(u)) + N
f_{i+1}^{j'}+Nf_{i-1}^{j}\geq
d_i(v',z_i(v'))-d_i(u',z_i(u'))+d_i(v'',z_i(v''))-d_i(u'',z_i(u'')) +
(N-1) f_{i+1}^{j'} + (N-1) f_{i-1}^{j}> (N-n) f_{i+1}^{j'}+ (N-n)
f_{i-1}^{j}=(N-n)(|R(L_{i-1}(u),L_{i-1}(v))|+|R(L_{i+1}(u),L_{i+1}(v))|) \geq 0$.

Suppose now that $u$ and $v$ do not lie in a region delimited by two
consecutive lines of color $i-1$ and/or in a region delimited by two
consecutive lines of color $i+1$.  One can easily find distinct
vertices $w_0,\ldots, w_r$ ($w_i$, $1\leq i< r$  chosen at
intersections of monochromatic lines of colors $i-1$ and $i+1$) such
that $w_0=u$, $w_r=v$, and for $0\leq \ell \leq r-1$, we have
$ R_i(w_\ell)\subsetneq R_i(w_{\ell+1})$ and $w_\ell,w_{\ell+1}$ are both in a
region delimited by two consecutive lines of color $i-1$ and in a
region delimited by two consecutive lines of color $i+1$. Thus by the
first part of the proof, $(w_\ell)_i-(w_{\ell+1})_i>
(N-n)(|R(L_{i-1}(w_{\ell+1}),L_{i-1}(w_\ell))|+|R(L_{i+1}(w_{\ell+1}),L_{i+1}(w_\ell))|)$.
Thus $v_i-u_i> (N-n)\sum_\ell
(|R(L_{i-1}(w_{\ell+1}),L_{i-1}(w_\ell))|+|R(L_{i+1}(w_{\ell+1}),L_{i+1}(w_\ell))|)$.
For any $a,b,c$ such $R_i(a)\subseteq R_i(b)\subseteq R_i(c)$, we have
$|R(L_{j}(a),L_{j}(b))|+|R(L_{j}(b),L_{j}(c))|=|R(L_{j}(a),L_{j}(c))|$.
Thus we obtain the result by summing the size of the regions.

\noindent
\emph{$\bullet$ Case 2: The Schnyder wood is of type 2.i.}

Suppose first that $u$ and $v$ are both in a region delimited by two
consecutive lines of color $i+1$. 

Let $L_{i+1}^{j},L_{i+1}^{j+1}$ be these lines such that
$L_{i+1}^{j+1}$ is on the positive side of $L_{i+1}^{j}$, and
$L_{i+1}^\ell \in \mathcal L_{i+1}^\ell$.  We can assume that we do
not have both $u$ and $v$ in $L_{i+1}^{j+1}$ (by eventually choosing
other consecutive lines of color $i+1$). We consider two cases:

\noindent
\emph{$\star$ Case 2.1: $v\notin L_{i+1}^{j+1}$.}  Then by
Lemma~\ref{lem:linesintersection},
$L_{i+1}^{j}=L_{i+1}(u)=(L_{i-1}(u))^{-1}=L_{i+1}(v)=(L_{i-1}(v))^{-1}$.
Then $v_i-u_i=d_i(v,z_i(v))-d_i(u,z_i(u))=d_i(v,u)$. Thus clearly, if
$R_i(u)\subseteq R_i(v)$, then $u_i\leq v_i$ and if $R_i(u)\subsetneq
R_i(v)$, then
$v_i-u_i>0=(N-n)(|R(L_{i-1}(u),L_{i-1}(v))|+|R(L_{i+1}(u),L_{i+1}(v))|)$.

\noindent
\emph{$\star$ Case 2.2: $v\in L_{i+1}^{j+1}$.}  Then
$L_{i+1}^{j+1}=L_{i+1}(v)=(L_{i-1}(v))^{-1}$ and $d_i(v,z_i(v))=0$.
By assumption $u\notin L_{i+1}^{j+1}$ and by
Lemma~\ref{lem:linesintersection},
$L_{i+1}^{j}=L_{i+1}(u)=(L_{i-1}(u))^{-1}$.  Then
$v_i-u_i=d_i(v,z_i(v))-d_i(u,z_i(u))+
N(f_{i+1}(L_{i+1}(v))-f_{i+1}(L_{i+1}(u)))-
N(f_{i-1}(L_{i-1}(v))-f_{i-1}(L_{i-1}(u))) = - d_i(u,z_i(u))+ 2N
f_{i+1}^{j}$.  Let $L_i$ and $L'_i$ be two consecutive $i$-lines such
that $u$ lies in the region between them and $L'_i$ is on the right of
$L_i$.  Let $u'$ be the intersection of $P_{i+1}(u)$ with $L_{i}$
(maybe $u=u'$).  Let $u''$ be the intersection of $P_{i-1}(u)$ with
$L'_{i}$ (maybe $u=u''$).  Then, by Lemma~\ref{lem:regionslongue},
$d_i(u',z_i(u'))<(n-1) f_{i+1}^{j}$ and $d_i(u'',z_i(u''))<(n-1)
f_{i+1}^{j}$.  Thus we have $d_i(u,z_i(u))\leq
d_i(u',z_i(u'))+d_i(u'',z_i(u''))+ f_{i+1}^{j}< (2(n-1)+1)
f_{i+1}^{j}$.  So finally, $v_i-u_i> -(2n-1) f_{i+1}^{j} + 2N
f_{i+1}^{j}> 2(N-n)f_{i+1}^{j} =
(N-n)(|R(L_{i-1}(u),L_{i-1}(v))|+|R(L_{i+1}(u),L_{i+1}(v))|)\geq 0$.

If $u$ and $v$ do not lie in a region delimited by two consecutive
lines of color $i+1$, then as in case 1, one can find intermediate
vertices to obtain the result.
\end{proof}

 \begin{lemma}
\label{lem:edgesbounded}
   If two vertices $u,v$ are adjacent, then for each color $i$, we
   have $|v_i - u_i |\leq 2Nf$.
 \end{lemma}

\begin{proof}
  Since $u,v$ are adjacent, they are both in a region delimited by two
  consecutive lines of color $i-1$ and in a region delimited by two
  consecutive lines of color $i+1$.  Let $L_{i-1}^j,L_{i-1}^{j+1}$ be
  these two consecutive lines of color $i-1$ and
  $L_{i+1}^{j'},L_{i+1}^{j'+1}$ these two consecutive lines of color
  $i+1$ with $L_k^\ell \in \mathcal L_k^\ell$, $L_{i-1}^{j+1}$ is on
  the positive side of $L_{i-1}^j$ and $L_{i+1}^{j'+1}$ is on the
  positive side of $L_{i+1}^{j'}$ (see
  Figure~\ref{fig:regionrectangle} when the Schnyder wood is not of
  type 2.i).  If the Schnyder wood is of type 2.i we assume that
  $L_{i-1}^{j+1}=(L_{i+1}^{j'})^{-1}$ and
  $L_{i-1}^{j}=(L_{i+1}^{j'+1})^{-1}$.  Let $z$ be a vertex on the
  intersection of $L_{i-1}^{j+1}$ and $L_{i+1}^{j'}$.  Let $z'$ be a
  vertex on the intersection of $L_{i-1}^{j}$ and $L_{i+1}^{j'+1}$.
  Thus we have $R_i(z)\subseteq R_i(u)\subseteq R_i(z')$ and
  $R_i(z)\subseteq R_i(v)\subseteq R_i(z')$. So by
  Lemma~\ref{lem:regionorder}, $z_i\leq u_i\leq z'_i$ and $z_i\leq
  v_i\leq z'_i$. So $|v_i-u_i|\leq
  z'_i-z_i=N(f_{i+1}(L_{i+1}^{j'+1})-f_{i+1}(L_{i+1}^{j'})-
  N(f_{i-1}(L_{i-1}^{j})-f_{i-1}(L_{i-1}^{j+1})) =N
  f_{i+1}^{j'}+Nf_{i-1}^{j}\leq 2Nf$.
\end{proof}

We are now able to prove the following:

\begin{theorem}
\label{th:geodesic}
If $G$ is a toroidal map given with an intersecting Schnyder wood,
then the mapping of each vertex of $G^\infty$ on its region vector
gives a geodesic embedding of $G$.
\end{theorem}

\begin{proof}
  By Lemmas~\ref{lem:periodic} and~\ref{lem:notcolinear}, the mapping
  of $G^\infty$ on its region vector is periodic with respect to $S$,
  $S'$ that are not collinear.  For any pair $u,v$ of distinct vertices
  of $G^\infty$, by Lemma~\ref{lem:regionss}.(iii), there exists $i,j$
  with $R_i(u) \subsetneq R_i(v)$ and $R_j(v) \subsetneq R_j(u)$ thus,
  by Lemma~\ref{lem:regionorder}, $u_i<v_i$ and $v_j<u_j$. So
  $\mv^\infty$ is a set of pairwise incomparable elements of $\mr^3$.

  (D1) $\mv^\infty$ is a set of pairwise incomparable elements so
  the mapping between vertices of $G^\infty$ and
  $\mv^\infty$ is a bijection.

  (D2) Let $e=uv$ be an edge of $G^\infty$. We show that $w=u\vee v$
  is on the surface $S_{\mathcal V^\infty}$. By definition $u\vee v$
  is in $\mathcal D_{\mathcal V^\infty}$. Suppose, by contradiction
  that $w\notin S_{\mathcal V^\infty}$. Then there exist $x\in
  \mathcal V^\infty$ with $x<w$.  Let $x$ also denote the
  corresponding vertex of $G^\infty$.  Edge $e$ is in a region
  $R_i(x)$ for some $i$. So $u,v\in R_i(x)$ and thus by
  Lemma~\ref{lem:regionss}.(i), $R_i(u)\subseteq R_i(x)$ and
  $R_i(v)\subseteq R_i(x)$. Then by Lemma~\ref{lem:regionorder},
  $w_i=\max(u_i,v_i)\leq x_i$, a contradiction.  Thus the elbow
  geodesic between $u$ and $v$ is  on the surface.

  (D3) Consider a vertex $v\in\mathcal V$ and a color $i$.  Let $u$ be
  the extremity of the arc $e_i(v)$. We have $u\in R_{i-1}(v)$ and
  $u\in R_{i+1}(v)$, so by Lemma~\ref{lem:regionss}.(i),
  $R_{i-1}(u)\subseteq R_{i-1}(v)$ and $R_{i+1}(u)\subseteq
  R_{i+1}(v)$. Thus by Lemma~\ref{lem:regionss}.(iii),
  $R_{i}(v)\subsetneq R_{i}(u)$. So, by Lemma~\ref{lem:regionorder},
  $v_i<u_i$, $u_{i-1}\leq v_{i-1}$ and $u_{i+1}\leq v_{i+1}$. So the
  orthogonal arc of vertex $v$ in direction of the basis vector $e_i$
  is part of the elbow geodesic of the edge $e_i(v)$.

  (D4) Suppose there exists a pair of crossing edges $e=uv$ and
  $e'=u'v'$ on the surface $S_{\mathcal V^\infty}$. The two edges
  $e,e'$ cannot intersect on orthogonal arcs so they intersects on a
  plane orthogonal to one of the coordinate axis. Up to symmetry we
  may assume that we are in the situation of Figure~\ref{fig:crossing}
  with $u_1=u'_1$, $u_2>u'_2$ and $v_2<v'_2$. Between $u$ and $u'$,
  there is a path consisting of orthogonal arcs only. With (D3), this
  implies that there is a bidirected path $P^*$ colored $0$ from $u$
  to $u'$ and colored $2$ from $u'$ to $u$. We have $u\in R_2(v)$, so
  by Lemma~\ref{lem:regionss}.(i), $R_2(u)\subseteq R_2(v)$.  We have
  $u'\in R_2(u)$, so $u'\in R_2(v)$.  If $P_0(v)$ contains $u'$, then
  there is a directed cycle containing $v,u,u'$ in $G^\infty_1\cup
  (G^\infty_{0})^{-1}\cup (G^\infty_{2})^{-1}$, contradicting
  Lemma~\ref{lem:nodirectedcycle}, so $P_0(v)$ does not contain
  $u'$.  If $P_1(v)$ contains $u'$, then $u'\in P_1(u)\cap P_0(u)$,
  contradicting Lemma~\ref{lem:nocommongeneral}.  So $u'\in
  R_2^\circ(v)$. Thus the edge $u'v'$ implies that $v'\in R_2(v)$. So
  by Lemma~\ref{lem:regionorder}, $v'_2\leq v_2$, a contradiction.
\end{proof}

\begin{figure}[h!]
\center
\includegraphics[scale=0.5]{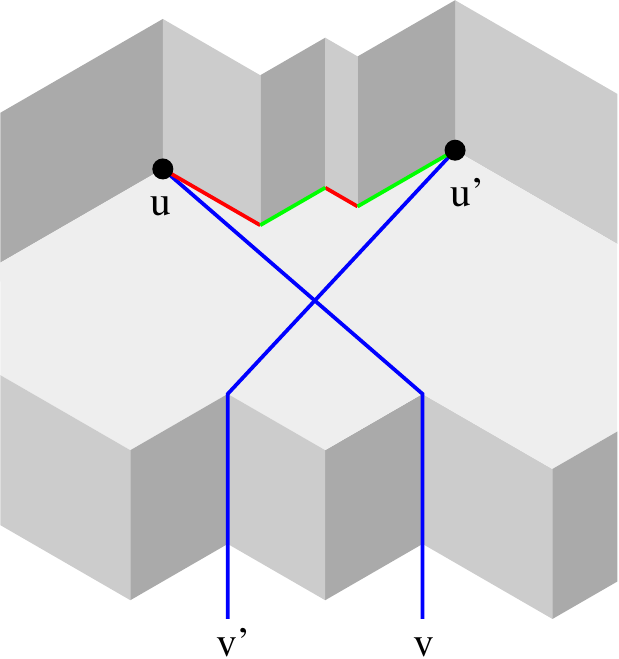}
\caption{A pair of crossing elbow geodesic}
\label{fig:crossing}
\end{figure}

Theorem~\ref{th:existencebasic} and~\ref{th:geodesic} imply
Theorem~\ref{th:triortho}.
 
One can ask what is the ``size'' of the obtained geodesic embedding of
Theorem~\ref{th:geodesic}?  Of course this mapping is infinite so
there is no real size, but as the object is periodic one can consider
the smallest size of the vectors such that the mapping is periodic
with respect to them.  There are several such pairs of vectors, one is
$S,S'$.  Recall that $S_i=N(c_{i+1}-c_{i-1})f$ and
$S'_i=N(c'_{i+1}-c'_{i-1})f$. Unfortunately the size of $S,S'$ can be
arbitrarily large. Indeed, the values of $c_{i+1}-c_{i-1}$ and
$c'_{i+1}-c'_{i-1}$ are unbounded as a toroidal map can be
artificially ``very twisted'' in the considered parallelogram
representation (independently from the number of vertices or faces).
Nevertheless we can prove existence of bounded size vectors for which
the mapping is periodic with respect to them.

\begin{lemma}
\label{th:gamma}
If $G$ is a toroidal map given with an intersecting Schnyder wood, then the mapping
of each vertex of $G^\infty$ on its region vector gives a periodic
mapping of $G^\infty$ with respect to non collinear vectors $Y$ and
$Y'$ where the size of $Y$ and $Y'$ is in $\mathcal O(\omega
Nf)$.  In general we have $\omega\leq n$ and in the case where $G$ is
a simple toroidal triangulation given with a Schnyder wood obtained by
Theorem~\ref{th:schnydersimple}, we have $\omega=1$.
\end{lemma}

\begin{proof}
  By Lemma~\ref{lem:copies}, the vectors $Z_{i-1},Z_{i+1}$ (when the
  Schnyder wood is not of type 2.i) span a subset of $S,S'$ (it can
  happen that this subset is strict).  Thus in the parallelogram
  delimited by the vectors $Z_{i-1},Z_{i+1}$ (that is a parallelogram
  by Lemma~\ref{lem:dim}), there is a parallelogram with sides $Y,Y'$
  containing a copy of $V$. The size of the vectors $Z_i$ is in
  $\mathcal O(\omega Nf)$ and so $Y$ and $Y'$ also.

  In general we have $\omega_i\leq n$ as each intersection between two
  monochromatic cycles of $G$ of color $i-1$ and $i+1$ corresponds to
  a different vertex of $G$ and thus $\omega \leq n$.  In the case of
  simple toroidal triangulation given with a Schnyder wood obtained by
  Theorem~\ref{th:schnydersimple}, we have, for each color $i$,
  $\omega_i=1$, and thus $\omega=1$.
\end{proof}

Note that the geodesic embeddings of Theorem~\ref{th:geodesic} are not
necessarily rigid. A geodesic embedding is \emph{rigid}~\cite{FZ08,Mil02} if for every
pair $u,v\in \mathcal V$ such that $u\vee v$ is in $\mathcal
S_{\mathcal V}$, then $u$ and $v$ are the only elements of $\mathcal
V$ that are dominated by $u\vee v$.  The
geodesic embedding of Figure~\ref{fig:example-primal} is not rigid has
the bend corresponding to the loop of color $1$ is dominated by three
vertices of $G^\infty$. We do not know if it is possible to build a
rigid geodesic embedding from an intersecting Schnyder wood of a toroidal
map. Maybe a technique similar to the one presented in \cite{FZ08} can
be generalized to the torus. 

It has been already mentioned that in the geodesic embeddings of
Theorem~\ref{th:geodesic} the points corresponding to vertices are not
coplanar. The problem to build a coplanar geodesic embedding from the
Schnyder wood of a toroidal map is open. In the plane, there are some
examples of maps $G$~\cite{FZ08} for which it is no possible to
require both rigidity and co-planarity. It is not difficult to
transform these examples into toroidal maps such that the same is true
is the toroidal case.

Another question related to co-planarity is whether one can require that
the points of the orthogonal surface corresponding to edges of the
graph (i.e. bends) are coplanar.  This property is related to contact
representation by homologous triangles \cite{FZ08}. It is known
that in the plane, not all Schnyder woods are supported by such
surfaces.  Kratochvil's conjecture~\cite{Kra07}, 
proved in~\cite{GLP11b}, states that every 4-connected planar
triangulation admits a contact representation by homothetic
triangles. Can this be extended to the torus?

\section{Example of a geodesic embedding }

We use the example of the toroidal map $G$ of
Figure~\ref{fig:example-dual-tore} to illustrate the region vector
method.  This toroidal map has $n=3$ vertices, $f=4$ faces and $e=7$
edges. Let $N=n=3$. There are two edges that are oriented in two
directions. The Schnyder wood is of type 1, with two $1$-cycles.  We
choose as origin the three bold monochromatic lines of
Figure~\ref{fig:example-coordinate}. We give them value $0$ and
compute the other values $f_i(L)$ as explained in
 Section~\ref{sec:regionvector}, i.e. we compute the ``number'' of
faces between $L$ and the origin $L^*$ and put a minus if we are on
the left of $L^*$ (this corresponds to values indicated on the border
of Figure~\ref{fig:example-coordinate}, that are values $0,-4,-8,-12$
for lines of color $0$, values $-4,-2,0,2,4$ for lines of color $1$
and values $0,4$ for lines of color $2$).  Then we compute the region
vector of the points according to
Definition~\ref{def:regionvector}. For example, the point $v$ of
Figure~\ref{fig:example-coordinate} as the following values (the three
points $z_i(v)$ are represented on the figure):
$v_0\ =\ d_0(v,z_0(v))\ +\
N\times\big(f_{1}(L_{1}(v))-f_{2}(L_{2}(v))\big)= 0 + 3(0-0)=0$,
$v_1\ =\ d_1(v,z_1(v))\ +\
N\times\big(f_{2}(L_{2}(v))-f_{0}(L_{0}(v))\big)=0+3(0-(-4))=12$,
$v_2\ =\ d_2(v,z_2(v))\ +\
N\times\big(f_{0}(L_{0}(v))-f_{1}(L_{1}(v))\big)=1+3(-4-0)=-11$.
We compute similarly the region vectors of all the vertices that are
in the black box representing one copy of the graph and let
$\mathcal V=\{(0,0,0), (0,12,-11), (6,12,-18)\}$ be the set of these
vectors. Then we compute the $c_i$'s and $c'_i$'s by algebraically
counting the number of times the monochromatic cycles crosses the
sides of the black box. A monochromatic cycle of color $0$ goes $-1$
time from right to left and $-2$ times from bottom to top. So $c_0=-1$
and $c'_0=-2$. Similarly $c_1=0$, $c'_1=1$, $c_2=1$, $c'_2=0$. Then we
compute $S_i=N(c_{i+1}-c_{i-1})f$ and $S'_i=N(c'_{i+1}-c'_{i-1})f$ and
obtain $S=(-12,24,-12)$, $S'=(12,24,-36)$.  Then the region vectors of
the vertices of $G^\infty$ are
$\{u\in \mr^3\,|\,\exists\, v \in \mv,\ k_1,k_2\in \mz$ such that
$u=v+kS+k'S\}$. In this example, the points are not coplanar, they lie
on the two different planes of equation $x+y+z=0$ and $x+y+z=1$.  The
geodesic embedding that is obtained by mapping each vertex to its
region vector is the geodesic embedding of
Figure~\ref{fig:example-primal}. The black parallelogram has sides the
vectors $S,S'$ and represent a basic tile.

\begin{figure}[h!]
\center
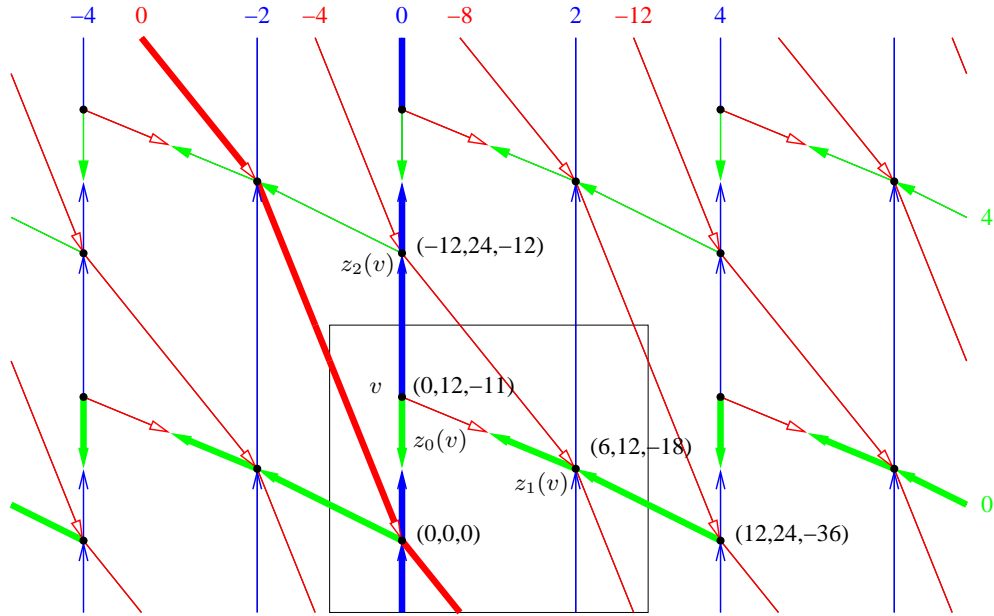
\caption{Computation of the region vectors for the Schnyder wood of
 Figure~\ref{fig:example-dual-tore}.}
\label{fig:example-coordinate}
\end{figure}

\section{Geodesic embedding and duality}
\label{sec:dualortho}

Given an orthogonal surface generated by $\mv$, let $\mathcal
F_{\mathcal V}$ be the maximal points of $\ms$, i.e. the points of
$\ms$ that are not dominated by any vertex of $\ms$.  If $A,B\in
\mathcal F_{\mathcal V}$ and $A\wedge B\in \ms$, then $\ms$ contains
the union of the two line segments joining $A$ and $B$ to $A\wedge
B$. Such arcs are called \emph{dual elbow geodesic}.  The \emph{dual
  orthogonal arc} of $A\in \mathcal F_{\mathcal V}$ in the direction
of the standard basis vector $e_i$ is the intersection of the ray
$A+\lambda e_i$ with $\ms$.

Given a toroidal map $G$, let $G^{\infty*}$ be the dual of $G^\infty$.
A \emph{dual geodesic embedding} of $G$ is a drawing of $G^{\infty*}$
on the orthogonal surface $\mathcal S_{\mathcal V^\infty}$, where
$\mathcal V^\infty$ is a periodic mapping of $G^\infty$ with respect
to two non collinear vectors, satisfying the following (see example of
Figure~\ref{fig:example-dual}):

\begin{itemize}
\item[(D1*)] There is a bijection between the vertices of
  $G^{\infty*}$ and $\mathcal F_{\mathcal V^\infty}$.

\item[(D2*)] Every edge of $G^{\infty*}$ is a dual elbow geodesic.

\item[(D3*)] Every dual orthogonal arc in $\mathcal S_{\mathcal V^\infty}$ is part of an edge of $G^{\infty*}$.

\item[(D4*)] There are no crossing edges in the embedding of $G^{\infty*}$ on $\mathcal S_{\mathcal V^\infty}$.
\end{itemize}

\begin{figure}[h!]
\center
\includegraphics[scale=0.2]{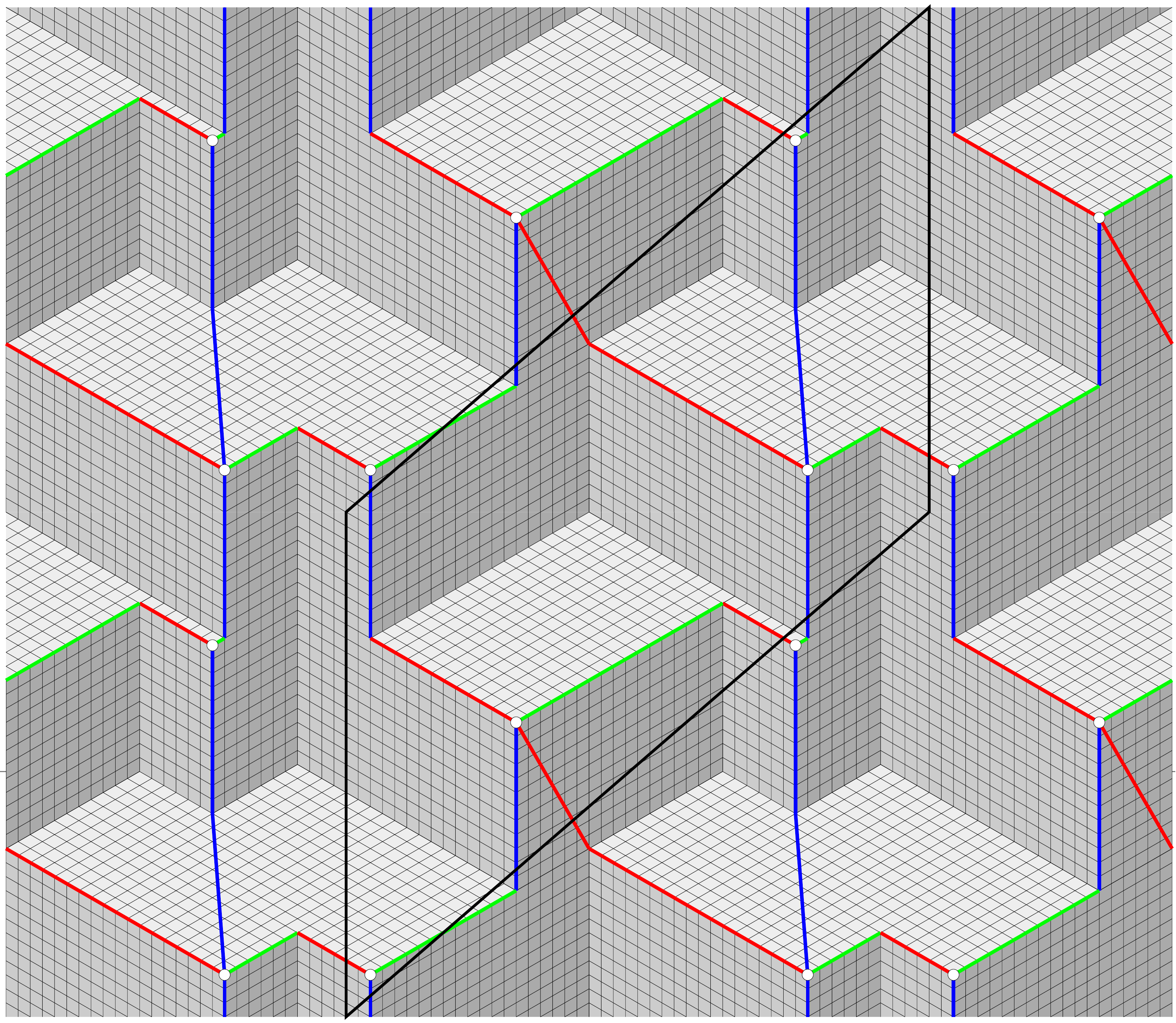}
\caption{Dual geodesic embedding of the toroidal map of 
  Figure~\ref{fig:example-dual-tore}.}
\label{fig:example-dual}
\end{figure}

Let $G$ be a toroidal map given with a intersecting Schnyder
wood. Consider the mapping of each vertex on its region vector.  We
consider the dual of the Schnyder wood of $G$.  By
Proposition~\ref{th:dualtorecross}, it is an intersecting Schnyder
wood of $G^*$.  A face $F$ of $G^{\infty}$ is mapped on the point
$\bigvee_{v\in F}v$.  Let $\overline{G^\infty}$ be a simultaneous
drawing of $G^\infty$ and $G^{\infty*}$ such that only dual edges
intersect. To avoid confusion, we note $R_i$ the regions of the primal
Schnyder wood and $R_i^*$ the regions of the dual Schnyder wood.

\begin{lemma}
\label{lem:maximalpoint}
  For any face $F$ of $G^{\infty}$, we have that $\bigvee_{v\in F}v$ is
  a maximal point of $\mathcal S_{\mathcal V^\infty}$.
\end{lemma}

\begin{proof}
  Let $F$ be a face of $G^{\infty}$.  For any vertex $u$ of $\mathcal
  V^\infty$, there exists a color $i$, such that the face $F$ is in
  the region $R_i(u)$. Thus for $v\in F$, we have $v\in R_i(u)$. By
  Lemma~\ref{lem:regionorder}, we have $v_i\leq u_i$ and so $F_i\leq
  u_i$.  So $F=\bigvee_{v\in F} v$ is a point of $\mathcal S_{\mathcal
    V^\infty}$.

  Suppose, by contradiction, that $F$ is not a maximal point of
  $\mathcal S_{\mathcal V^\infty}$. Then there is a point
  $\alpha \in \mathcal S_{\mathcal V^\infty}$ that dominates $F$ and
  for at least one coordinate $j$, we have $F_j < \alpha_j$. Consider
  the angle labeling corresponding to the considered Schnyder
  wood. Every face is of type $1$ so the
  angles at $F$ form, in counterclockwise order, nonempty intervals of
  $0$'s, $1$'s, and $2$'s.  For each color, let $z^i$ be a vertex of
  $F$ with angle $i$. We have $F$ is in the region $R_i(z^i)$. So
  $z^{i-1}\in R_i(z^i)$ and by Lemma~\ref{lem:regionss}.(i), we have
  $R_i(z^{i-1})\subseteq R_i(z^i)$.  Since $F$ is in
  $R_{i-1}(z^{i-1})$, it is not in $R_i(z^{i-1})$ and thus
  $R_i(z^{i-1})\subsetneq R_i(z^i)$.  Then by
  Lemma~\ref{lem:regionorder}, we have $(z^{i-1})_i<(z^{i})_i$ and
  symmetrically $(z^{i+1})_i<(z^{i})_i$.  So
  $F_{j-1}=(z^{j-1})_{j-1}>(z^j)_{j-1}$ and
  $F_{j+1}>(z^j)_{j+1}$. Thus $\alpha$ strictly dominates $z^j$, a
  contradiction to $\alpha \in \mathcal S_{\mathcal V^\infty}$. Thus
  $F$ is a maximal point of $\mathcal S_{\mathcal V^\infty}$
\end{proof}

\begin{lemma}
\label{lem:dualorder}
   If two faces $A,B$ are such that $R_i^*(B)\subseteq R_i^*(A)$,
    then $A_i\leq B_i$.
\end{lemma}

\begin{proof}
  Consider the angle labeling corresponding to the considered Schnyder
  wood. Let $v\in B$ be a vertex whose angle at $B$ is labeled $i$. We
  have $v\in R^*_i(B)$ and so $v\in R^*_i(A)$. In
  $\overline{G^\infty}$, the path $P_{i}(v)$ cannot leave $R^*_i(A)$,
  the path $P_{i+1}(v)$ cannot intersect $P_{i+1}(A)$ and the path
  $P_{i-1}(v)$ cannot intersect $P_{i-1}(A)$. Thus $P_{i+1}(v)$
  intersect $P_{i-1}(A)$ and the path $P_{i-1}(v)$  intersect
  $P_{i+1}(A)$.  So $A\in R_i(v)$. Thus for all $u\in A$, we have
  $u\in R_i(v)$, so $R_i(u)\subseteq R_i(v)$, and so $u_i\leq
  v_i$.
  Then $A_i=\max_{u\in A} u_i\leq v_i\leq \max_{w\in B} w_i=B_i$.
\end{proof}

\begin{theorem}
\label{th:dualgeodesic}
If $G$ is a toroidal map given with an intersecting Schnyder wood and
each vertex of $G^\infty$ is mapped on its region vector, then the
mapping of each face of $G^{\infty*}$ on the point
$\bigvee _{v\in F} v$ gives a dual geodesic embedding of $G$.
\end{theorem}

\begin{proof}
  By Lemmas~\ref{lem:periodic} and~\ref{lem:notcolinear}, the mapping
  is periodic with respect to non collinear vectors.

  (D1*) Consider a counting of elements on the orthogonal surface,
  where we count two copies of the same object just once (note that we
  are on an infinite and periodic object).  We have that the sum of
  primal orthogonal arcs plus dual ones is exactly $3m$. There are
  $3n$ primal orthogonal arcs and thus there are $3m-3n=3f$ dual
  orthogonal arcs.  Each maximal point of $\mathcal S_{\mathcal
    V^\infty}$ is incident to $3$ dual orthogonal arcs and there is no
  dual orthogonal arc incident to two distinct maximal points. So
  there is $f$ maximal points. Thus by Lemma~\ref{lem:maximalpoint},
  we have a bijection between faces of $G^\infty$ and maximal points
  of $\mathcal S_{\mathcal V^\infty}$.

  Let $\mathcal V^{\infty*}$ be the maximal points of $\mathcal
  S_{\mathcal V^\infty}$.  Let $\mathcal D^*_{\mathcal V^\infty}
  =\{A\in\mathbb R^3\ |\ \exists\, B\in \mathcal V^{\infty*}$ such
  that $A\leq B\}$. Note that the boundary of $\mathcal D^*_{\mathcal V^\infty}$
  is $\mathcal S_{\mathcal V^\infty}$.

  (D2*) Let $e=AB$ be an edge of $G^{\infty*}$.  We show that
  $w=A\wedge B$ is on the surface $S_{\mathcal V^\infty}$.  By
  definition $w$ is in $\mathcal D^*_{\mathcal V^\infty}$.  Suppose,
  by contradiction that $w\notin S_{\mathcal V^\infty}$. Then there
  exist $C$ a maximal point of $S_{\mathcal V^\infty}$ with $w<C$.  By
  the bijection (D1*) between maximal point and vertices of
  $G^{\infty*}$, the point $C$ corresponds to a vertex of
  $G^{\infty*}$, also denoted $C$.  Edge $e$ is in a region $R_i^*(C)$
  for some $i$. So $A,B\in R_i^*(C)$ and thus, by
  Lemma~\ref{lem:regionss}.(i), $R_i^*(A)\subseteq R_i^*(C)$ and
  $R_i^*(B)\subseteq R_i^*(C)$. Then by Lemma~\ref{lem:dualorder}, we
  have $C_i\leq \min(A_i,B_i)=w_i$, a contradiction.  Thus the dual
  elbow geodesic between $A$ and $B$ is also on the surface.

  (D3*) Consider a vertex $A$ of $G^{\infty*}$ and a color $i$.  Let
  $B$ be the extremity of the arc $e_i(A)$. We have $B\in
  R_{i-1}^*(A)$ and $B\in R_{i+1}^*(A)$, so by
  Lemma~\ref{lem:regionss}.(i), $R_{i-1}^*(B)\subseteq R_{i-1}^*(A)$
  and $R_{i+1}^*(B)\subseteq R_{i+1}^*(A)$. So by
  Lemma~\ref{lem:dualorder}, $A_{i-1}\leq B_{i-1}$ and $A_{i+1}\leq
  B_{i+1}$. As $A$ and $B$ are distinct maximal point of $\mathcal
  S_{\mathcal V^\infty}$, they are incomparable, thus $B_i<A_i$.  So
  the dual orthogonal arc of vertex $A$ in direction of the basis
  vector $e_i$ is part of edge $e_i(A)$.

  (D4*) Suppose there exists a pair of crossing edges $e=AB$ and
  $e'=A'B'$ of $G^{\infty*}$ on the surface $S_{\mathcal
    V^\infty}$.
  The two edges $e,e'$ cannot intersect on orthogonal arcs so they
  intersects on a plane orthogonal to one of the coordinate axis. Up
  to symmetry we may assume that we are in the situation $A_1=A'_1$,
  $A'_0>A_0$ and $B'_0<B_0$. Between $A$ and $A'$, there is a path
  consisting of orthogonal arcs only. With (D3*), this implies that
  there is a bidirected path $P^*$ colored $2$ from $A$ to $A'$ and
  colored $0$ from $A'$ to $A$. We have $A\in R_0(B)$, so by
  Lemma~\ref{lem:regionss}.(i), $R_0(A)\subseteq R_0(B)$.  We have
  $A'\in R_0(A)$, so $A'\in R_0(B)$.  If $P_2(B)$ contains $A'$, then
  there is a cycle containing $A,A',B$ in
  $(G^\infty_1)^*\cup (G^\infty_{0})^{*-1}\cup (G^\infty_{2})^{*-1}$,
  contradicting Lemma~\ref{lem:nodirectedcycle}, so $P_2(B)$ does not
  contain $A'$.  If $P_1(B)$ contains $A'$, then
  $A'\in P_1(A)\cap P_2(A)$, contradicting
  Lemma~\ref{lem:nocommongeneral}.  So $A'\in R_0^\circ(B)$. Thus the
  edge $A'B'$ implies that $B'\in R_0(B)$. So by
  Lemma~\ref{lem:dualorder}, $B'_0\geq B_0$, a contradiction.
\end{proof}

Theorems~\ref{th:geodesic} and~\ref{th:dualgeodesic} can be combined to
obtain a simultaneous representation of a Schnyder wood and its dual
on an orthogonal surface. The projection of this 3-dimensional object
on the plane of equation $x+y+z=0$ gives a representation of the
primal and the dual where edges are allowed to have one bend and two
dual edges have to cross on their bends (see example of
Figure~\ref{fig:example-primal-dual}). 

\begin{theorem}
\label{cor:primal-dual}
An essentially 3-connected toroidal map admits a simultaneous flat
torus representation of the primal and the dual where edges are
allowed to have one bend and two dual edges have to cross on their
bends. Such a representation is contained in a (triangular) grid of
size $\mathcal O(n^2f)\times \mathcal O(n^2f)$ in general and
$\mathcal O(nf)\times\mathcal O(nf)$ if the map is a simple
triangulation. Furthermore the length of the edges are in $\mathcal
O(nf)$.
\end{theorem}

\begin{proof}
  Let $G$ be an essentially 3-connected toroidal map.  By
  Theorems~\ref{th:existencebasic} (or Theorem~\ref{th:schnydersimple} if
  $G$ is a simple triangulation), $G$ admits an intersecting Schnyder
  wood (where monochromatic cycles of different colors intersect just
  once if $G$ is simple).  By Theorems~\ref{th:geodesic}
  and~\ref{th:dualgeodesic}, the mapping of each vertex of $G^\infty$
  on its region vector gives a primal and dual geodesic
  embedding. Thus the projection of this embedding on the plane of
  equation $x+y+z=0$ gives a representation of the primal and the dual
  of $G^\infty$ where edges are allowed to have one bend and two dual
  edges have to cross on their bends.

  By Lemma~\ref{th:gamma}, the obtained mapping is a periodic
  mapping of $G^\infty$ with respect to non collinear vectors $Y$ and
  $Y'$ where the size of $Y$ and $Y'$ is in $\mathcal O(\omega Nf)$,
  with $\omega\leq n$ in general and $\omega=1$ in case of a simple
  triangulation. Let $N=n$. The embedding gives a representation in
  a parallelogram of sides $Y,Y'$ where the size of the vectors $Y$ and
  $Y'$ is in $\mathcal O(n^2f)$ in general and in $\mathcal O(nf)$ if
  the graph is simple and the Schnyder wood is obtained by
  Theorem~\ref{th:schnydersimple}.  By Lemma~\ref{lem:edgesbounded}
  the length of the edges in this representation are in $\mathcal
  O(nf)$.
\end{proof}

\begin{figure}[h!]
\center
\includegraphics[scale=0.2]{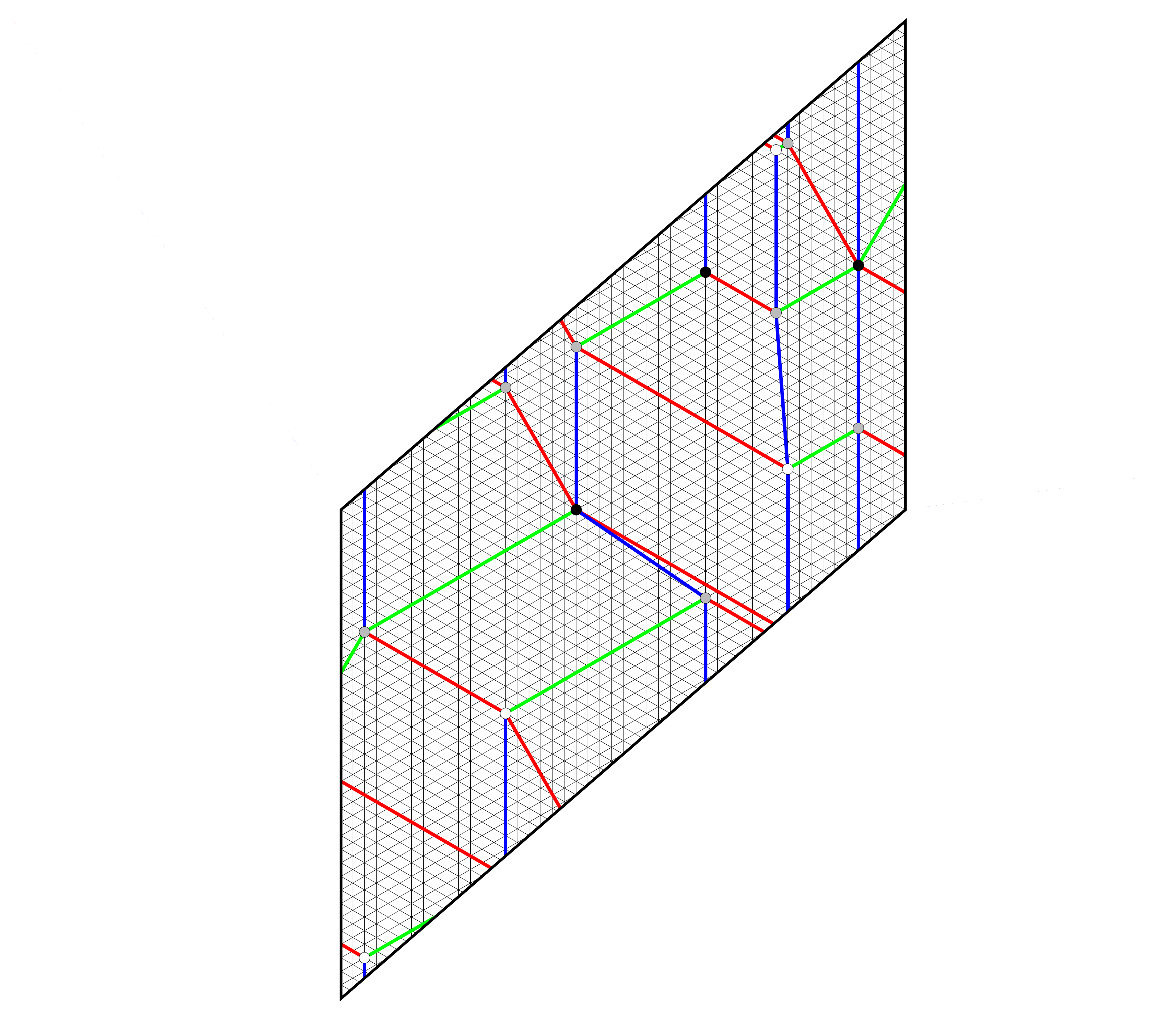}
\caption{Simultaneous representation of the primal and the dual of the
  toroidal map of Figure~\ref{fig:example-dual-tore} (see also
  Figure~\ref{fig:example-superposition}) with edges having one bend
  (in gray).}
\label{fig:example-primal-dual}
\end{figure}

\section{Straight-line representation}
\label{sec:straight} 

A \emph{straight-line  representation} of a toroidal map
$G$ is the restriction to a parallelogram of a periodic straight-line
representation of $G^{\infty}$. 
In the embeddings obtained by
Theorem~\ref{th:triortho}, vertices are not coplanar but we prove that
for toroidal triangulations one can project the vertices on a plane to
obtain a periodic straight-line representation of $G^{\infty}$ (see
Figure~\ref{fig:example-droit-primal}).

\begin{figure}[h!]
\center
\includegraphics[scale=0.2]{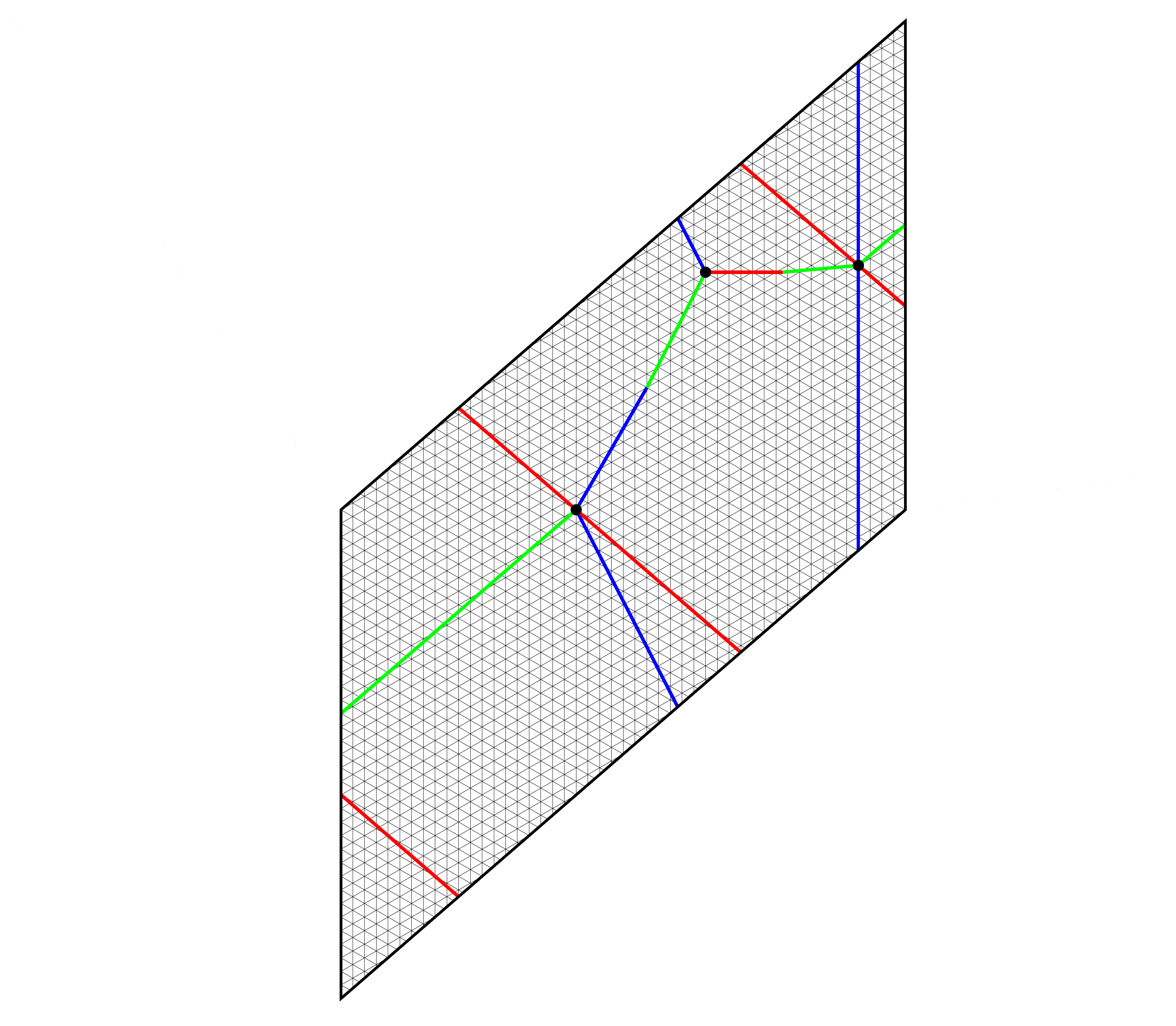}
\caption{Straight-line representation of the map of 
  Figure~\ref{fig:example-dual-tore} obtained by projecting the
  geodesic embedding  of Figure~\ref{fig:example-primal}}
\label{fig:example-droit-primal}
\end{figure}

For this purpose, we have to choose $N$ bigger than previously.  Note
that Figure~\ref{fig:example-droit-primal} is the projection of the
geodesic embedding of Figure~\ref{fig:example-primal} obtained with
the value of $N=n$. In this particular case this gives a straight-line
representation but in this section we only prove that such a technique
works for triangulations and for $N$ sufficiently large. To obtain a
straight-line representation of a general toroidal map, one first has
to triangulate it.

Let $G$ be a toroidal triangulation given with an intersecting
Schnyder wood and $V^\infty$ the set of region vectors of vertices of
$G^\infty$.  The Schnyder wood is crossing by
Proposition~\ref{prop:crossing-trieq}.  Recall that $\omega_i$ is the
integer such that two monochromatic cycles of $G$ of colors $i-1$ and
$i+1$ intersect exactly $\omega_i$ times.

\begin{lemma}
\label{lem:boundedtriangle}
For any vertex $v$, the number of faces in the bounded region
delimited by the three lines $L_{i}(v)$, for $i\in\{0,1,2\}$, is strictly less than
$(5\min_i(\omega_i)+ \max_i(\omega_i))f$.
\end{lemma}

\begin{proof}
  Suppose by symmetry that $\min(\omega_i)=\omega_1$.  Let
  $L_i=L_i(v)$ and $z_i=z_i(v)$. Let $T$ be the bounded region
  delimited by the three monochromatic lines $L_{i}$.  If $T$ is non
  empty, then its boundary is a cycle $C$ oriented clockwise or
  counterclockwise. Assume that $C$ is oriented counterclockwise (the
  proof is similar if oriented clockwise).  The region $T$ is on the
  left sides of the lines $L_i$.  We have $z_{i-1}\in P_{i}(z_{i+1})$.

  We define, for $j,k \in \mathbb{N}$, monochromatic lines $L_2(j)$,
  $L_0(k)$ and vertices $z(j,k)$ as follows (see
  Figure~\ref{fig:triangleprooff}).  Let $L_2(1)$ be the first
  $2$-line intersecting $L_0\setminus\{z_1\}$ while walking from
  $z_1$, along $L_0$ in the direction of $L_0$.  Let $L_0(1)$ be the
  first $0$-line of color $0$ intersecting $L_2\setminus\{z_1\}$ while
  walking from $z_1$, along $L_2$ in the reverse direction of $L_2$.
  Let $z(1,1)$ be the intersection between $L_2(1)$ and $L_0(1)$.  Let
  $z(j,1)$, $j\geq 0$, be the consecutive copies of $z(1,1)$ along
  $L_0(1)$ such that $z(j+1,1)$ is after $z(j,1)$ in the direction of
  $L_0(1)$. Let $L_2(j)$, $j\geq 0$, be the $2$-line of color $2$
  containing $z(j,1)$. Note that we may have $L_2=L_2(0)$ but in any
  case $L_2$ is between $L_2(0)$ and $L_2(1)$.  Let $z(j,k)$, $k\geq
  0$, be the consecutive copies of $z(j,1)$ along $L_2(j)$ such that
  $z(j,k+1)$ is after $z(j,k)$ in the reverse direction of $L_2(j)$.
  Let $L_0(k)$, $k\geq 0$, be the $0$-line containing $z(1,k)$. Note
  that we may have $L_0=L_0(0)$ but in any case $L_0$ is between
  $L_0(0)$ and $L_0(1)$.  Let $S(j,k)$ be the region delimited by
  $L_2(j),L_2(j+1),L_0(k),L_0(k+1)$. All the region $S(j,k)$ are
  copies of $S(0,0)$. The region $S(0,0)$ may contain several copies
  of a face of $G$ but the number of copies of a face in $S(0,0)$ is
  equal to $\omega_1$.  Let $R$ be the unbounded region situated on
  the right of $L_0(1)$ and on the right of $L_2(1)$. As $P_0(v)$
  cannot intersect $L_0(1)$ and $P_2(v)$ cannot intersect $L_2(1)$,
  vertex $v$ is in $R$.  Let $P(j,k)$ be the subpath of $L_0(k)$
  between $z(j,k)$ and $z(j+1,k)$. All the lines $L_0(k)$ are composed
  only of copies of $P(0,0)$.  The interior vertices of the path
  $P(0,0)$ cannot contains two copies of the same vertex, otherwise
  there will be a vertex $z(j,k)$ between $z(0,0)$ and $z(1,0)$. Thus
  all interior vertices of a path $P(j,k)$ corresponds to distinct
  vertices of $G$.

  The Schnyder wood is crossing, thus $1$-lines are crossing
  $0$-lines.  As a line $L_0(k)$ is composed only of copies of
  $P(0,0)$, any path $P(j,k)$ is crossed by a $1$-line. Let $L'_1$ be
  the first $1$-line crossing $P(1,1)$ on a vertex $x$ while walking
  from $z(1,1)$ along $L_0(1)$. By the Schnyder property, line $L'_1$ is not
  intersecting $R\setminus \{z(1,1)\}$. As $v\in R$ we have $L_1$ is
  on the left of $L'_1$ (maybe $L_1=L'_1$). Thus the region $T$ is
  included in the region $T'$ delimited by $L_0,L'_1,L_2$.

  Let $y$ be the vertex where $L'_1$ is leaving $S(1,1)$.  We claim
  that $y\in L_2(1) $.  Note that by the Schnyder property, we have
  $y\in L_2(1)\cup P(1,2)$.  Suppose by contradiction that $y$ is an
  interior vertex of $P(1,2)$. Let $d_x$ be the length of the subpath
  of $P(1,1)$ between $z(1,1)$ and $x$.  Let $d_y$ be the length of
  the subpath of $P(1,2)$ between $z(1,2)$ and $y$.  Suppose
  $d_y<d_x$, then there should be a distinct copy of $L'_1$
  intersecting $P(1,1)$ between $z(1,1)$ and $x$ on a copy of $y$, a
  contradiction to the choice of $L'_1$. So $d_x\leq d_y$. Let $A$ be
  the subpath of $L'_1$ between $x$ and $y$. Let $B$ be the subpath of
  $P(1,1)$ between $x$ and the copy of $y$ (if $d_x=d_y$, then $B$ is
  just a single vertex).  Consider all the copies of $A$ and $B$
  between lines $L_2(1)$ and $L_2(2)$, they form an infinite line $L$
  situated on the right of $L_2(1)$ that prevents $L'_1$ from crossing
  $L_2(1)$, a contradiction.

  By the position of $x$ and $y$. We have $L'_1$ intersects $S(0,1)$
  and $S(1,0)$. We claim that $L'_1$ cannot intersect both $S(0,3)$
  and $S(3,0)$.  Suppose by contradiction that $L'_1$ intersects both
  $S(0,3)$ and $S(3,0)$. Then $L'_1$ is crossing $S(0,2)$ without
  crossing $L_2(0)$ or $L_2(1)$.  Similarly $L'_1$ is crossing
  $S(2,0)$ without crossing $L_0(0)$ or $L_0(1)$. Thus by superposing
  what happen in $S(0,2)$ and $S(2,0)$ in a square $S(j,k)$, we have
  that there are two crossing $1$-lines, a contradiction. Thus $L'_1$
  intersects at most one of $S(0,3)$ and $S(3,0)$.

  Suppose that $L'_1$ does not intersect $S(3,0)$.  Then the part of
  $T'$ situated right of $L_0(2)$ (left part on
  Figure~\ref{fig:triangleprooff}) is strictly included in
  $(S(0,0)\cup S(1,0) \cup S(2,0)\cup S(0,1)\cup S(1,1))$.  Thus this
  part of $T'$ contains at most $5\omega_1 f$ faces.  Now consider the
  part of $T'$ situated on the left of $L_0(2)$ (right part on
  Figure~\ref{fig:triangleprooff}).  Let $y'$ be the intersection of
  $L'_1$ with $L_2$. Let $Q$ be the subpath of $L'_1$ between $y$ and
  $y'$. By definition of $L_2(1)$, there are no $2$-lines between
  $L_2$ and $L_2(1)$. So $Q$ cannot intersect a $2$-line on one of its
  interior vertices.  Thus $Q$ is crossing at most $\omega_2$
  consecutive $0$-lines (that are not necessarily lines of type
  $L_0(k)$). Let $L'_0$ be the $\omega_2+1$-th consecutive $0$-line
  that is on the left of $L_0(2)$ (counting $L_0(2)$).  Then the part
  of $T'$ situated on the left of $L_0(2)$ is strictly included in the
  region delimited by $L_0(2),L'_0,L_2,L_2(1)$, and thus contains at
  most $\omega_2$ copies of a face of $G$.  Thus $T'$ contains at most
  $(\omega_2+5\omega_1)f$ faces.

  Symmetrically if $L'_1$ does not intersect $S(0,3)$ we have that
  $T'$ contains at most $(\omega_0+5\omega_1)f$ faces.  Then in any
  case, $T'$ contains at most $(\max(\omega_0,\omega_2)+5\omega_1)f$
  faces and the lemma is true.
\end{proof}

\begin{figure}[!h]
\center
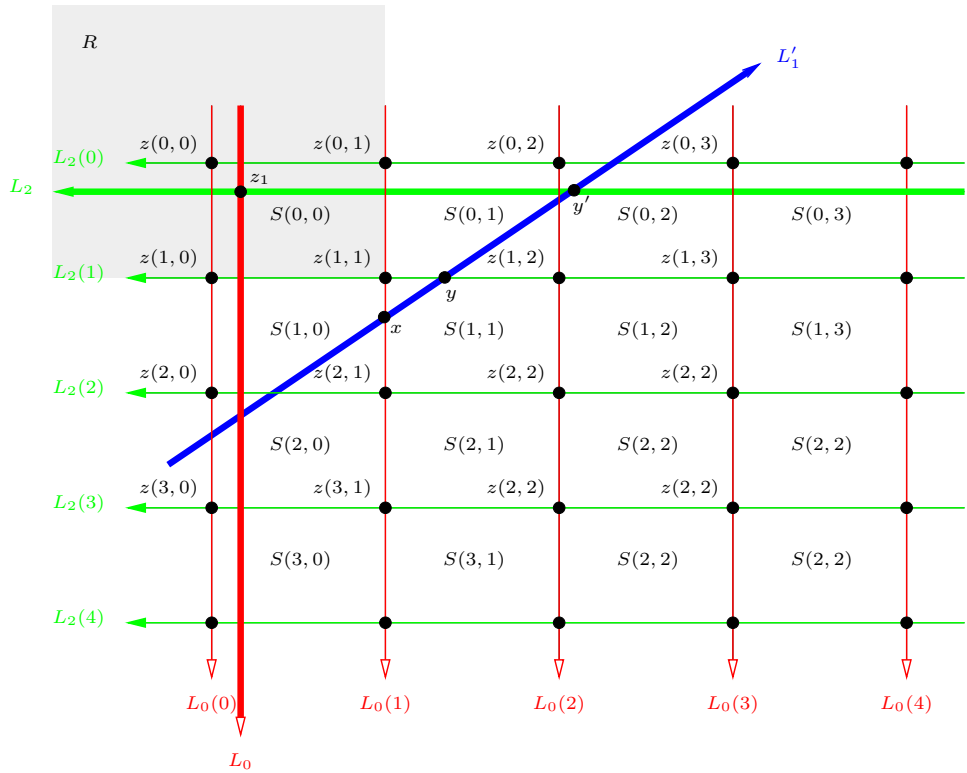
\caption{Notations of the proof of Lemma~\ref{lem:boundedtriangle}.}
\label{fig:triangleprooff}
\end{figure}

The bound of Lemma~\ref{lem:boundedtriangle} is somehow sharp. In the
example of Figure~\ref{fig:exampletriangle}, the rectangle represent a
toroidal triangulation $G$ and the universal cover is partially
represented.  For each value of $k\geq 0$, there is a toroidal
triangulation $G$ with $n=4(k+1)$ vertices, where the gray region,
representing the region delimited by the three monochromatic lines
$L_{i}(v)$ contains $4\sum_{j=1}^{2k+1}+3(2k+2)=\Omega(n\times f)$
faces. Figure~\ref{fig:exampletriangle} represent such a triangulation
for $k=2$.

 \begin{figure}[!h]
 \center
 \includegraphics[scale=0.3]{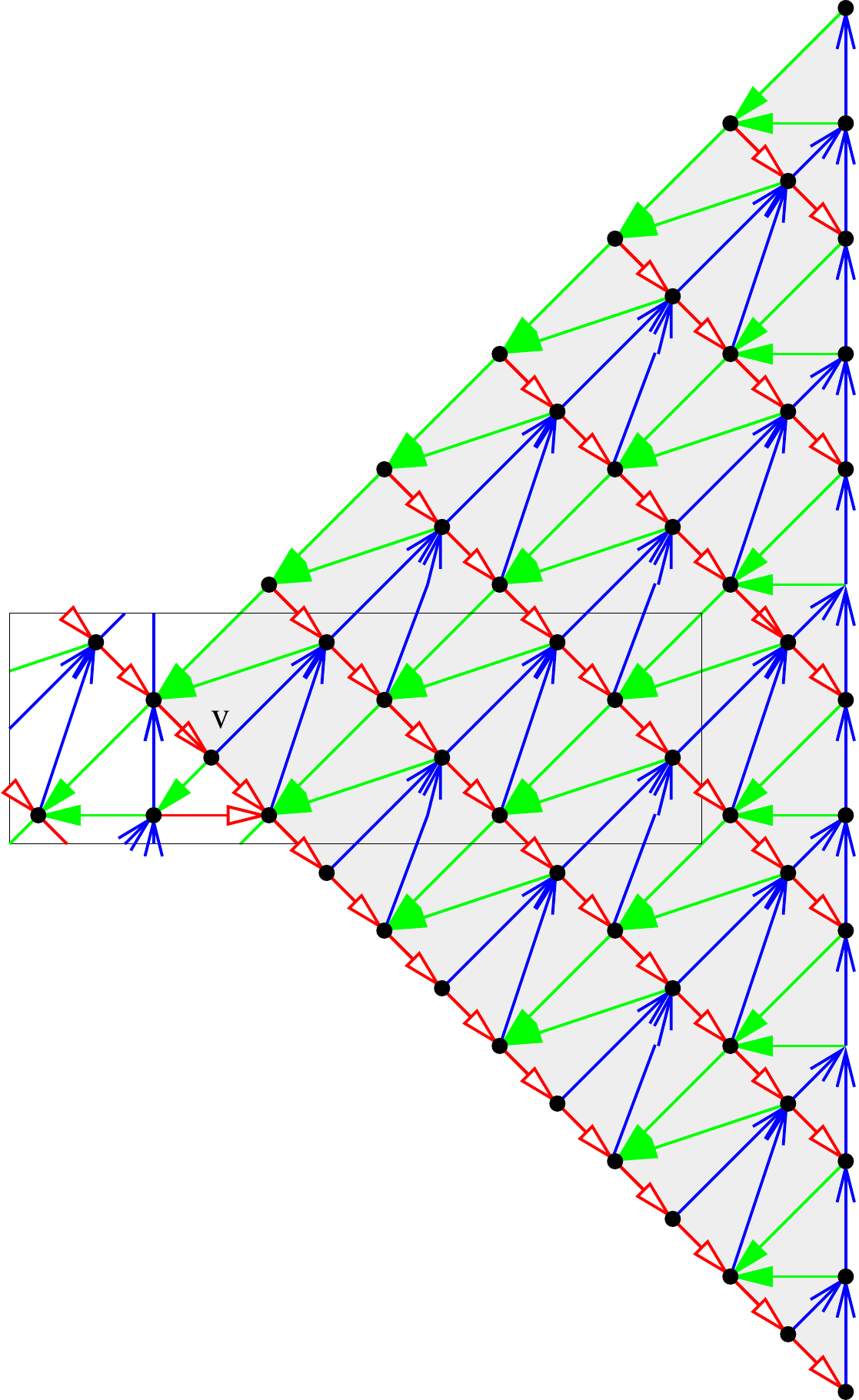}
 \caption{Example of a toroidal triangulation where the number of faces
   in the region delimited by the three monochromatic lines $L_{i}(v)$
   contains $\Omega(n\times f)$ faces.}
 \label{fig:exampletriangle}
 \end{figure} 

 For planar maps the region vector method gives vertices that all lie on
 the same plane. This property is very helpful in proving that the
 position of the points on $P$ gives straight-line representations.  In the
 torus, things are more complicated as our generalization of the
 region vector method does not give coplanar points.  But
 Lemma~\ref{lem:sum} and~\ref{lem:boundedtriangle} show that all the
 points lie in the region situated between the two planes of equation
 $x+y+z=0$ and $x+y+z=t$, with $t=(5\min(\omega_i)+
 \max(\omega_i))f$. Note that $t$ is bounded by $6nf$ by
 Lemma~\ref{th:gamma} and this is independent from $N$. Thus from
 ``far away'' it looks like the points are coplanar and by taking $N$
 sufficiently large, non coplanar points are ``far enough'' from each
 other to enable the region vector method to give straight-line
 representations.  

 Let $N=t+n$.

\begin{lemma}
  \label{lem:shortregionedge}
  Let $u,v$ be two vertices of $G^\infty$ such that $e_{i-1}(v)=uv$,
  $L_i=L_i(u)=L_i(v)$, and such that both $u$, $v$ are in the region
  $R(L_i,L'_i)$ for $L'_i$ a $i$-line consecutive to $L_i$. Then
  $v_{i+1}-u_{i+1}<|R(L_i,L'_i)|$ and $e_{i-1}(v)$ is going
  counterclockwise around the closed disk bounded by
  $\{e_{i-1}(v)\}\cup P_i(u)\cup P_i(v)$.
\end{lemma}

\begin{proof}
  Let $y$ be the first
  vertex of $P_i(v)$ that is also in $P_i(u)$.  Let $Q_u$
  (resp. $Q_v$) the part of $P_i(u)$ (resp. $P_i(v)$) between $u$
  (resp. $v$) and $y$.

  Let $D$ be the closed disk bounded by the cycle $C=(Q_v)^{-1}\cup
  \{e_{i-1}(v)\}\cup Q_u$. If $C$ is going clockwise around $D$, then
  $P_{i+1}(v)$ is leaving $v$ in $D$ and thus has to intersect $Q_u$
  or $Q_v$. In both cases, there is a cycle in $G^\infty_{i+1}\cup
  (G^\infty_i)^{-1}\cup (G^\infty_{i-1})^{-1}$, a contradiction to
  Lemma~\ref{lem:nodirectedcycle}.  So $C$ is going
  clockwise around $D$.

  As $L_i(u)=L_i(v)$ and $L_{i-1}(u)=L_{i-1}(v)$, we have
  $v_{i+1}-u_{i+1}=d_{i+1}(v,u)$ and this is equal to the number of
  faces in $D$.  We have $D\subsetneq R(L_i,L'_i)$.  Suppose $D$
  contains two copies of a given face. Then, these two copies are on
  different sides of a $1$-line. By the Schnyder property, it
  is not possible to have a $1$-line entering $D$.  So $D$ contains at
  most one copy of each face of $R(L_i,L'_i)$.
\end{proof}

\begin{lemma}
\label{lem:vectoriel}
For any face $F$ of $G^\infty$, incident to vertices $u,v,w$ (given in
counterclockwise order around $F$), the cross product
$\overrightarrow{vw}\wedge \overrightarrow{vu}$ has strictly positive
coordinates.
\end{lemma}

\begin{proof}
  Consider the angle labeling corresponding to the Schnyder wood.  The
  angles at $F$ are labeled in counterclockwise order $0,1,2$.  As
  $\overrightarrow{uv}\wedge
  \overrightarrow{uw}=\overrightarrow{vw}\wedge
  \overrightarrow{vu}=\overrightarrow{wu}\wedge \overrightarrow{wv}$,
  we may assume that $u,v,w$ are such that $u$ is in the angle labeled
  $0$, vertex $v$ in the angle labeled $1$ and vertex $w$ in the angle
  labeled $2$.  The face $F$ is either a cycle completely directed
  into one direction or it has two edges oriented in one direction and
  one edge oriented in the other.  Let
$$\overrightarrow{X}=\overrightarrow{vw}\wedge
\overrightarrow{vu}= \begin{pmatrix}
  (w_1-v_1)(u_2-v_2)-(w_2-v_2)(u_1-v_1)\\
  -(w_0-v_0)(u_2-v_2)+(w_2-v_2)(u_0-v_0)\\
  (w_0-v_0)(u_1-v_1)-(w_1-v_1)(u_0-v_0)
\end{pmatrix}$$

By symmetry, we consider the following two cases:

 \begin{figure}[!h]
 \center
\begin{tabular}{ccc}
 \includegraphics[scale=.5]{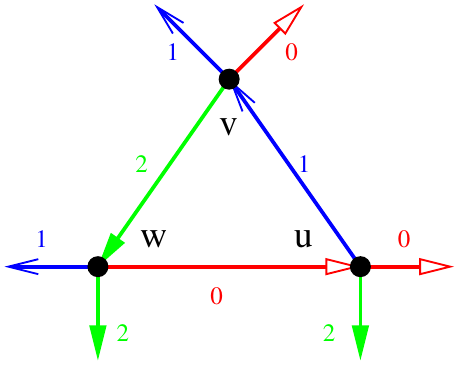}& \hspace{2em} &
 \includegraphics[scale=.5]{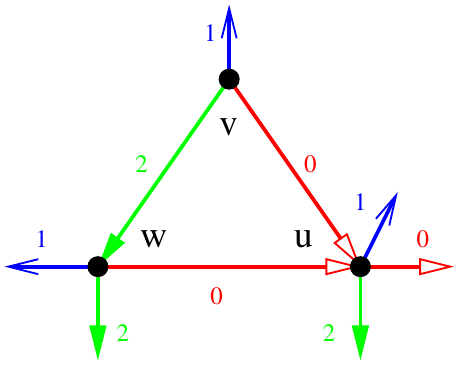}\\
(a)& &(b) \\
\end{tabular}
\caption{(a) case 1 and (b) case 2 of the proof of Lemma~\ref{lem:vectoriel}}
 \label{fig:vectoriel}
 \end{figure} 

 \noindent $\bullet$ \emph{Case 1: the edges of the face $F$ are in
   counterclockwise order $e_1(u)$, $e_2(v)$, $e_0(w)$ (see
   Figure~\ref{fig:vectoriel}.(a)).}

 We have $v\in P_1(u)$, so $v\in R_0(u)\cap R_2(u)$ and $u\in
 R_1^{\circ}(v)$ (as there is no edges oriented in two direction). By
 Lemma~\ref{lem:regionss}, we have $R_0(v)\subseteq R_0(u)$ and
 $R_2(v)\subseteq R_2(u)$ and $R_1(u)\subsetneq R_1(v)$. In fact the
 first two inclusions are strict as $u\notin R_0(v)\cup R_2(v)$. So by
 Lemma~\ref{lem:regionorder}, we have $v_0<u_0$, $v_2<u_2$,
 $u_1<v_1$. We can prove similar inequalities for the other pairs of
 vertices and we obtain $w_0<v_0<u_0$, $u_1<w_1<v_1$,
 $v_2<u_2<w_2$. By just studying the signs of the different terms
 occurring in the value of the coordinates of $\overrightarrow{X}$, it
 is clear that $\overrightarrow{X}$ as strictly positive
 coordinates. (For the first coordinates, it is easier if written in
 the following form $X_0=(u_1-w_1)(v_2-w_2)-(u_2-w_2)(v_1-w_1)$.)

  \noindent $\bullet$ \emph{Case 2: the edges of the face $F$ are in
    counterclockwise order $e_0(v)$, $e_2(v)$, $e_0(w)$.(see
    Figure~\ref{fig:vectoriel}.(b)).}

  As in the previous case, one can easily obtain the following
  inequalities: $w_0<v_0<u_0$, $u_1<w_1<v_1$, $u_2<v_2<w_2$ (the only
  difference with case $1$ is between $u_2$ and $v_2$).  Exactly like
  in the previous case, it is clear to see that $X_0$ and $X_2$ are
  strictly positive.  But there is no way to reformulate $X_1$ to have
  a similar proof.  Let $A=w_2-v_2$, $B=u_0-v_0$, $C=v_0-w_0$
  and $D=v_2-u_2$, so $X_1=AB-CD$ and $A,B,C,D$ are all strictly
  positive.

  Vertices $u,v,w$ are in the region $R(L_1,L'_1)$ for $L'_1$ a
  $1$-line consecutive to $L_1$. We consider two cases depending on
  equality or not between $L_1(u)$ and $L_1(v)$.

\noindent\emph{$\star$ Subcase 2.1: $L_1(u)=L_1(v)$.}

We have $X_1=A(B-D)+D(A-C)$.  

We have $B-D=(u_0+u_2)-(v_0+v_2)=(v_1-u_1)+(\sum u_i - \sum
v_i)$. Since $u\in P_0(v)$, we have $L_0(u)=L_0(v)$.  Suppose that
$L_2(u)=L_2(v)$, then by Lemma~\ref{lem:sum}, we have $\sum u_i = \sum
v_i$, and thus $B-D=v_1-u_1>0$. Suppose now that $L_2(u)\neq
L_2(v)$. By Lemmas~\ref{lem:sum} and~\ref{lem:boundedtriangle}, $\sum
u_i - \sum v_i>-t$.  By Lemma~\ref{lem:regionorder}, $v_1-u_1>
(N-n)|R(L_2(u),L_2(v))|\geq N-n$. So $B-D>N-n-t\geq 0$.

We have $A-C=(w_0+w_2)-(v_0+v_2)=(v_1-w_1)+(\sum w_i - \sum v_i)>\sum
w_i - \sum v_i$.  Suppose that $L_1(v)=L_1(w)$, then by
Lemma~\ref{lem:sum}, we have $\sum v_i = \sum w_i$ and thus
$A-C=v_1-w_1>0$. Then $X_1>0$.  Suppose now that $L_1(v)\neq L_1(w)$.
By Lemma~\ref{lem:shortregionedge}, $D=v_2-u_2<|R(L_1,L'_1)|$.  By
Lemma~\ref{lem:regionorder}, $A=w_2-v_2>(N-n)|R(L_1,L'_1)|$.  By
Lemmas~\ref{lem:sum} and~\ref{lem:boundedtriangle}, $\sum w_i - \sum
v_i> -t$, so $A-C>-t$.  Then $X_1>(N-n-t)|R(L_1,L'_1)|>0$.

\noindent\emph{$\star$ Subcase 2.2: $L_1(u)\neq L_1(v)$.}

We have $X_1=B(A-C)+C(B-D)$.  

Suppose that $L_1(w)\neq L_1(v)$. Then $L_1(w)= L_1(u)$. By
Lemma~\ref{lem:shortregionedge}, $e_0(w)$ is going counterclockwise
around the closed disk $D$ bounded by $\{e_0(w)\}\cup P_1(w)\cup
P_1(u)$. Then $v$ is inside $D$ and $P_1(v)$ has to intersect
$P_1(w)\cup P_1(u)$, so $L_1(v)=L_1(u)$, contradicting our
assumption. So $L_1(v)= L_1(w)$.

By Lemma~\ref{lem:regionorder}, $B=u_0-v_0>(N-n)|R(L_1,L'_1)|$.  We
have $A-C=(w_0+w_2)-(v_0+v_2)=(v_1-w_1)+(\sum w_i - \sum v_i)$. By
Lemma~\ref{lem:sum}, we have $\sum v_i = \sum w_i$ and thus
$A-C=v_1-w_1>0$.  By (the symmetric of)
Lemma~\ref{lem:shortregionedge}, $C=v_0-w_0<|R(L_1,L'_1)|$.  By
Lemmas~\ref{lem:sum} and~\ref{lem:boundedtriangle},
$B-D=(u_0+u_2)-(v_0+v_2)=(v_1-u_1)+(\sum u_i - \sum v_i)>-t$.  So
$X_1>(N-n-t)|R(L_1,L'_1)|>0$.
\end{proof}

Let $G$ be an essentially 3-connected toroidal map.  Consider a
periodic mapping of $G^{\infty}$ embedded graph $H$ (finite or
infinite) and a face $F$ of $H$. Denote $(f_1,f_2,\ldots,f_t)$ the
counterclockwise facial walk around $F$. Given a mapping of the
vertices of $H$ in $\mathbb{R}^2$, we say that $F$ is \emph{correctly
  oriented} if for any triplet $1\le i_1 < i_2 < i_3 \le t$, the
points $f_{i_1}$, $f_{i_2}$, and $f_{i_3}$ form a counterclockwise
triangle.  Note that a correctly oriented face is drawn as a convex
polygon.

\begin{lemma}
\label{th:faceorientation}
  Let $G$ be an essentially 3-connected toroidal map given with a
  periodic mapping of $G^{\infty}$ such that every face of
  $G^{\infty}$ is correctly oriented. This mapping gives a
  straight-line representation of $G^{\infty}$.
\end{lemma}

\begin{proof}
  We proceed by induction on the number of vertices $n$ of $G$. Note
  that the theorem holds for $n=1$, so we assume that $n>1$.  Given
  any vertex $v$ of $G$, let $(u_0,u_1,\ldots,u_{d-1})$ be the
  sequence of its neighbors in counterclockwise order (subscript
  understood modulo $d$). Every face being correctly oriented, for
  every $i\in [0,d-1]$ the oriented angle (oriented counterclockwise)
  $(\overrightarrow{vu_i},\overrightarrow{vu_{i+1}}) < \pi$.  Let the
  winding number $k_v$ of $v$ be the integer such that $2k_v\pi =
  \sum_{i\in[0,d-1]} (\overrightarrow{vu_i},\overrightarrow{vu_{i+1}})
  $. It is clear that $k_v \ge 1$. Let us prove that $k_v=1$ for every
  vertex $v$.
\begin{claim}
For any vertex $v$, its winding number $k_v=1$.
\end{claim}
\begin{proofclaim}
  In a parallelogram representation of $G$, we can sum up all the angles
  by grouping them around the vertices or around the faces.
  $$\sum_{v\in V(G)} \sum_{u_i\in N(v)}
(\overrightarrow{vu_i},\overrightarrow{vu_{i+1}})
=
\sum_{F\in F(G)} \sum_{f_i\in F} (\overrightarrow{f_if_{i-1}},\overrightarrow{f_if_{i+1}})$$
The face being correctly oriented, they form convex polygons. Thus the
angles of a face $F$ sum at $(|F|-2)\pi$.
$$\sum_{v\in V(G)} 2k_v\pi = \sum_{F\in F(G)} (|F|-2)\pi$$
$$\sum_{v\in V(G)} k_v =\frac{1}{2}\sum_{F\in F(G)} |F| -f$$
$$\sum_{v\in V(G)} k_v = m - f$$
So by Euler's formula $\sum_{v\in V(G)} k_v =n$, and thus $k_v=1$ for
every vertex $v$.
\end{proofclaim}

Let $v$ be a vertex of $G$ that minimizes the number of loops whose
ends are on $v$. Thus either $v$ has no incident loop, or every vertex
is incident to at least one loop.

Assume that $v$ has no incident loop.  Let $v'$ be any copy of $v$ in
$G^\infty$ and denote its neighbors $(u_0,u_1,\ldots,u_{d-1})$ in
counterclockwise order. As $k_v =1$, the points
$u_0,u_1,\ldots,u_{d-1}$ form a polygon $P$ containing the point $v'$
and the segments $[v',u_i]$ for any $i\in [0,d-1]$.  It is well known
that any polygon, admits a partition into triangles by adding some of
the chords. Let us call $O$ the  outerplanar graph with outer
boundary $(u_0,u_1,\ldots, u_{d-1})$, obtained by this
``triangulation'' of $P$. Let us now consider the toroidal map $G' =
(G\setminus \{v\})\cup O$ and its periodic embedding obtained from the
mapping of $G^\infty$ by removing the copies of $v$.  It is easy to
see that in this embedding every face of $G'$ is correctly oriented
(including the inner faces of $O$, or the faces of $G$ that have been
shortened by an edge $u_iu_{i+1}$). Thus by induction hypothesis, the
mapping gives a straight-line representation of $G'^\infty$.  It is
also a straight-line representation of $G^\infty$ minus the copies of
$v$ where the interior of each copy of the polygons $P$ are pairwise
disjoint and do not intersect any vertex or edge. Thus one can add the
copies of $v$ on their initial positions and add the edges with their
neighbors without intersecting any edge. The obtained drawing is thus
a straight-line representation of $G^\infty$.

Assume now that every vertex is incident to at least one loop.  Since
these loops are non-contractible and do not cross each other, they
form homothetic cycles.  Thus $G$ is as depicted in
Figure~\ref{fig:polygonG}, where the dotted segments stand for edges
that may be in $G$ but not necessarily.  Since the mapping is periodic
the edges corresponding to loops of $G$ form several parallel lines,
cutting the plane into infinite strips.  Since for any $1\leq i\leq
n$, $k_{v_i}=1$, a line of copies of $v_i$ divides the plane, in such
a way that their neighbors which are copies of $v_{i-1}$ and their
neighbors which are copies of $v_{i+1}$ are in distinct
half-planes. Thus adjacent copies of $v_i$ and $v_{i+1}$ are on two
lines bounding a strip. Then one can see that the edges between copies
of $v_i$ and $v_{i+1}$ are contained in this strip without
intersecting each other.  Thus the obtained mapping of $G^\infty$ is a
straight-line representation.
\end{proof}

\begin{figure}[!h]
\begin{center}
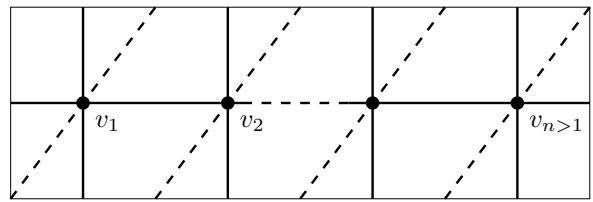
\caption{The graph $G$ if every vertex is incident to a loop.}
\label{fig:polygonG}
\end{center}
\end{figure}

A plane is \emph{positive} if it has equation $\alpha x +\beta y +
\gamma z=0$ with $\alpha,\beta,\gamma \geq 0$.

\begin{theorem}
\label{th:straightline}
If $G$ is a toroidal triangulation given with a crossing Schnyder wood, and
$V^\infty$ the set of region vectors of vertices of $G^\infty$. Then
the projection of $V^\infty$ on a positive plane gives a straight-line
representation of $G^\infty$.
\end{theorem}

\begin{proof}
  Let $\alpha,\beta,\gamma \geq 0$ and consider the projection of
  $V^\infty$ on the plane $P$ of equation $\alpha x +\beta y + \gamma
  z=0$. A normal vector of the plane is given by the vector
  $\overrightarrow{n}=(\alpha,\beta,\gamma)$. Consider a face $F$ of
  $G^\infty$. Suppose that $F$ is incident to vertices $u,v,w$ (given
  in counterclockwise order around $F$).  By
  Lemma~\ref{lem:vectoriel}, $(\overrightarrow{uv}\wedge
  \overrightarrow{uw}).\overrightarrow{n}$ is positive. Thus the
  projection of the face $F$ on $P$ is correctly oriented. So by
  Lemma~\ref{th:faceorientation}, the projection of $V^\infty$ on
  $P$ gives a straight-line representation of $G^\infty$.
\end{proof}

Theorems~\ref{th:cross-tri} and~\ref{th:straightline} implies the following:

\begin{theorem}
\label{cor:straightline}
A toroidal map admits a straight-line flat torus representation in a
polynomial size grid.
\end{theorem}

Indeed, any toroidal map $G$ can be
transformed into a toroidal triangulation $G'$ by adding a linear
number of vertices and edges and such that $G'$ is simple if and only
if $G$ is simple (see for example the proof of Lemma~2.3
of~\cite{Moh96}).  Then by Theorem~\ref{th:cross-tri},
$G'$ admits a crossing Schnyder wood.
By Theorem~\ref{th:straightline}, the projection of the set of region
vectors of vertices of $G'^\infty$ on a positive plane gives a
straight-line representation of $G'^\infty$.  The grid where the
representation is obtained can be the triangular grid, if the
projection is done on the plane of equation $x+y+z=0$, or the square
grid, if the projection is done on one of the plane of equation $x=0$,
$y=0$ or $z=0$.
By Lemma~\ref{th:gamma}, and the choice of $N$, the obtained mapping
is a periodic mapping of $G^\infty$ with respect to non collinear
vectors $Y$ and $Y'$ where the size of these vectors is in $\mathcal
O(\omega^2 n^2)$ with
$\omega\leq n$ in general and $\omega=1$ if the graph is simple and
the Schnyder wood obtained by Theorem~\ref{th:schnydersimple}.
By Lemma~\ref{lem:edgesbounded}, the
length of the edges in this representation are in $\mathcal O(n^3)$ in
general and in $\mathcal O(n^2)$ if the graph is simple.  When the
graph is not simple, there is a non-contractible cycle of length $1$
or $2$ and thus the size of one of the two vectors $Y$, $Y'$ is in
$\mathcal O(n^3)$. Thus the grid obtained in
Theorem~\ref{cor:straightline} has size in $\mathcal O(n^3)\times
\mathcal O(n^4)$ in general and $\mathcal O(n^2)\times\mathcal O(n^2)$
if the graph is simple.

The method presented here gives a polynomial algorithm to obtain 
straight-line representation of any toroidal maps in polynomial size
grids. Indeed, all the proofs leads to polynomial
algorithms, even the proof of Theorem~\ref{th:fij}~\cite{Fij} which  uses
results from Robertson and Seymour~\cite{RS86} on disjoint paths
problems.

It would be nice to extend Theorem~\ref{th:straightline}
to obtain convex straight-line representation for essentially
3-connected toroidal maps.

Note that the results presented in this chapter
motivated~Castelli~Aleardi, Devillers and~Fusy~\cite{CF12} to develop
direct methods to obtain straight line representations for toroidal
maps. They manage to generalize planar canonical ordering to the
cylinder to obtain straight-line representation of toroidal
triangulations in grids of size
$\mathcal O(n)\times \mathcal O(n^\frac{3}{2})$.

\chapter*{Conclusion}
\addcontentsline{toc}{part}{Conclusion}
\markboth{}{}

In this manuscript we propose a generalization of Schnyder woods to
higher genus oriented surfaces.  For this purpose we keep the
simplicity of the local definition of Schnyder woods in the plane and
extend it to surfaces. A drawback is the loss of the main global
property of the planar case : the partition into three trees (see
Figures~\ref{fig:ex-tree} and~\ref{fig:3c-tree}).  Nevertheless in the
toroidal case we are able to define some other global properties that
a Schnyder wood may have or may not have (crossing, balanced,
minimality, etc.) that enable to understand the structure of the
manipulated objects and use it for nice applications like graph
drawing or optimal encoding.  Many questions were raised all along
the manuscript concerning toroidal Schnyder woods, let us recall one
of the simplest property that one may ask: does a simple toroidal
triangulation admit a Schnyder wood where each of the three colors
induces a connected subgraph?  Among the numerous applications of
planar Schnyder woods that are still to be generalized to the torus,
we hope that the obtained bijection can be used to generate random
toroidal triangulations with uniform distribution.

In the toroidal case ($g=1$), Euler's formula ensures that every vertex
plays exactly the same role. When going to higher genus ($g\geq 2$),
some strange things happen with the need of special vertices
satisfying ``several times'' the local property (see
Figure~\ref{fig:369}). The lattice structure of homologous
orientations is a general result that is valid in this context
but we are not yet able to find a proof of existence of the considered
objects. We only conjecture that, when $g\geq2$, a map admits a
Schnyder wood if and only if it is essentially 3-connected. Also we
have no guess of what would be some global interesting property in
this case. Indeed, generalizing the global properties that we found
for the torus is not straightforward since one has to deal with the
special vertices. For example the definition of $\gamma$ is not
anymore ``stable by homology'', i.e. if two non-contractible cycles
$C,C'$ are homologous, we may have $\gamma(C)\neq\gamma(C')$ whereas
we have equality on the torus by Lemma~\ref{lem:gammahomologyequalbasis}.

Originally Schnyder woods were defined for planar triangulations and
then, by allowing edges to be oriented in two directions, they have
been generalized to (internally) 3-connected planar maps. Thus this
manuscript is written for the general situation of (essentially)
3-connected maps except when we prove results that are specific to
triangulations.  We may have adopted a different point of view where
instead of generalizing toward 3-connected maps we would have
generalized toward $d$-angulations.  

Indeed, there are many kinds of planar maps that admit orientations
that are similar to Schnyder woods. For example planar 4-connected
triangulations admit so-called transversal structures~\cite{Fus09},
quadrangulations admit 2-orientations~\cite{BH10}, $d$-angulations of
girth $d$ admit $\frac{d}{d-2}$-orientations~\cite{BF12}, even-degree
maps admit Eulerian orientations~\cite{Sch97}, corner triangulations
admits partition into $3$ bipolar orientations~\cite{EM14}, etc. Among
these different kind of orientations, Schnyder woods (for
triangulations) and 2-orientations (for quadrangulations) can be seen
has particular cases of $\frac{d}{d-2}$-orientations (for
$d$-angulations of girth $d$).  We do not see a common way to present
these different orientations in a unified framework so we choose to
present this manuscript on Schnyder woods only as this is certainly
one of the most famous objects of this kind (and also the one that we
started to look at).

Nevertheless it seems that several structural
properties that have been presented in this manuscript for Schnyder
woods in higher genus are still valid in the context of
$\frac{d}{d-2}$-orientations. 
Figure~\ref{fig:d-angul-example} gives examples of toroidal
$d$-angulations for $d=3,4,5,6$ given with
\emph{$\frac{d}{d-2}$-orientations}, i.e. assignments of positive
integers, on the half-edges of the map $G$ such that, for every edge,
the weights of its two half-edges sum to $d-2$, and for every vertex,
the weights of its incident half-edges sum to $d$.  Such a fractional
orientation of $G$ can be seen as a classical orientation in the
\emph{$(d-2)$-multigraph} $H$ obtained from $G$ by replacing every
edge by $d-2$ parallel edges.  Indeed having a
$\frac{d}{d-2}$-orientation of $G$ is equivalent of having an
orientation of the edges of $H$ (in one direction only) where every
vertex has outdegree exactly $d$.  A notion (and existence !) of
balanced orientation can then be defined similarly as in this
manuscript. Then the lattice structure in the multigraph naturally
define a minimal element that can be used for bijection purpose.  A
similar story can be told for transversal structures (see
Figure~\ref{fig:tts}) by looking at the angle graph to define the
balanced property (and also maybe for other kind of orientations?) but
we are at a point where this manuscript should stop and of course
research continues ...

\begin{figure}[!h]
\center

\begin{tabular}{cccc}
\includegraphics[scale=0.34]{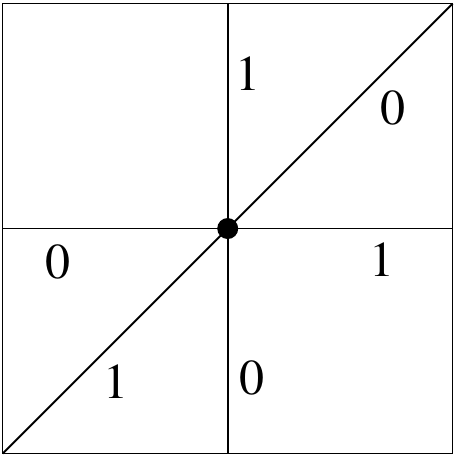}  \ & \
\includegraphics[scale=0.34]{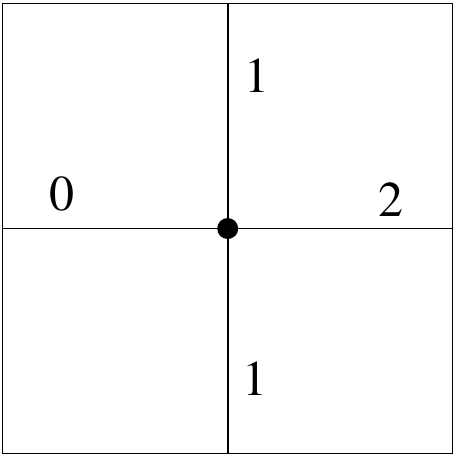}  \ & \
\includegraphics[scale=0.34]{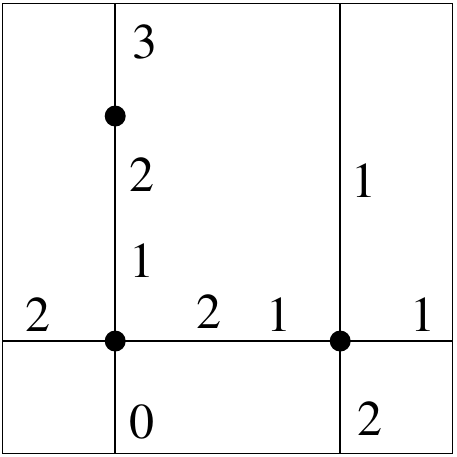}  \ & \
\includegraphics[scale=0.34]{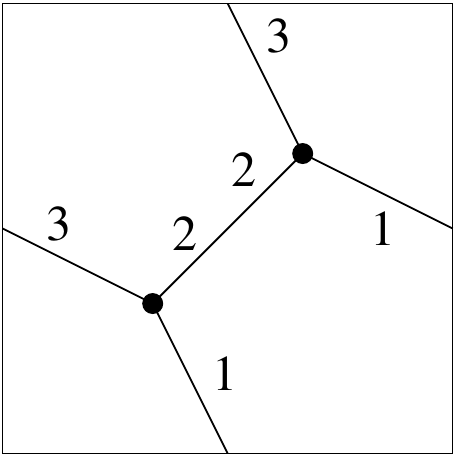}  \\
\\
\includegraphics[scale=0.34]{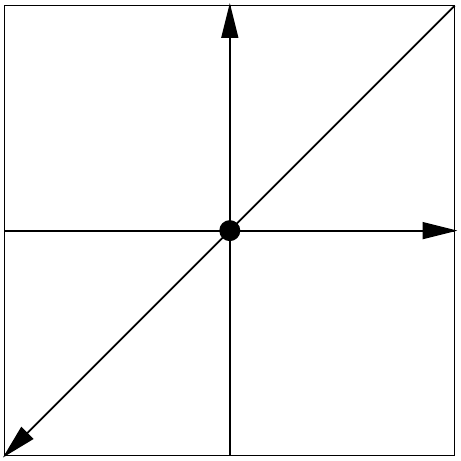}  \ & \
\includegraphics[scale=0.34]{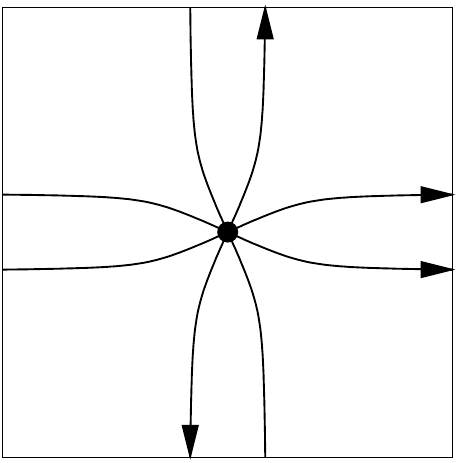}  \ & \
\includegraphics[scale=0.34]{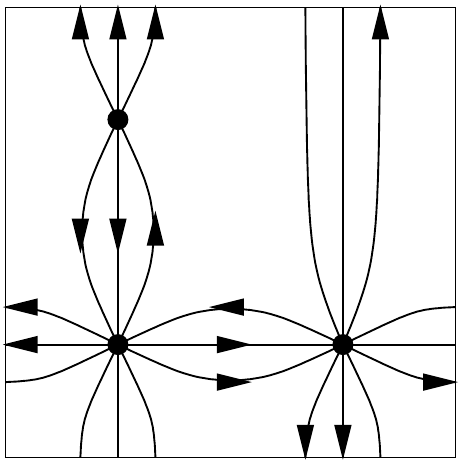}  \ & \
\includegraphics[scale=0.34]{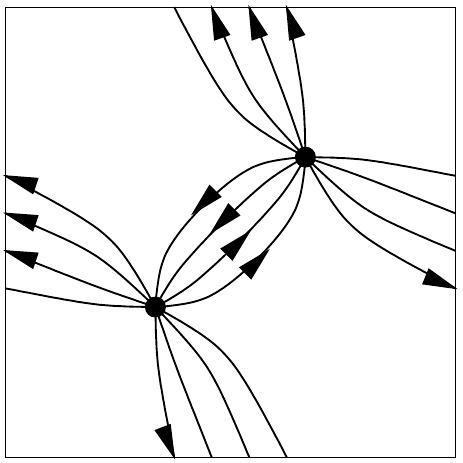}  \\
\\
 $d=3$ \ & \  $d=4$ \ & \  $d=5$ \ & \  $d=6$ \\
\end{tabular}
\caption{Example of $\frac{d}{d-2}$-orientations of toroidal maps for
  $d=3,4,5,6$.}
\label{fig:d-angul-example}
\end{figure}

\begin{figure}[!h]
\center
\includegraphics[scale=0.5]{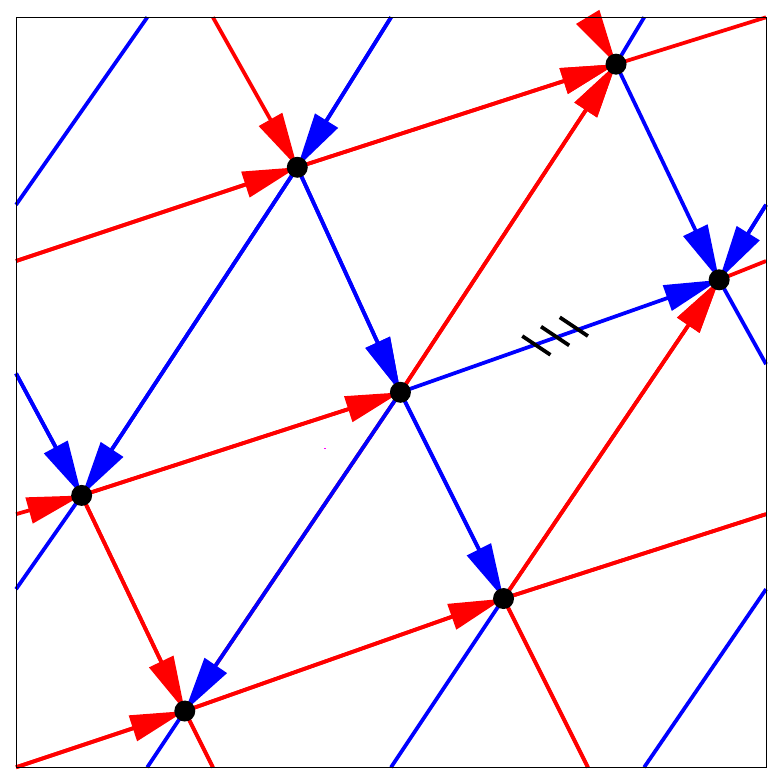}
\caption{The minimal balanced transversal structure of $K_7$ w.r.t.~to the
  rooted (dashed) edge.}
\label{fig:tts}
\end{figure}


\begin{thebibliography}{00}


\bibitem{AGK14} B.~Albar, D.~Gon\c{c}alves, K.~Knauer, Orienting
  triangulations, \emph{Journal of Graph Theory} 83 (2016) 392-405.

\bibitem{AP13} M. Albenque, D. Poulalhon, Generic method for
  bijections between blossoming trees and planar maps, 
\emph{Electronic Journal of Combinatorics} 22 (2015) P2.38. 


\bibitem{BT06} J.~Bar\'at, C.~Thomassen, Claw-decompositions and
  tutte-orientations, \emph{Journal of Graph Theory} 52 (2006) 135-146.

\bibitem{Bar12} J.~Barbay, L.~Castelli Aleardi, M.~He, J.~I.~Munro,
  Succinct representation of labeled graphs, \emph{Algorithmica} 61 (2012)
  224-257.

\bibitem{BH10} L. Barrière, C. Huemer, 4-labelings and grid embeddings
  of plane quadrangulations, GD 2010, \emph{Lecture Notes in Computer
    Science} 5849 (2010) 413-414.

\bibitem{Ber07} O. Bernardi, Bijective counting of tree-rooted maps
  and shuffles of parenthesis systems, \emph{Electronic Journal of
  Combinatorics} 14 (2007) R9.

\bibitem{BC11} O. Bernardi, G. Chapuy
A bijection for covered maps, or a shortcut between Harer-Zagier's and
Jackson's formulas, \emph{Journal of Combinatorial Theory A} 118 (2011) 1718-1748.

\bibitem{BF12} O. Bernardi, E. Fusy, A bijection for triangulations, quadrangulations, pentagulations, etc.,
\emph{Journal of Combinatorial Theory A} 119 (2012) 218-244.

\bibitem{BGH03} N. Bonichon, C. Gavoille, N. Hanusse, An
  information-theoretic upper bound of planar graphs using
  triangulation, STACS 2003, \emph{Lecture Notes in Computer Science}
  2607 (2003) 499-510.

\bibitem{Bon05} N.~Bonichon, A bijection between realizers of maximal
  plane graphs and pairs of non-crossing dyck paths, \emph{Discrete
    Mathematics} 298 (2005) 104-114.


\bibitem{BFM07} N. Bonichon, S. Felsner, M. Mosbah, Convex drawings of
  3-connected plane graphs, {\it Algorithmica} 47 (2007) 399-420.

\bibitem{BGHI10} N. Bonichon, C. Gavoille, N. Hanusse, D. Ilcinkas,
  Connections between theta-graphs, delaunay triangulations, and
  orthogonal surfaces, WG 2010, \emph{Lecture Notes in Computer
    Science} 6410 (2010) 266-278.


\bibitem{CFL09} L.~Castelli Aleardi, E.~Fusy, T.~Lewiner, Schnyder
  woods for higher genus triangulated surfaces, with applications to
  encoding, \emph{Discrete and Computational Geometry} 42 (2009)
  489-516.



\bibitem{CFL10} L. Castelli Aleardi, E. Fusy, T. Lewiner, Optimal
  encoding of triangular and quadrangular meshes with fixed topology,
    CCCG 2010.

  \bibitem{CF12} L.~Castelli Aleardi, O.~Devillers, E.~Fusy, Canonical
    ordering for triangulations on the cylinder, with applications to
    periodic straight-line drawings, GD 2012, \emph{Lecture Notes in
      Computer Science} 7704 (2012) 376-387.


\bibitem{CEG11} E.~Chambers, D.~Eppstein, M.~Goodrich, M.~L\"offler,
  Drawing graphs in the plane with a prescribed outer face and
  polynomial area, GD 2010, \emph{Lecture Notes in Computer Science} 6502
  (2011) 129-140.

\bibitem{DGL15}
  V. Despr\'e, D. Gon\c{c}alves, B. L\'ev\^eque,  Encoding toroidal triangulations,
   \emph{Discrete and Computational Geometry}
 (2016) 1-38.

\bibitem{DPS13} E. Duchi, D. Poulalhon, G. Schaeffer, Uniform random
  sampling of simple branched coverings of the sphere by itself, SODA
  2014 294-304.

\bibitem{DGK11} C.~Duncan, M.~Goodrich, S.~Kobourov, Planar drawings
  of higher-genus graphs, \emph{Journal of Graph Algorithms and
    Applications} 15 (2011) 13-32.


\bibitem{EM14} D. Eppstein, E. Mumford, Steinitz theorems for simple
  orthogonal polyhedra, \emph{Journal of Computational Geometry} 5
  (2014) 179-244.

\bibitem{Fel01} S. Felsner, Convex drawings of planar graphs and the
  order dimension of 3-polytopes, {\it Order} 18 (2001) 19-37.

\bibitem{Fel03} S. Felsner, Geodesic embeddings and planar graphs,
  {\it Order} 20 (2003) 135-150.

\bibitem{Fel04} S. Felsner, Lattice structures from planar graphs,
  {\it  Electronic Journal of Combinatorics} 11 (2004) R15.

\bibitem{Fel-book} S. Felsner, \emph{Geometric Graphs and
    Arrangements}, Vieweg, 2004.


\bibitem{FT05} S.~Felsner, W.~T.~Trotter, Posets and planar graphs,
  \emph{Journal  Graph Theory} 49 (2005) 273-284.


\bibitem{FZ08} S. Felsner, F. Zickfeld, Schnyder woods and orthogonal
  surfaces, \emph{Discrete and Computational Geometry} 40 (2008) 103-126.


\bibitem{Fel09} S. Felsner, K.~Knauer, ULD-lattices and
  $\Delta $-bonds, {\it Combinatorics, Probability and Computing} 18 
  (2009) 707-724.


\bibitem{Fij} G. Fijavz, personal communication, 2011.



\bibitem{FOR94} H. de Fraysseix, P. Ossona de Mendez, P. Rosenstiehl,
  On triangle contact graphs, {\it Combinatorics, Probability and
    Computing} 3 (1994) 233-246.


\bibitem{FO01} H. de Fraysseix, P. Ossona de Mendez, On topological
  aspects of orientations, {\it Discrete Mathematics} 229 (2001)
  57-72.



\bibitem{Fus09} E. Fusy, Transversal structures on triangulations: a
  combinatorial study and straight-line drawings, \emph{Discrete
    Mathematics} 309 (2009) 1870-1894.




\bibitem{Gib10} P.~Giblin, {\it Graphs, surfaces and homology},
  Cambridge University Press, Cambridge, third edition, 2010.


\bibitem{GLP11} D. Gon\c{c}alves, B. L\'ev\^eque, A. Pinlou, Triangle
  contact representations and duality, \emph{Discrete and
    Computational Geometry} 48 (2012) 239-254.


\bibitem{GLP11b} D. Gon\c{c}alves, B. L\'ev\^eque, A. Pinlou, Triangle
  contact representations and duality, GD 2010, \emph{Lecture Notes in
    Computer Science} 6502 (2011) 262-273.

\bibitem{GL13} D.~Gon\c{c}alves, B.~L\'ev\^eque, Toroidal maps:
  Schnyder woods, orthogonal surfaces and straight-line
  representations, \emph{Discrete and Computational Geometry} 51
  (2014) 67-131.


   \bibitem{GKL15} D.~Gon\c{c}alves, K.~Knauer, B.~L\'ev\^eque,
     On the structure of Schnyder woods on orientable surfaces,
     manuscript, 2015. arXiv:1501.05475




\bibitem{Kan96} G.~Kant, Drawing planar graphs using the canonical
  ordering, \emph{Algorithmica} 16 (1996) 4-32.



\bibitem{Kra07} J. Kratochv\'il, Bertinoro Workshop on Graph Drawing,
  2007.

\bibitem{Massey} W. S. Massey, \emph{Algebraic Topology: An
    Introduction}, Harcourt, Brace and World, New York, 1967



\bibitem{Meu} F. Meunier, personal communication, 2015.





\bibitem{Mil02} E. Miller, Planar graphs as minimal resolutions of
  trivariate monomial ideals, {\it Documenta Mathematica} 7 (2002)
  43-90.


\bibitem{Moh96} B. Mohar, Straight-line representations of maps on the
  torus and other flat surfaces, \emph{Discrete Mathematics} 155
  (1996) 173-181.



\bibitem{MR98} B. Mohar, P. Rosenstiehl, Tessellation and visibility
  representations of maps on the torus, \emph{Discrete and
    Computational Geometry} 19 (1998) 249-263.
    
\bibitem{Oss94} P.~Ossona de Mendez, \emph{Orientations bipolaires},
  PhD Thesis, manuscript, 1994.

    
\bibitem{PS06} D.~Poulalhon, G.~Schaeffer, Optimal coding and sampling
  of triangulations, \emph{Algorithmica} 46 (2006) 505-527.

\bibitem{Pro93} J. Propp, Lattice structure for orientations of
  graphs, manuscript, 1993. arXiv:math/0209005


\bibitem{RS86} N.~Robertson, P.D~Seymour, Graph minors. VI. Disjoint
  paths across a disc,\emph{Journal of Combinatorial Theory B} 41
  (1986) 115-138.

\bibitem{Rose89} P. Rosenstiehl, Embedding in the plane with
  orientation constraints: The angle graph, \emph{Annals of New York Academy of
  Sciences} 555 (1989) 340-346.

\bibitem{Sch97} G. Schaeffer, Bijective census and random generation
  of eulerian planar maps with prescribed vertex degrees,
  \emph{Electronics Journal of Combinnatorics} 4 (1997) R20.

\bibitem{Sch89} W. Schnyder, Planar graphs and poset dimension, {\it
    Order} 5 (1989) 323-343.
  
\bibitem{Sch90} W.~Schnyder, Embedding planar graphs on the grid, SODA
  1990 138-148.



\bibitem{UecPHD} T. Ueckerdt, Geometric representations of graphs with low
polygonal complexity, PhD thesis, manuscript, 2011.


\end{thebibliography}
\end{document}